\documentclass[10pt]{article}
\usepackage{graphicx}    % needed for including graphics e.g. EPS, PS
\usepackage[export]{adjustbox}
\usepackage{amsmath}
\usepackage{amsfonts}
\usepackage{hyperref}
\usepackage{subcaption}
\usepackage{wrapfig}
\usepackage{overpic}
\usepackage{amssymb}
\usepackage{amsthm}
\usepackage{mathtools}
\usepackage{latexsym}
\usepackage{enumerate}
\usepackage{colortbl}
\usepackage[space]{grffile}
\usepackage{thmtools,thm-restate}
\usepackage[margin=1in]{geometry}
\newtheorem{theorem}{Theorem}[section]
\newtheorem{lemma}[theorem]{Lemma}

\newtheorem{corollary}[theorem]{Corollary}
\newtheorem{observation}[theorem]{Observation}
\newtheorem{problem}{Problem}
\newtheorem{property}{Property}

\usepackage{footnote}
\makesavenoteenv{tabular}
\makesavenoteenv{table}

\theoremstyle{definition}
\newtheorem{definition}{Definition}

\def\emph#1{\textbf{\textit{\boldmath #1}}}

\let\realbfseries=\bfseries
\def\bfseries{\realbfseries\boldmath}

{\makeatletter
 \gdef\xxxmark{%
   \expandafter\ifx\csname @mpargs\endcsname\relax % in minipage?
     \expandafter\ifx\csname @captype\endcsname\relax % in figure/caption?
       \marginpar{xxx}% not in a caption or minipage, can use marginpar
     \else
       xxx % notice trailing space
     \fi
   \else
     xxx % notice trailing space
   \fi}
 \gdef\xxx{\@ifnextchar[\xxx@lab\xxx@nolab}
 \long\gdef\xxx@lab[#1]#2{\textbf{[\xxxmark #2 ---{\sc #1}]}}
 \long\gdef\xxx@nolab#1{\textbf{[\xxxmark #1]}}
 \long\gdef\xxx@lab[#1]#2{}\long\gdef\xxx@nolab#1{}%
}

\newif\ifabstract
\abstracttrue
\abstractfalse
\newif\iffull
\ifabstract \fullfalse \else \fulltrue \fi

\newcounter{section-preserve}
\newcounter{theorem-preserve}
\newcommand{\blank}[1]{}
\newtoks\magicAppendix
\magicAppendix={}
\newtoks\magictoks
\newif\iflater
\laterfalse
\long\def\later#1{\magictoks={#1}%
  \edef\magictodo{\noexpand\magicAppendix={\the\magicAppendix \par
    \the\magictoks%
  }}
  \magictodo}
\long\def\both#1{\magictoks={#1}%
  \edef\magictodo{\noexpand\magicAppendix={\the\magicAppendix \par
    \noexpand\setcounter{theorem-preserve}{\noexpand\arabic{theorem}}%
    \noexpand\setcounter{theorem}{\arabic{theorem}}%
    \noexpand\setcounter{section-preserve}{\noexpand\arabic{section}}%
    \noexpand\setcounter{section}{\arabic{section}}%
    \noexpand\let\noexpand\oldsection=\noexpand\thesection
    \noexpand\def\noexpand\thesection{\thesection}
    \noexpand\let\noexpand\oldlabel=\noexpand\label
    \noexpand\let\noexpand\label=\noexpand\blank
    \the\magictoks%
    \noexpand\setcounter{theorem}{\noexpand\arabic{theorem-preserve}}%
    \noexpand\setcounter{section}{\noexpand\arabic{section-preserve}}%
    \noexpand\let\noexpand\thesection=\noexpand\oldsection
    \noexpand\let\noexpand\label=\noexpand\oldlabel
  }}
  \magictodo
  \the\magictoks}
\def\magicappendix{\latertrue \the\magicAppendix}

\def\abstractlater#1{\ifabstract\later{#1}\fi}

\iffull
  \long\def\both#1{#1}
  \let\later=\both
  \def\magicappendix{}
\fi

\newenvironment{proofsketch}{\begin{proof}[Proof sketch]}{\end{proof}}

\newcommand{\start}{\raisebox{-0.25ex}{\includegraphics[height=2ex]{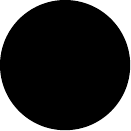}}}
\newcommand{\antibody}{\includegraphics[height=2ex]{figures/antibody}}
\newcommand{\hexagon}{\includegraphics[height=2ex]{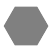}}

\newcommand{\broken}{\raisebox{0.75ex}{\includegraphics[height=0.5ex]{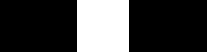}}}
\renewcommand{\square}{\includegraphics[height=2ex]{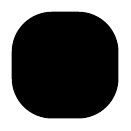}}
\renewcommand{\star}{\includegraphics[height=2ex]{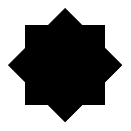}}
\renewcommand{\triangle}[1]{\ifcase#1%
  \or\includegraphics[height=1.75ex]{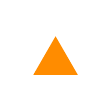}
  \or\includegraphics[height=1.75ex]{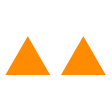}
  \or\includegraphics[height=1.75ex]{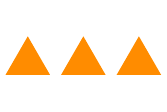}
  \fi}
\newcommand{\tetris}{\includegraphics[height=2ex,trim=10 10 10 10]{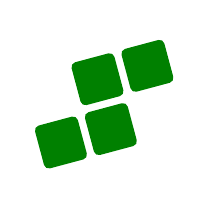}}
\newcommand{\tetrisfix}{\includegraphics[height=2ex,trim=10 10 10 10]{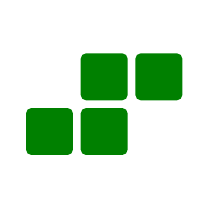}}
\newcommand{\antitetris}{\includegraphics[height=2ex,trim=10 10 10 10]{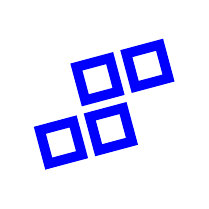}}
\newcommand{\dominovert}{\includegraphics[height=1.5ex,trim=15 15 15 15]{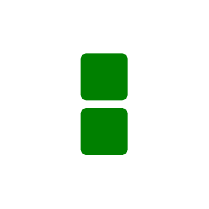}}
\newcommand{\dominofree}{\includegraphics[height=1.5ex,trim=15 15 15 15]{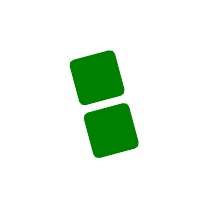}}
\newcommand{\monomino}{\includegraphics[height=1.5ex,trim=20 20 20 20]{figures/monomino}}
\newcommand{\antimonomino}{\includegraphics[height=1.5ex,trim=20 20 20 20]{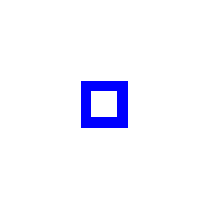}}

\title{\textbf{Who witnesses The Witness?} \\[0.5ex]
  \Large Finding witnesses in The Witness is hard and sometimes impossible%
  \thanks{A preliminary version of this paper appeared at the
    \textit{9th International Conference on Fun with Algorithms}, 2018~\cite{DBLP:conf/fun/AbelBDHHKLR18}.}
}
\author{
Zachary Abel\thanks{MIT EECS Department, 50 Vassar St.,
  Cambridge, MA 02139, USA, \protect\url{zabel@mit.edu}} \and
Jeffrey Bosboom\thanks{MIT Computer Science and Artificial Intelligence
  Laboratory, 32 Vassar Street, Cambridge, MA 02139, USA,
  \protect\url{{jbosboom,mcoulomb,edemaine,jkopin,jaysonl}@mit.edu},
  \protect\url{mrudoy@gmail.com}} \and
Michael Coulombe\footnotemark[3]\and
Erik D. Demaine\footnotemark[3]\and
Linus Hamilton\thanks{MIT Mathematics Department, 77 Massachusetts Avenue,
  Cambridge, MA 02139, USA, \protect\url{{luh,achester}@mit.edu}} \and
Adam Hesterberg\footnotemark[4]
  \thanks{Now at Google Inc.} \and
Justin Kopinsky\footnotemark[3] \and
Jayson Lynch\footnotemark[3] \and
Mikhail Rudoy\footnotemark[3] \footnotemark[5] \and
Clemens Thielen\thanks{Department of Mathematics,
  University of Kaiserslautern, Paul-Ehrlich-Str.~14,
  D-67663 Kaiserslautern, Germany, \protect\url{thielen@mathematik.uni-kl.de}}
}

\date{}

\begin{document}         
\maketitle

\vspace*{-4ex}
%\vspace*{-2ex}
\begin{abstract}
  We analyze the computational complexity of the many types of
  pencil-and-paper-style puzzles featured in the 2016 puzzle video game
  \emph{The Witness}.  In all puzzles, the goal is to draw a simple path in a
  rectangular grid graph from a start vertex to a destination vertex.
  The different puzzle types place different constraints on the path:
  preventing some edges from being visited (broken edges);
  forcing some edges or vertices to be visited (hexagons);
  forcing some cells to have certain numbers of incident path edges
  (triangles); or
  forcing the regions formed by the path to be
  partially monochromatic (squares),
  have exactly two special cells (stars), or
  be singly covered by given shapes (polyominoes) and/or
  negatively counting shapes (antipolyominoes).
  We show that any \textit{one} of these clue types (except the first)
  is enough to make path finding NP-complete
  (``witnesses exist but are hard to find''),
  even for rectangular boards.
  Furthermore, we show that a final clue type (antibody), which necessarily
  ``cancels'' the effect of another clue in the same region, makes path finding
  $\Sigma_2$-complete (``witnesses do not exist''), even with a single
  antibody (combined with many anti/polyominoes), and the problem
  gets no harder with many antibodies.
  On the positive side, we give a polynomial-time algorithm for monomino clues,
  by reducing to hexagon clues on the boundary of the puzzle,
  even in the presence of broken edges,
  and solving ``subset Hamiltonian path'' for terminals on the boundary of an
  embedded planar graph in polynomial time.
\end{abstract}
\vspace*{-2ex}

%\thispagestyle{empty}

%\vspace*{-2ex}

\definecolor{hard}{rgb}{1,0.85,0.85}
\definecolor{open}{rgb}{1,1,0.85}
\definecolor{easy}{rgb}{0.85,0.85,1}
\definecolor{header}{rgb}{0.85,0.85,0.85}

\arrayrulewidth=0.75pt
\begin{table}[t]
  \centering
  \vspace*{-3ex}
  \def\FONT{\footnotesize}
  \def\ALLOW{\checkmark}
  \def\OPEN{\emph{OPEN}}
  \def\STACK#1#2{$\vcenter{\hbox{#1}\hbox{#2}}$}
  \tabcolsep=3pt
  \begin{tabular}{|cccccccc||c|>{\footnotesize}c|}
    \hline
    \rowcolor{header}
    \multicolumn{1}{|c|}{\FONT broken edge} & \multicolumn{1}{c|}{\FONT hexagon} & \multicolumn{1}{c|}{\FONT square} & \multicolumn{1}{c|}{\FONT star} & \multicolumn{1}{c|}{\FONT triangle} & \multicolumn{1}{c|}{\FONT polyomino} & \multicolumn{1}{c|}{\FONT antipolyomino} & \multicolumn{1}{c||}{\FONT antibody} & \multicolumn{1}{c|}{} & %\raisebox{-1ex}{thm}
    \\
    \rowcolor{header}
    \multicolumn{1}{|c|}{\broken} & \multicolumn{1}{c|}{\hexagon} & \multicolumn{1}{c|}{\square} & \multicolumn{1}{c|}{\star} & \multicolumn{1}{c|}{\triangle2} & \multicolumn{1}{c|}{\tetris} & \multicolumn{1}{c|}{\antitetris} & \multicolumn{1}{c||}{\antibody} & \multicolumn{1}{c|}{\raisebox{2ex}{complexity}} & \multicolumn{1}{c|}{\normalsize\raisebox{2ex}{ref}}
    \\
    \hline
    \hline
    % Just broken edges is in L (trivial)
    \rowcolor{easy}
    \ALLOW &&&&&&&& $\in{}$L & Obs~\ref{thm:broken-edges}
    \\
    % Vertex hexagons + broken edges is NP-complete (trivial)
    \rowcolor{hard}
    \ALLOW & \ALLOW \FONT vertices &&&&&&& NP-complete & Obs~\ref{thm:broken edges+vertex hexagons}
    \\
    % Vertex hexagons is OPEN
    \rowcolor{open}
    & \ALLOW \FONT vertices &&&&&&& \OPEN & Prob~\ref{open:vertex-hexagons}
    \\
    % Edge hexagons is NP-complete
    \rowcolor{hard}
    & \ALLOW \FONT edges &&&&&&& NP-complete & Thm~\ref{thm:edge hexagons}
    \\
    % Hexagons on boundary is polynomial
    \rowcolor{easy}
    \ALLOW & \multicolumn{3}{l}{\ALLOW \FONT on boundary} &&&&& $\in{}$P & Thm~\ref{thm:monominoes-dp}
    \\
    \hline
    % Squares of one color is in P (trivial)
    \rowcolor{easy}
    && \ALLOW \FONT 1 color &&&&&& always \textsc{yes} & Obs~\ref{obs:1 square}
    \\
    % Squares of two colors is NP-complete
    \rowcolor{hard}
    && \ALLOW \FONT 2 colors &&&&&& NP-complete & \clap{\STACK{Thm~\ref{thm:squares 2 color}}{\!\!\cite{eurocg}\!}}
    \\
    \hline
    % Stars of 1 color is OPEN
    \rowcolor{open}
    &&& \ALLOW \FONT 1 color &&&&& \OPEN & Prob~\ref{open:1 star}
    \\
    % Stars of many colors is NP-complete
    \rowcolor{hard}
    &&& \ALLOW \FONT $n$ colors &&&&& NP-complete & Thm~\ref{thm:stars}
    \\
    \hline
    %% k-triangles + broken edges is NP-complete
    \rowcolor{hard}
    \ALLOW &&&& \ALLOW \FONT any &&&& NP-complete & \!\cite{slitherlink}\!
    \\
    % 1-triangles is NP-complete
    \rowcolor{hard}
    &&&& \ALLOW \triangle1 &&&& NP-complete & Thm~\ref{thm:triangles1}
    \\
    %% 2-triangles + broken edges is NP-complete
    %\rowcolor{hard}
    %\ALLOW &&&& \ALLOW \triangle2 &&&& NP-complete & \!\cite{slitherlink}\!
    %\\
    %% 2-triangles is OPEN
    %\rowcolor{open}
    %&&&& \ALLOW \triangle2 &&&& \OPEN & \ref{open:2-triangles}
    %\\
    % 2-triangles is NP-complete
    \rowcolor{hard}
    &&&& \ALLOW \triangle2 &&&& NP-complete & Thm~\ref{thm:triangles2}
    \\
    % 3-triangles is NP-complete
    \rowcolor{hard}
    &&&& \ALLOW \triangle3 &&&& NP-complete & Thm~\ref{thm:triangles3}
    \\
    \hline
    % Monominoes alone is in P
    \rowcolor{easy}
    &&&&& \ALLOW \monomino &&& $\in{}$P & \cite{eurocg}\!
    \\
    % Monominoes + broken edges is in P
    \rowcolor{easy}
    \ALLOW &&&&& \ALLOW \monomino &&& $\in{}$P & Thm~\ref{MonominoesAndBrokenEdgesTheorem}
    %\rowcolor{open}
    %\ALLOW &&&&& \ALLOW \monomino &&& \OPEN
    \\
    % Monominoes + anti-monominoes is NP-complete
    \rowcolor{hard}
    &&&&& \ALLOW \monomino & \ALLOW \antimonomino && NP-complete & Thm~\ref{anti-monominoes}
    \\
    % Rotatable dominoes is NP-complete
    \rowcolor{hard}
    &&&&& \ALLOW \dominofree &&& NP-complete & Thm~\ref{rotating dominoes}
    \\
    % Nonrotatable vertical dominoes is NP-complete
    \rowcolor{hard}
    &&&&& \ALLOW \dominovert &&& NP-complete & Thm~\ref{vertical dominoes}
    \\
    \hline
    % Anything except antibodies is in NP
    \rowcolor{easy}
    \ALLOW & \ALLOW & \ALLOW & \ALLOW & \ALLOW & \ALLOW & \ALLOW && $\in{}$NP & Obs~\ref{all-but-antibody-np}
    \\
    % Antibodies + anything except Tetris and Antitetris is in NP
    \rowcolor{easy}
    \ALLOW & \ALLOW & \ALLOW & \ALLOW & \ALLOW &&& \ALLOW $n$ & $\in{}$NP & Thm~\ref{all-but-poly-antipoly-np}
    \\
    % Two antibodies + Tetris is Sigma_2-complete
    \rowcolor{hard}
    &&&&& \ALLOW && \ALLOW 2 & $\Sigma_2$-complete & Thm~\ref{thm:antibody-poly}
    \\
    % One antibody + Tetris + no Antitetris is in NP
    \rowcolor{easy}
    \ALLOW & \ALLOW & \ALLOW & \ALLOW & \ALLOW & \ALLOW && \ALLOW $1$ & $\in{}$NP & Thm~\ref{thm:antibody-np}
    \\
    % One antibody + Tetris + Antitetris is Sigma_2-complete
    \rowcolor{hard}
    &&&&& \ALLOW & \ALLOW & \ALLOW 1 & $\Sigma_2$-complete & Thm~\ref{thm:antibody-antipoly}
    \\
    % Antibodies + anything is in Sigma_2
    \rowcolor{easy}
    \ALLOW & \ALLOW & \ALLOW & \ALLOW & \ALLOW & \ALLOW & \ALLOW & \ALLOW $n$ & $\in \Sigma_2$ & Thm~\ref{in Sigma_2}
    \\
    % Most puzzles (excepting antibodies) are in XP w.r.t. number of rows,
    %   and FPT if the number of squares and stars is O(1).
    % Metapuzzles with switches + powered doors are PSPACE-complete
    % Metapuzzles with sliding bridges are PSPACE-complete
    % Metapuzzles with marsh up/down reconfiguration are PSPACE-complete
    \hline
  \end{tabular}
  \vspace*{-2ex}
  \caption{Our results for one-panel puzzles in The Witness:
    computational complexity with various sets of
    allowed clue types (marked by {\ALLOW}).
    Allowed polyomino clues are either arbitrary (\ALLOW),
    or restricted to be monominoes (\ALLOW \protect\monomino),
    vertical dominoes (\ALLOW \protect\dominovert),
    or rotatable dominoes (\ALLOW \protect\dominofree).}
    % A mark of {\MAYBE} indicates that the clue type could be allowed or not.}
  \label{2D results}
  %\vspace*{-10ex}
\end{table}

\section{Introduction}
\label{sec:intro}

\begin{wrapfigure}{r}{3.5in}
  \centering
  \vspace*{-9ex}
  \includegraphics[width=\linewidth]{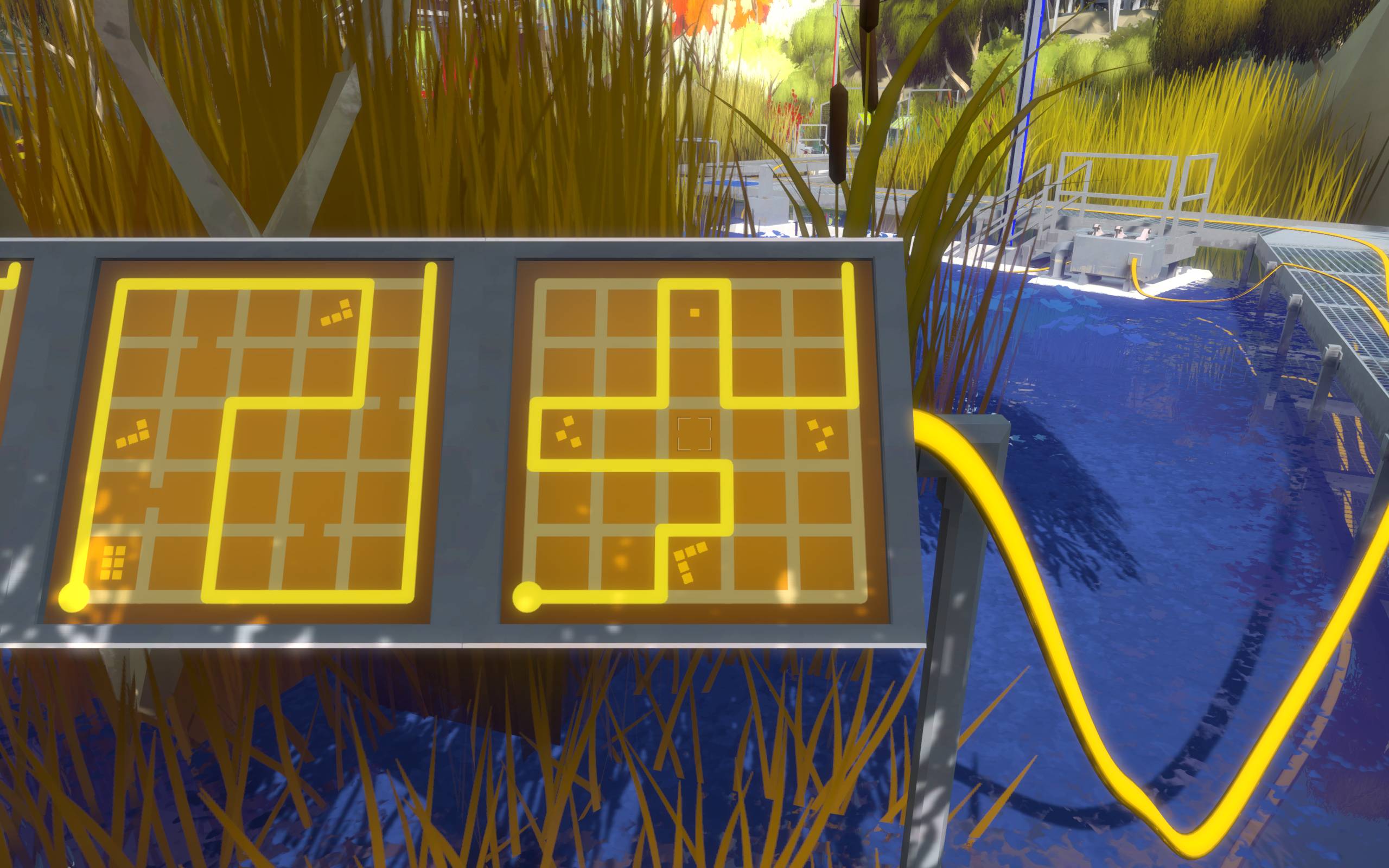}
  \caption{A screenshot from The Witness, featuring 2D puzzles in a 3D world.
    }
  \vspace*{-3ex}
  \label{screenshot puzzle}
\end{wrapfigure}

\emph{The Witness}%
\footnote{The Witness is a trademark owned by Jonathan Blow.
  % http://tmsearch.uspto.gov/bin/showfield?f=doc&state=4805:7ahng.2.42
  Screenshots and game elements are included under Fair Use
  for the educational purposes of describing video games and
  illustrating theorems.}
\cite{wikipedia} is an acclaimed 2016 puzzle video game designed by
Jonathan Blow (who originally became famous for designing the 2008 platform
puzzle game Braid, which is undecidable \cite{Hamilton-2014-braid}).
The Witness is a first-person adventure game, but the main mechanic of the game
is solving 2D puzzles presented on flat panels (sometimes CRT monitors)
within the game; see Figure~\ref{screenshot puzzle}.
The 2D puzzles are in a style similar to pencil-and-paper puzzles,
such as Nikoli puzzles.  Indeed, one clue type in The Witness (triangles)
is very similar to the Nikoli puzzle \emph{Slitherlink}
(which is NP-complete \cite{slitherlink}).

\begin{wrapfigure}{r}{2.5in}
  \centering
  %\vspace*{-3ex}
  \includegraphics[width=0.47\linewidth]{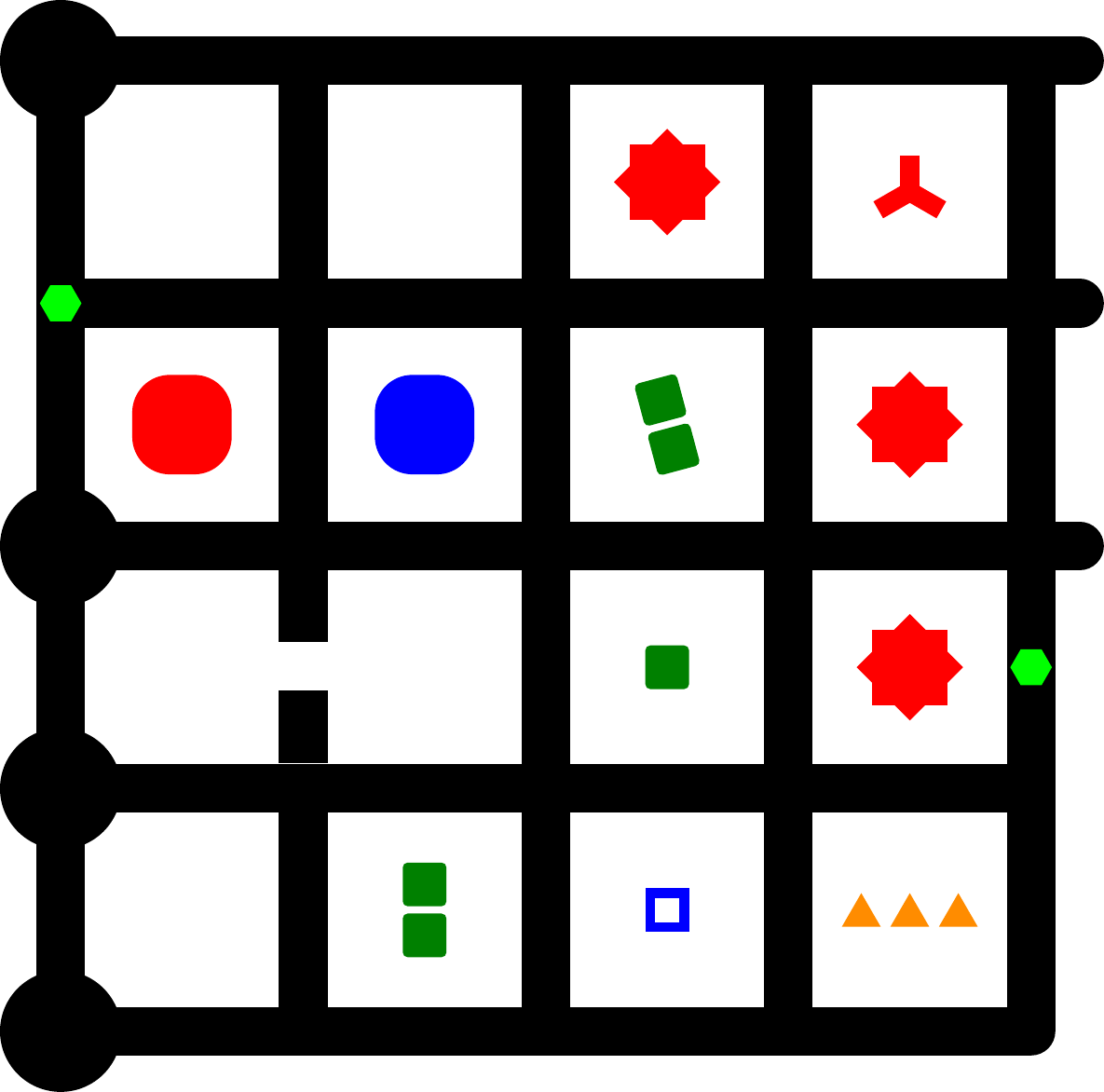}\hfill
  \includegraphics[width=0.47\linewidth]{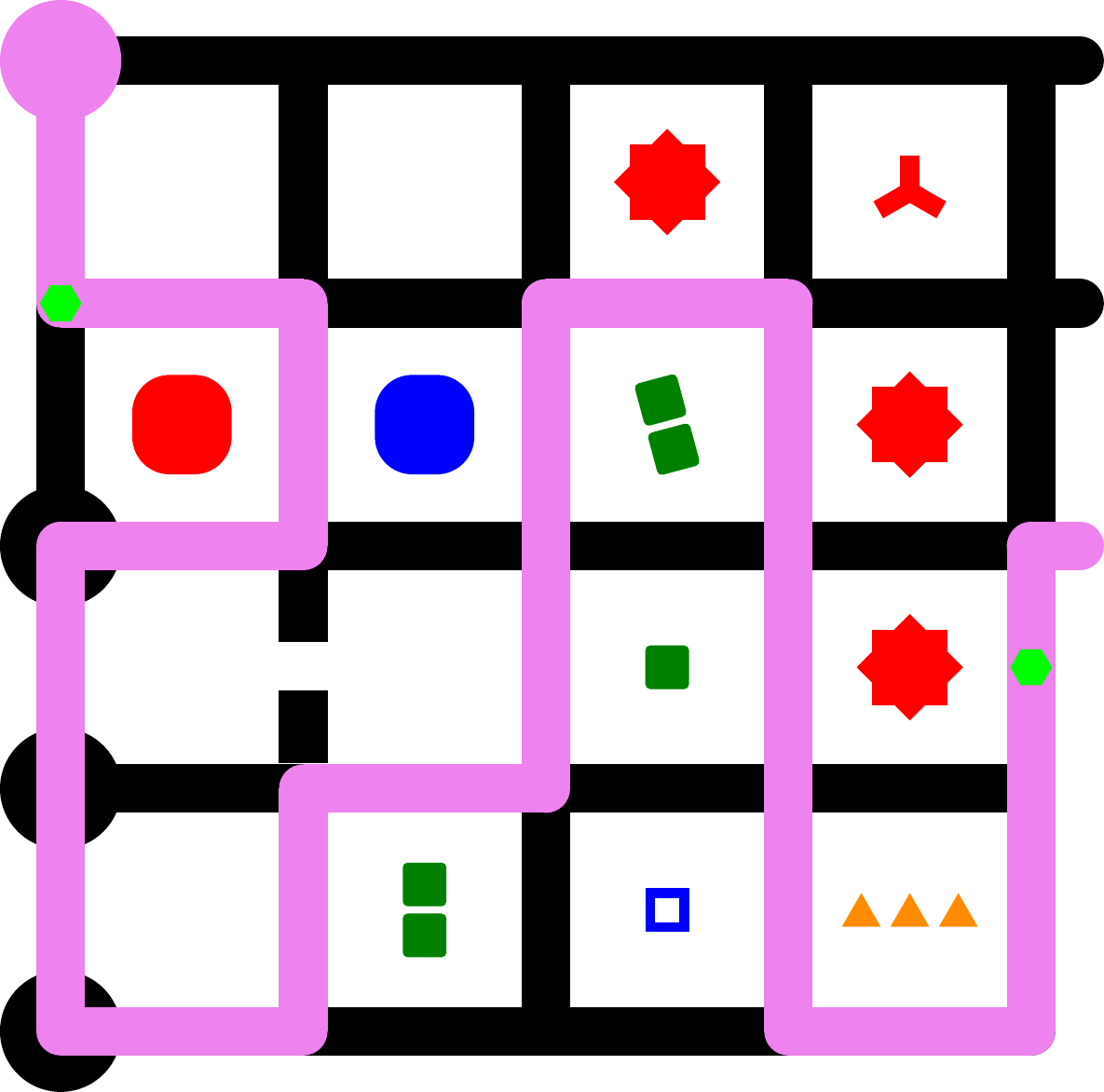}
  \caption{A small Witness puzzle featuring all clue types
           (left) and its solution (right). (Not from the actual video game.)}
  \vspace*{-3ex}
  \label{simple}
\end{wrapfigure}

In this paper, we perform a systematic study of the computational complexity
of all single-panel puzzle types in The Witness, as well as some of the
3D ``metapuzzles'' embedded in the environment itself.
Table~\ref{2D results} summarizes our single-panel results,
which range from polynomial-time algorithms (as well as membership in~L)
to completeness in two complexity classes, NP (i.e., $\Sigma_1$)
and the next level of the polynomial hierarchy, $\Sigma_2$.
%to two levels of complexity, NP-completeness and the (presumably) harder
%$\Sigma_2$-completeness.
Table~\ref{3D results} summarizes our metapuzzle results, where
PSPACE-completeness typically follows immediately.

\paragraph{Witness puzzles.}
Single-panel puzzles in The Witness (which we refer to henceforth as
\emph{Witness puzzles}) consist of an $x \times y$ full rectangular grid;%
\footnote{While most Witness puzzles have a rectangular boundary,
  some lie on a general grid graph.  This generalization is mostly equivalent
  to having broken-edge clues (defined below) on all the non-edges of the
  grid graph, but the change in boundary can affect the decomposition into
  regions.  We focus here on the rectangular case because it is most common
  and makes our hardness proofs most challenging.}
one or more \emph{start circles} (drawn as a large dot, \start);
one or more \emph{end caps}
(drawn as half-edges leaving the rectangle boundary);
and zero or more \emph{clues} (detailed below)
each drawn on a vertex, edge, or cell%
\footnote{We refer to the unit-square faces of the rectangular grid as
  \emph{cells}, given that ``squares'' are a type of clue and ``regions''
  are the connected components outlined by the solution path and
  rectangle boundary.  ``Pixels'' will be used for unit squares of
  polyomino clues.}
of the rectangular grid.
Figure~\ref{simple} shows a small example and its solution.
The goal of the puzzle is to find a simple path that starts at one of the
start circles, ends at one of the end caps, and satisfies all the constraints
imposed by the clues (again, detailed below).
We generally focus on the case of a single start circle and single end cap,
which makes our hardness proofs the most challenging.

We now describe the clue types and their corresponding constraints.
Table~\ref{clue types} lists the clues by what they are drawn on ---
grid edge, vertex, or cell --- which we refer to as ``this'' edge,
vertex, or cell.
While the last five clue types are drawn on a cell,
their constraint applies to the \emph{region} that contains that cell
(referred to as ``this region''),
where we consider the regions of cells in the rectangle as decomposed by
the (hypothetical) solution path and the rectangle boundary.

\begin{table}
\centering
%\begin{tabular}{llcp{3.65in}}
\begin{tabular}{llcp{4in}}
  \rowcolor{header}
  clue        & drawn on & symbol & constraint
  \\ \hline
  broken edge & edge     & \broken
    & The solution path cannot include this edge.
  \\
  hexagon     & edge     & \hexagon
    & The solution path must include this edge.
  \\
  hexagon     & vertex   & \hexagon
    & The solution path must visit this vertex.
  \\
  triangle    & cell     & \triangle2\!\!
    & There are three kinds of triangle clues (\triangle1\!\!, \triangle2\!,
      \triangle3\!).  For a clue with $i$ triangles, the path must include
      exactly $i$ of the four edges surrounding this cell.
  \\
  square      & cell     & \square
    & A square clue has a color.
      This region must not have any squares of a color different from this clue.
  \\
  star        & cell     & \star
    & A star clue has a color.
      This region must have exactly one other star, exactly one square, or
      exactly one antibody of the same color as this clue.
  \\
  polyomino   & cell     & \tetris
    & A polyomino clue has a specified polyomino shape, and
      is either nonrotatable (if drawn orthogonally, like \tetrisfix)
      or rotatable by any multiple of $90^\circ$
      (if drawn at $15^\circ$, like \tetris).
      Assuming no antipolyominoes, this region must be perfectly packable
      by the polyomino clues within this region.
  \\
  antipolyomino & cell   & \antitetris
    & Like polyomino clues, an antipolyomino clue has a specified
      polyomino shape and is either rotatable or not.
      For some $i \in \{0,1\}$, each cell in this region must be coverable by exactly $i$ layers,
      where polyominoes count as $+1$ layer and antipolyominoes count as
      $-1$ layer (and thus must overlap), with no positive or negative layers
      of coverage spilling outside this region.
  \\
  antibody    & cell     & \antibody
    & Effectively ``erases'' itself and another clue in this region.
      This clue also must be necessary, meaning that the solution path
      should not otherwise satisfy all the other clues.
      See Section~\ref{Antibodies} for details.
  \\
\end{tabular}
\caption{Witness puzzle clue types and the definitions of their constraints.}
\label{clue types}
\end{table}

The solution path must satisfy \textit{all} the constraints given by all the
clues.  (The meaning of this statement in the presence of antibodies is
complicated; see Section~\ref{Antibodies}.)
Note, however, that if a region has no clue constraining it in a particular
way, then it is free of any such constraints.  For example, a region without
polyomino or antipolyomino clues has no packing constraint.

As summarized in Table~\ref{2D results}, we prove that most clue types
\textit{by themselves} are enough to obtain NP-hardness.
The exceptions are broken edges, which alone just define a graph search
problem; and vertex hexagons, which are related to Hamiltonian path in
rectangular grid graphs as solved in \cite{Itai-Papadimitriou-Luiz-1982}
but remain open.
But vertex hexagons are NP-hard when we also add broken edges.
On the other hand, vertex and/or edge hexagons restricted to the boundary
of the puzzle (even for nonrectangular boundaries) are polynomial;
this result more generally solves ``subset Hamiltonian path''
(find a simple path visiting a specified subset of vertices and/or edges)
when the subset is on the outside face of a planar graph.
For squares, we determine that exactly two colors are needed for hardness.
For stars, we do not know whether one or any constant number of colors are hard.
For triangles, each single kind of triangle clue alone suffices for hardness,
though the proofs differ substantially between kinds.
For polyominoes, monominoes alone are easy to solve,
% \cite{eurocg},
but monominoes plus antimonominoes are hard, as are rotatable dominoes
by themselves and vertical nonrotatable dominoes by themselves.
All problems without antibodies or without (anti)polyominoes are in NP.
Antibodies combined with (anti)polyominoes push the complexity up to
$\Sigma_2$-completeness, but no further.

\paragraph{Nonclue constraints.}
In Section~\ref{sec:nonclue-constraints}, we consider how additional features
in The Witness beyond clues can affect (the complexity of) Witness puzzles.
Specifically, these features include visual obstruction caused by the
surrounding 3D environment (which blocks some edges), symmetry puzzles
(where inputting one path causes a symmetric copy to be drawn as well),
and intersection puzzles (where multiple puzzles must be solved
simultaneously by the same solution path).

\paragraph{Metapuzzles.}
We also consider some of the \emph{metapuzzles} formed by the 3D environment
in The Witness, where the traversable geometry changes according to
2D single-panel puzzles.
See Section~\ref{Metapuzzles} for details of these interaction models.
Table~\ref{3D results} lists our metapuzzle results,
which are all PSPACE-completeness
proofs following the infrastructure of \cite{Nintendo_TCS} (from FUN 2014).

\begin{table}
  \centering
  \tabcolsep=4pt
  \begin{tabular}{|c||c|c|}
    \hline
    \rowcolor{header}
    features & complexity & \multicolumn{1}{c|}{ref}
    \\
    \hline
    \hline
    % Metapuzzles with sliding bridges are PSPACE-complete
    \rowcolor{hard}
    sliding bridges & PSPACE-complete & Thm~\ref{sliding-bridges}
    \\
    % Metapuzzles with marsh up/down reconfiguration are PSPACE-complete
    \rowcolor{hard}
    elevators and ramps & PSPACE-complete & Thm~\ref{thm:elevator}
    \\
    % Metapuzzles with switches + powered doors are PSPACE-complete
    \rowcolor{hard}
    power cables and doors & PSPACE-complete & Thm~\ref{thm:power cables}
    \\
    % Metapuzzles with switches + powered doors are PSPACE-complete
    \rowcolor{open}
    light bridges & OPEN & Prob~\ref{open:light bridge}
    \\
    \hline
  \end{tabular}
  \caption{Our results for metapuzzles in The Witness:
    computational complexity with various sets of environmental features.}
  \label{3D results}
\end{table}

\paragraph{Puzzle design.}
In Section~\ref{sec:puzzle design}, we consider designing 2D Witness
puzzles using the variety of clues available.  In particular, we introduce
the \emph{Witness puzzle design problem} where the goal is to find a puzzle
whose solution is a specific path or set of paths.  This problem naturally
arises in metapuzzles, or if we were to design a Witness puzzle font
(in the style of \cite{Fonts_TCS,SpiralGalaxies_MOVES2017}).
Although many versions of this problem remain open, we present some basic
universality results and limitations.

\paragraph{Open Problems.}
Finally, Section~\ref{sec:conclusion} collects together the main open
problems from this paper.

\section{Hamiltonicity Reduction Framework}
\label{sec:Hamiltonicity Reduction Framework}

We introduce a framework for proving NP-hardness of Witness puzzles
by reduction from Hamiltonian cycle in a grid graph $G$ of maximum degree~$3$.
Roughly speaking, we scale $G$ by a constant scale factor~$s$,
and replace each vertex by a block called a chamber;
refer to Figure~\ref{Hamiltonicity graph}.
%mapping a vertex $v$ of $G$ from coordinates $(x,y)$ to $(s x, s y)$.
Precisely, for each vertex $v$ of $G$ at coordinates $(x,y)$,
we construct a $(2r+1) \times (2r+1)$ subgrid of vertices
%whose vertices have coordinates $(s x + a, s y + b)$ for all $-r \leq a, b \leq r$
$\{s x - r, \dots, s x + r\} \times \{s y - r, \dots, s y + r\}$,
and all induced edges between them, called a \emph{chamber}~$C_v$.
This construction requires $2 r < s$ for chambers not to overlap.
For each edge $e = \{v,w\}$ of $G$,
we construct a straight path in the grid from $s v$ to $s w$, and
define the \emph{hallway} $H_{v,w}$ to be
the subpath connecting the boundaries of $v$'s and $w$'s chambers,
which consists of $s - 2 r$ edges.
Figure~\ref{Hamiltonicity graph} illustrates this construction
on a sample graph $G$.
% with $r = 1$ and $s = 4$.

\begin{figure}
  \centering
  \def\scale{0.18}
  $\vcenter{\hbox{\subcaptionbox{\label{grid Hamiltonicity} Instance of grid-graph Hamiltonicity}{\includegraphics[scale=0.3]{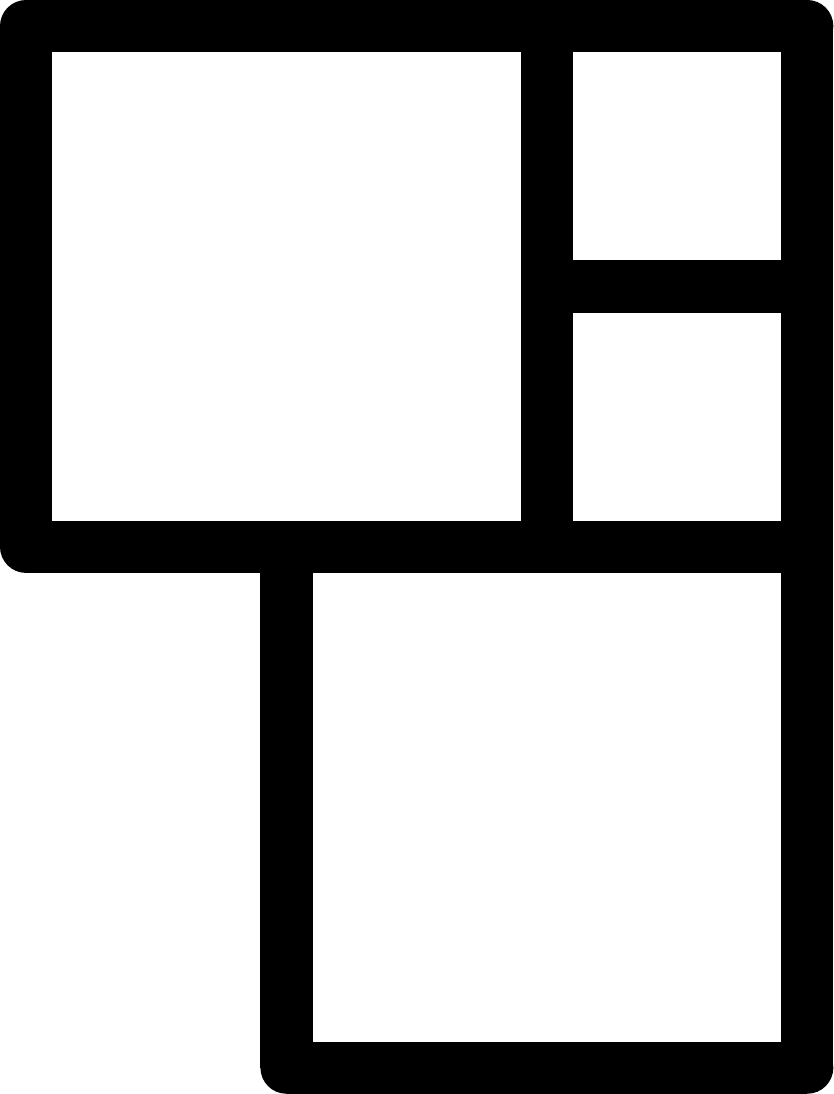}}}}$
  \hfil\hfil
  $\vcenter{\hbox{\subcaptionbox{\label{grid Hamiltonian} A possible solution to (a)}{\includegraphics[scale=0.3]{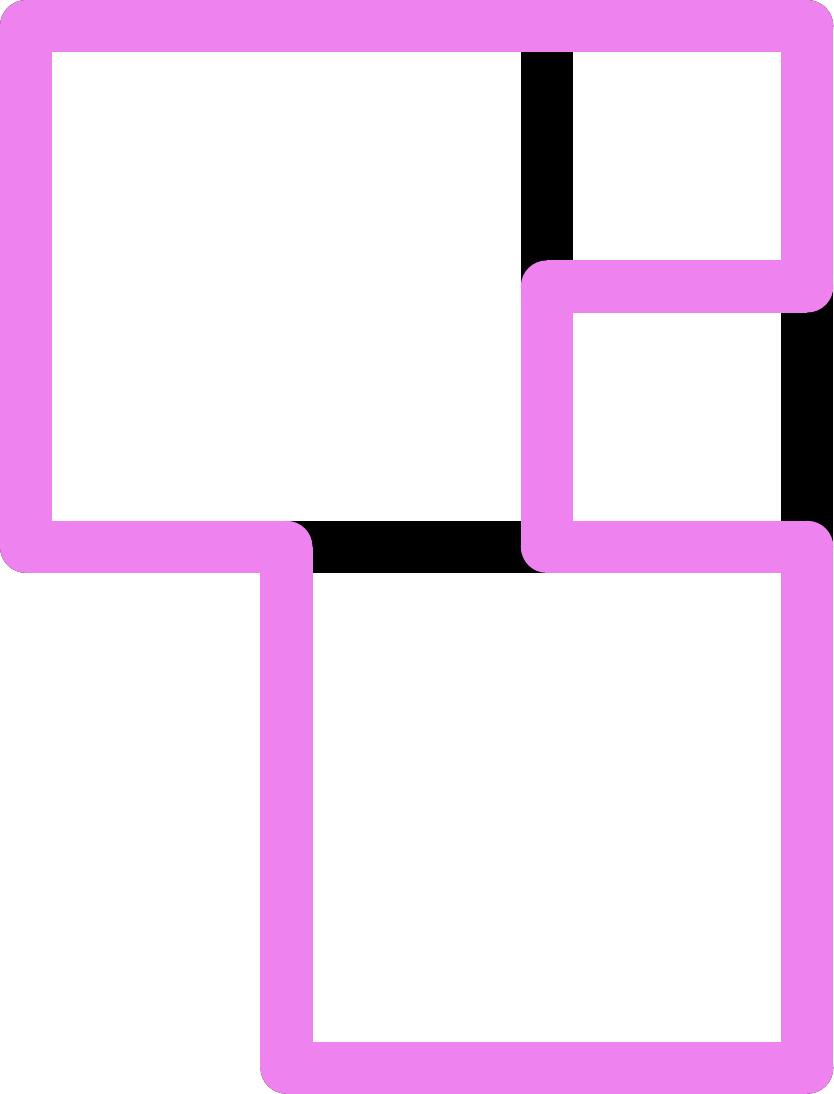}}}}$
  \hfil\hfil
  $\vcenter{\hbox{\subcaptionbox{Corresponding chambers (turquoise) and hallways (pink)}{\includegraphics[scale=0.18]{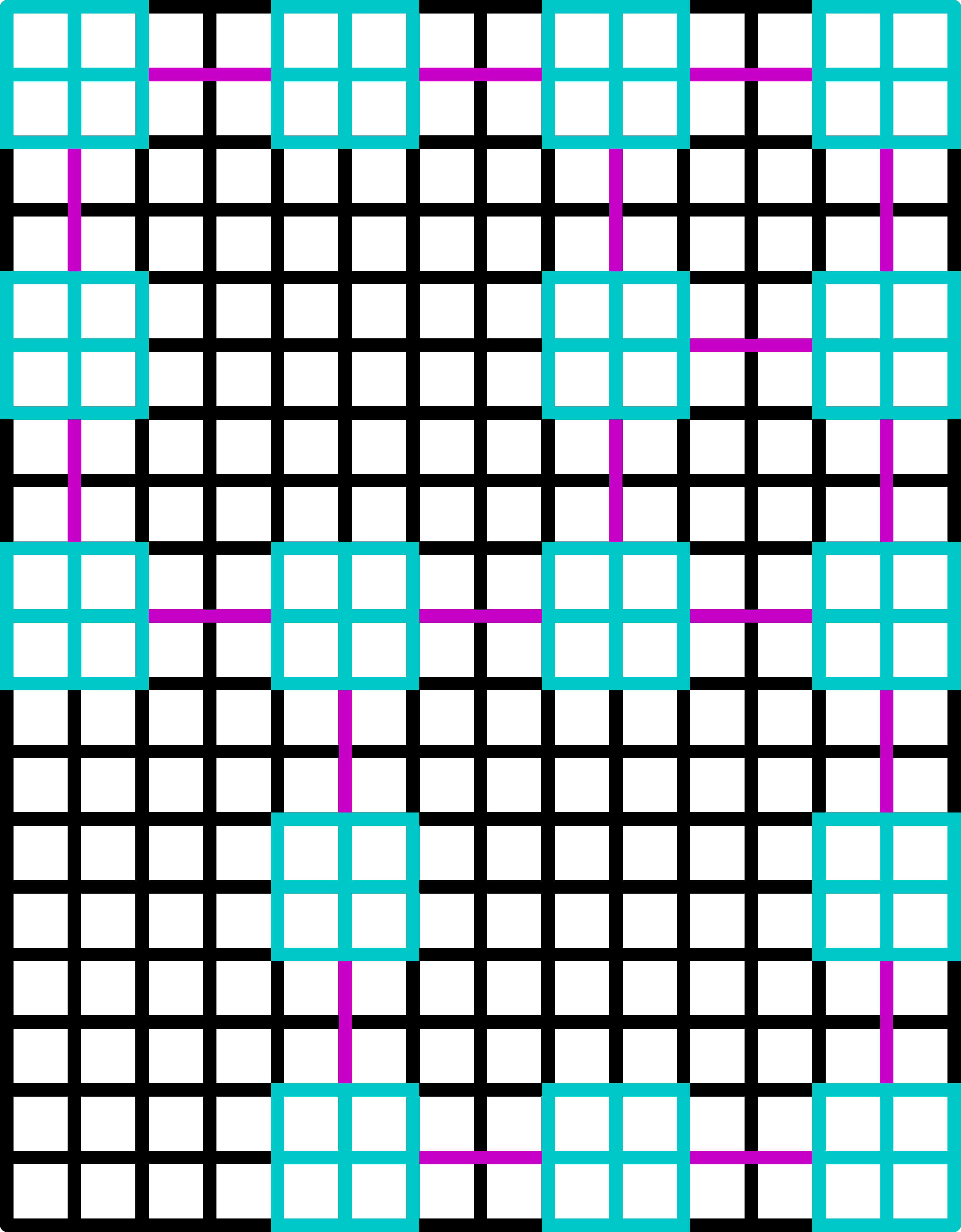}}}}$
  \caption{An example of the Hamiltonicity framework with $r=1$ and $s=4$.}
  \label{Hamiltonicity graph}
\end{figure}

In each reduction, we define constraints to force
the solution path to visit (some part of) each chamber at least once,
to alternate between visiting chambers and traversing hallways
that connect those chambers,
and to traverse each hallway at most once.
Because $G$ has maximum degree~$3$, these constraints imply that
each chamber is entered exactly once and exited exactly once.
Next to one chamber on the boundary of $G$, called the \emph{start/end chamber},
we place the start circle and end cap of the Witness puzzle.
Thus, any solution to the Witness puzzle induces a Hamiltonian cycle
in~$G$.
To show that any Hamiltonian cycle in $G$ induces a solution to the
Witness puzzle, we simply need to show that a chamber can be traversed in
each of the $3 \choose 2$ ways.

\subsection{Simple Applications of the Hamiltonicity Framework}
\label{sec:simple hamiltonicity framework}

%In Section~\ref{sec:simple hamiltonicity framework},
In this section, we briefly present some simple constructions based
on this Hamiltonicity framework.  All of these results are subsumed by
stronger results presented formally in later sections, so we omit the
details of these simpler NP-hardness proofs.
Membership in NP for all puzzles except antibodies will be proved
in Observation~\ref{all-but-antibody-np}.

\later{
\begin{corollary}
It is NP-complete to solve Witness puzzles containing broken edges and squares of two colors.
\end {corollary}
\begin{proof}
We make hallways with broken edges and chambers with one square of each color. See Figure~\ref{Hamiltonicity square and broken edges}.
\begin{figure}
  \centering
  \subcaptionbox{Unsolved}{\includegraphics[scale=0.172]{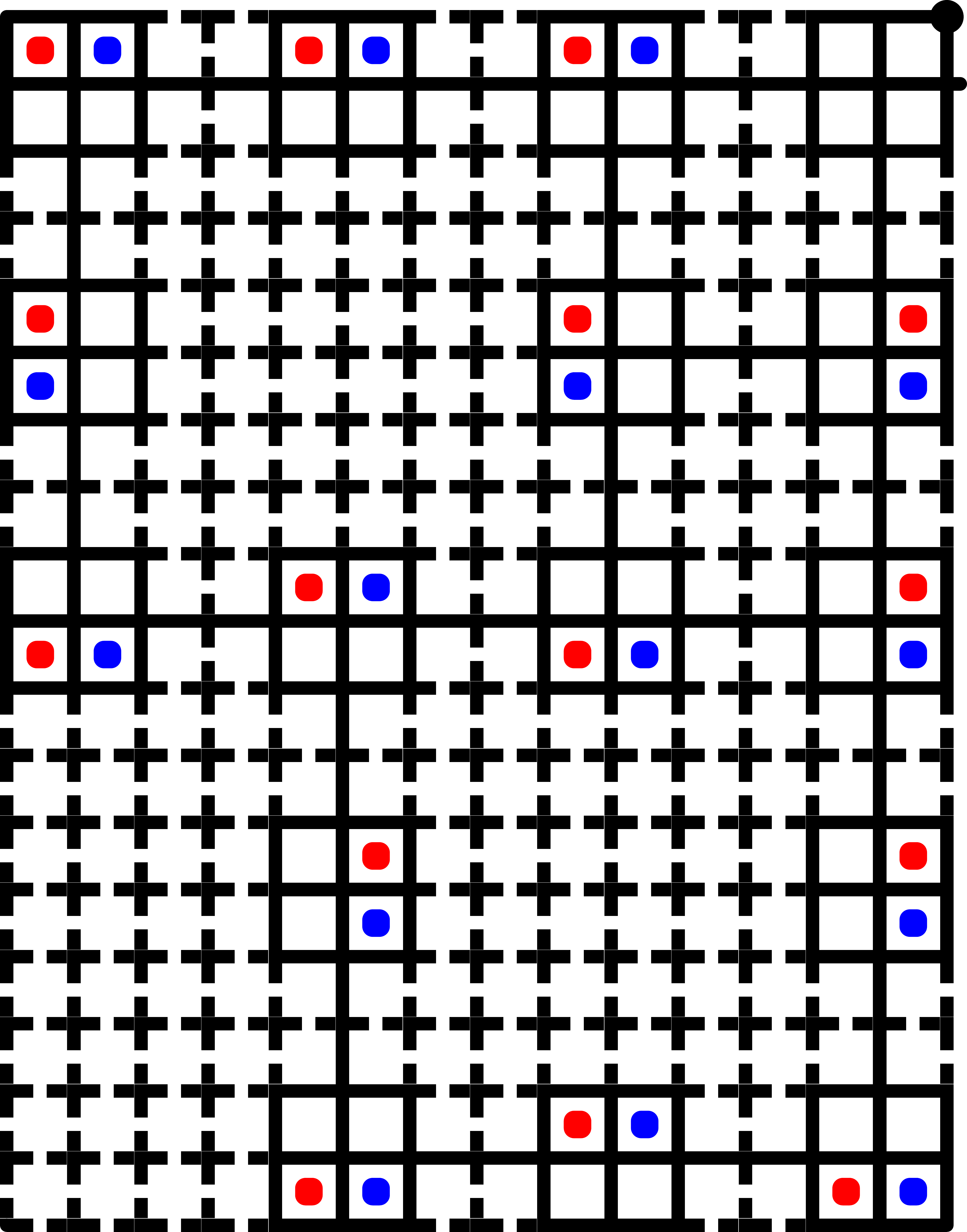}}
  \hfil\hfil
  \subcaptionbox{Solved}{\includegraphics[scale=0.172]{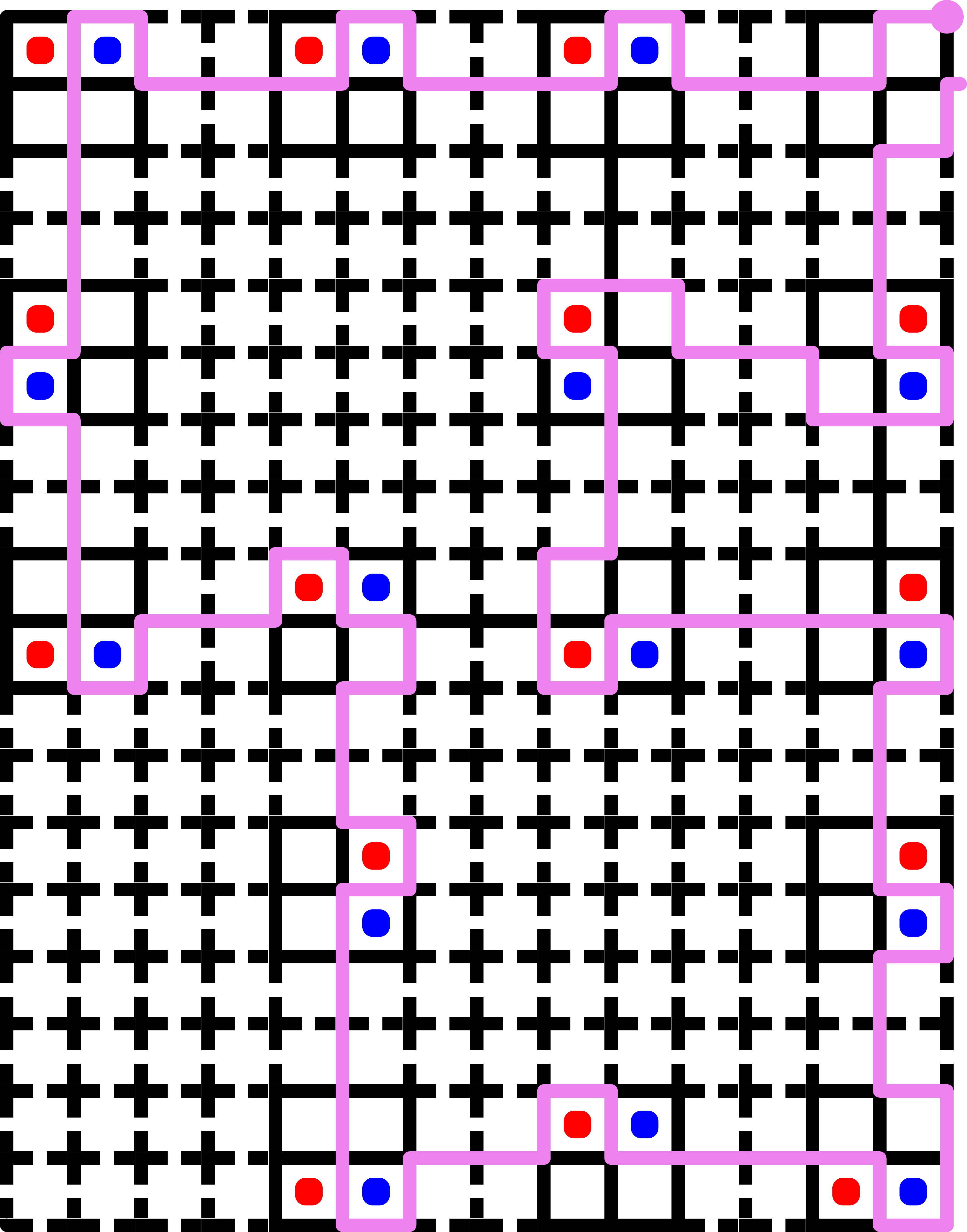}}
  \caption{Example of the Hamiltonicity framework applied to Witness with
           squares and broken edges.}
  \label{Hamiltonicity square and broken edges}
\end{figure}
\end{proof}

\begin{corollary}
It is NP-complete to solve Witness puzzles containing broken edges and stars of
arbitrarily many colors.
\end{corollary}
\begin{proof}
We make hallways with broken edges and chambers with four stars of one
color (a different color for each chamber). See Figure~\ref{Hamiltonicity star and broken edges}.
\begin{figure}
  \centering
  \subcaptionbox{Unsolved}{\includegraphics[scale=0.172]{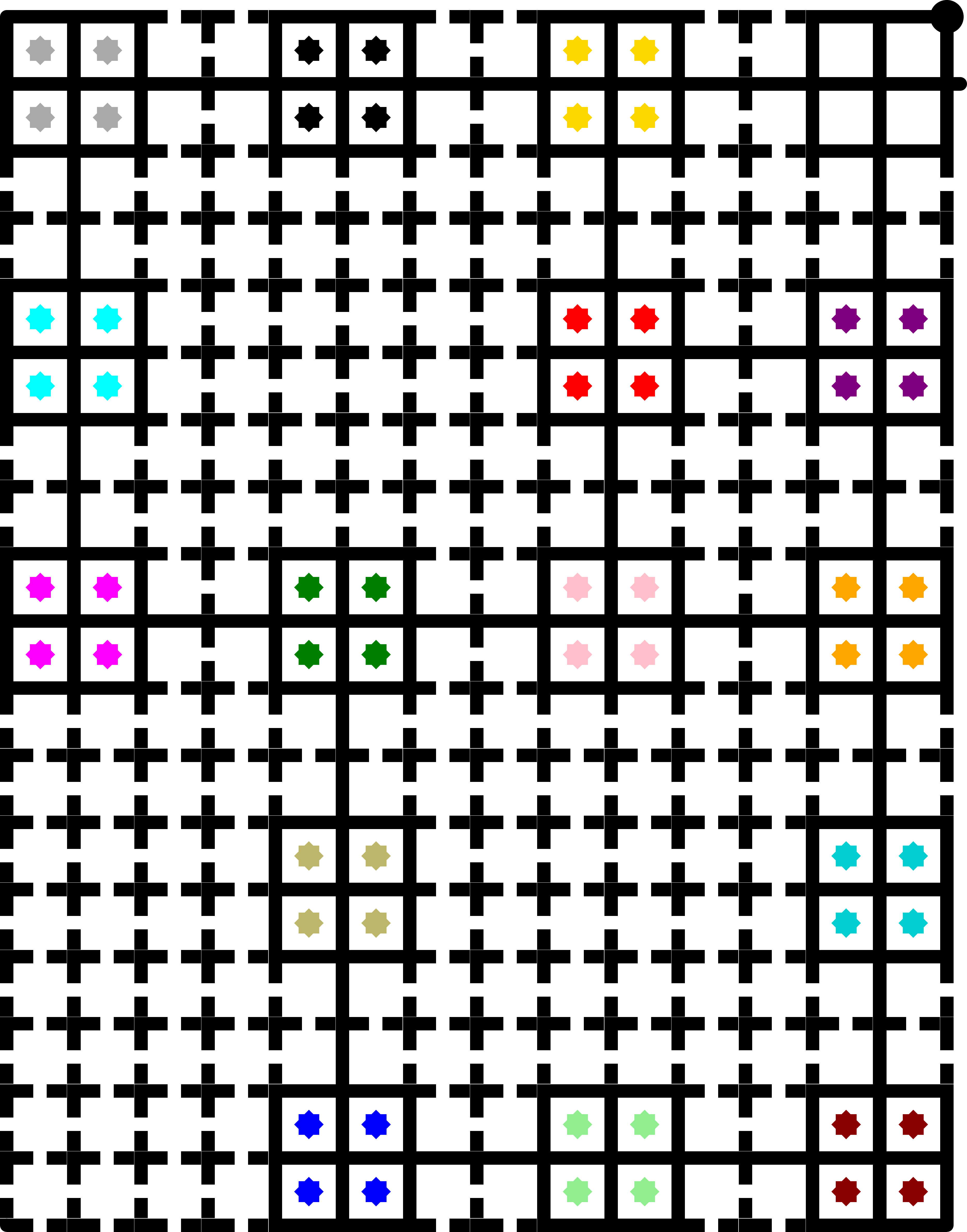}}
  \hfil\hfil
  \subcaptionbox{Solved}{\includegraphics[scale=0.172]{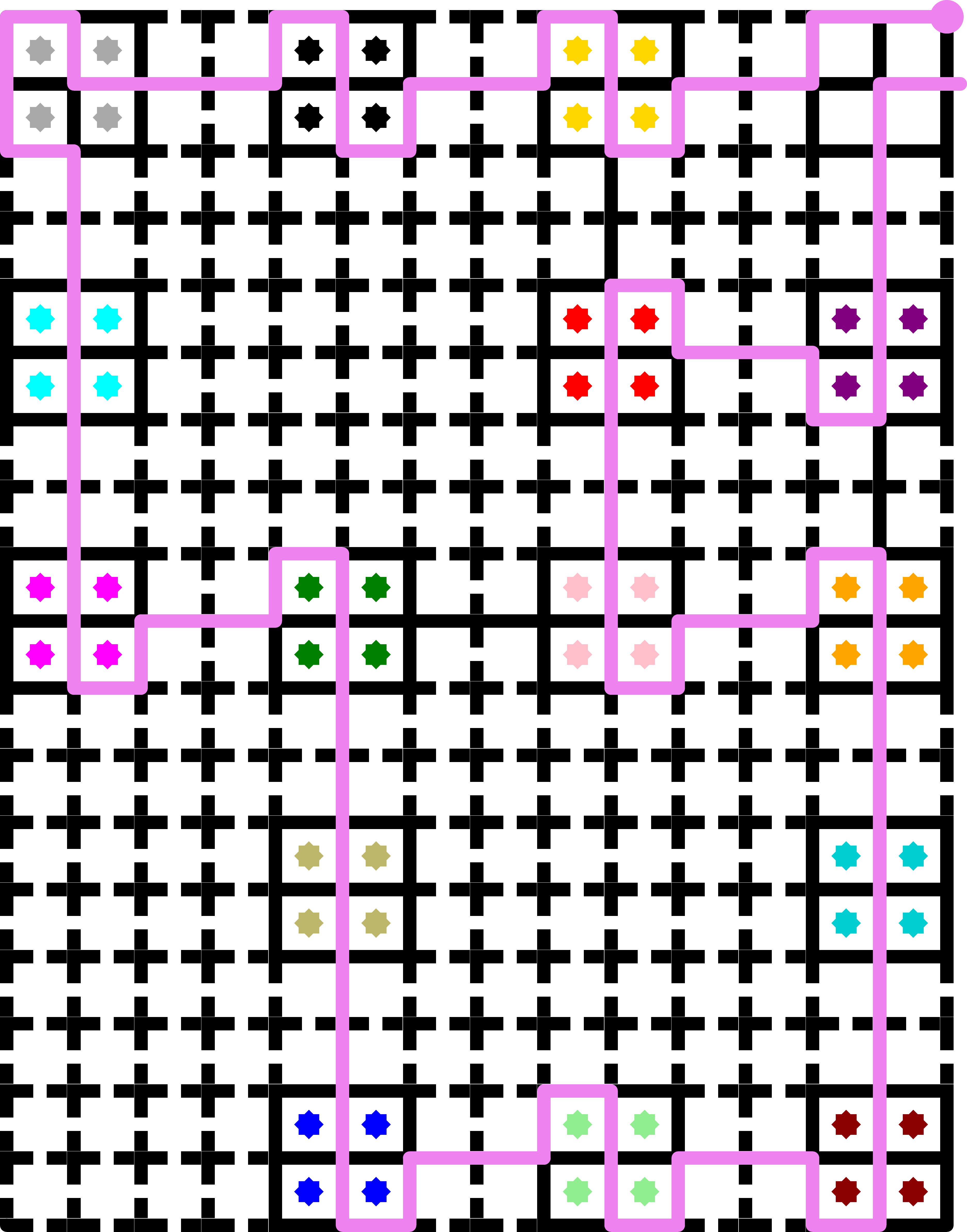}}
  \caption{Example of the Hamiltonicity framework applied to Witness with
           stars and broken edges.}
  \label{Hamiltonicity star and broken edges}
\end{figure}
\end{proof}

\begin{corollary}\label{thm:triangles-0-2} \label{cor:broken edges+triangles}
It is NP-complete to solve Witness puzzles containing broken edges and $k$-triangle clues, for any (single) $k \in \{1,2,3\}$. 
\end{corollary}
\begin{proof}
We make hallways with broken edges and chambers with a single
$k$-triangle clue. See Figure~\ref{Hamiltonicity triangle and broken edges}.
\begin{figure}
  \centering
  \subcaptionbox{Unsolved}{\includegraphics[scale=0.172]{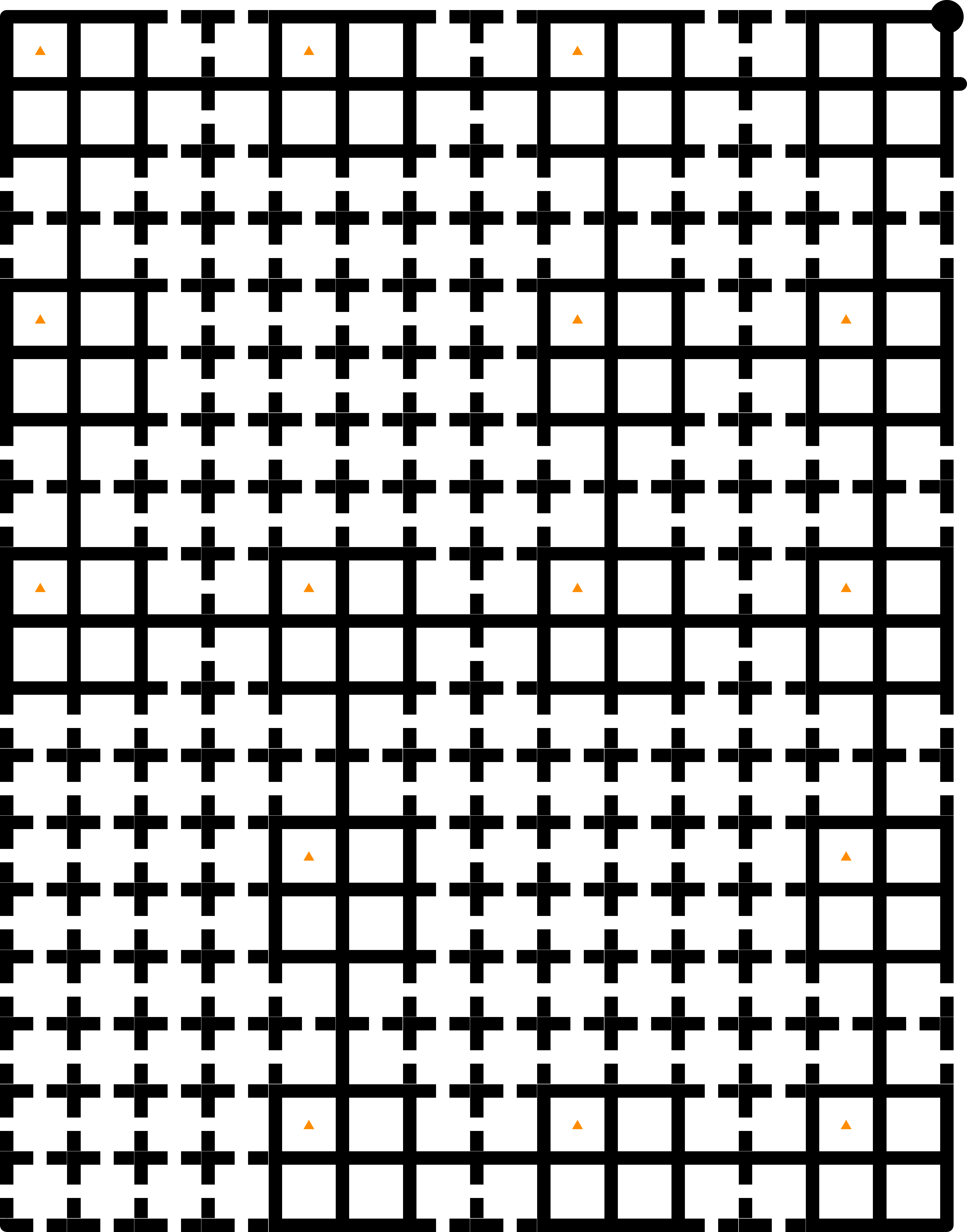}}
  \hfil\hfil
  \subcaptionbox{Solved}{\includegraphics[scale=0.172]{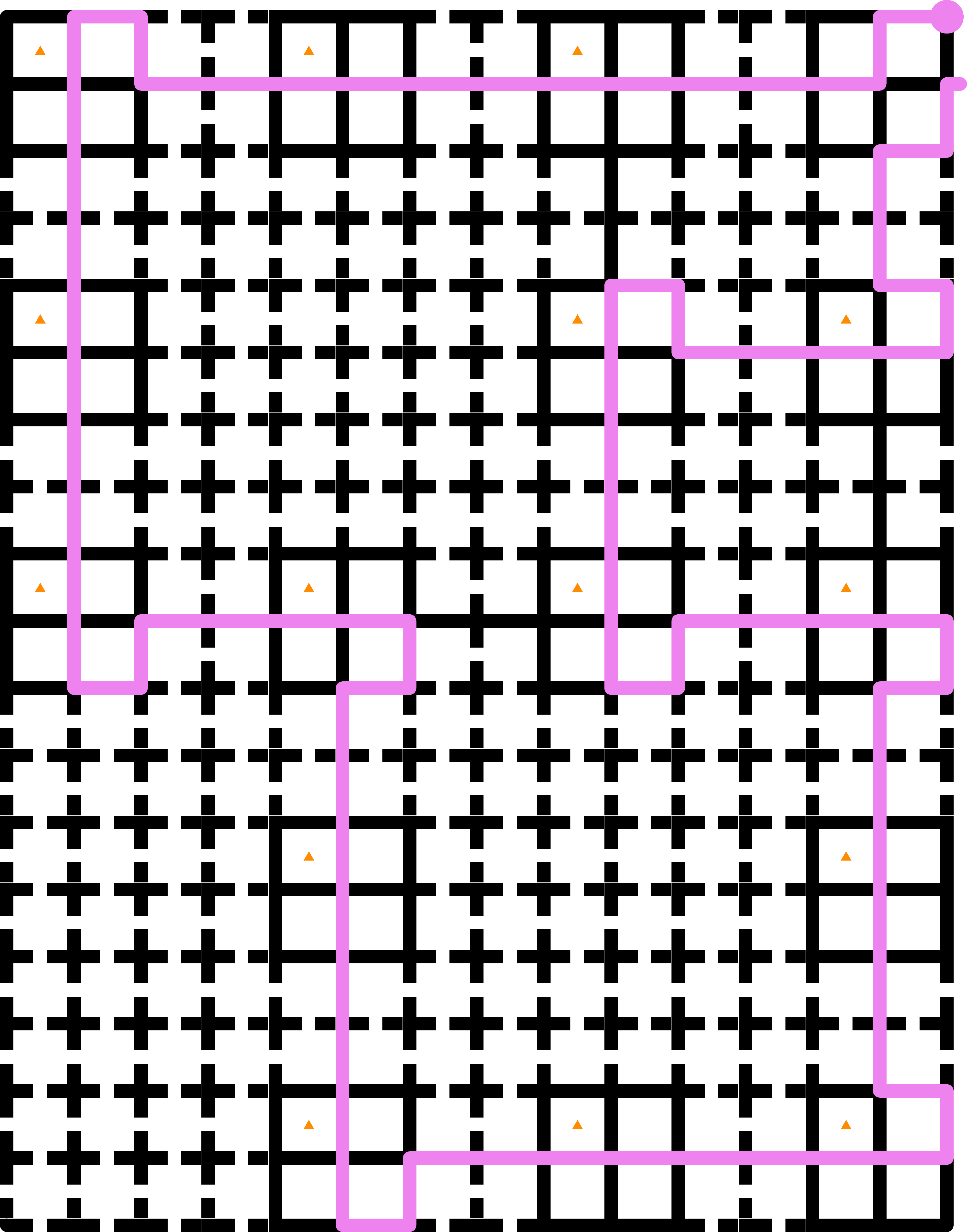}}
  \caption{Example of the Hamiltonicity framework applied to Witness with
           1-triangles and broken edges.}
  \label{Hamiltonicity triangle and broken edges}
\end{figure}
\end{proof}
}

\section{Hexagons and Broken Edges}

\abstractlater{
  \section{Proofs: Hexagons}
  \label{appendix:hexagons}
}

Hexagons are placed on vertices or edges of the graph and
require the path to pass through all of the hexagons.
Broken edges are edges which cannot be included in the path.
In this section, we show two positive results and two negative results.
On the positive side, we show that
puzzles with just broken edges are solvable in $L$
(Section~\ref{sec:broken edges}),
and puzzles with hexagons just on the boundary of the puzzle
(even when the boundary is not rectangle) and arbitrary broken edges
are solvable in~$P$ (Section~\ref{sec:boundary hexagons}).
On the negative side, we show that
puzzles with just hexagons on vertices and broken edges are NP-complete
(Section~\ref{sec:hexagon+broken}), and
puzzles with just hexagons on edges (and no broken edges) are NP-complete
(Section~\ref{sec:edge hexagons}).
We leave open the complexity of puzzles with just hexagons on vertices
(and no broken edges):

\begin{restatable}{open}{vertexhexagons} \label{open:vertex-hexagons}
  Is there a polynomial-time algorithm to solve Witness puzzles containing only hexagons on vertices?
\end{restatable}

\subsection{Just Broken Edges}
\label{sec:broken edges}

\begin{observation} \label{thm:broken-edges}
Witness puzzles containing only broken edges, multiple start circles, and 
multiple end caps are in L.
\end{observation}
\begin{proof}
%Add super source $s$ and super sink $t$ to represent the multiple possible starts and ends in the puzzle.
We keep two pointers and a counter to track which pairs of starts and ends we have tried. For each start and end pair, we run an $(s,t)$ path existence algorithm, which is in L. If any of these return \textsc{yes}, then the answer is \textsc{yes}. Thus, we have solved the problem with a quadratic number of calls to a log-space algorithm, a constant number of pointers, and a counter, all of which only require logarithmic space.
\end{proof}

\subsection{Hexagons and Broken Edges}
\label{sec:hexagon+broken}

The following trivial result motivates Open Problem~\ref{open:vertex-hexagons}
(do we need broken edges?).

\begin{observation} \label{thm:broken edges+vertex hexagons}
It is NP-complete to solve Witness puzzles containing only broken edges and hexagons on vertices.
\end{observation}
\begin{proof}
Hamiltonian path in grid graphs \cite{Itai-Papadimitriou-Luiz-1982}
is a strict subproblem.
\end{proof}

For edge hexagons, we first present a simple application of the Hamiltonicity
framework that uses broken edges; this result will be subsumed by
Theorem~\ref{thm:edge hexagons} by a more involved application that avoids
broken edges.

\iffull
  \begin{theorem} \label{broken+hexagons}
  It is NP-complete to solve Witness puzzles containing only broken edges and hexagons on edges.
  \end{theorem}
  \begin{proof}
  Apply the Hamiltonicity framework with scale factor $s=5$ and
  chamber radius $r=1$.
  All edges within chambers and hallways are unbroken,
  and all other edges are broken,
  forcing hallways to be traversed at most once.
  Within each chamber, we place a hexagon on one of the edges
  incident to the center of the chamber.
  This hexagon forces the chamber to be visited, while having enough empty
  space to enable connections however desired.
  %as shown in Figure~\ref{3x3}.
  %A chamber is simply a large enough region of unbroken edges,
  %with a hexagon on an edge in the middle of the region.
  In fact, a smaller $2 \times 3$ chamber with a hexagon on the middle edge
  also suffices, as shown in Figure~\ref{2x3}.
  \end{proof}

  \begin{figure}
    \centering
    \includegraphics[scale=0.172]{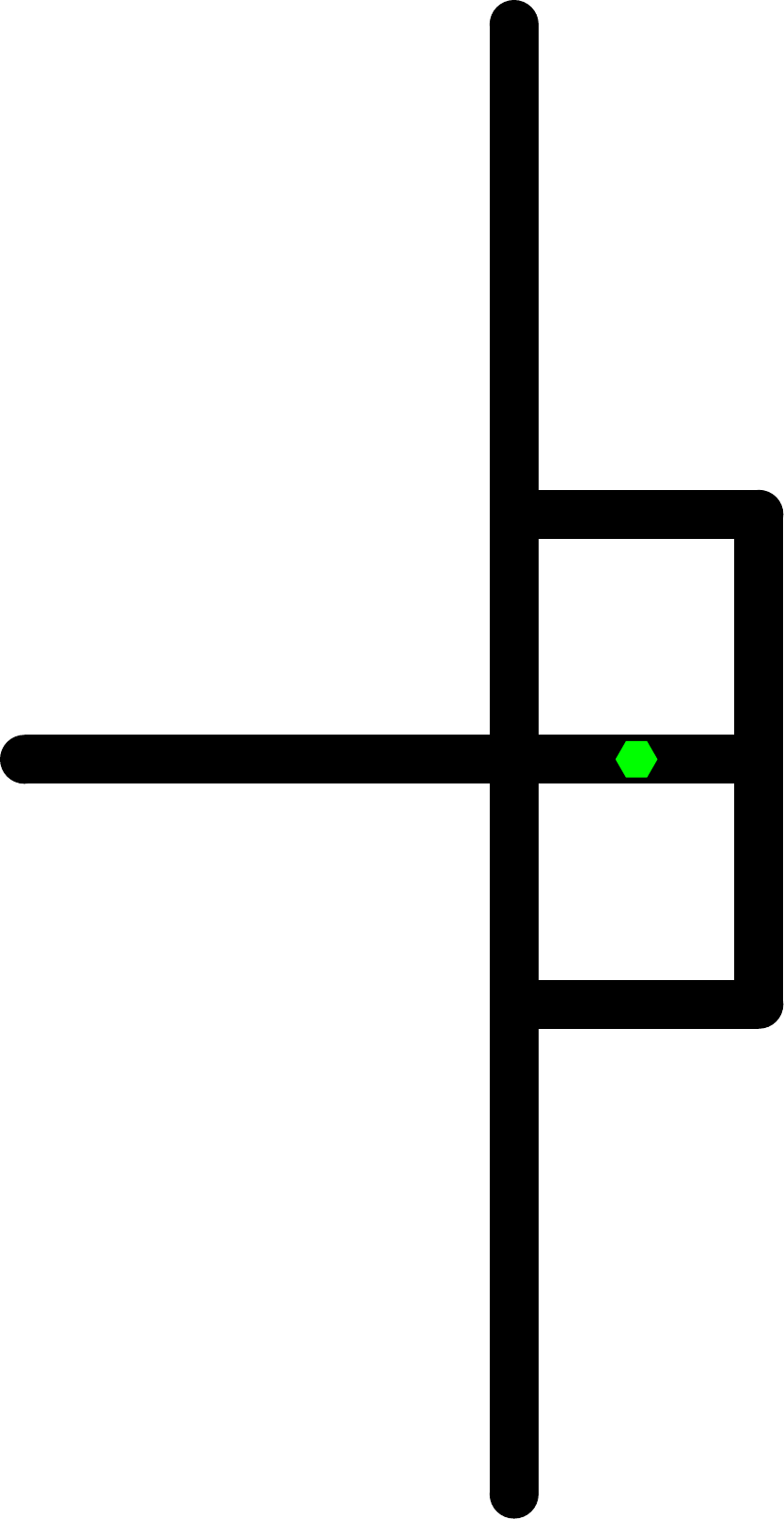}
    \caption{$2 \times 3$ vertex gadget.}
    \label{2x3}
  \end{figure}

  %\begin{verbatim}
  %        |
  %        |_
  %________|_|
  %        |_|
  %        |
  %        |
  %\end{verbatim}
\fi

\subsection{Just Edge Hexagons}
\label{sec:edge hexagons}

\both{
\begin{theorem} \label{thm:edge hexagons}
It is NP-complete to solve Witness puzzles containing only hexagons on edges (and no broken edges).
\end{theorem}
}

\ifabstract
\begin{proofsketch}
We use the Hamiltonicity framework. Noting that two edge hexagons incident to the same vertex must be consecutively
traversed by the solution path, we carefully force the solution path to traverse the boundary of every chamber separate
from the decision of which hallways to use. As with other Hamiltonicity framework reductions, we force each chamber to be visited
with an edge hexagon in its center and can deduce the corresponding Hamiltonian cycle in the original grid graph from the set of used hallways.
The full proof can be found in Appendix~\ref{appendix:hexagons}.
\end{proofsketch}
\fi

\later{
\begin{proof}
Apply the Hamiltonicity framework with scale factor $s=8$ and
chamber radius $r=2$; refer to Figure~\ref{Hamiltonicity edge hexagon}.
As before, within each chamber, we place a hexagon on one of the edges
incident to the center of the chamber.
Consider the grid graph $G$ formed by the chambers and hallways,
and its complement grid graph $\bar G$ (induced by all grid points
in $\mathbb Z^2 \setminus G$).

\begin{figure}
  \centering
  \subcaptionbox{Instance corresponding to Figure~\ref{grid Hamiltonicity}}{\includegraphics[scale=0.092]{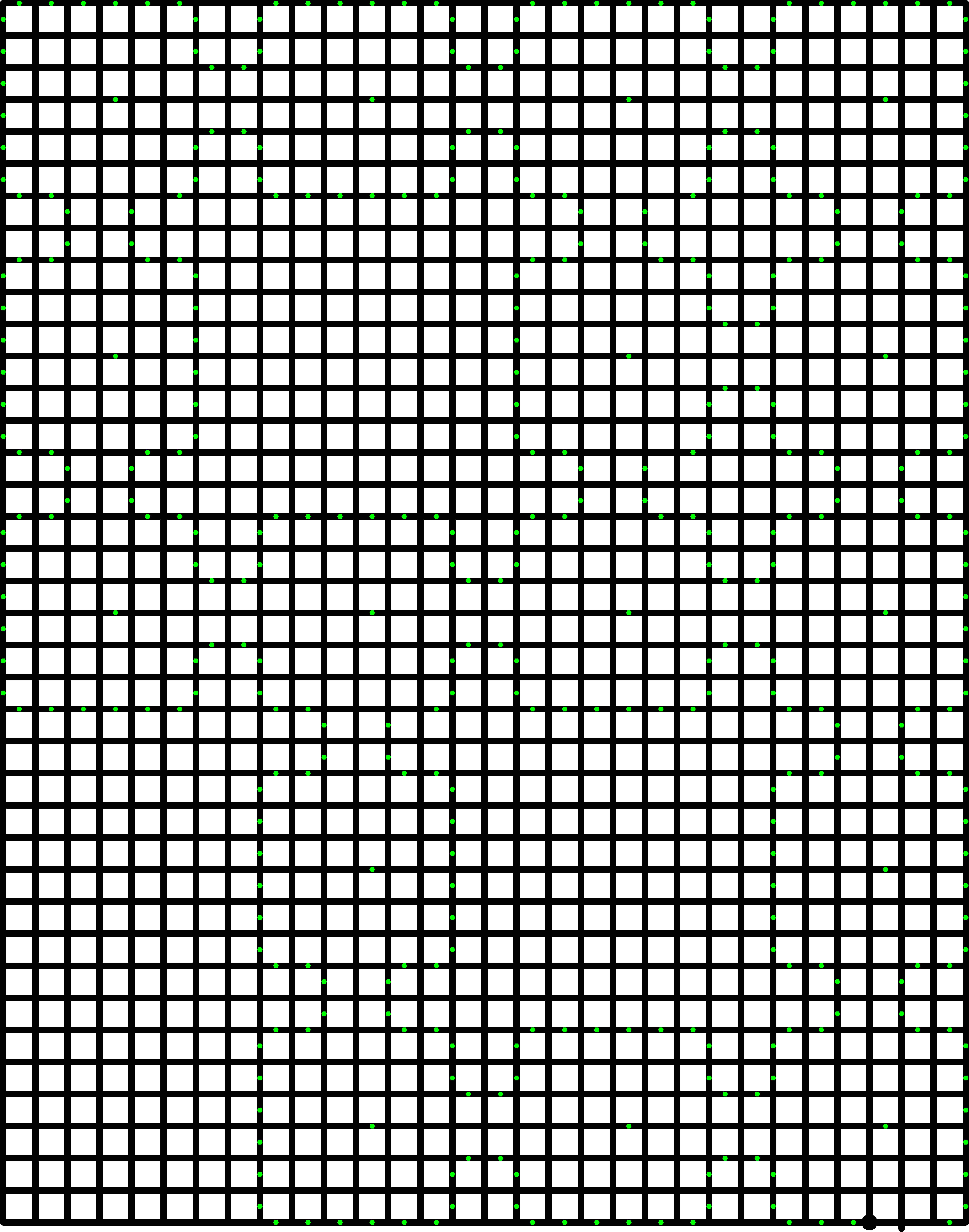}}
  \hfil\hfil
  \subcaptionbox{Solution corresponding to Figure~\ref{grid Hamiltonian}}{\includegraphics[scale=0.092]{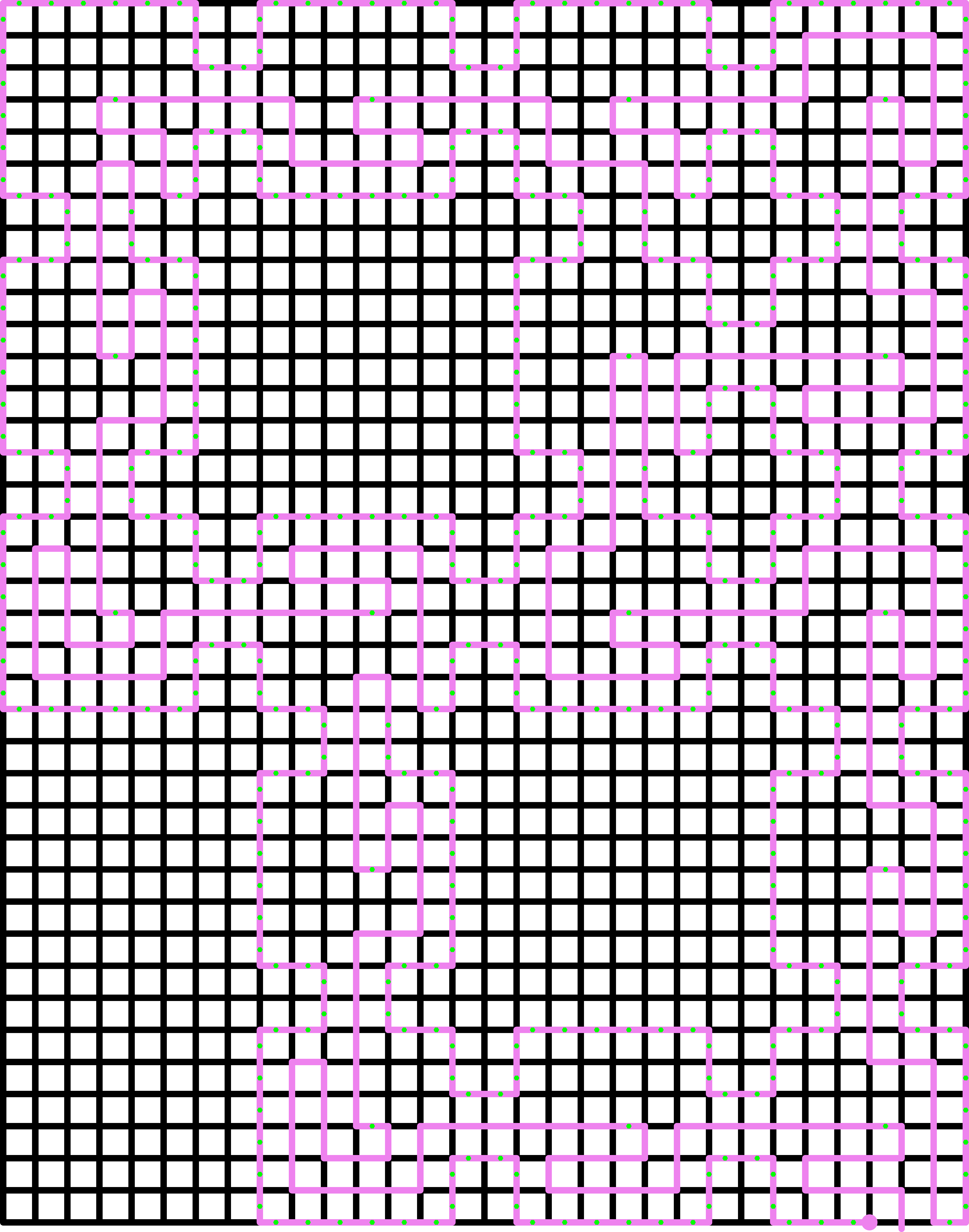}}
  \caption{Example of the Hamiltonicity framework applied to Witness with
           edge hexagons.}
  \label{Hamiltonicity edge hexagon}
\end{figure}

To constrain the solution path to remain mostly within~$G$,
we add hexagons as follows.
For each connected component $C$ of $\bar G$
(including the region exterior to~$G$),
we add a cycle of hexagons on the boundary edges of~$C$.
(These cycles outline the chambers and hallways, without intersecting them.)
We break each cycle by removing one or two consecutive hexagons,
leaving a path of hexagons called a \emph{wall}.
The removed hexagon(s) are adjacent to the \emph{holey chamber} of the
cycle: for a non-outer cycle, the holey chamber is the leftmost topmost chamber
adjacent to the cycle, and for the outer cycle, the holey chamber is the
rightmost bottommost chamber adjacent to the cycle.
For each non-outer cycle, we remove one hexagon, 
from the horizontal edge immediately
below the bottom-right cell of the holey chamber.
For the outer cycle, we remove two hexagons,
from the horizontal edges immediately
below the two rightmost cells in the bottom row of the holey chamber;
these two horizontal edges share a vertex called the \emph{gap vertex}.
Thus, in all cases, the removed hexagon(s) of a cycle are below the bottom
right of its holey chamber, so each chamber is the holey chamber of at
most one cycle.

We place the start circle and end cap at the bottom of the diagram,
with the start circle at the left endpoint of the outer wall,
and the end cap below the gap vertex.
(Thus, the rightmost bottommost chamber is the start/end chamber.)

\paragraph{Witness solution $\to$ Hamiltonian cycle.}

We claim that any solution to this Witness puzzle contains each wall as a
contiguous subpath.  Let $e_1, e_2, \dots, e_k$ be the path of edges with
hexagon clues forming a wall, and let $v_0, v_1, \dots, v_k$ be the
corresponding vertices on the path.
When the solution path visits an edge $e_i$, where $1 \leq i < k$,
the solution path must visit edge $e_{i+1}$ immediately before or after;
otherwise, $e_{i+1}$ could not be visited at another time on the solution
path (contradicting its hexagon constraint) because its endpoint $v_i$
has already been visited.
By induction, any solution path visits the wall edges consecutively,
as either $e_1, e_2, \dots, e_k$ or $e_k, e_{k-1}, \dots, e_1$.

Next we claim that any solution path must enter and exit each non-outer wall
from its corresponding holey chamber.
The cycle corresponding to the wall is the boundary of a
connected component $C$ of~$\bar G$.
The wall contains all vertices of the boundary of~$C$,
so the only way for the solution path to enter or exit the interior of $C$
is by entering or exiting the wall.
But the start circle and end cap are not interior to $C$ or on the wall,
so the solution path must enter and exit the wall from the exterior of~$C$.
The only such neighbors of the wall endpoints are in the holey chamber.

By the previous two claims, any solution to this Witness puzzle cannot go
strictly inside any connected component of $\bar G$, except the outer
component, and traversing a wall starts and ends in the same chamber.
Hence, a solution can traverse from chamber to chamber only via hallways,
and it must visit every chamber to visit the hexagon in the middle.
The solution effectively begins and ends at the rightmost bottommost chamber,
the gap vertex being the only way in from the outside.
Therefore, any solution can be converted into a Hamiltonian cycle in the
original grid graph.

\paragraph{Hamiltonian cycle $\to$ Witness solution.}

To convert any Hamiltonian cycle into a Witness solution, first we route
the Hamiltonian cycle within the chambers and hallways so that, in every
chamber, the routed cycle visits the hexagon in the middle of the chamber
as well as the bottom edge of the bottom-right cell of the chamber.
The routing along each hallway is uniquely defined.
To route within a chamber, we connect one visited hallway along a straight line
to the central vertex of the chamber; then traverse the incident edge
with the hexagon unless we just did; then walk clockwise or counterclockwise
around the remainder of the chamber in order to visit the bottom edge of the
bottom-right cell of the chamber before reaching the other visited hallway.

Next we modify this routed cycle into a Witness solution path by including
the walls.  For each wall, within the corresponding holey chamber,
we replace the bottom edge $e$ of the bottom-right cell of the
chamber by the two vertical edges below~$e$.
For non-outer walls, the two endpoints of the wall attach to these two
vertical edges, effectively replacing the edge $e$ with the wall path.
Thus, before we modify the holey chamber of the outer wall,
we still have a cycle.
For the outer wall, the right endpoint of the wall attaches
to the right vertical edge, while the left endpoint of the wall is the start
circle; the left vertical edge attaches to the gap vertex, which leads to the
end cap.  Thus, we obtain a path from the start circle to the end cap.
\end{proof}
} %\later

%\begin{theorem}
%There's a polynomial-time algorithm to solve Witness puzzles with hexagons on vertices.
%\end{theorem}
%\begin{proof}
%Hopefully we'll prove this soon.
%\end{proof}

\subsection{Boundary Hexagons and Broken Edges}
\label{sec:boundary hexagons}

In this section, we solve Witness puzzles with arbitrary broken edges
and hexagons just on boundary of the puzzle.
To make this result more interesting,
we allow a generalized type of Witness puzzle (also present in the real game)
where the board consists of an arbitrary simply connected set $C$ of cells
(instead of just a rectangle), and hexagons can be placed on any
vertices and/or edges on the outline of~$C$.
\xxx{Should define outline somewhere.}

This result is essentially a polynomial-time algorithm for solving
\emph{subset Hamiltonian path} ---
find a simple path visiting a specified subset of vertices and/or edges ---
on planar graphs when the subset lies entirely on the outside face
of the graph.
This problem is a natural variation of \emph{subset TSP} ---
find a minimum-length not-necessarily-simple cycle visiting
a specified subset of vertices.
A related result is that \emph{Steiner tree} ---
find a minimum-length tree visiting a specified subset of vertices ---
can be solved in polynomial time on planar graphs when the subset
lies entirely on the outside face of the graph \cite{boundary-steiner-tree}.
This result is a key step in the first PTAS for Steiner tree in planar graphs
\cite{Borradaile-Mathieu-Klein-2007}.
It seems that simple paths are trickier to find than trees,
so our algorithm is substantially more complicated.
Hopefully, our result will also find other applications
in planar graph algorithms.

%\subsubsection{Maximum-Remainder Paths}

\label{forced division}
Define a \emph{forced division} $(C,F)$ to consist of a simply connected set $C$
of empty cells (possibly with broken edges between them)
together with a set $F$ of \emph{forced} vertices and edges
(i.e., vertices and/or edges with hexagon clues)
on the outline of~$C$.
%A path \emph{respects} $F$ if it traverses all edges in~$F$.
A \emph{$(C,F)$-path} is a simple path using only vertices and edges
incident to cells in $C$ that traverses all vertices and edges in~$F$.
\xxx{Is ``simple'' automatically included in ``path''? If so, could remove,
  though it might be good to stress it here.}

To enable a dynamic program for finding $(C,F)$-paths when they exist,
we prove a strong structural result about $(C,F)$-paths that leave the
``most room'' for future paths, which may be of independent interest.

For any two vertices $u$ and $v$ on the outline of~$C$,
any $(C,F)$-path $P$ from $u$ to $v$
decomposes the cells of $C$ into one or more connected components.
For each vertex $t$ on the outline of $C$, if $P$ does not visit $t$,
then there is a unique connected component $R_t(P)$ incident to $t$.
We call $R_t(P)$ the \emph{$t$-remainder} of~$P$.
The intent is for the path to continue on from $v$ into the $t$-remainder,
so we are interested in the case where $v$ is incident to $R_t(P)$,
which is equivalent to preventing the path from going ``backward''
along the outline of $C$ between $v$ and~$t$.

%We will build a desired path from $s$ to $t$ out of subpaths from $s$ to
%intermediate vertices of~$t$.  To leave the most room for the rest of the path,
%we will search for \emph{maximum-remainder} paths, that is, paths that leave
%a remainder with the most cells possible.

% 3 proofs of weaker version (F = \emptyset):
% * comparing two paths:
%   * L_3 (simplified)
%   * flood fill from (u,v)
% * compare DFS path to arbitrary path
% But weaker version doesn't seem to suffice.

\begin{lemma} \label{thm:monominoes-canonical}
%Given a maximal connected component $C$ of empty cells (possibly with broken
%edges), with some of its edge outline marked as needing traversal,
Given a forced division $(C,F)$ with possibly broken edges,
for distinct vertices $u,v,t$ appearing in clockwise order on the outline of~$C$,
if there is any $(C,F)$-path from $u$ to $v$ that does not visit the outline
of $C$ in the clockwise interval $(v,t]$, then
there exists a unique \emph{maximum $t$-remainder} $R^*$ over all such paths.
More precisely, the $t$-remainder of any such path is a subset of~$R^*$, and
there is such a path with $t$-remainder exactly~$R^*$.
%Furthermore, for $F = \emptyset$, a (not necessarily unique) \emph{maximum-$t$-remainder} path $P^*$ from $u$ to $v$
%can be found (or proven to not exist) in polynomial time.
The same result holds for any forbidden interval $(v,x)$ of the outline of $C$
that contains $(v,t]$.
\end{lemma}

\begin{proof}
\begin{figure}
  \centering
  \subcaptionbox{$P_1$ and $P_1$-depth.}{\includegraphics[width=0.49\textwidth]{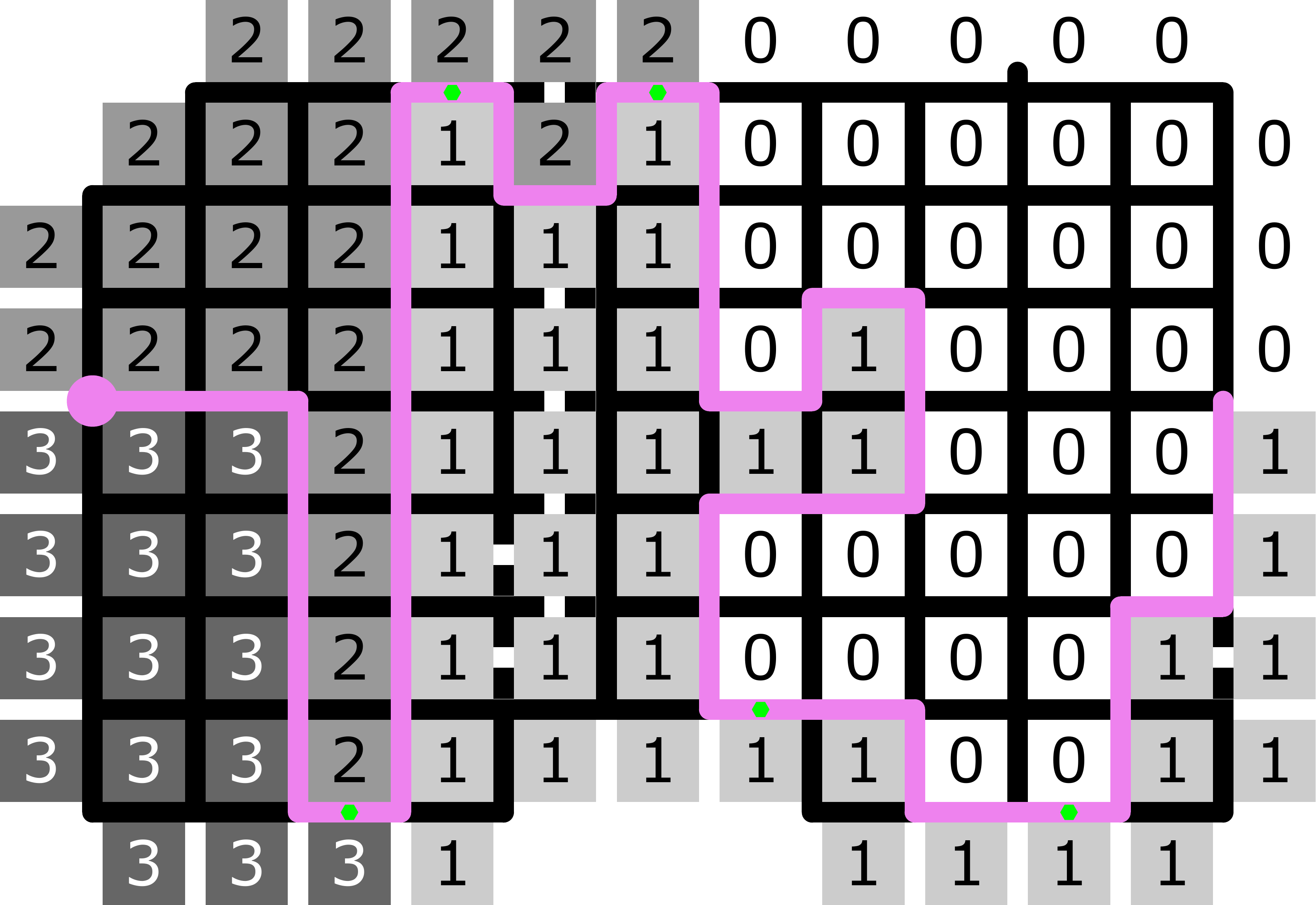}}
  \hfil\hfil
  \subcaptionbox{$P_2$ and $P_2$-depth.}{\includegraphics[width=0.49\textwidth]{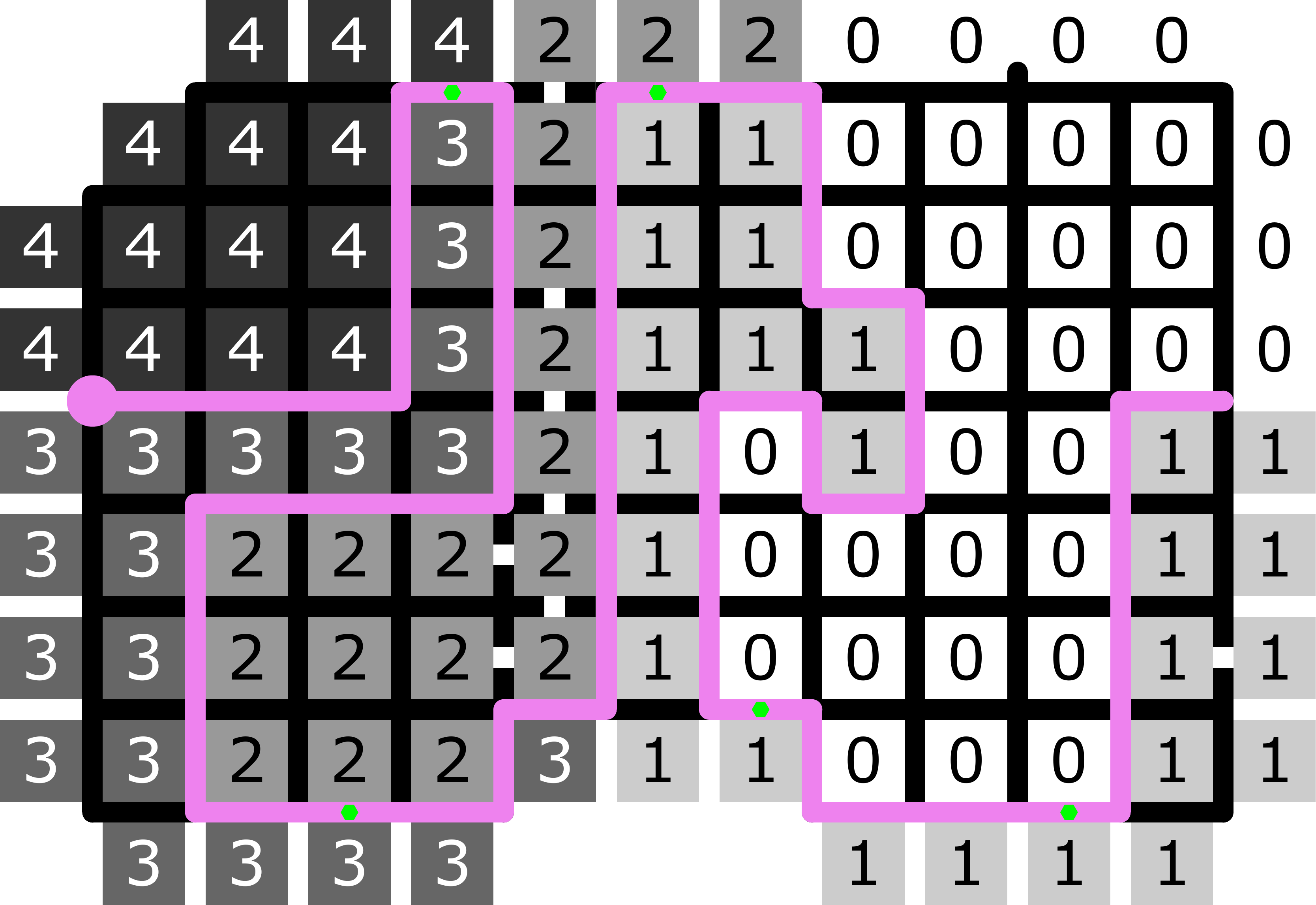}}

  \subcaptionbox{Depth (the minimum of $P_1$-depth and $P_2$-depth).  Note the $1$-cavity inside the depth-$0$ area.}{\includegraphics[width=0.49\textwidth]{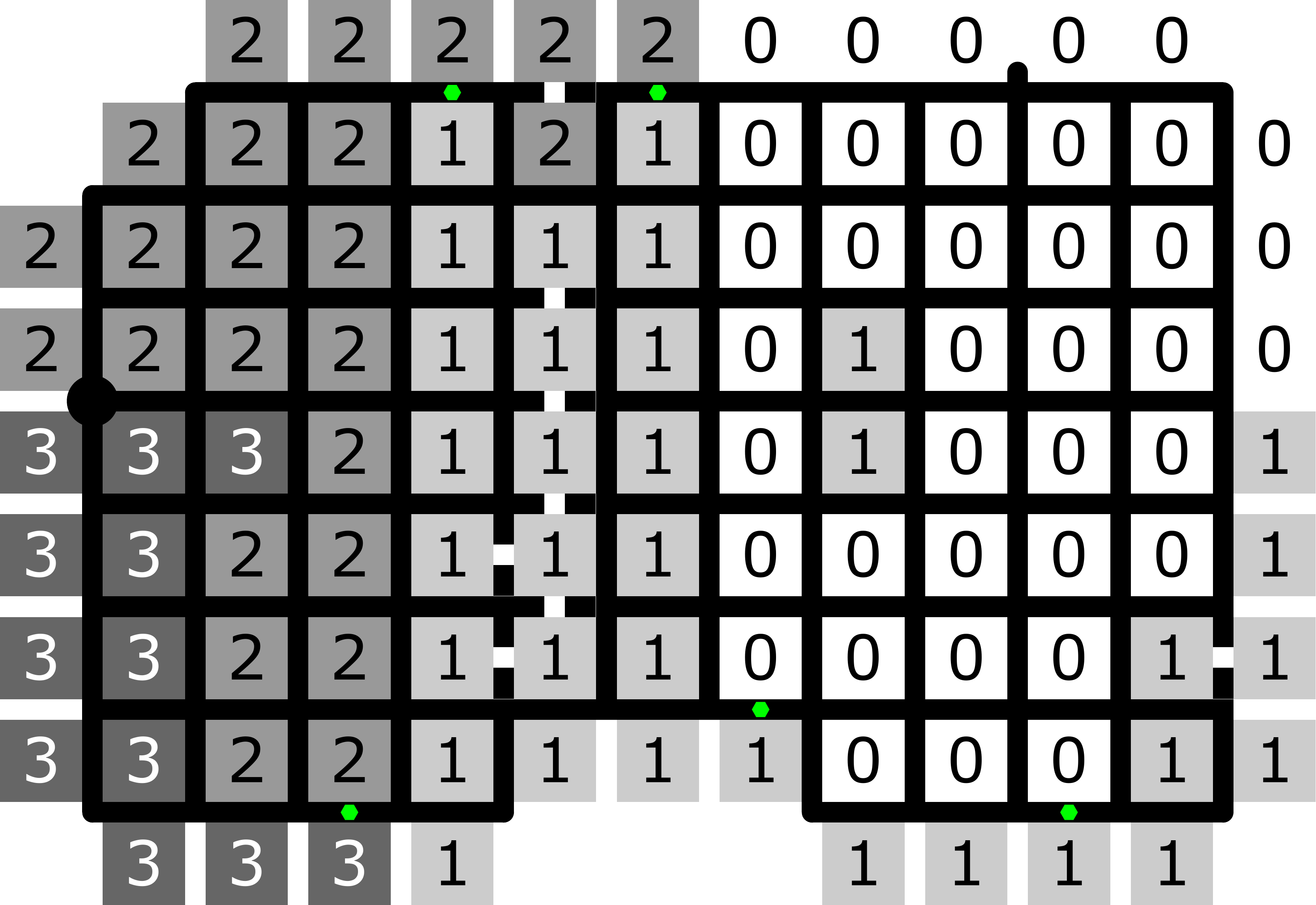}}
  \hfil\hfil
  \subcaptionbox{$P_3$.  After filling the cavity, the ledges form a path from $u$ to $v$.}{\includegraphics[width=0.49\textwidth]{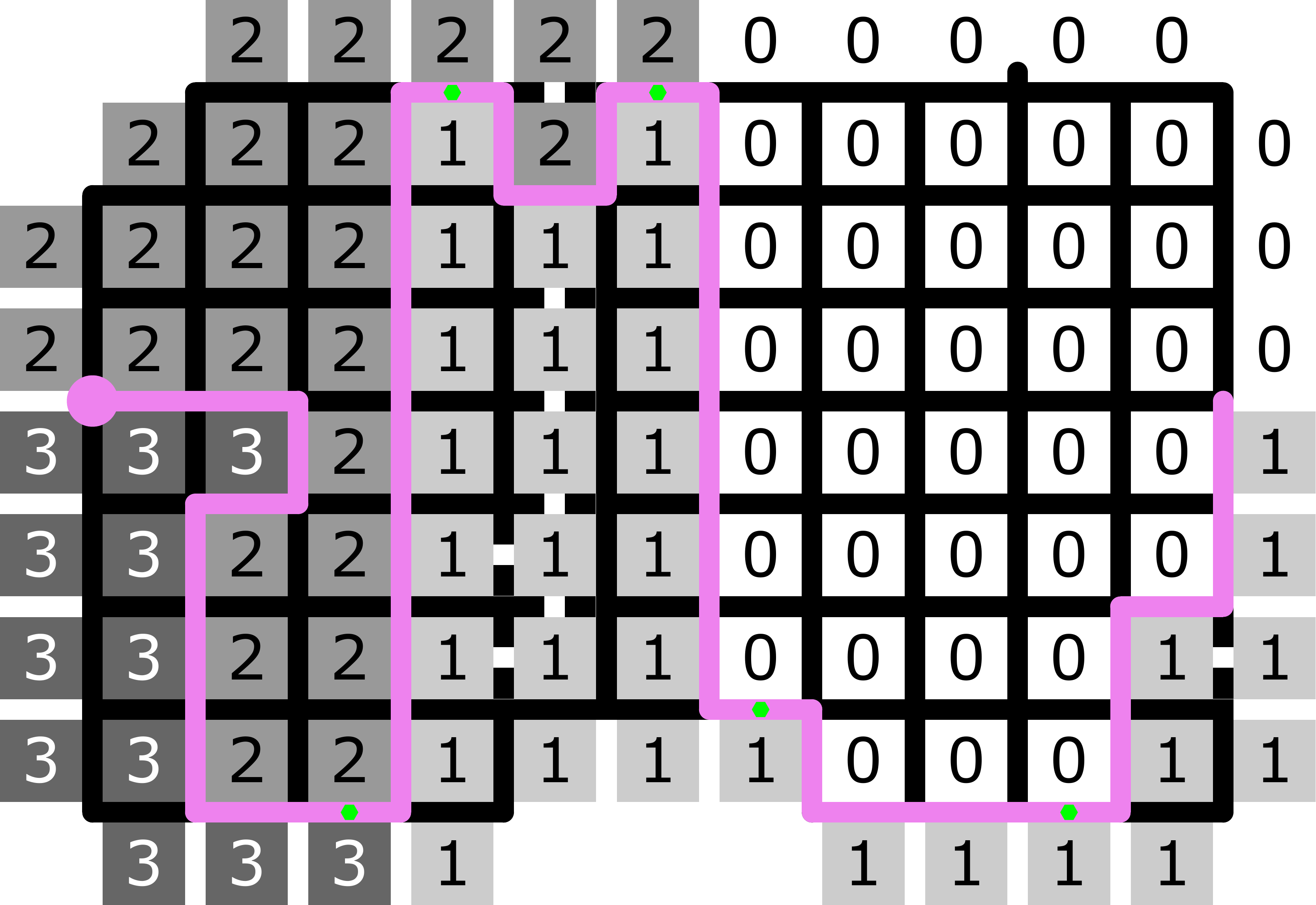}}

  \caption{Example remainder computation, where $F$ is represented by hexagons, $u$ is the start circle, $v$ is the other endpoint of the paths, and $t$ is the end cap.  Maximizing the remainder leaves the maximum freedom for completing the path.  (Note that the two edges in each reflex corner have two exterior depths; they happen to be the same in this example, but they may differ in general.)}
  \label{fig:remainder}
\end{figure}

Let $P_1$ and $P_2$ be any two $(C,F)$-paths from $u$ to $v$
that do not visit the outline of $C$ in the clockwise interval $(v,t]$
(or any larger forbidden interval).
We will construct another such path $P_3$ (built entirely from edges of $P_1$
and~$P_2$) whose $t$-remainder is a superset of the $t$-remainders of
$P_1$ and~$P_2$.  Refer to Figure~\ref{fig:remainder} for an example.
By applying this construction repeatedly to all such paths,
we obtain a single such path whose $t$-remainder is a superset of the
$t$-remainder of all such paths.

Define the \emph{$P_i$-depth} of a cell $c \in C$ to be the minimum possible
number of crossings of $P_i$ by a curve within $C$
from a point interior to $c$ to~$t$.
(This number is most easily seen to be well-defined by allowing the curve
to cross cell boundaries only at edges, not vertices.)
Define the \emph{depth} of a cell $c \in C$ to be the minimum of the
$P_1$-depth and $P_2$-depth of~$c$.
%, and the depth of every face outside of $C$
%(i.e., all puzzle cells outside $C$ and the outside face of the entire puzzle)
%to be~$0$.

%--- The following works only for $x$ interior to $C$; can differ by up to
%    $\pm 2$ for outline of $C$.  Luckily, we don't seem to need it!
%We claim that two cells $c_1,c_2 \in C$ incident to a common vertex $x$
%have depths that differ by at most $\pm 1$.  Consider a curve $Q$
%from $c_1$ to $c_2$ going around and staying arbitrarily close to $x$
%while avoiding exiting~$C$ (which is possible because $C$ is simply connected,
%so its outline has no pinch points).
%Because each $P_i$ is a simple path, $P_i$ visits $x$ at most once,
%so $Q$ crosses $P_i$ at most twice.
%% can't appeal to a curve Q' in the opposite direction of Q, because it may not exist in C if x is on C's outline.
%If $Q$ crosses $P_i$ twice, the second crossing returns to the first side of $P_i$, so $c_1$ and $c_2$ have equal $P_i$-depth.
%Otherwise, $Q$ crosses $P_i$ at most once,
%so the $P_i$-depths of $c_1$ and $c_2$ can differ by at most~$\pm 1$,
%and thus the depths of $c_1$ and $c_2$ also differ by at most~$\pm 1$.

Define a \emph{$k$-cavity} to be a connected set of cells of depth $\geq k$
(with at least one cell of depth exactly~$k$)
``surrounded'' by cells of depth $< k$:
every edge on the outline of the cavity must have an incident
exterior cell that has depth $< k$.
The outline of a cavity therefore shares no edges with the outline of~$C$.
(On the other hand, this definition allows the outline of a cavity to share a
vertex with the outline of~$C$, though we will argue later that this is
impossible.)
Furthermore, by this definition, every cavity is a \textit{maximal}
connected set of cells of depth $\geq k$.

Define the \emph{filled depth} of cells $c \in C$
to be the result of the following process.
Start with filled depth equal to depth.
If there is a $k$-cavity with the current notion of (filled) depth,
set the filled depth of all cells in the cavity to $k-1$.
Repeat until there are no more cavities.

%--- This fails in the same way as the commented-out $\pm 1$ claim above.
%We claim that two cells $c_1,c_2 \in C$ incident to a common vertex $x$
%have filled depths that differ by at most $\pm 1$.
%We prove that the previous claim is invariant over a sequence of
%cavity filling.  Before we fill a $k$-cavity, every edge on the outline of a
%$k$-cavity
%% used to say "but not the outline of C" here, but that's always true anyway
%has depth $k$ on its inside and
%depth $k-1$ on its outside; after we fill the cavity, every edge has
%filled depth $k-1$ on both sides.
%% maybe obvious: everything entirely interior to a cavity continues to share a depth (everyone changed so difference didn't), and everything entirely exterior didn't change
%Thus, cavity filling can only decrease the set of depths among cells incident to $x$, so difference by at most $\pm 1$ is preserved.

\begin{figure}
  \centering
  \subcaptionbox{\label{fig:remainder kkkk interior} Alternating filled depths $< k, \geq k, < k, \geq k$.}
    {\includegraphics[scale=0.333]{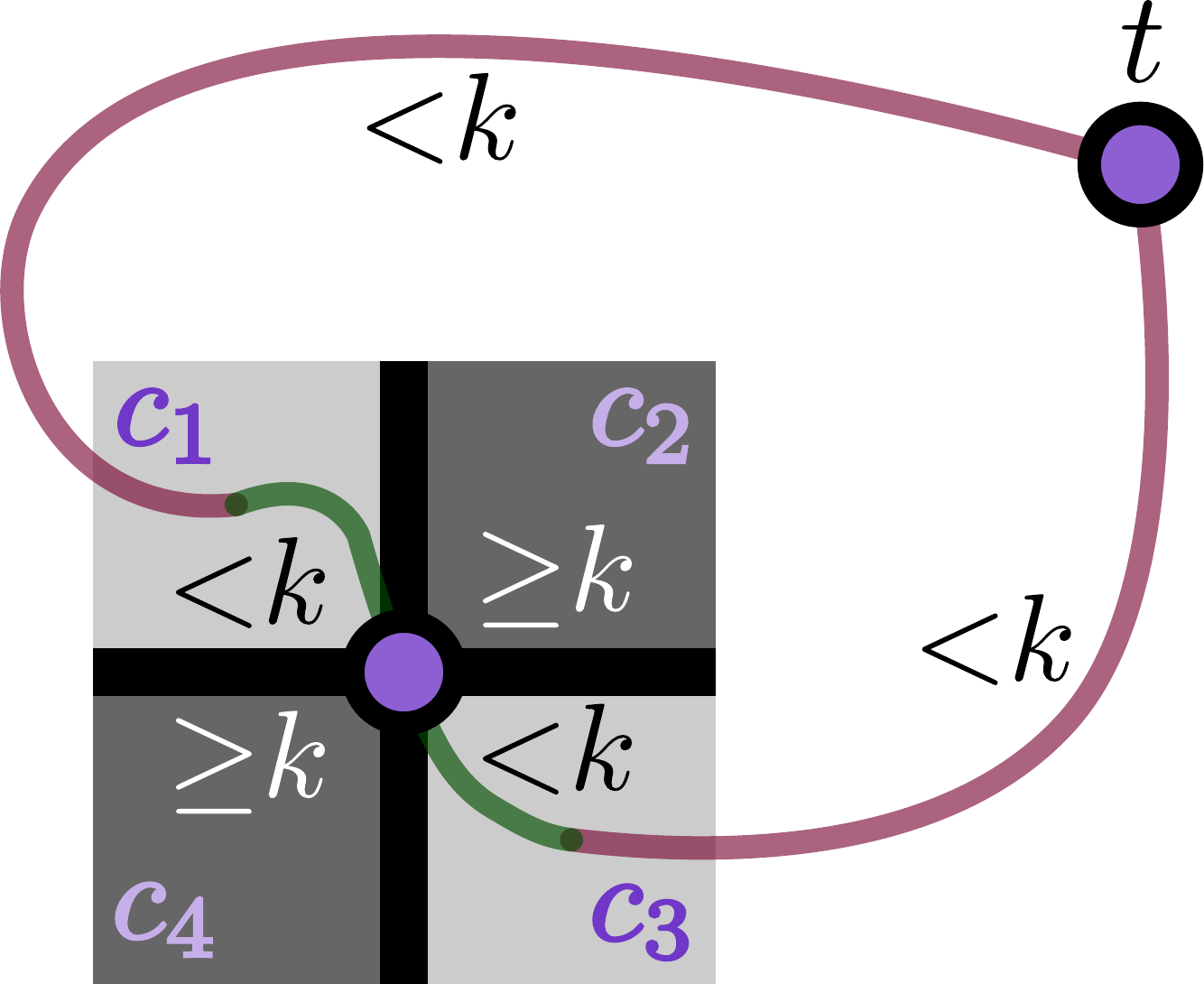}}
  \qquad
  \subcaptionbox{\label{fig:remainder kkkk boundary} Alternating filled depths $< k, \geq k, < k$ on the boundary.}
    {\includegraphics[scale=0.333]{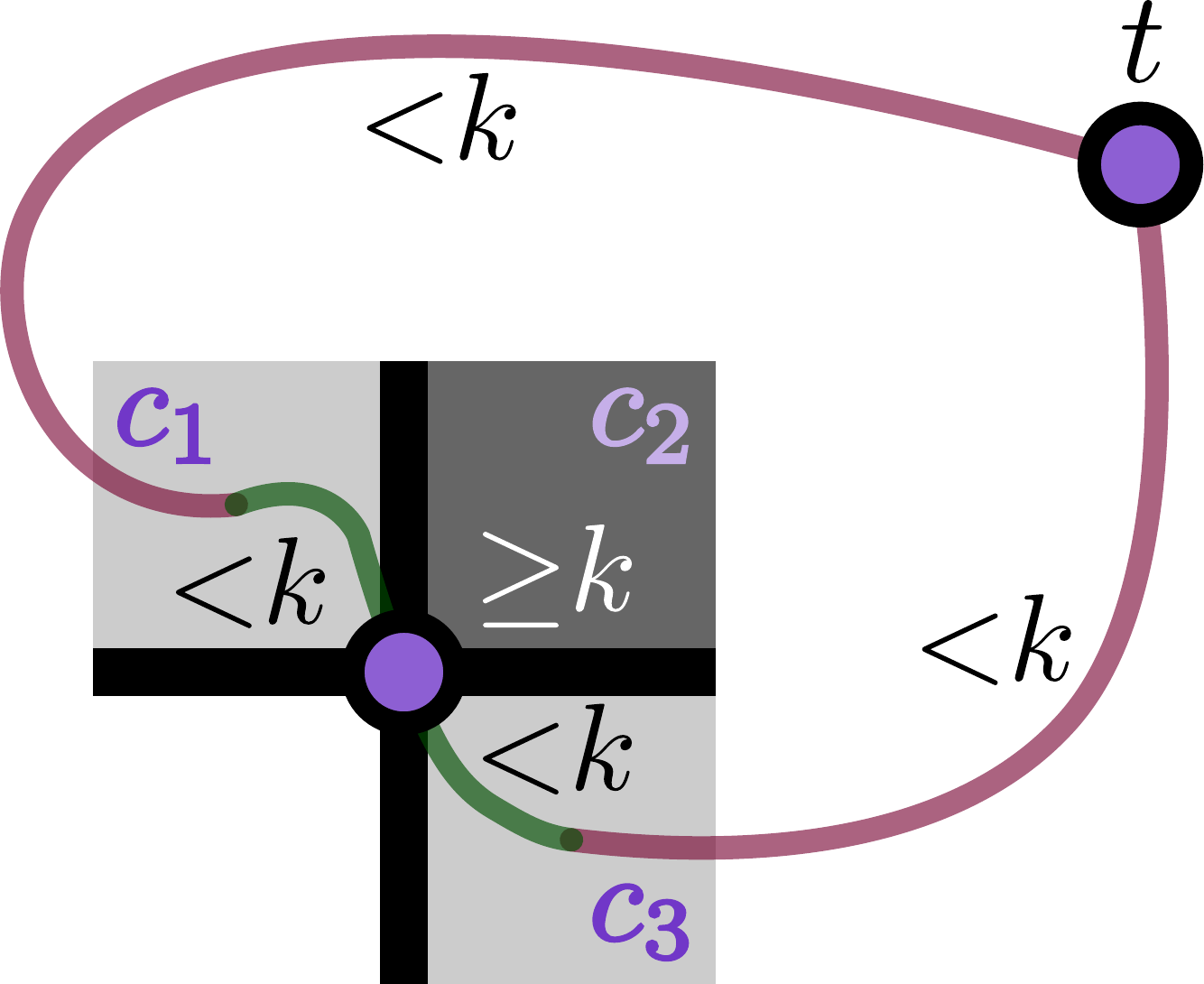}}
  \caption{Two impossible situations in the proof of
    Lemma~\ref{thm:monominoes-canonical}: alternating filled depths.}
  \label{fig:remainder kkkk}
\end{figure}

\xxx{Numbering claims in claim environment would be nice}

Next we claim that we never have cells $c_1, c_2, c_3, c_4$ in clockwise order
around a vertex $x$ with filled depths of $< k, \geq k, < k, \geq k$
respectively; refer to Figure~\ref{fig:remainder kkkk interior}.%
\footnote{For simplicity, we assume here that every vertex has degree
  at most $4$, as in The Witness, though this argument can easily be
  generalized to arbitrary planar graphs:
  replace each $c_i$ with an interval of cells around~$x$,
  so that $c_i$ is indeed adjacent to $c_{i+1}$,
  and generalize to $\geq 4$ alternations.}
%
%First we reduce the filled-depth claim to the depth claim.
%Because cavity filling only decreases depths, cells $c_2$ and $c_4$ also
%have (unfilled) depths $\geq k$.
%We further claim that cells $c_1$ and $c_3$ were not filled, i.e.,
%have depth equal to filled depth which is $< k$.
%Otherwise, the last cavity filling involving $c_1$ or $c_3$
%must have been a $\leq k$-cavity containing $c_1$ or $c_3$ but neither $c_2$ nor
%$c_4$, but $c_2$ and $c_4$ have always had filled depth $\geq k$ and one of them
%surrounds a cell of the cavity ($c_1$ or $c_3$), contradicting the definition
%of $\leq k$-cavity.
First, neither $c_1$ nor $c_3$ belongs to a filled cavity, or else
$c_2$ and $c_4$ would belong to the same cavity,
so filling this cavity would place $c_2$ and $c_4$ at the same depth as
either $c_1$ or~$c_3$.
Thus, $c_1$ and $c_3$ have depths equal to filled depths which are $< k$.
By definition of depth, for $i \in \{1,3\}$,
we can draw a simple curve from $c_i$ to $t$
that remains in depth $< k$.
Because filling only decreases depths,
these simple curves also remain in filled depth $< k$.
Connecting these two curves at $x$, we obtain a simple closed curve
through $x$ and $t$ that traverses cells only of filled depth $< k$.
By planarity, this closed curve must contain either $c_2$ or $c_4$
of filled depth $\geq k$, so there must in fact be a $k$-cavity that is not yet
filled, a contradiction, proving the claim.

Similarly, we claim that we never have a vertex $x$ on the outline of $C$
with incident cells $c_1, c_2, c_3$ in clockwise order interior to $C$
with filled depths of $< k, \geq k, < k$ respectively;
refer to Figure~\ref{fig:remainder kkkk boundary}.
Otherwise, as above, for $i \in \{1,3\}$, $c_i$ was not filled,
and we can draw a simple curve from $c_i$
to $t$ that remains in depth $< k$; connect these curves
into a simple closed curve through $x$ and $t$ that traverses cells only of
filled depth $< k$; and by planarity this cycle must contain $c_2$
of filled depth~$k$, forming an unfilled $k$-cavity and a contradiction.

As a consequence of the previous claim,
no $k$-cavity can touch a vertex $x$ of the outline of~$C$.
Furthermore, the previous claim holds for $P_i$-depths as well as filled depths:
%because $c_1,c_2,c_3$ are incident to the outline of $C$ so cannot have been
%filled, so
by the same argument, we get a simple closed curve (through $x$ and $t$)
that traverses cells only of $P_i$-depth $< k$, yet containing a cell of
$P_i$-depth $\geq k$, contradicting that $P_i$ is a simple path.

For an edge $e$ on the outline of~$C$, define the
\emph{interior depth} of $e$ to be the depth of the incident cell of~$C$; the
\emph{exterior $P_i$-depth} of $e$ to be the minimum possible number of crossings
of $P_i$ by a curve that starts just outside $e$, immediately crosses $e$, and
continues within $C$ to~$t$; and
the \emph{exterior depth} of $e$ to be the minimum of the exterior
$P_1$-depth and exterior $P_2$-depth of $e$.
Note that the exterior depth of $e$ is always at least the interior depth
of~$e$.
(These notions do not need to distinguish between depth and filled depth,
because we never fill a cell incident to an edge on the outline of~$C$.)

We claim that the exterior $P_1$-depth and exterior $P_2$-depth of $e$
have the same parity.
By the definition of exterior $P_i$-depth,
we can draw a simple curve starting at a point $q$ just outside~$C$,
immediately crossing~$e$, and continuing within $C$ to~$t$,
crossing $P_i$ that many times.
We can close these curves without adding $P_i$ crossings
by adding a simple curve exterior to $C$ from $q$ to~$t$.
The resulting simple closed curves for both $P_1$ and $P_2$
enclose the same interval of the outline of $C$,
so $u$ is either inside or outside of both closed curves, and similarly for $v$.
If $u$ and $v$ are on the same side of the closed curves,
then $P_i$ from $u$ to $v$ must cross the closed curve an even number of times,
so $e$ has even exterior $P_i$-depth;
while if $u$ and $v$ are on different sides, $e$ has odd exterior $P_i$-depth.
In either case, the parities are the same for both~$i$.

We claim that two consecutive edges $e_1,e_2$ of the outline of $C$
have the opposite exterior depth parity if and only if
their common endpoint is $u$ or $v$.
Consider the curve proceeding around the common endpoint $x$
from just outside $e_1$, immediately crossing $e_1$,
crossing all cells incident to~$x$, and crossing $e_2$ to just outside~$e_2$.
This curve crosses $P_i$ zero or two times
unless $x$ is one of $P_i$'s endpoints, in which case it crosses exactly once.
Thus, the exterior $P_i$-depths of $e_1$ and $e_2$ have the same parity
unless $x$ is $u$ or $v$, in which case they have opposite parity.
The exterior depths of $e_1$ and $e_2$ therefore satisfy the same parity
relationship.
%By the definition of depth (and noting that both or neither $P_i$ ends at~$x$),
%$e_1$'s and $e_2$'s exterior depths have the same parity relation.
% Note: use depth here, not filled depth.

Define \emph{ledges} as follows: an edge $e$ on the outline of $C$ is a ledge
if it has different interior and exterior depths, and an
edge $e$ interior to $C$ is a ledge if its two incident cells have
different filled depths.
Filled depth changes exactly at ledges, by $\pm 1$, so in particular
every ledge is an edge of $P_1$ or $P_2$ or both.
%and every shared edge of $P_1$ and $P_2$ is a ledge. -- unclear if filled

We claim that, at every vertex of $C$ except $u$ or $v$, the number of incident
ledges is $0$ or $2$; and at $u$ and $v$, the number of incident ledges is~$1$.
By the $< k,\geq k,< k,\geq k$ and $< k, \geq k, < k$ claims,
every vertex has at most two incident ledges.
For a vertex interior to~$C$, we must therefore have zero or two incident ledges:
there cannot be just one change in a cycle of numbers.
A vertex $x$ on the outline of $C$ might have zero, one, or two incident ledge.
Let $e_1$ and $e_2$ be the two (consecutive) edges of the outline of $C$
incident to~$x$.
As in the previous claim, consider the curve proceeding around $x$
from just outside $e_1$, immediately crossing $e_1$,
crossing all cells incident to~$x$, and crossing $e_2$ to just outside~$e_2$.
As we traverse this curve, the filled/exterior depth changes by $\pm 1$
at ledge edges, and at no other edges (by definition of ledge).
Thus, $x$ has exactly one incident ledge if and only if
the exterior depths of $e_1$ and $e_2$ have opposite parity,
which by the previous claim must happen exactly when $x \in \{u,v\}$;
otherwise, $x$ has zero or two incident ledges as desired.

By the previous claim, the set of ledges forms a path $P_3$ between $u$ and $v$
plus zero or more disjoint simple cycles.
We claim that, in fact, there can be no cycles of ledges.
Suppose for contradiction that there were such a cycle~$X$.
There are two cases:
\begin{enumerate}
\item
If $X$ touches the outline of $C$ at just one vertex (which cannot be $t$
by its degree bound) or not at all,
then $X$ is surrounded by cells of a constant filled depth~$k$
(because filled depth changes exactly at ledges), and the filled depths
of cells interior and incident to $X$ must have filled depth
either $< k$ or $> k$.
In fact, no cells interior to $X$ can have filled depth $< k$,
as they would have a curve to $t$ (which is exterior to $X$)
that only visits cells with filled depth $< k$,
contradicting the filled depth of the cells surrounding~$X$.
But if all cells interior to $X$ have filled depth $> k$,
then $X$ is the outline of a $> k$-cavity,
contradicting that we already filled all cavities.
%Therefore $X$ touches the outline of $C$ at two or more vertices.

\item
If $X$ touches the outline of $C$ at two or more vertices then,
because $X$ is disjoint from $P_3$, there is a subpath $Y$ of $X$
whose endpoints $y_1,y_2$ lie on the outline of $C$ separating the rest of $X$
from $P_3$ and thus $u$ and~$v$.
Refer to Figure~\ref{fig:remainder case 2}.
Label $y_1$ and $y_2$ to be closer to $u$ and $v$ respectively
on the outline of~$C$, i.e.,
so that their clockwise order is either $y_1,u,v,y_2$ or $u,y_1,y_2,v$.
Because $P_1$ and $P_2$ and thus all ledges are not on
the outline of $C$ in the clockwise interval $(v,t]$,
$t$~is on the same side of $Y$ as~$v$.
Therefore $Y$ is a ledge transition from some filled depth $k$
on the $u,v,t$ side to filled depth $k+1$ on the $X$ side
(locally on either side of~$Y$).
%as filled depth changes by at most $1$ at any edge.
Indeed, the entire $X$ side of $Y$ has filled/exterior depth $> k$;
otherwise, there would be a curve to $t$ (and thus crossing~$Y$)
of filled depths $\leq k$.
Therefore, the other (nonempty) path of~$X$, $X \setminus Y$,
must be a ledge transition from filled depth $k+1$ to $k+2$, so there
must be a filled/exterior depth of $k+2$ on the $X$ side of~$Y$.

\begin{figure}[htbp]
  \centering
  \includegraphics{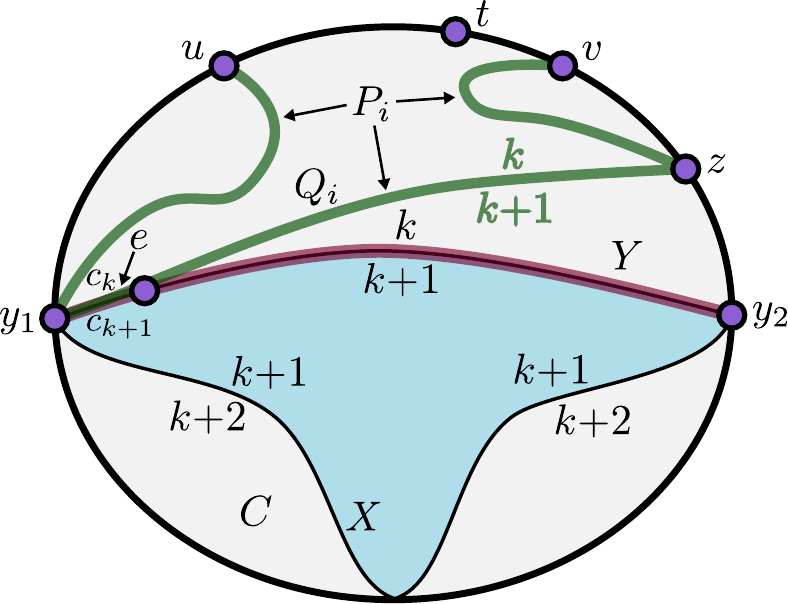}
  \caption{Case 2 of the proof of Lemma~\ref{thm:monominoes-canonical}.}
  \label{fig:remainder case 2}
\end{figure}

Now consider the transition from filled depth $k$ to $k+1$ on the edge $e$
of $Y$ incident to~$y_1$, and let $c_k$ and $c_{k+1}$ be the cells of $C$
incident to $e$ of filled depths $k$ and $k+1$
respectively (which are exterior and interior respectively to~$X$).
In fact, $c_k$ and $c_{k+1}$ have depth $k$ and $k+1$ respectively:
cell $c_k$ could not have been filled as it has a neighbor $c_{k+1}$
of higher depth, and cell $c_{k+1}$ could not have been filled
because that would have made its filled depth match $c_k$'s.
By definition of depth, for some $i \in \{1,2\}$,
the $P_i$-depth of $c_k$ is $k$ while
the $P_i$-depth of $c_{k+1}$ is $k+1$.
To achieve this transition, $P_i$ must have a subpath $Q_i$
between $y_1$ and some vertex $z$ on the outline of~$C$,
with $e$ as its first or last edge,
where one side of $Q_i$ locally has $P_i$-depth $k$ while
the other side of $Q_i$ locally has $P_i$-depth $k+1$.
Because the entire $X$ side of $Y$ has filled/exterior depth $> k$,
%the $X$ side of $Y$ must be entirely on the right side of~$Q_i$.
$Q_i$ cannot strictly enter the $X$ side of~$Y$,
so $Q_i$ must be nonstrictly on the $t$ side of~$Y$.
Viewed from the other side,
$X$ must be nonstrictly on the $k+1$ side of $Q_i$,
while $t$ must be strictly on the $k$ side of~$Q_i$.
By our labeling of $y_1$ and~$y_2$,
$z$~lies on the interval of the outline of $C$
between $t$ and~$y_2$ not containing $y_1$ (or~$u$),
so $u$ and $v$ lie on the same ($k$) side of $Q_i$ as~$t$.
Now beyond the endpoints $y_1,z$ of $Q_i$, the simple path $P_i$
must proceed on the $k$ side of $Q_i$ which contains $u$ and~$v$,
so $P_i$ never goes strictly on the $k+1$ side of~$Q_i$,
%and thus never goes strictly on the $X$ side of~$Y$.
which is nonstrictly on the $t$ side of~$Y$.
But this contradicts that there must be a transition from $P_i$-depth $k+1$
to $k+2$ on $X \setminus Y$, which is on the $X$ side of~$Y$.

\end{enumerate}
%
%We claim that there are no cycles of ledges.
%Any such cycle would have cells of filled depth $k_i$ in its interior and cells of filled depth $k_e$ on its exterior (with $k_i \neq k_e$ or there would not be ledges).  If $k_i < k_e$, the cycle encloses a cavity that we would have filled.  If $k_e > k_i$ (a ``mountain''),
%\xxx{no ``mountain'': filled depth $< k$ surrounded by $\geq k$:
%  find someone of original depth $< k$ (any boundary cell?),
%  old curve of depth $< k$ to $t$ will be preserved}
%
%We claim that the ledges form a connected graph.  Even stronger,
%we claim that any simple curve (not meeting any vertices of~$C$)
%that connects two points on the outline of $C$ and separates $u$ from~$v$
%crosses at least one ledge.
%
Therefore the set of ledges forms exactly a path $P_3$ between $u$ and~$v$.

We claim that $P_3$ visits all forced vertices and edges in~$F$.
First, $P_3$ visits every forced edge because such an edge is in both
$P_1$ and $P_2$ and thus is a ledge via differing internal and external depths,
which are unaffected by cavity filling.
Second, we claim that $P_3$ visits every vertex $x$ on the outline of $C$
that is visited by both $P_1$ and $P_2$, and thus every forced vertex.
Specifically, we show that $x$ is incident to two different depths
(cell or exterior), which implies (because cavities cannot touch a vertex $x$
on the outline of~$C$) that $x$ is incident to two different filled depths,
and thus incident to a ledge.
Suppose for contradiction that
every depth incident to $x$ is equal to the same~$k$, and thus
every incident $P_i$-depth is $\geq k$ and every cell or edge exterior
has either $P_1$-depth or $P_2$-depth equal to~$k$;
refer to Figure~\ref{fig:remainder forced vertex}.
\begin{figure}
  \centering
  \includegraphics{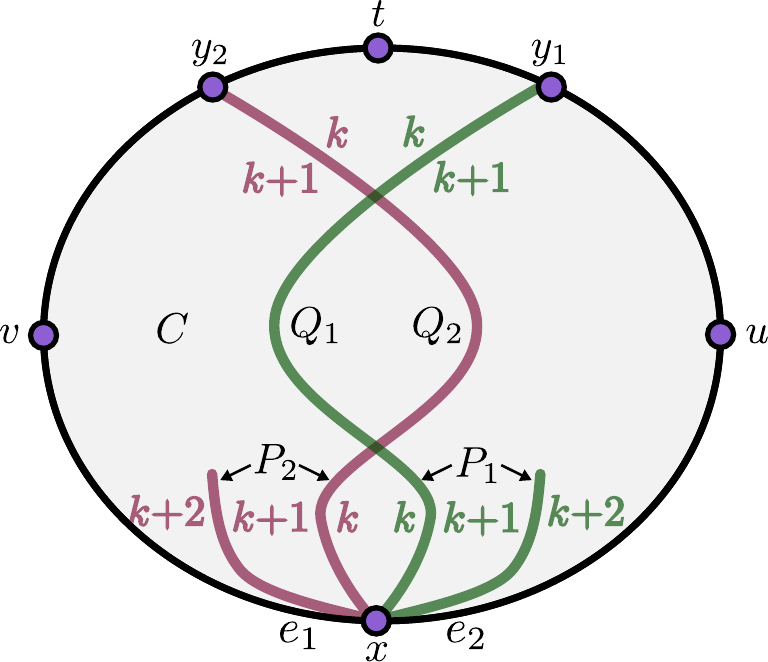}
  \caption{An impossible situation in the proof of
    Lemma~\ref{thm:monominoes-canonical}: vertex $x$ visited by both
    paths $P_1$ and $P_2$ incident to only a single depth~$k$.}
  \label{fig:remainder forced vertex}
\end{figure}
Because the $P_i$-depth changes by $\pm 1$ across each edge of~$P_i$,
$x$ must be incident to a $P_i$-depth of $> k$ as well.
Let $e_1$ and $e_2$ be the two edges of the outline of $C$ incident to~$x$
in counterclockwise order around the outline of~$C$, so that the edges
incident to $x$ proceed clockwise from $e_1$ to~$e_2$.
Assume by symmetry that the exterior $P_1$-depth of $e_1$ is~$k$.
By the $< k, \geq k, < k$ claim applied to $P_1$-depth at~$x$,
the exterior $P_1$-depth of $e_2$ cannot be~$k$; otherwise, there would be
a $P_1$-depth of $> k$ in between $e_1$ and~$e_2$ of exterior $P_1$-depths~$k$.
Thus, the $P_1$-depths clockwise around $x$ must proceed $k, k+1, k+2$.
The exterior depth of $e_2$ is~$k$,
while the exterior $P_1$-depth of $e_2$ is $k+2$,
so the exterior $P_2$-depth of $e_2$ must be~$k$.
By a symmetric argument,
the $P_2$-depths clockwise around $x$ must proceed $k+2, k+1, k$.
Because there are no incident depths $> k$,
the transition of $P_2$-depth from $k+2$ to $k$ must occur fully before
(counterclockwise of) any of the transition of $P_1$-depth from $k$ to $k+2$.
%Hence, there is a cell $c$ incident to $x$
%having $P_1$-depth $k$ and $P_2$-depth $k$.
Because of the transition from $P_i$-depth $k$ to $P_i$-depth $k+1$,
each $P_i$ has a subpath $Q_i$ between $x$ and a point $y_i$ on
the outline of $C$, without visiting any vertices on the outline of $C$
in between, where $Q_i$ has $P_i$-depth $k$
locally on one side and $P_i$-depth $k+1$ locally on the other side.
Because of the transition from $P_i$-depth $k+1$ to $P_i$-depth $k+2$,
locally at~$x$, $P_i$ proceeds strictly to the $k+1$ side of~$Q_i$,
and by planarity and simplicity of $P_i$, one endpoint of $P_i$ ($u$ or~$v$)
must be strictly on the $k+1$ side of~$P_i$.
Because $u,v,t$ appear in clockwise order on the outline of $C$,
$u$ must be strictly on the $k+1$ side of $Q_1$
while $v$ must be strictly on the $k+1$ side of~$Q_2$.
Because $t$ is on the $k$ side of both $Q_1$ and $Q_2$, we must have
$t, y_1, u, x, v, y_2, t$ appearing in clockwise order on the outline of~$C$.
Because $P_2$ does not visit the clockwise interval $(v,t]$,
we must have $v = y_2$.
But then $P_2$ consists of $Q_2$ preceded by a path strictly on the
$k+1$ side of~$Q_2$, so $P_2$ cannot reach a vertex on the outline of $C$
on the $k+1$ side of~$Q_1$, so in particular cannot reach~$u$,
a contradiction.

Now that $P_3$ is a path from $u$ to $v$ that visits all forced vertices and
edges, we just need to check a few more properties.
The remainder of $P_3$ contains all cells of filled depth $0$
because such a cell can be connected to $t$ via a curve
that does not cross any ledges; therefore, the remainder of $P_3$
includes all cells of depth $0$, and thus all cells of $P_i$-depth $0$,
which is the remainder of~$P_i$.
Every edge in path $P_3$ is a ledge and thus an edge of $P_1$ or $P_2$,
so $P_3$ uses no broken edges and
shares the property of not visiting the outline of $C$ in the
clockwise interval $(v,t]$ (or any larger forbidden interval).
Therefore, $P_3$ is the path we were searching for.
\end{proof}

Now we give our dynamic programming algorithm for generalized Witness
puzzles with boundary hexagons and broken edges.

\begin{theorem} \label{thm:monominoes-dp}
% Writing note: trying to avoid 'region' and 'boundary' as they have preexisting meanings.
%Given a maximal connected component $C$ of empty cells (possibly with broken edges), with some of its edge outline marked as needing traversal,
Given a forced division $(C,F)$ with possibly broken edges,
for any two vertices $s$ and $t$ on the outline of $C$, we can decide in polynomial time whether there exists a $(C,F)$-path from $s$ to $t$.
\end{theorem}

\begin{proof}
%Let $f_1, f_2, \dots, f_\ell$ be the marked edges of the outline of $C$
%in clockwise order starting from~$t$.
%
We use dynamic programming to solve this problem.
For every clockwise interval $[a,b]$ of the outline of $C$
containing $s$ and not strictly containing~$t$ (i.e., $t \notin (a,b)$),
we define two subproblems based on Lemma~\ref{thm:monominoes-canonical}:
\begin{itemize}
\item find a maximum-remainder $(C,F \cap [a,b])$-path from $s$ to $b$
      that does not visit the outline of $C$ in the clockwise interval $(b,a)$,
      or report that no such path exists;
  and
\item find a maximum-remainder $(C,F \cap [a,b])$-path from $s$ to $a$
      that does not visit the outline of $C$ in the clockwise interval $(b,a)$,
      or report that no such path exists.
\end{itemize}
%In fact, we will only consider a restricted subset of subproblems where
%$a$ is either $s$ or $t$ or the clockwise endpoint of some $f_i$,
%and $b$ is either $s$ or $t$ or the counterclockwise endpoint of some $f_j$.
%\xxx{``near'' endpoint would be better term: nearest to $s$, cutting at $t$}
%%(where $i < j$ when $a,b \neq s$).
%We will denote these intervals by
%$[f_i,f_j]$ or $[s,f_j]$ or $[f_i,s]$ or $[s,s]$.
To solve the original problem, we apply either subproblem on the clockwise
interval $[t,t]$
(understood to mean the entire outline of~$C$, starting and ending at~$t$).
In this case, $a=b=t$ and the clockwise interval $(b,a) = \emptyset$,
and the remainder is undefined, so the subproblem statement is exactly to
find a $(C,F)$-path from $s$ to~$t$.
%By Lemma~\ref{thm:monominoes-canonical}, if there is a $(C,F \cap [a,b])$-path
%from $s$ to~$t$, then there is a path ...

To find a maximum-remainder $(C,F \cap [a,b])$-path $P^*$ from $s$ to $b$
not visiting $(b,a)$, if such a path exists, we guess (try all options for)
the last vertex $x$ before $b$ of the outline of $C$ visited by the path~$P^*$.
Such a vertex $x \in [a,b)$ exists because $s$ is a candidate for $x$,
except in the base case $s=b$,
where we simply return the empty path from $s$ to~$s$.
Vertex $x$ splits $P^*$ into a path $P^*_1$ from $s$ to $x$
followed by a path $P^*_2$ from $x$ to~$b$.
The latter path $P^*_2$ visits the outline of $C$ only at its endpoints,
except when $P^*_2$ is a single edge $(b-1,b)$
where $b-1$ denotes the clockwise previous vertex
before $b$ on the outline of~$C$.
%Let $f_i, \dots, f_j$ be the forced edges in clockwise order
%in the interval $S = [a,b]$.
%Because vertex $x$ lies in the clockwise interval $[a,b)$ of the outline of~$C$,
We divide into two overlapping cases based on the relation between
$x \in [a,b)$ and $s$:
\begin{enumerate}
\item If $x$ is in the clockwise interval $[s,b)$,
  then there must not be any forced vertices/edges
  in the clockwise interval $(x,b)$,
  except possibly a single forced edge $(b-1,b)$ when $x=b-1$.
  Otherwise, $P^*_2$ could not include any such forced vertices/edges
  (as it cannot visit the outline of $C$ in that interval)
  and $P^*_2$ effectively ``hides'' such forced vertices/edges from $P^*_1$
  (as $s$ is outside the cycle formed by $P^*_2$ and $[x,b]$).
  Thus, $P^*_1$ can visit forced vertices/edges
  only in the interval $I = [a,x]$.
\item If $x$ is in the clockwise interval $[a,s]$,
  then there must not be any forced vertices/edges
  in the clockwise interval $(a,x)$.
  Otherwise, $P^*_2$ cannot include any such forced vertices/edges
  (as it cannot visit the outline of $C$ in that interval)
  and $P^*_2$ effectively ``hides'' such forced vertices/edges from $P^*_1$
  (as $s$ is outside the cycle formed by $P^*_2$ and $[b,x]$).
  Furthermore, the edge $(b-1,b)$ cannot be a forced edge unless $x=b-1=s$:
  it could not be visited by $P^*_2$
  (as it includes no edges of the outline of~$C$)
  nor by $P^*_1$ (or else $P^*$ would be a cycle, not a simple path).
  Thus, $P^*_1$ can visit forced edges only in the interval $I = [x,b-1]$.
\item If $x = s$, then the constraints of both above cases must hold:
  there must not be any forced vertices/edges in the clockwise interval $(a,b)$,
  except possibly a single forced vertex $s$ and a
  single forced edge $(b-1,b)$ when $x=b-1=s$.
  Furthermore, $P^*_1$ must be the trivial path from $s$ to~$s$, so
  it can visit forced vertices/edges only in the trivial interval $I = [s,s]$.
\end{enumerate}
If any of the stated constraints do not hold, then this choice of $x$ fails.

Otherwise, we recursively find a maximum-remainder $(C,F \cap I)$ path
$P_1$ from $s$ to $x$ that does not visit the outline of $C$
outside the interval $I$ defined above in each case.
Note that $I$ is strictly contained in $[a,b]$ in all cases, so the recursive
calls cannot form a cycle.
By a cut-and-paste argument, we can assume $P_1 = P^*_1$, i.e., that
the path $P_1$ we find matches the prefix $P^*_1$ of the desired path~$P^*$.
Otherwise, modify $P^*$ by replacing $P^*_1$ with $P_1$, resulting in
an equally suitable maximum-remainder $(C,F \cap [a,b])$-path $P^*$
that does not visit $(b,a)$.
First, the remainder of $P^*$ is the remainder of $P^*_2$, so remains unchanged.
Second, $P^*_2$ lies in the remainder of $P^*_1$, which is contained within
the (maximum) remainder of~$P_1$, so $P^*$ remains non-self-crossing.

Next we depth-first search using the left-hand rule
to find a maximum-remainder path $P_2$ from $x$ to $b$
while avoiding all broken edges, all vertices on the outline of~$C$,
and all vertices of~$P_1$.
If the depth-first search fails to find such a path,
then this choice of $x$ fails.
Otherwise, we claim that $P_2 = P^*_2$.
At each step, the depth-first search makes the leftmost choice
that can be completed into a path.  Hence,
if $P^*_2$ deviates from $P_2$ at any step,
$P^*_2$ must be to the right of $P_2$ at that step.
But that deviation is incident to a cell in $P$'s remainder
that is not in $P^*$'s remainder
(because $P_2$ does not visit the outline of $C$ until~$b$,
and the interval $(b,a)$ has not been visited by $P_1$ or $P_2$),
contradicting that $P^*$ is maximum-remainder.
Here we crucially use Lemma~\ref{thm:monominoes-canonical} that
there is a unique maximum remainder according to the subset relation,
not two possible remainders that are mutually incomparable.

Finally, we concatenate $P_1$ and $P_2$ to form a path $P$,
which is a candidate for $P^*$ that (as argued above)
equals $P^*$ assuming our guess for $x$ was correct.
By trying all choices for $x$, and returning the resulting path $P$
that has the maximum remainder, we are sure to find $P^*$ if it exists.

The other type of subproblem,
to find a maximum-remainder $(C,F \cap [a,b])$-path $P^*$ from $s$ to $a$
not visiting $(b,a)$, is symmetric.

Overall, if $C$ has $n$ cells and $m = O(n)$ edges on the outline,
then there are $O(m^2)$ subproblems, $O(m)$ choices per subproblem, and
the depth-first search spends $O(n)$ time per subproblem and choice.
Therefore, the total running time is $O(m^3 n) = O(n^4)$.
\end{proof}

\section{Squares}

\abstractlater{
  \section{Proofs: Squares}
  \label{appendix:squares}
}

Each square clue has a color and is placed on a cell of the puzzle.
Each region formed by the solution path and puzzle boundary must have at most
one color of squares.  In this section, we prove NP-hardness of square clues
of two colors.  This result is tight given the following:

\begin{observation} \label{obs:1 square}
  Witness puzzles containing only squares of one color
  do not constrain the path, so are always solvable.
\end{observation}

\def\TRVB{Tree-Residue Vertex Breaking}
\def\trvb{tree-residue vertex breaking}

\subsection{\TRVB}

Our reduction is from \emph{\trvb} \cite{TRVB}.
Define \emph{breaking} a vertex of degree $d$ to be the operation of replacing
that vertex with $d$ vertices, each of degree $1$, with the neighbors of the vertex becoming neighbors of these replacement vertices in a one-to-one way.
The input to the \trvb\ problem is a planar multigraph
in which each vertex is labeled as ``breakable'' or ``unbreakable''.
The goal is to determine whether there exists a subset of the breakable vertices
such that breaking those vertices (and no others) results in the graph
becoming a tree (i.e., destroying all cycles without losing connectivity).
This problem is NP-hard even if all vertices are degree-$4$ breakable
vertices \cite{TRVB}.

\subsection{Squares with Squares of Two Colors}

\begin{theorem} \label{thm:squares 2 color}
  It is NP-complete to solve Witness puzzles containing only squares of two colors.
\end{theorem}

Concurrent work~\cite{eurocg} also proves this theorem.
However, we prove this by showing that the stronger \emph{Restricted Squares
  Problem} is also hard, which will be useful to reduce from in Section~\ref{sec:stars}.

\begin{problem}[Restricted Squares Problem]
An instance of the \emph{Restricted Squares Problem} is a Witness puzzle
containing only squares of two colors (red and blue), where each cell in the
leftmost and rightmost columns, and each cell in the
  topmost or bottommost rows, contains a square clue; and of these square clues,
  exactly one is blue, and that square clue is not in a corner cell; and the
  start vertex and end cap are the two boundary vertices incident to that
  blue square; see Figure~\ref{squares boundary}.
  % 1 in coauthor: https://coauthor.csail.mit.edu/6.890/m/ugxMtkgm9wEiWWiWh
\end{problem}

\both{
\begin{theorem}
The Restricted Squares Problem is NP-complete.
\end{theorem}
}

\ifabstract
\begin{proofsketch}
We reduce from \trvb{} and construct gadgets for an unbreakable degree~3 vertex (Figure~\ref{squares-unbreakable})
and a breakable degree~4 vertex (Figure~\ref{squares-breakable}) out of squares. We force the solution
path to take an Euler tour of these gadgets, which can only be done if the underlying \trvb{} graph is a tree.
The full proof can be found in Appendix~\ref{appendix:squares}.
\end{proofsketch}
\fi

%\begin{figure}
%\centering
% casework about the black stars
%\subcaptionbox{Unsolved gadget.}{%
%  \includegraphics[width=3.5cm]{figures/squares_stars_unbreakable_degree_3.pdf}%
%  \label{squares unbreakable degree 3}%
%}~~~
%\subcaptionbox{The unique solution path.}{
%  \hspace{.5cm}%
%  \includegraphics[width=3.5cm]{figures/squares_stars_unbreakable_degree_3_solution.pdf}%
%  \hspace{.5cm}%
%  \label{squares unbreakable degree 3 solution}%
%}
%\caption{Unbreakable degree-3 vertex gadget}.
% \label{squares-unbreakable}
%\end{figure}
  
\begin{figure}
\centering
% casework about the black stars
\subcaptionbox{\label{squares breakable degree 4} Unsolved gadget}{%
  \includegraphics[width=.3\textwidth]{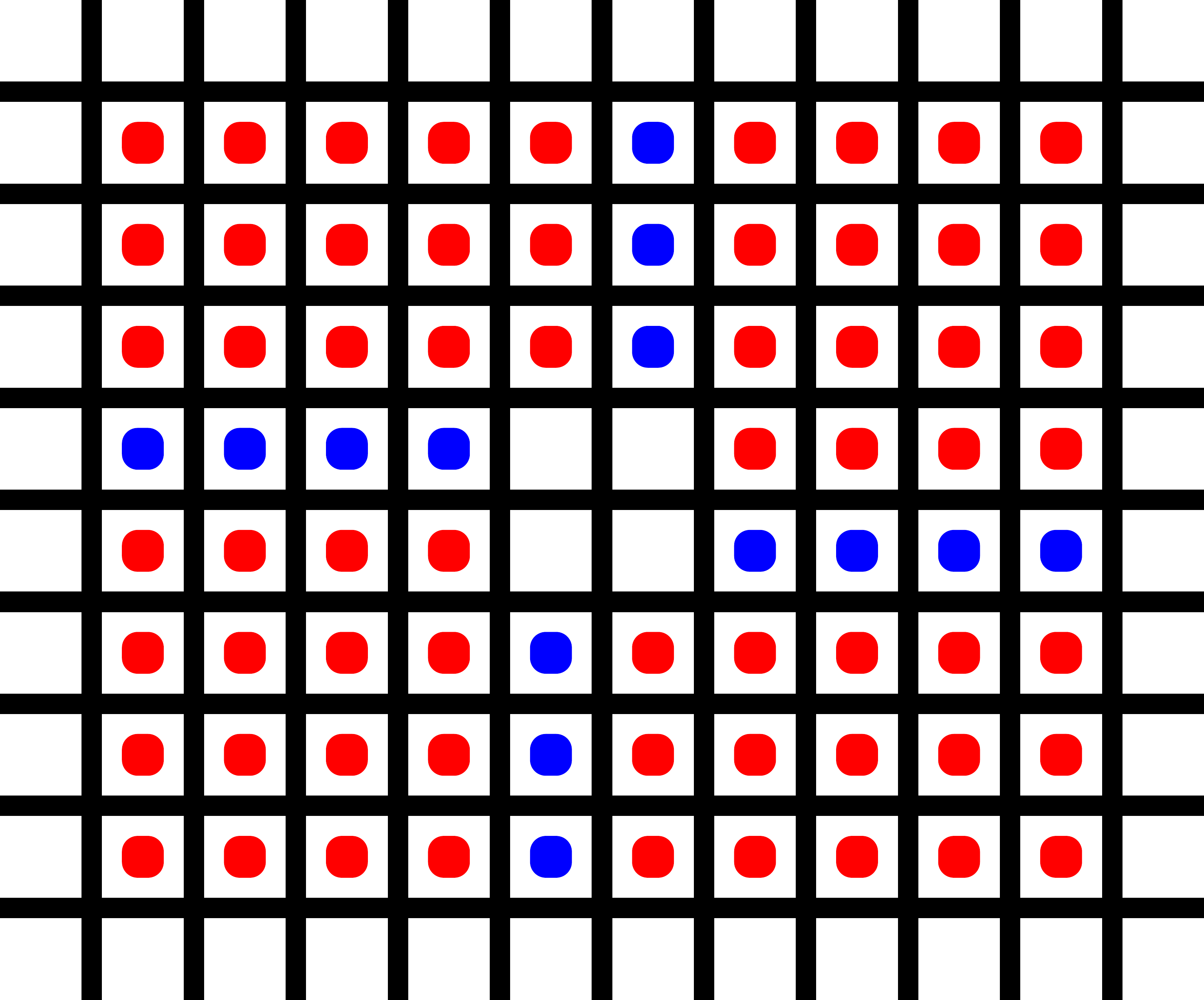}%
}\hfill
\subcaptionbox{\label{squares breakable degree 4 unbroken} Unbroken solution path.}{%
  \includegraphics[width=.3\textwidth]{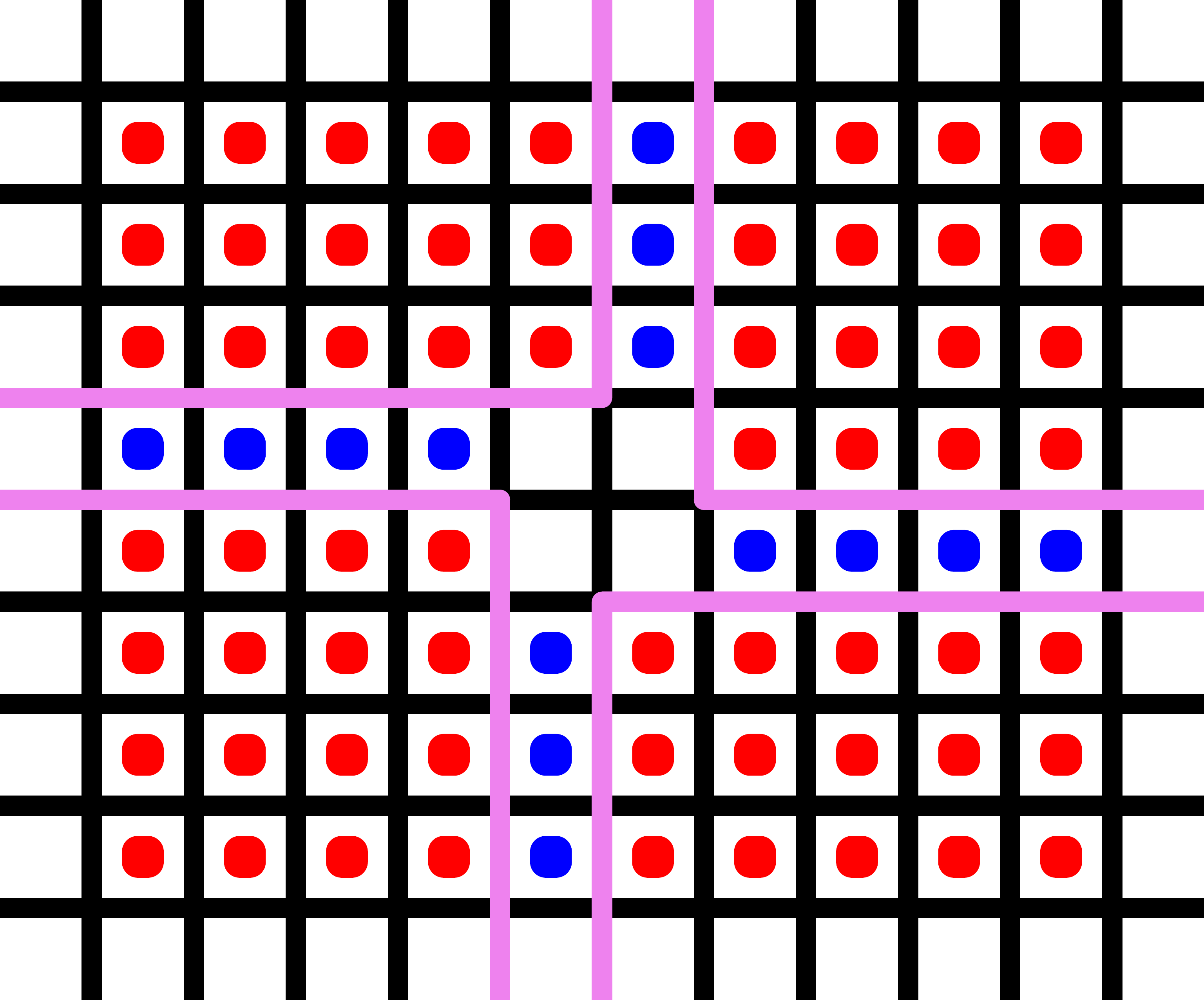}%
}\hfill
\subcaptionbox{\label{squares breakable degree 4 broken} Broken solution path.}{%
  \includegraphics[width=.3\textwidth]{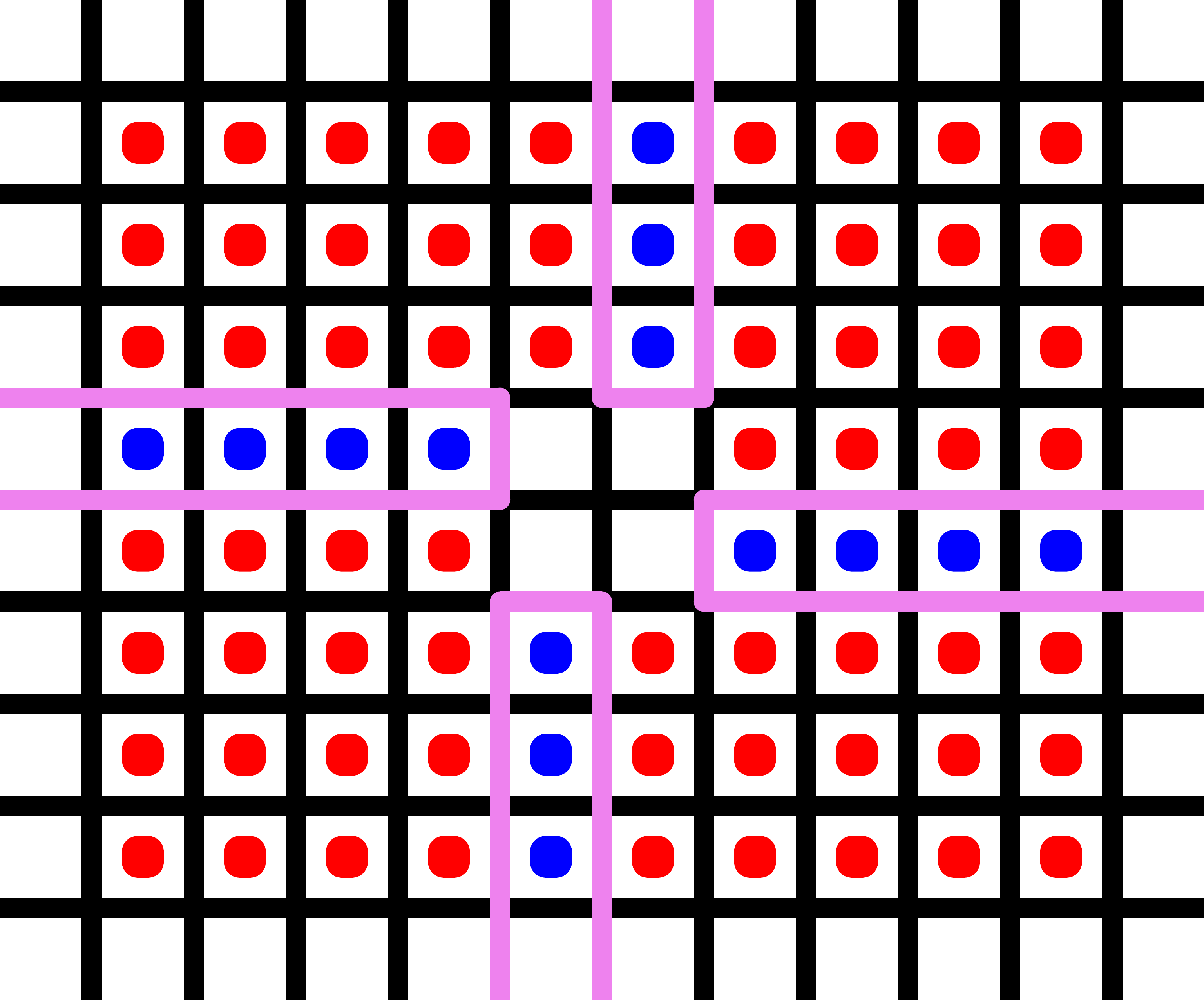}%
}
\caption{Breakable degree-4 vertex gadget.}
\label{squares-breakable}
\end{figure}

\begin{figure}
\centering
% casework about the black stars
\subcaptionbox{\label{squares wire} Edge gadget}{%
  \hspace{.5cm}%
  \includegraphics[width=.3\textwidth]{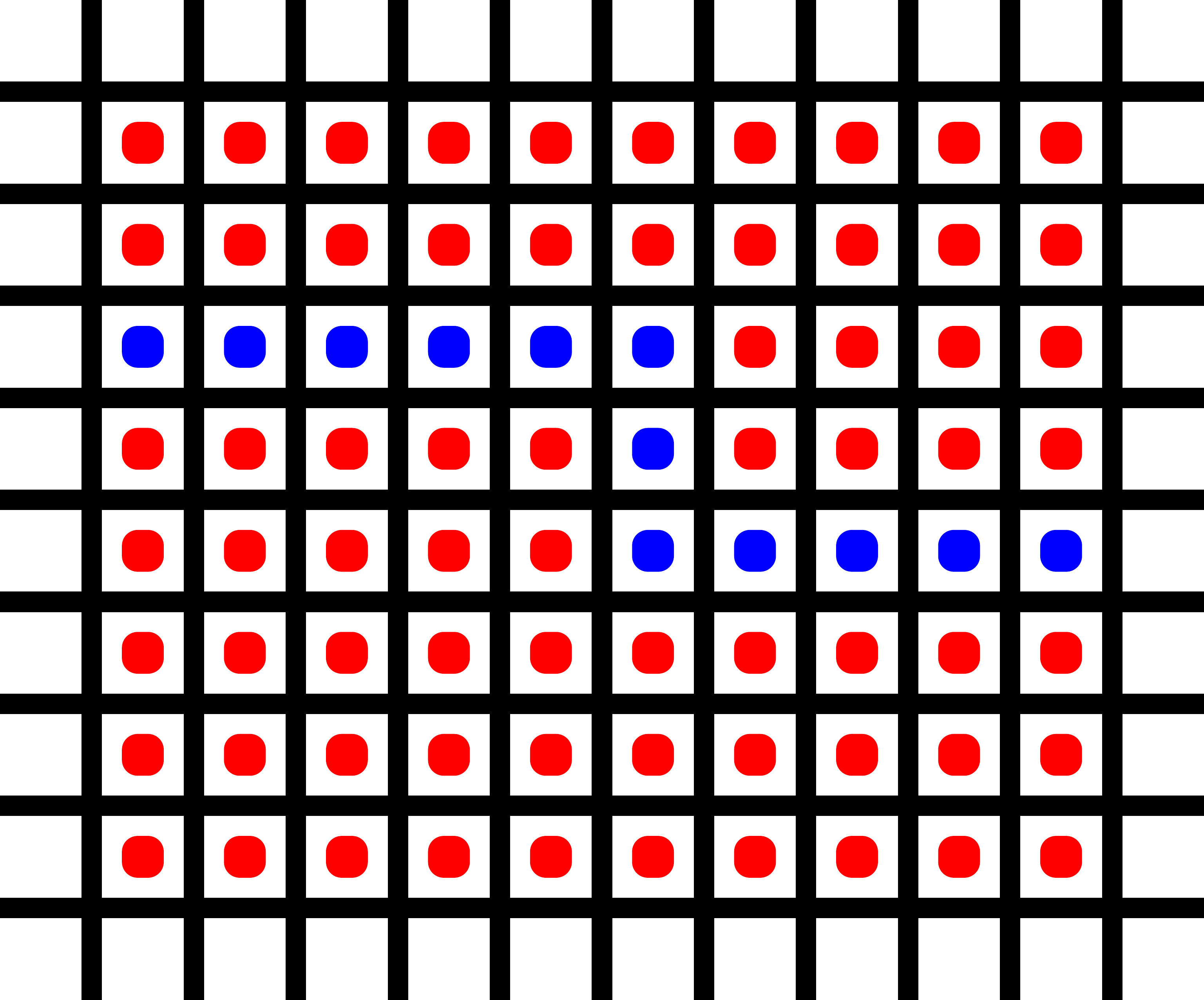}%
  \hspace{.5cm}%
}~~~
\subcaptionbox{\label{squares wire solution} Local solution}{%
  \hspace{.5cm}%
  \includegraphics[width=.3\textwidth]{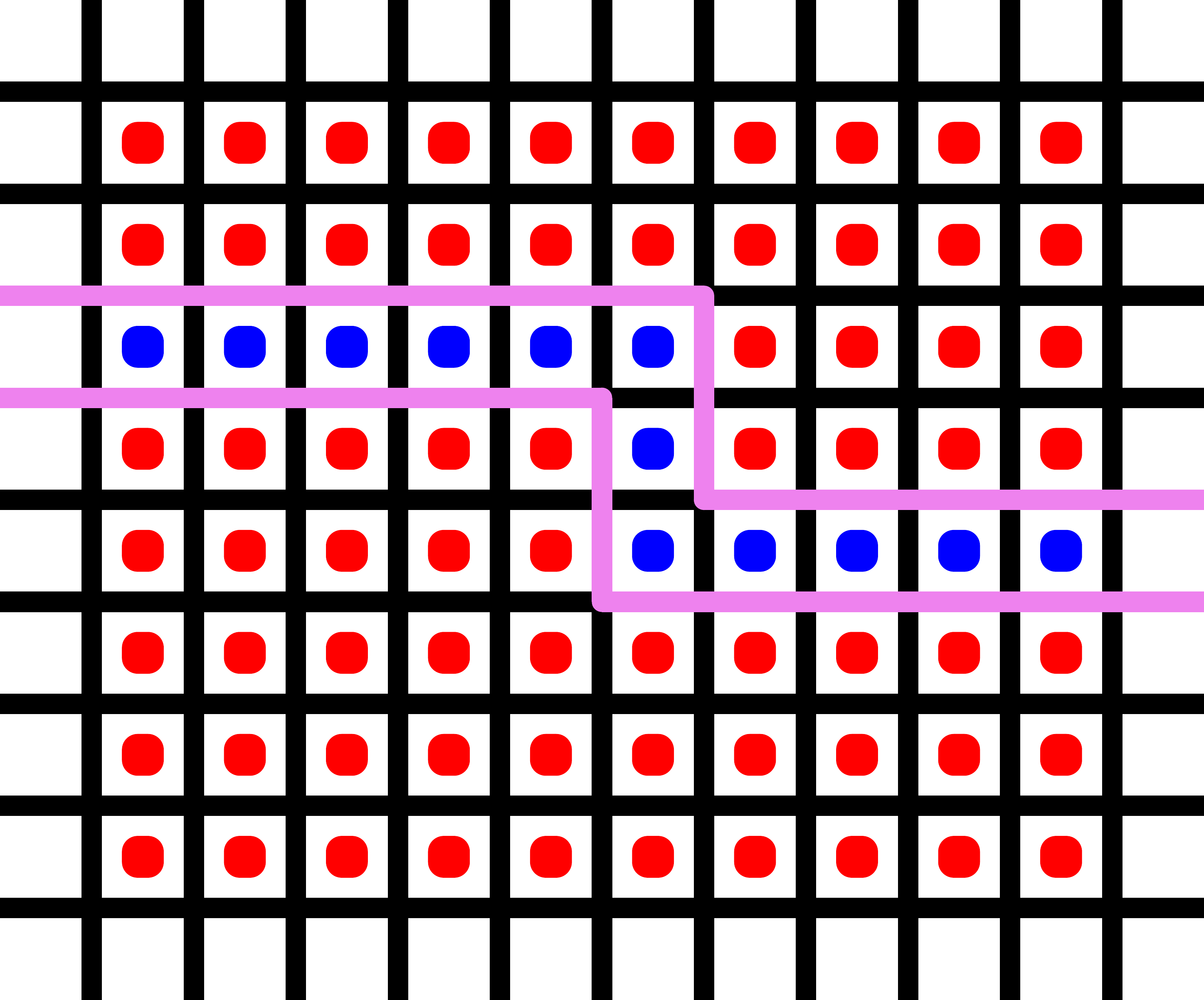}%
  \hspace{.5cm}%
}
\caption{Edge gadget.}
\label{squares-edge}
\end{figure}

\later{
\begin{proof}
  We reduce from \trvb\ on a planar multigraph with
  breakable vertices of degree~4.
  The overall plan is to fill the board with red squares, and then
  embed the multigraph into the strict interior of the board
  (i.e., without using the extreme rows and columns).
  %, outlined by the mostly red boundary given in the theorem statement.
  The gadgets representing the graph use a combination of blue squares and
  empty cells.
  Figures~\ref{squares-breakable} and~\ref{squares-edge} illustrate
  the gadgets for a breakable vertex of degree~4 and an edge (including turns), respectively.
  The square constraints dictate that any solution path must traverse the
  edge shared by two cells with squares of opposite color.
  By this property, the only possible local solutions to the gadgets
  are the ones shown in the figures.
 
\begin{figure}
\centering
% casework about the black stars
\subcaptionbox{\label{squares boundary} Boundary of the Restricted Squares Problem}{%
  \hspace{.5cm}%
  \includegraphics[width=.3\textwidth,angle=90]{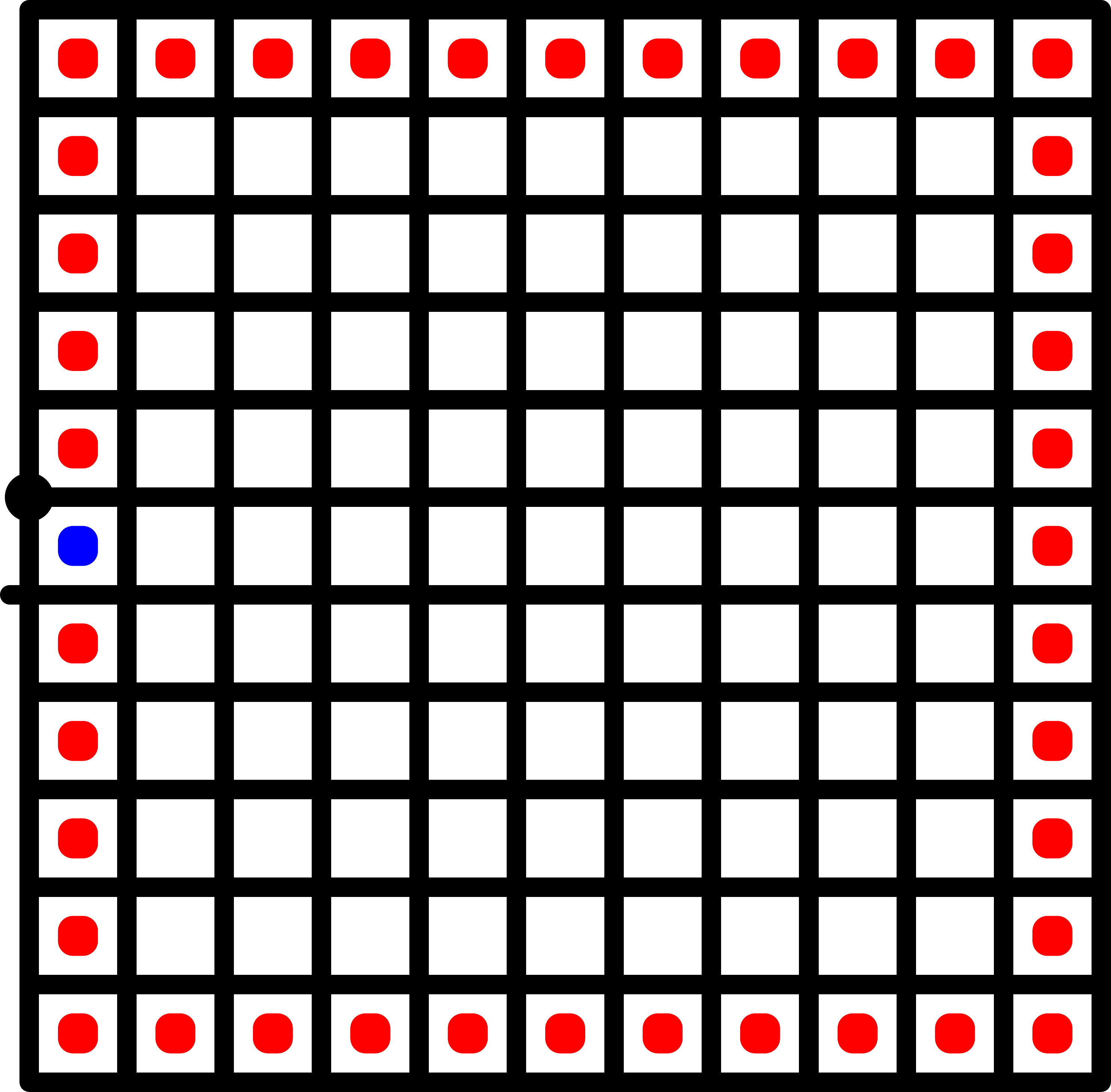}%
  \hspace{.5cm}%
}~~~
\subcaptionbox{\label{squares boundary connection} Connecting the boundary to the center}{%
  \includegraphics[width=3.5cm,angle=90]{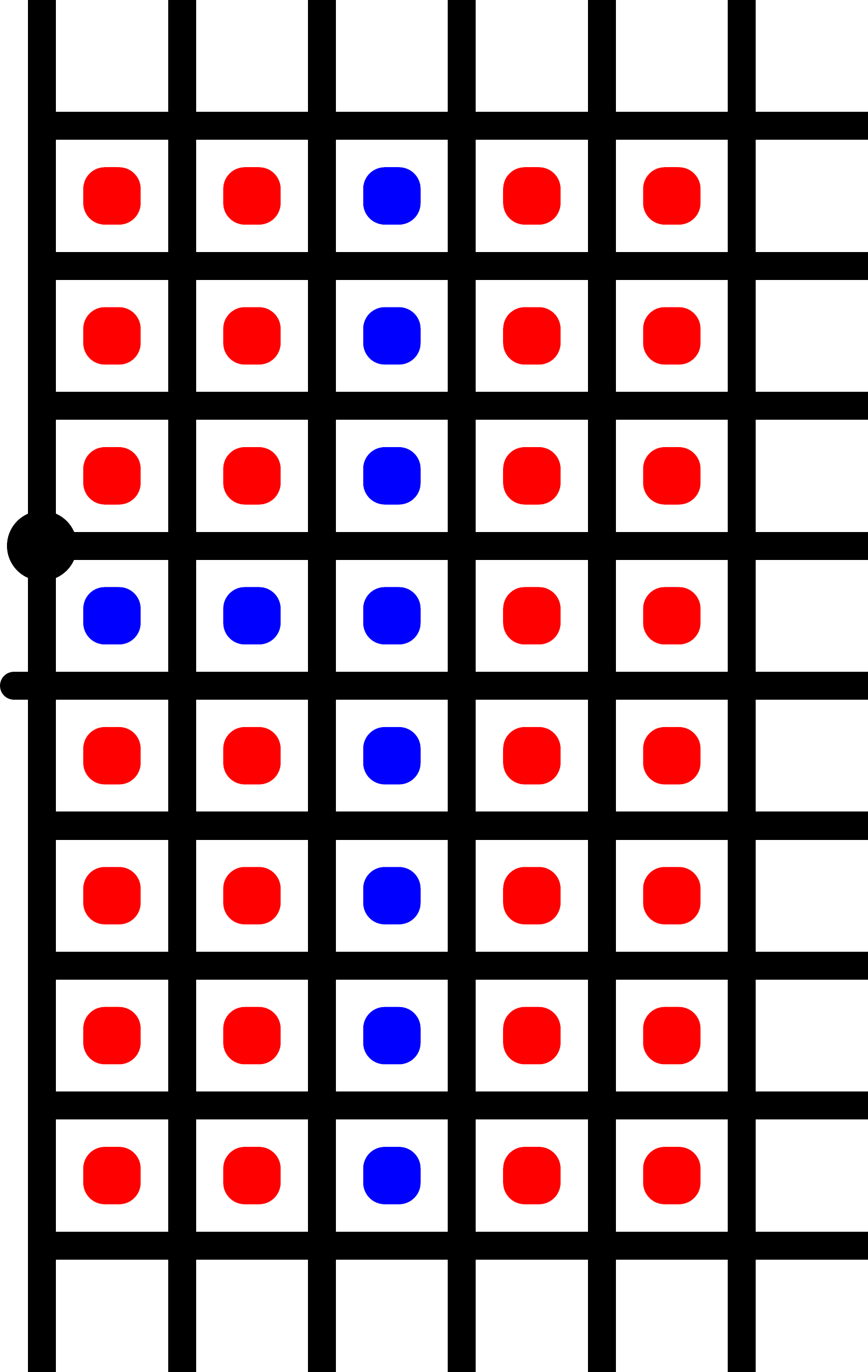}%
}
\caption{Boundary of the Restricted Squares Problem.}
\end{figure}

  To connect these gadgets together, we first lay out the multigraph
  orthogonally on a linear-size grid with turns along the edges
  (e.g., using \cite{Schaffter-1995}), scale by a constant factor,%
  \footnote{For example, a scale factor of $8$ in each dimension suffices.
    The center of a vertex gadget is at most $2 \times 2$, and three more rows
    and columns on each side suffice to wiggle the incident edges to a desired
    row or column to match an adjacent gadget.
    More efficient scaling is likely possible.}
  place the vertex gadgets at the vertices, and route the edges to connect
  the vertex gadgets (roughly following the drawing,
  but possibly adding extra turns).
  Finally, we choose any edge on the bounding box of the embedded multigraph and connect it to the boundary as shown in Figure~\ref{squares boundary connection}.  This effectively introduces an unbreakable vertex of degree 3, which does not affect the solvability of the \trvb\ instance.
  We place the start vertex and end cap at the two boundary vertices at the end of this boundary connection, as required by the Restricted Squares Problem. 
%  subdivide it with a new unbreakable vertex of degree~3,
%  and connect the free edge to the boundary,
%  as shown in Figure~\ref{squares boundary connection}.
  % 2 in coauthor
%  Note that this leaves the solvability of the vertex-breaking problem unchanged.
%  We place the start vertex and end cap at the two boundary vertices
%  at the end of this boundary connection, as required by the Restricted
%  Squares Problem.

  We claim that there is a bijection between solutions of the
  constructed Restricted Squares Problem and solutions to the given
  instance of \trvb: a vertex gets broken if and only if
  we choose the locally disconnected solution
  %(Figures~\ref{squares breakable degree 2 broken} and~\ref{squares breakable degree 4 broken})
  (Figure~\ref{squares breakable degree 4 broken})
  in the corresponding vertex gadget.
  It remains to show that breaking a subset of vertices results in a tree
  if and only if the corresponding local solutions form a global solution path
  to the Witness puzzle.
  The local solutions mimic an Euler tour of the planar graph
  after breaking the subset of vertices, and this Euler tour is a single
  connected cycle if and only if the graph is connected and acyclic.
  The one difference is the subdivided edge with an extra unbreakable vertex of
  degree 3 connected to the boundary (Figure~\ref{squares boundary connection}),
  which transforms an Euler tour into a solution path starting at the start
  vertex and ending at the end cap.
  Provided that a vertex breaking solution exists, the corresponding
  solution path will satisfy all square constraints because it will have
  all blue squares on its interior and all red squares on its exterior.
\end{proof}
}

\section{Stars}
\label{sec:stars}

\abstractlater{
  \section{Proofs: Stars}
  \label{appendix:stars}
}

Star clues are in cells of a puzzle. If a region formed by the solution
path and boundary of a puzzle has a star of a given color, then the
number of clues (stars, squares, or antibodies) of that color in that
region must be exactly two. A star imposes no constraint on clues with
colors different from that of the star.

\both{
\begin{theorem} \label{thm:stars}
  It is NP-complete to solve Witness puzzles containing only stars
  (of arbitrarily many colors).
\end{theorem}
}

\ifabstract
\begin{proofsketch}
We reduce from the Restricted Squares Problem. For every square in the source instance, $I$, we use exactly one pair of stars of a distinct color corresponding to that square,
   as well as ten auxiliary colors.
Figure~\ref{star scheme} shows the high level structure of the reduction.
A subrectangle, $S$, of the puzzle is designated for recreating
   $I$. For each pair of stars corresponding to a square, we place one of the two stars on the boundary of the puzzle, and the other
   in $S$ in the same position as the corresponding square in $I$. The solution path will be forced to divide the overall puzzle into exactly two regions---an ``inside'' and an
   ``outside''---such that all of the boundary stars corresponding to red squares are on the outside and all of the boundary stars corresponding to blue squares are on the inside.
   Then, inside of $S$, the solution path must ensure that all stars corresponding to red squares are in the outside region and all stars corresponding to blue squares are in the
   inside region, or else the star constraint will be violated. Then the solution path inside of $S$ must correspond exactly to a solution path in $I$.
The full proof is available in Appendix~\ref{appendix:stars}.
\end{proofsketch}
\fi

\begin{figure}
  \centering
  \includegraphics[width=\linewidth]{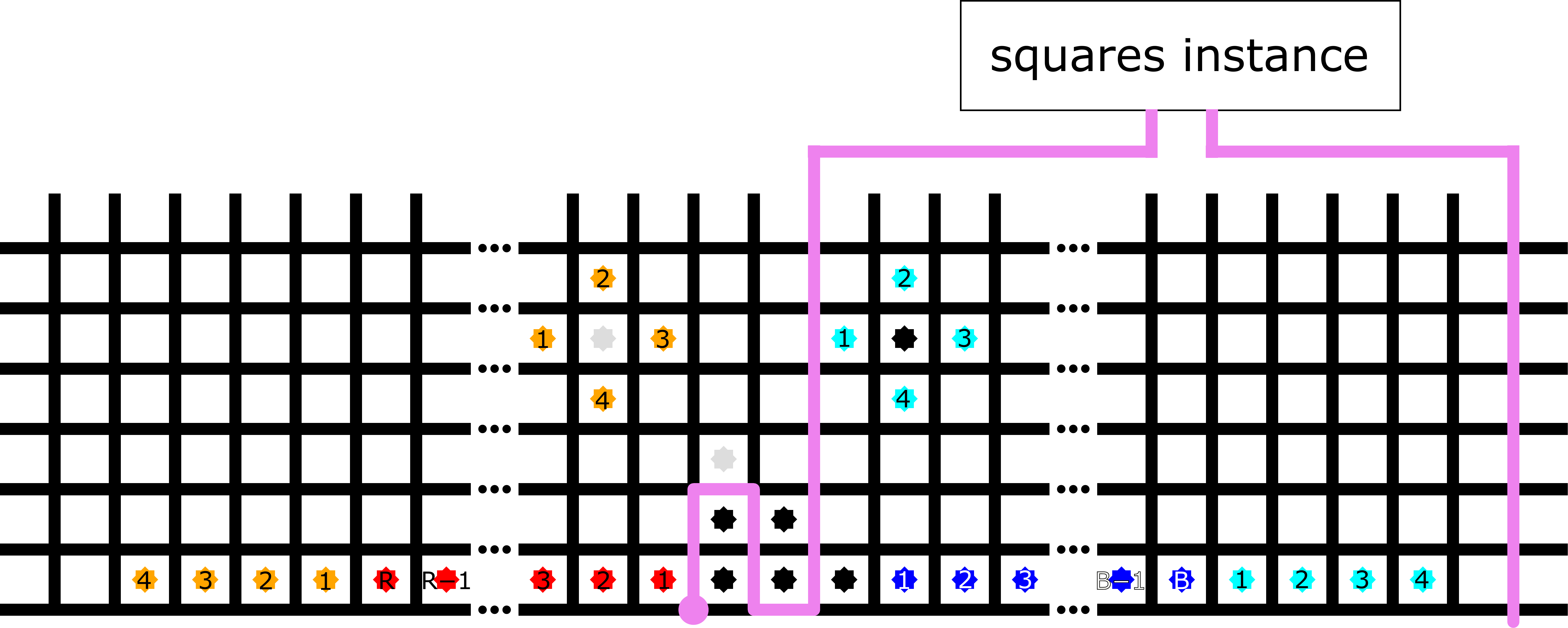}
  \caption{The boundary of the reduction. Each visual (color, number) pair represents a distinct color in the constructed instance.
  All stars depicted as blue correspond to blue squares in the source instance and must be in the inside region.
  Stars depicted as red correspond to red squares and must be in the outside region. The other stars enforce this.}
  \label{star scheme}
\end{figure}

\later{
\begin{proof}
  We reduce from the Restricted Squares Problem. %problem of Theorem~\ref{thm:squares 2 color}.
  Figure~\ref{star scheme} shows the general scheme of the reduction.
  % 3 in coauthor
  If the given squares instance has $R$ red squares and $B$ blue squares,
  then in the stars instance, we use $R+B+10$ colors:
  black,
  white,
  $c_1,c_2,c_3,c_4$   (all drawn as cyan),
  $o_1,o_2,o_3,o_4$   (all drawn as orange),
  $r_1,r_2,\dots,r_R$ (all drawn as red), and
  $b_1,b_2,\dots,b_B$ (all drawn as blue).

\begin{figure}
\centering
% casework about the black stars
\subcaptionbox{\label{stars casework noturn} After visiting the vertex in the center of four stars of the same color, the path cannot turn. (Rotations omitted.)}{%
  \hspace{.5cm}%
  \includegraphics[width=2cm]{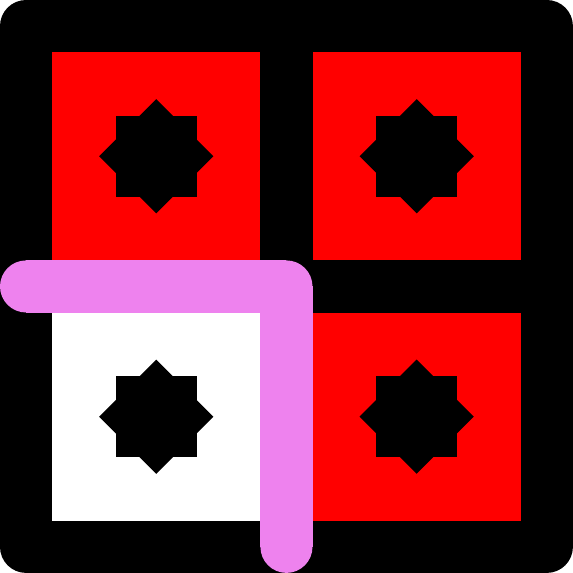}%
  \hspace{.5cm}%
}~~~
\subcaptionbox{\label{stars casework not horizontal} The star in the right column prevents the path from splitting the four stars horizontally.}{%
  \includegraphics[width=3cm]{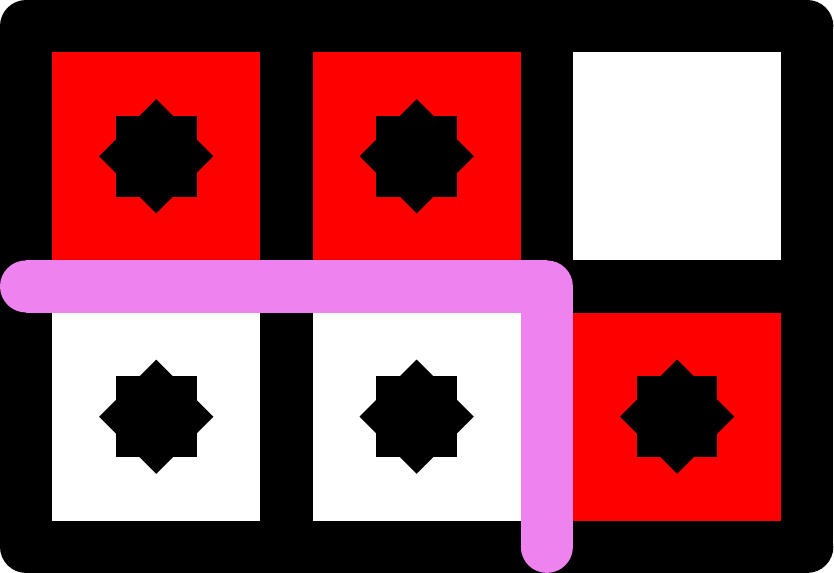}%
  ~~
  \includegraphics[width=3cm]{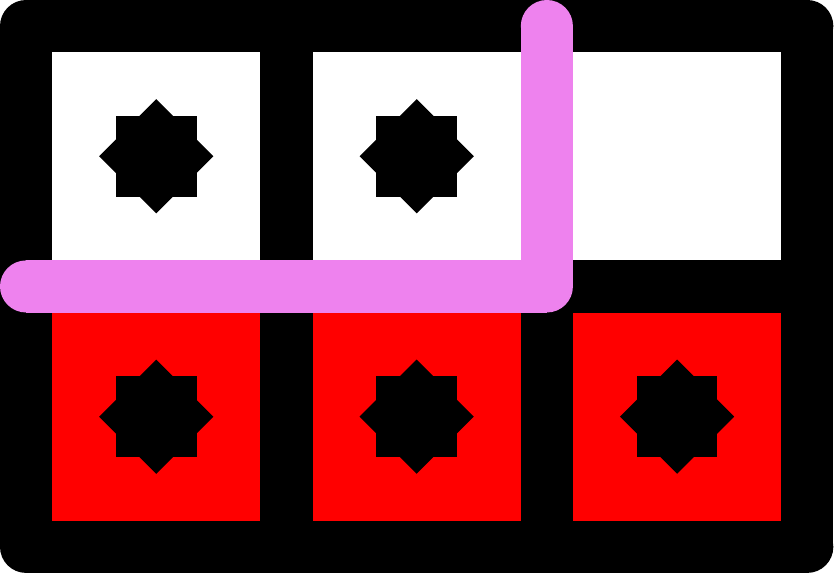}%
}~~~
\subcaptionbox{\label{stars casework forced} Thus, the path must split the four stars vertically and separate the right star by using these forced edges, so the rightmost star cannot be in the same region as the middle stars.}{%
  \hspace{.5cm}%
  \includegraphics[width=3cm]{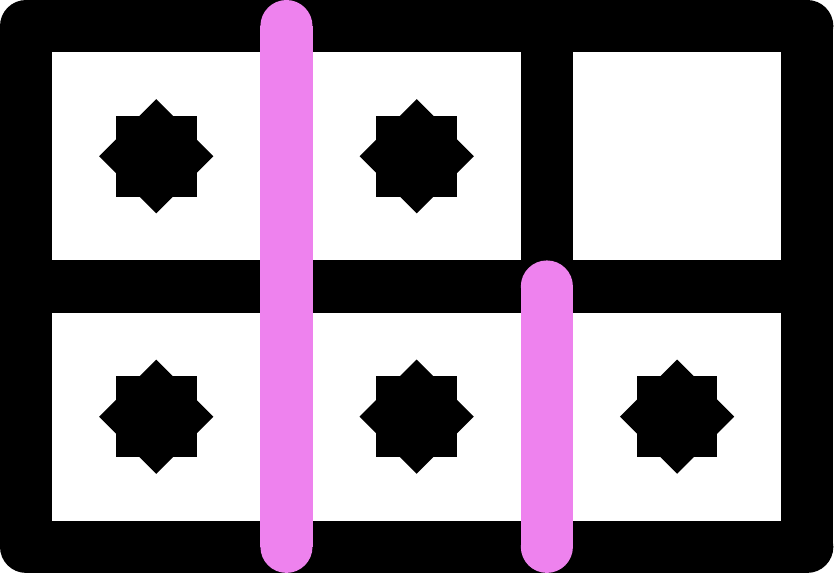}%
  \hspace{.5cm}%
}

% casework about the white star
\subcaptionbox{\label{stars casework white} No valid assignment of stars to regions places the white star in the same region as the rightmost black star.}{%
  \includegraphics[width=3cm]{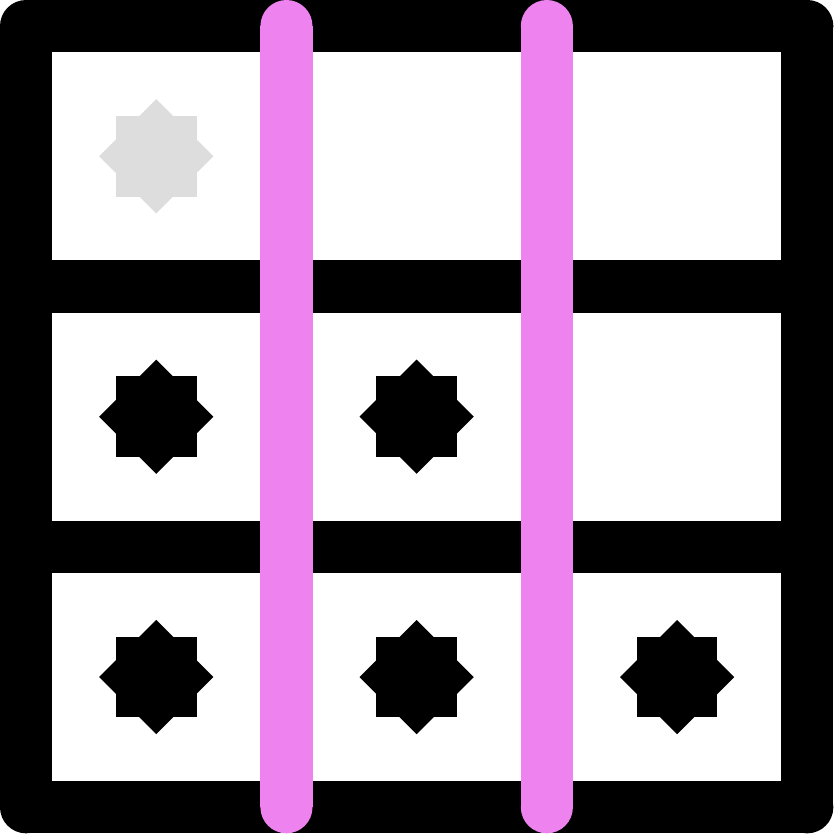}%
  ~~
  \includegraphics[width=3cm]{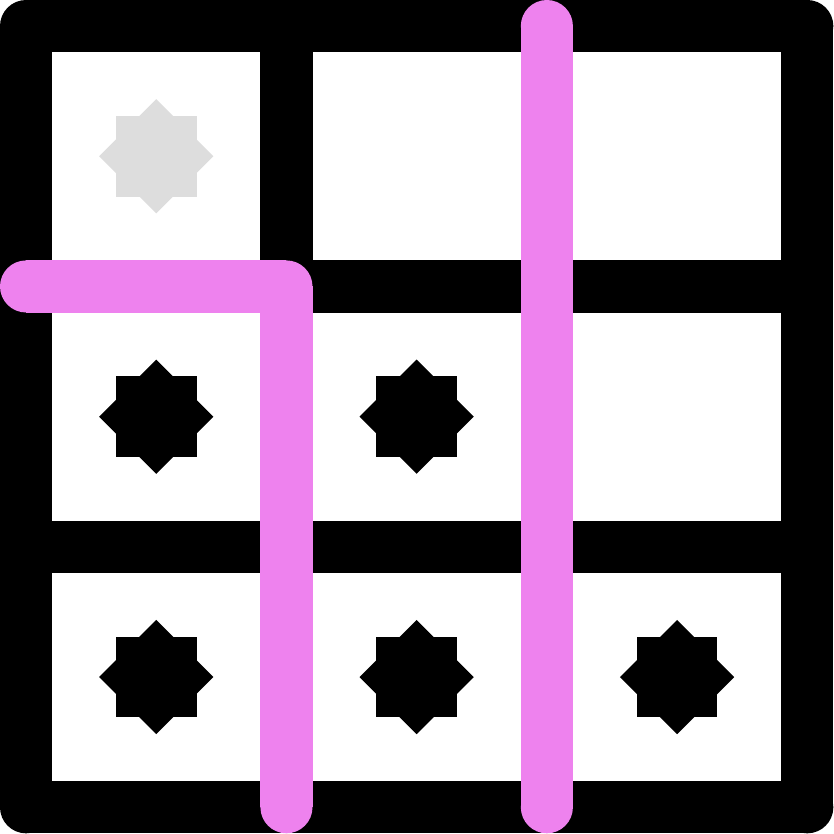}%
  % this one shows a different path, but the same effective regions as 0
  %\includegraphics[width=3cm]{figures/squares_stars_casework-white-2.pdf}
}
\caption{Analysis of the group of contiguous black stars.}
\label{stars casework}
\end{figure}

Consider the group of five contiguous black stars.  Figure~\ref{stars casework} shows that any solution must place the two leftmost black stars in the same region, the two middle black stars in another region, and the two rightmost black stars (including the star outside the group) in a third region.  Now consider the extragroup black star, which is adjacent to a $c_1$ star, a $c_2$ star, a $c_3$ star, and a $c_4$ star.  At least one of the four edges of the extragroup black star is absent, so the extragroup star is in the same region as at least one of its neighboring cyan stars, and thus also in the same region as the singular other star of that cyan color in the far right of the puzzle.  By transitivity, the rightmost intragroup black star (named $b$ in the following) is in the same region as at least one far-right cyan star.

Consider the edge $e$ to the left of $b$. Any solution path passes
through $e$ to
separate $b$ from the star to its left. Then consider the
maximal part $P$ of the solution path that contains $e$ but not any
boundary edges. Orient $P$ so that it starts with edge $e$; note that
this orientation may be opposite the orientation of the whole solution
path. If $P$'s other endpoint lies on the part of the boundary between
the
black stars and the far-right cyan stars, then $b$ and those stars are in
a different
regions, a contradiction. Therefore, all the boundary
cells between $b$ and the far-right cyan stars are in the same
region. In particular, the blue stars of colors $b_1, b_2, \dots,
b_B$ are all in the same region as $b$, that is, on the right of $P$.

Now consider the white star $w$ adjacent to the group of black stars.  As shown in Figure~\ref{stars casework white}, $w$ is not in the same region as $b$, so $w$ is to the left of $P$ (though not necessarily in the region immediately to the left of $P$), and thus the other (leftmost) white star is also left of $P$.  As with the cyan stars, at least one orange star among those adjacent to the leftmost white star lies on the left of $P$, and so does at least one far-left orange star. As with the cyan, black, and blue stars, this implies that the red stars of colors $r_1, r_2, \dots, r_R$ are on the left of $P$.

Thus, we have shown that the red stars are on a different side of
$P$ from the blue stars. This is possible only if the original
Restricted Squares Problem instance has a solution. %Thus, we have shown
                                %that if the original squares instance
                                %is impossible to solve then so is the
                                %constructed stars instance.
Conversely, if the original squares instance is possible to solve, then
we can augment that solution into a solution to the constructed stars
Witness puzzle as shown in Figure~\ref{star scheme}. \xxx{make figure specifically for this (rather than conflating with the schema figure); was figure 4 on
  coauthor/google sheets}
\end{proof}
} %\later

\begin{restatable}{open}{star} \label{open:1 star}
Is it NP-hard to solve Witness puzzles containing only a constant number of colors of stars? Or just a single color of stars?
\end{restatable}

\section{Triangles}

\abstractlater{
  \section{Proofs: Triangles}
  \label{appendix:triangles}
}
Triangles are placed in cells.
The number of edges on the solution path that are incident to that cell
must match the number of triangles. This constraint is similar to Slitherlink,
which is known to be NP-complete~\cite{slitherlink}.
Table~\ref{SlitherlinkTable} summarizes known and new results for Slitherlink
puzzles with clues chosen from $\{0,1,2,3\}$.
The one difference is that The Witness does not allow $0$-triangle clues.
Unfortunately, the proof in~\cite{slitherlink} relies critically on being able
to force zero edges around a cell using $0$ clues.
We can simulate $0$ clues using broken edges, as in Table~\ref{2D results}.
To avoid broken edges, we develop new proofs that $1$-triangle, $2$-triangle,
and $3$-triangle clues are NP-hard.
\ifabstract
Proofs omitted from this section can be found in Appendix~\ref{appendix:triangles}.
\fi

\begin{table}
\centering
\begin{tabular}{|l|l|}
\hline
\emph{Clue types} & \emph{Complexity}  \\ \hline
\hline
$0$ & P \cite{slitherlink} \\ \hline
\textbf{1} & \textbf{NP-complete [Theorem~\ref{thm:triangles1}]} \\ \hline
\textbf{2} & \textbf{NP-complete [Theorem~\ref{thm:triangles2}]} \\ \hline
\textbf{3} & \textbf{NP-complete [Theorem~\ref{thm:triangles3}]} \\ \hline
$4$ & P [trivial] \\ \hline
%\textbf{0} and \textbf{2} & \textbf{NP-complete [Corollary~\ref{thm:triangles-0-2}]} \\ \hline
$0$ and $(1$ or $2$ or $3)$ & NP-complete \cite{slitherlink}\footnote{The proof in \cite{slitherlink} is for $0$ and $1$ clues, but can be straightforwardly adapted to replace $1$ clues with $2$ or $3$ clues.} \\ \hline 
\end{tabular}
\caption{Summary of Slitherlink / Witness triangle constraints. New results are bold. Recall that ``$0$'' is not a valid clue in The Witness.}
\label{SlitherlinkTable}
\end{table}

\subsection{1-Triangle Clues}
\label{sec:1-triangles}

Proving hardness of Witness puzzles containing only 1-triangle clues is
challenging because it is impossible to (locally) force turns on the interior
of the puzzle.  Any rectangular subpuzzle with any set of 1-triangle clues can be
satisfied by a set of disjoint paths containing either every second row of
horizontal edges in that rectangle or every second column of vertical edges in that rectangle, regardless of the configuration of 1-triangle clues.
Therefore, any local arguments about gadgets in the interior of the puzzle
must confront the possibility of local solutions which are comprised of just
horizontal or vertical paths straight through.  2-triangle clues, discussed in
Section~\ref{sec:2-triangles}, present a similar challenge.
Nonetheless, we are able to prove NP-hardness:

\begin{theorem}
\label{thm:triangles1}
It is NP-complete to solve Witness puzzles containing only 1-triangle clues.
\end{theorem}

\begin{proofsketch}
We reduce from X3C, making use of the fact that the solution path must be a
single closed path.  Refer to Figure~\ref{1-triangles overview}.
Sets and elements are represented by sets of rows in a puzzle
almost completely filled with 1-triangle clues.  Boundary conditions
alone force any solution to form disconnected cycles and one path consisting
largely of horizontal lines, but we can build gadgets by deleting 1-triangle
clues in specific locations.  Set--element connection gadgets allow connecting
the horizontal lines of a chosen set exactly when we connect the horizontal
lines of a specific element, but only if the element is so covered exactly
once, and the 3-set gadget allows connecting the horizontal lines of an
unchosen set.  The cycles can be connected to form a single solution path
exactly when the X3C instance has a solution.
\end{proofsketch}

%% OUTLINE: All lines are doubled. 2 lines per clause. 2 lines per occurrence of variable. Every pair of lines should be broken exactly once. Local all-or-nothing gadget to satisfy variables in false case. Nonlocal between variable occurrence and clause to satisfy variable in true case.

\begin{figure}
\centering
\includegraphics[width=\linewidth]{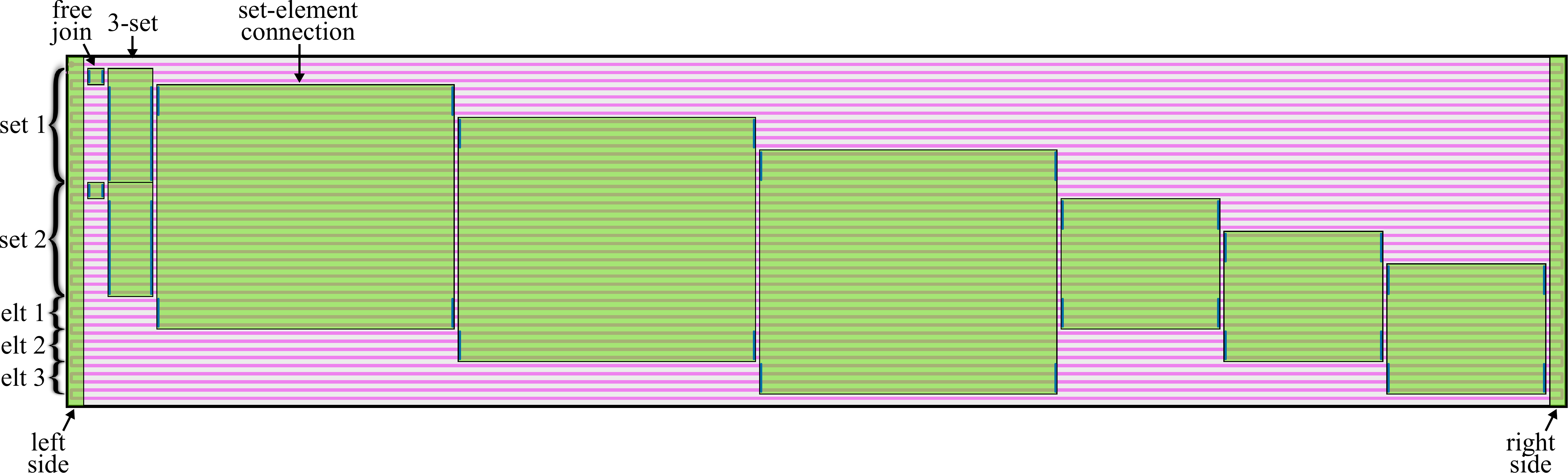}
\caption{Full example schematic (drawn roughly to scale) of the reduction
  from X3C to 1-triangle clues of Theorem~\ref{thm:triangles1},
  with three elements and two 3-sets each consisting of all three elements.
  Pink lines indicate the ``solution'' of the initial puzzle filled with
  1-triangles in all but the leftmost and rightmost columns.
  Each box represents one of the gadgets in the proof, with vertical lines
  to indicate the horizontal lines that could be connected by the gadget.}
\label{1-triangles overview}
\end{figure}

\begin{proof}
  We reduce from Exact Cover by 3-Sets, which we abbreviate X3C
  \cite{Garey-Johnson-1979}:
  given a set $X$ of $n$ elements, and
  given a family $C$ of $m$ cardinality-3 subsets of~$X$,
  decide whether there exists a subset $P \subseteq C$ such that $P$
  is a partition of~$X$ (i.e., $P$ covers every element in $X$ exactly once).
  Without loss of generality, let $X = \{0,1,\dots,n-1\}$.
  %We reduce from positive exactly-1 CNF SAT:
  %given $n$ variables $x_1, x_2, \dots, x_n$ and $m$ clauses (nonempty tuples
  %of variables), is there a truth assignment such that exactly one variable in each clause is true?
  Figure~\ref{1-triangles overview} shows how to fit together
  the gadgets we will now describe.

\paragraph{Left and right sides of the construction.}

We start from (and later modify) a Witness puzzle where all cells have a
1-triangle clue except for the leftmost and rightmost columns of cells,
as shown in Figure~\ref{1-triangles edges unsolved}.
The rightmost column follows the repeating pattern
(1-triangle, empty, empty, 1-triangle)$^k$, while the leftmost column
has the same pattern except for the top four rows which instead of
two empty cells have the start circle and end cap,
in the bottom-left of the first and third row of cells respectively.
The number of rows of cells is $4 k = 2 + 28 m + 8 n + 2$.

\begin{figure}
\centering
% 1-triangles rules
\subcaptionbox{\label{1-triangles edges unsolved} Unsolved.}{
  \hspace{.5cm}%
  \includegraphics[scale=0.3]{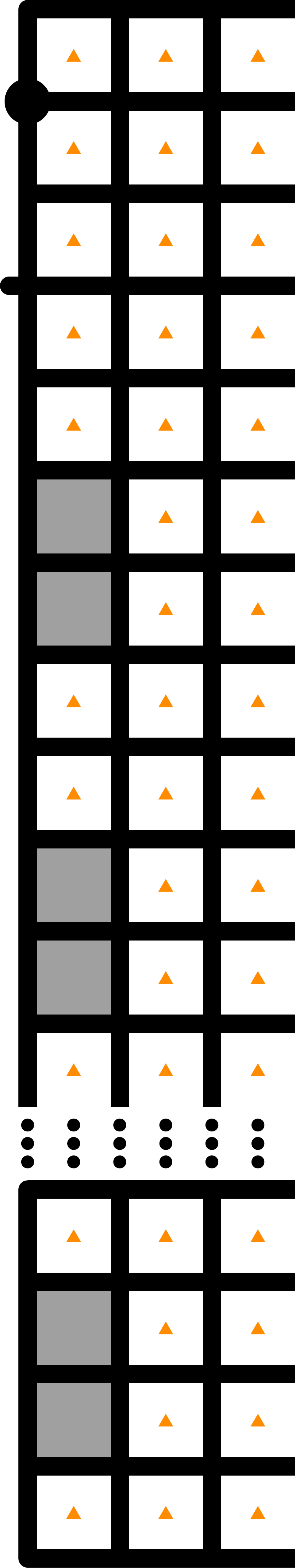}~~~
  \hspace{.25cm}%
  \includegraphics[scale=0.3]{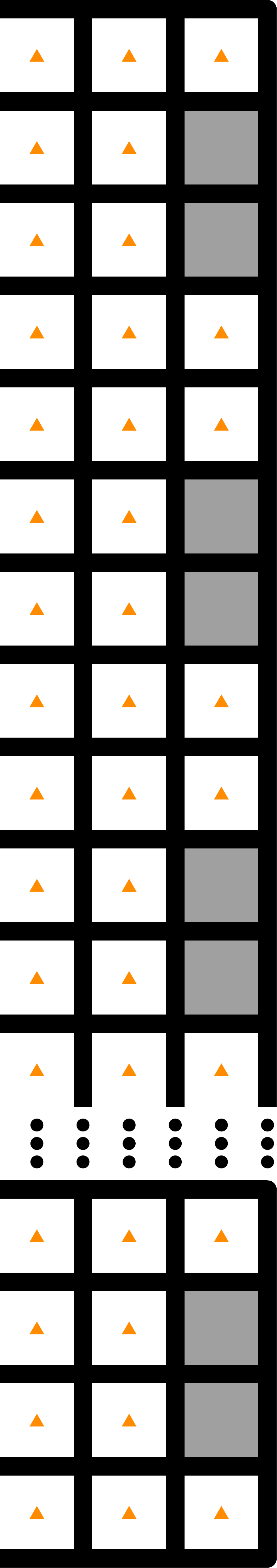}%
  \hspace{.5cm}%
}~~~
\subcaptionbox{\label{1-triangles edges solved} Solved.}{
  \hspace{.5cm}%
  \includegraphics[scale=0.3]{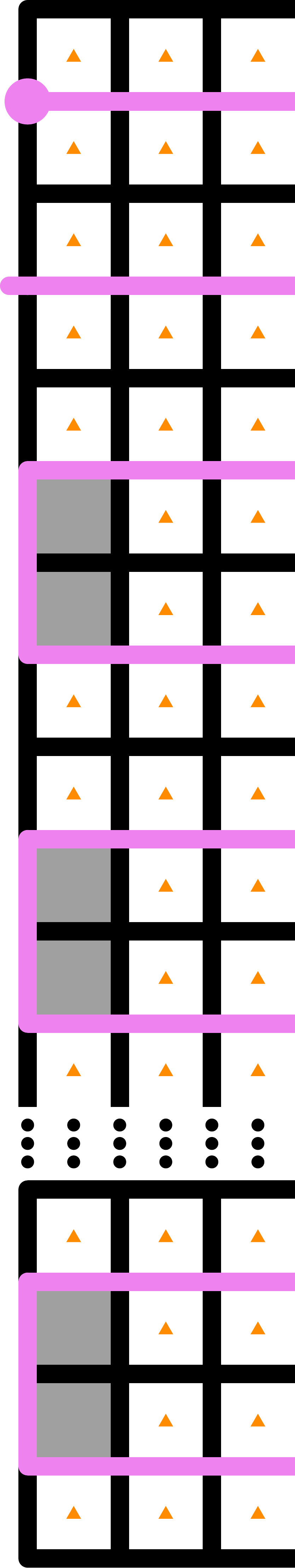}~~~
  \hspace{.25cm}%
  \includegraphics[scale=0.3]{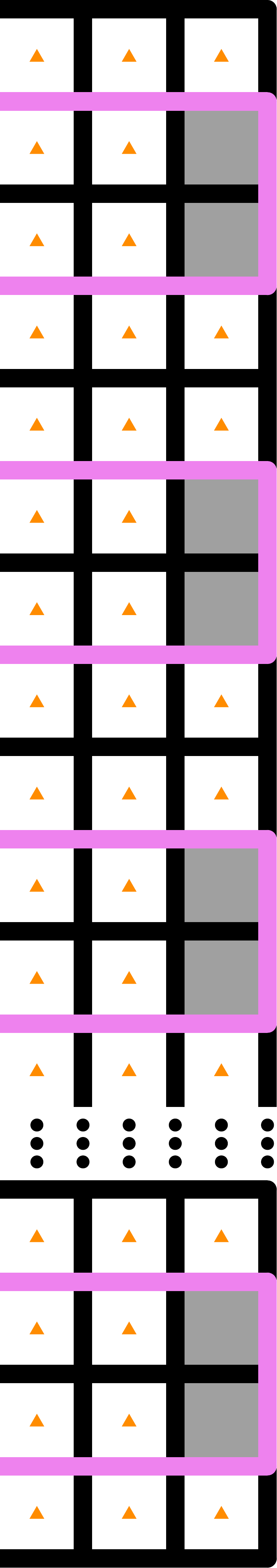}%
  \hspace{.5cm}%
}
\caption{Leftmost and rightmost three columns of the puzzle.
  Empty cells are shaded gray for emphasis.}
\label{1-triangles edges}
\end{figure}

If we relax the ``solution'' to allow multiple paths and/or cycles, then
we claim that this puzzle has the unique solution shown in
Figure~\ref{1-triangles edges solved}.
(Thus, with the usual constraint of a single path, there is no solution.)
Even stronger, for any modified puzzle still having
1-triangles throughout the topmost two rows and the same leftmost three columns
and rightmost three columns as Figure~\ref{1-triangles edges unsolved},
we claim that any solution must look like Figure~\ref{1-triangles edges solved}
in those columns.

% Figure~\ref{1-triangles edges unsolved}
% shows the leftmost and rightmost three columns of our construction.
%We claim any solution must locally look like Figure~\ref{1-triangles edges solved}, and so before considering the other gadgets, the path forms disconnected cycles at right-joined row pairs separated by right-separated row pairs.

First we claim that the path goes right from the start vertex;
refer to Figure~\ref{1-triangles left gadget nonsolutions}.
If the path goes up into the top-left corner, then it must then go right,
violating the 1-triangle clue in the top-left cell.
If the path goes down, the 1-triangle forces it to go down again,
but then it has reached the end cap without
satisfying most of the 1-triangles in the puzzle.
Therefore, the path must initially go right.

\begin{figure}
\centering
\subcaptionbox{Failed attempt to start by going up.}{%
  \hspace{.5cm}%
  \includegraphics[scale=0.3]{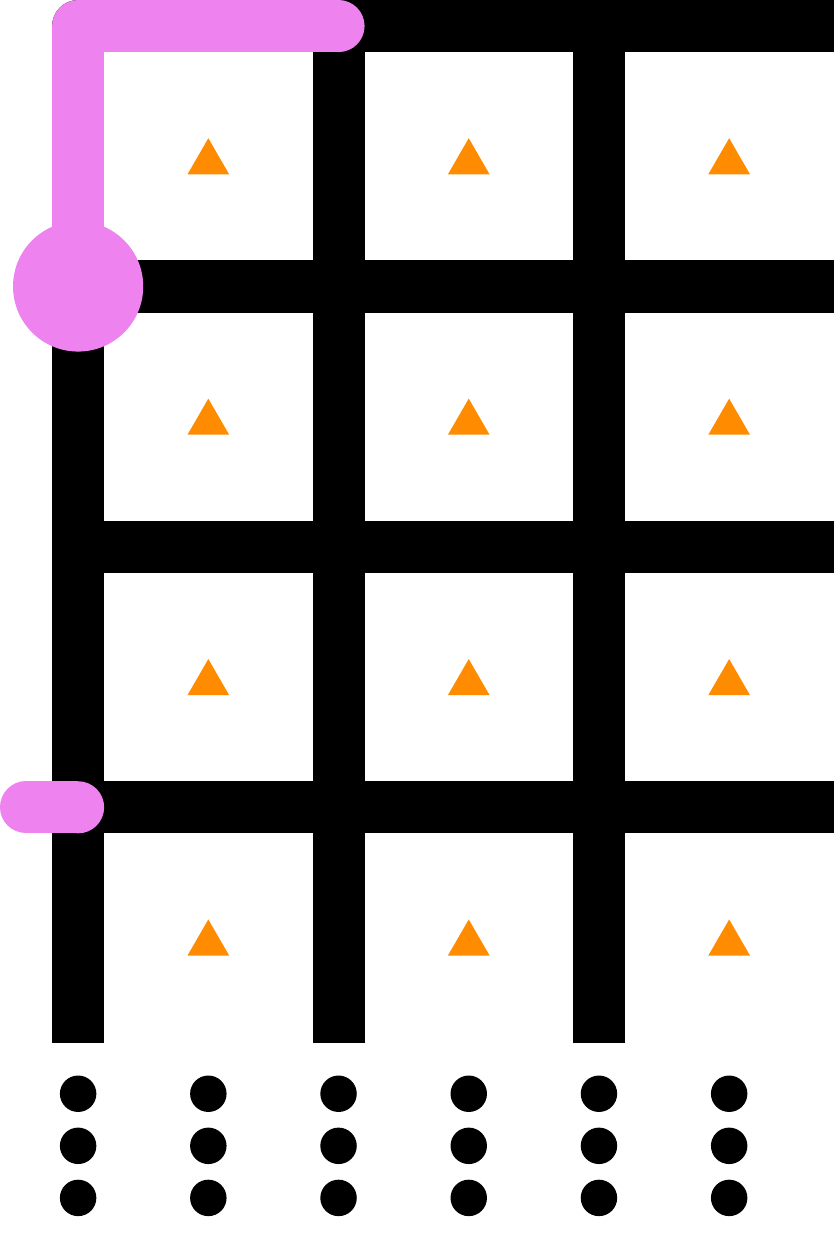}%
  \hspace{.5cm}%
  \label{1-triangles left gadget nonsolution 1}%
}~~~
\subcaptionbox{Failed attempt to start by going down then right.}{%
  \hspace{.5cm}%
  \includegraphics[scale=0.3]{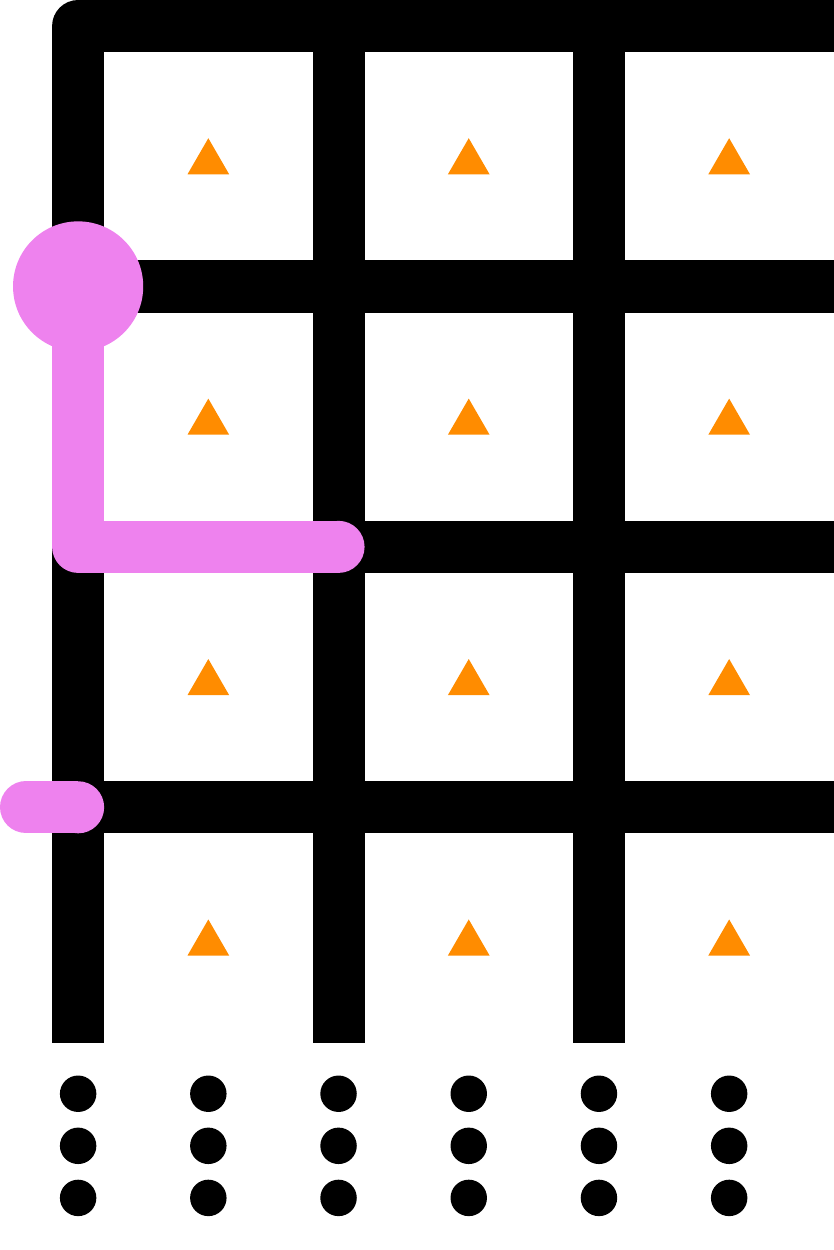}%
  \hspace{.5cm}%
  \label{1-triangles left gadget nonsolution 2}%
}~~~
\subcaptionbox{Failed attempt to start by going down twice.}{%
  \hspace{.5cm}%
  \includegraphics[scale=0.3]{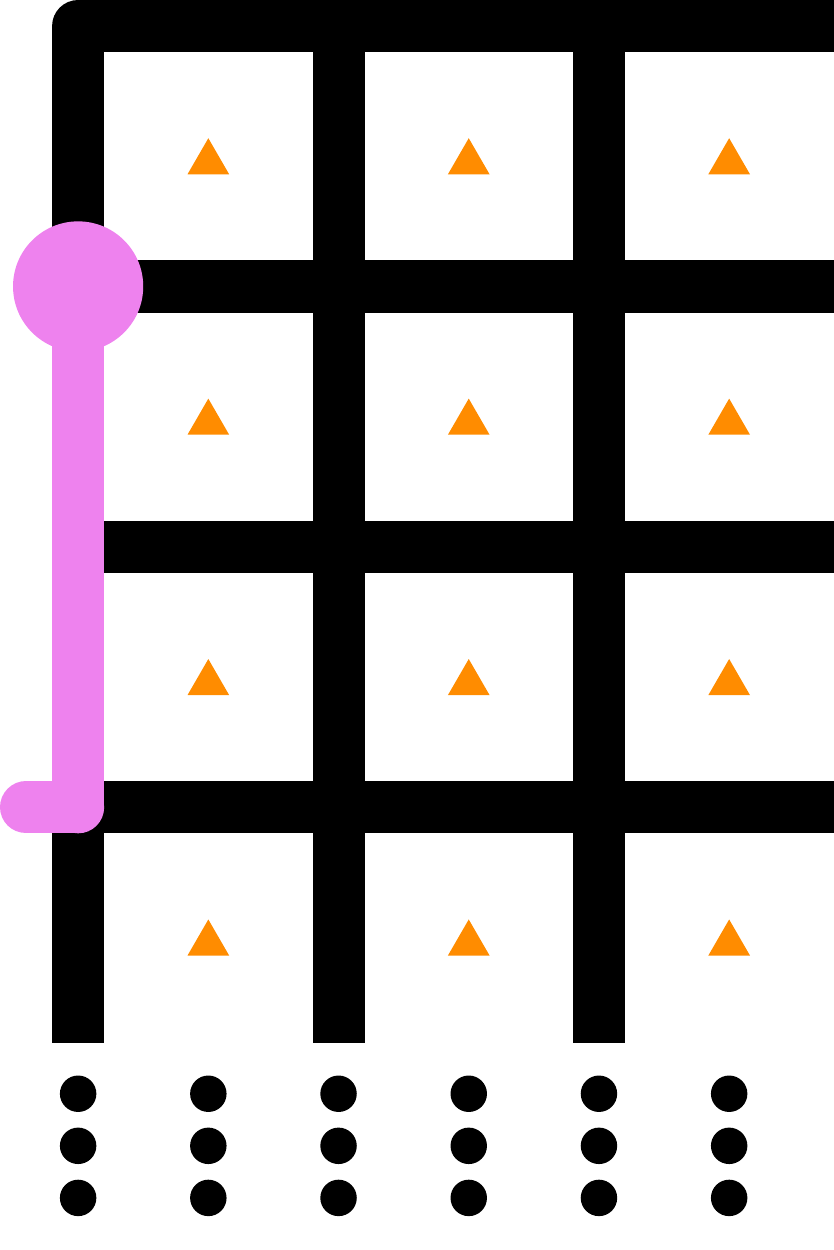}%
  \hspace{.5cm}%
  \label{1-triangles left gadget nonsolution 3}%
}
\caption{Nonsolutions to the leftmost columns.}
\label{1-triangles left gadget nonsolutions}
\end{figure}

Next, by repeated application of Rule~1 of Figure~\ref{1-triangles rules}
in the topmost two rows, the path must continue right from the start circle
until reaching the right side of the puzzle.
By repeated application of Rule~3 of Figure~\ref{1-triangles rules},
rows of cells alternate between having the path on their bottom and having the
path on their top, thereby making $2 k$ \emph{horizontal lines} of path.
If we perfectly pair up adjacent rows of cells
(pairing every odd row, including the first, with the row below it),
then there is exactly one horizontal line in between each pair of rows.
These horizontal lines must join using the edges incident to blank cells
in the leftmost and rightmost columns (except at the start and end cap).

  \begin{figure}
\centering
% 1-triangles rules
\subcaptionbox{\label{1-triangles rule 1} Rule 1: When an edge adjacent to two 1-triangle clues is present in the solution path, then it must extend at least one more step in both directions.}{%
  \hspace{.5cm}%
  \includegraphics[scale=0.35]{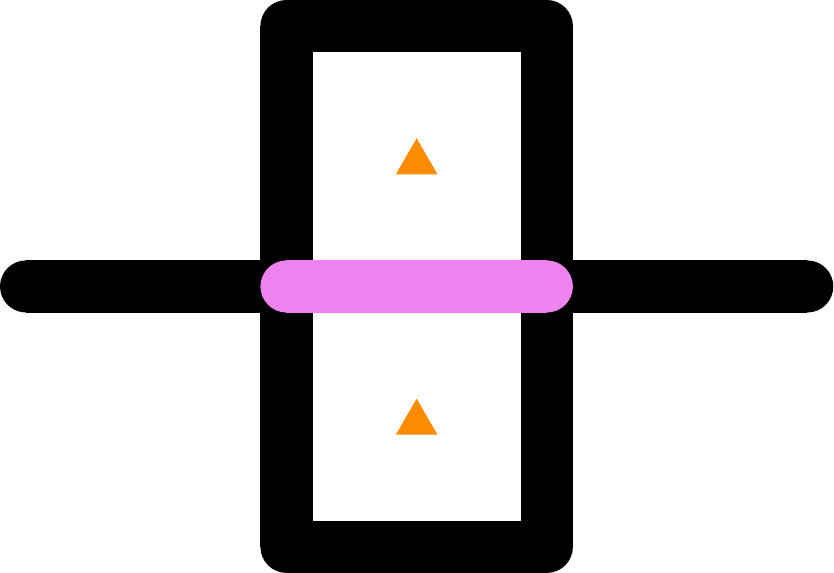}~~~
  \includegraphics[scale=0.35]{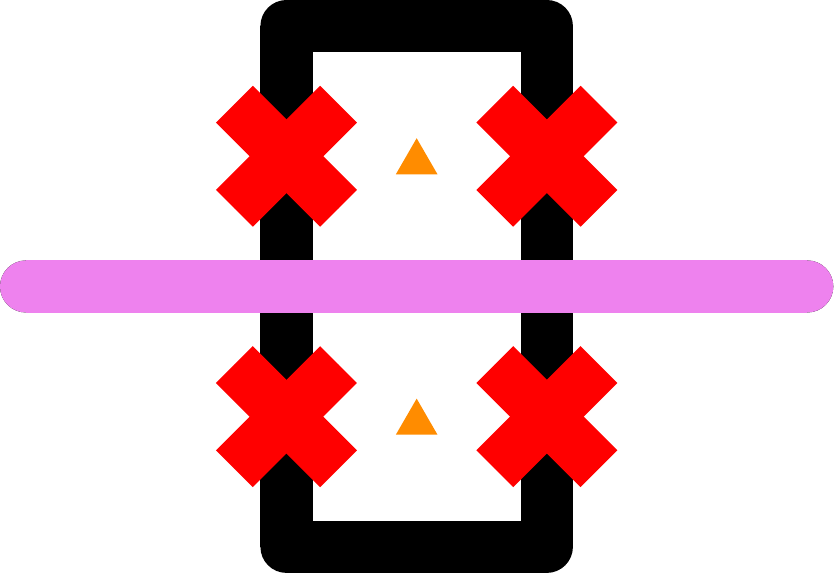}%
  \hspace{.5cm}%
}~~~
\subcaptionbox{\label{1-triangles rule 2} Rule 2: When the three edges along the long side of four 1-triangle clues arranged in a $T$ shape are present in the solution path, then the edge opposite them must also be present.}{%
  \includegraphics[scale=0.35]{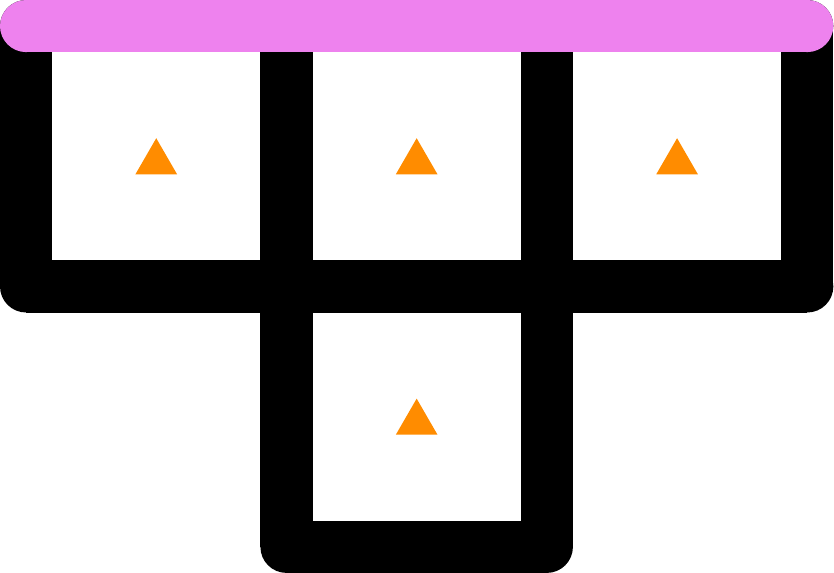}~~~
  \includegraphics[scale=0.35]{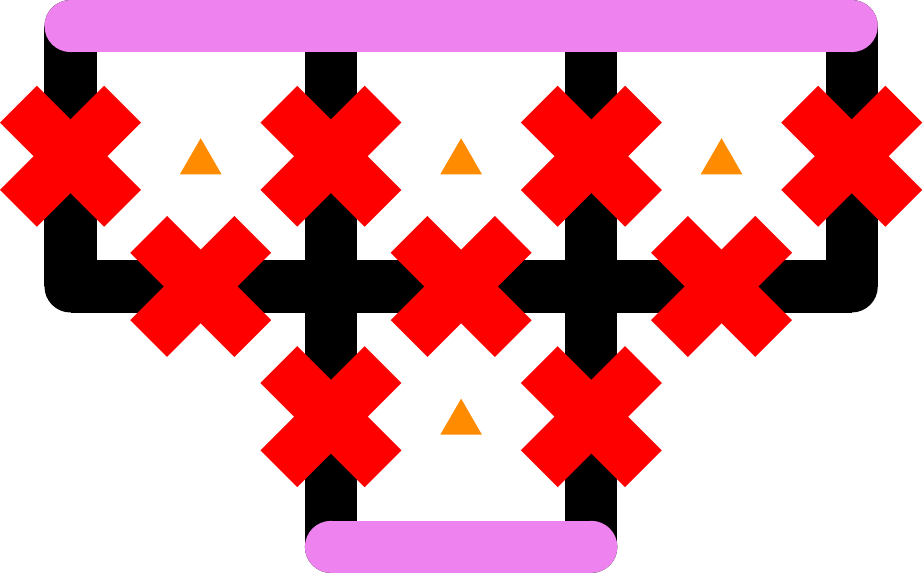}%
}\\
\subcaptionbox{\label{1-triangles rule 3} Rule 3: An application of Rule~2 followed by an application of Rule~1.}{%
  \includegraphics[scale=0.35]{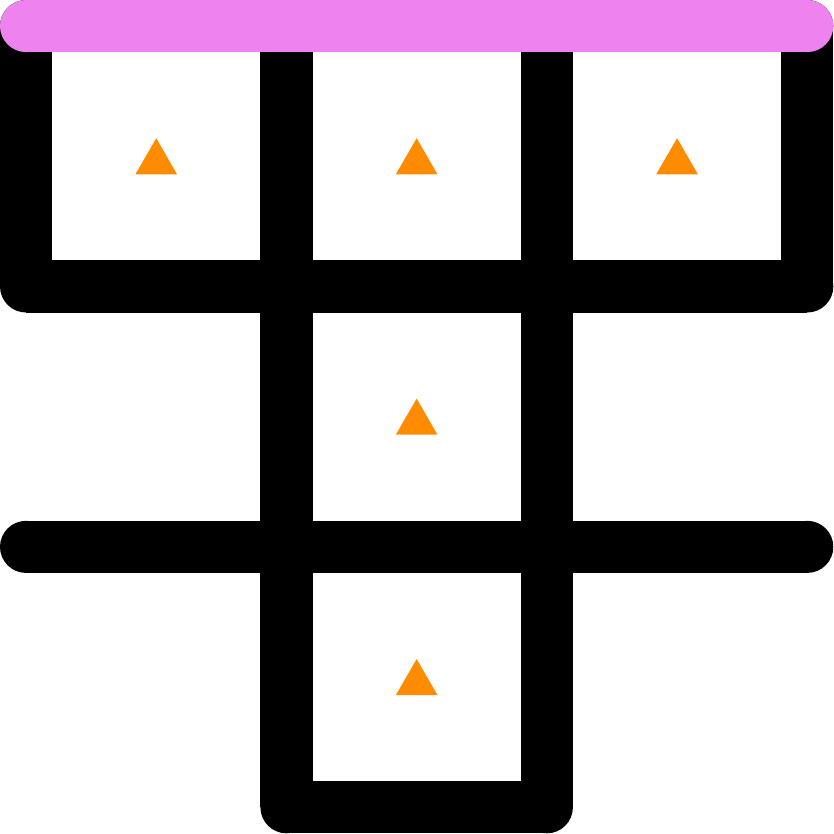}~~~
  \includegraphics[scale=0.35]{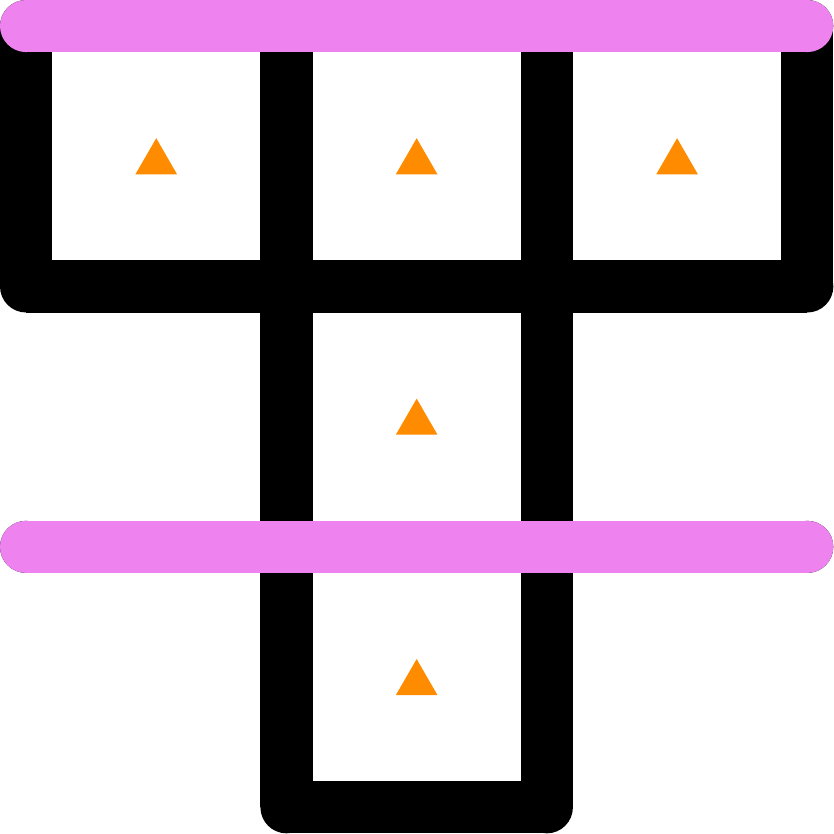}%
}
\caption{Local rules for 1-triangle clues. A red X means that the edge
  cannot be in the solution path due to constraints imposed by triangles.}
\label{1-triangles rules}
\end{figure}

Hence, the unique ``solution'' to this initial puzzle consists of $k$ connected
components, one start-to-end path in the top four rows and one cycle in
every following consecutive four rows.
In the remainder of the proof, we will place gadgets in the interior of the
puzzle (removing some 1-triangle clues) to enable connecting these solution
components together into one path exactly when the X3C instance has a solution.
This modification will not touch the topmost two rows, the bottommost two rows,
the leftmost three columns, or the rightmost three columns,
thereby still forcing those columns to look like
Figure~\ref{1-triangles edges solved}.

\paragraph{Free-join gadget.}
Figure~\ref{1-triangles joining lines} shows one gadget for joining
components of the ``solution'' together, called the \emph{free-join gadget}.
By placing a $4 \times 4$ ``circle'' of empty cells, we allow the solution
to optionally deviate from horizontal lines and thereby connect together
two components by joining the bottom line of one component with the
top line of the next component.
(This gadget, and all future gadgets, can be argued to have exactly two
 local solutions using Rules 1--3 of Figure~\ref{1-triangles rules}.)
Later gadgets will perform similar joins between adjacent pairs of components,
but with more constraints on when the join is allowed.
Figure~\ref{1-triangles lines twice joined} shows that performing two such
joins on the same two lines results in an isolated cycle in the middle of the
puzzle, which in our construction can never be connected to other components.
Therefore, to form one connected solution path,
we will need to make exactly one join between each even-numbered
horizontal line (in particular, excluding the first line)
and the line immediately below it.

  \begin{figure}
\centering
% 1-triangles rules
\subcaptionbox{\label{1-triangles lines unjoined} Lines unjoined.}{%
  \hspace{.5cm}%
  \includegraphics[width=.75\textwidth]{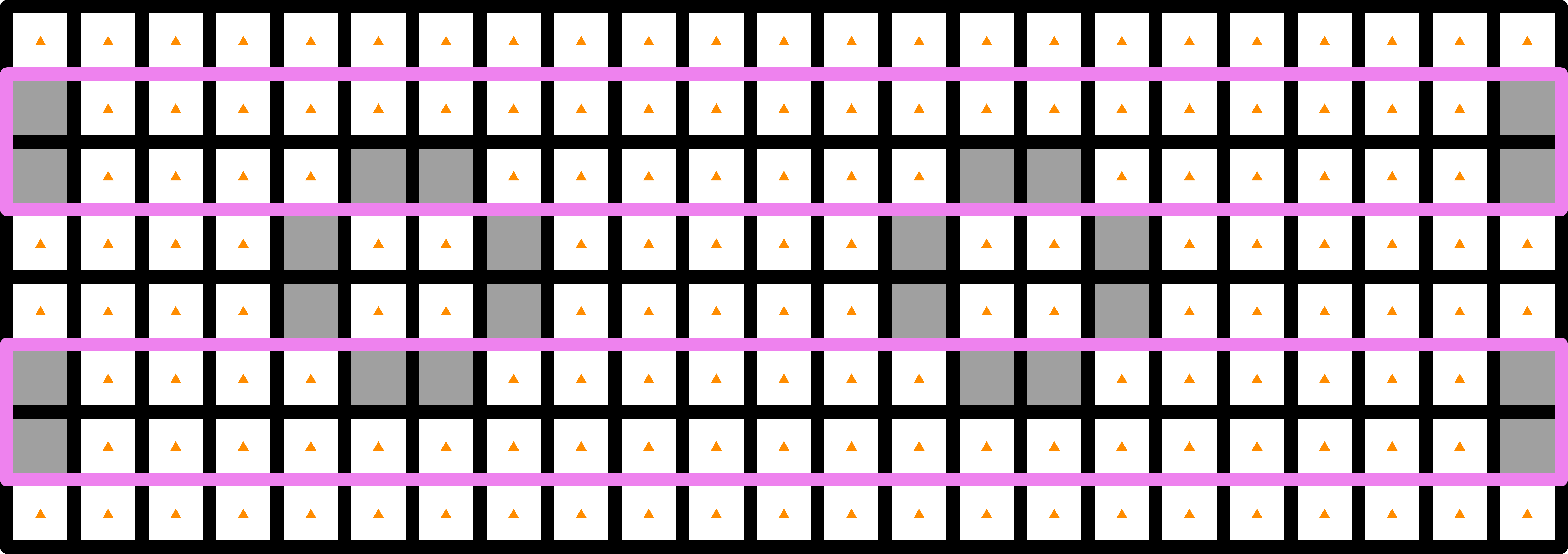}
  \hspace{.5cm}%
}\\
\subcaptionbox{\label{1-triangles lines once joined} Lines joined once.}{%
  \hspace{.5cm}%
  \includegraphics[width=.75\textwidth]{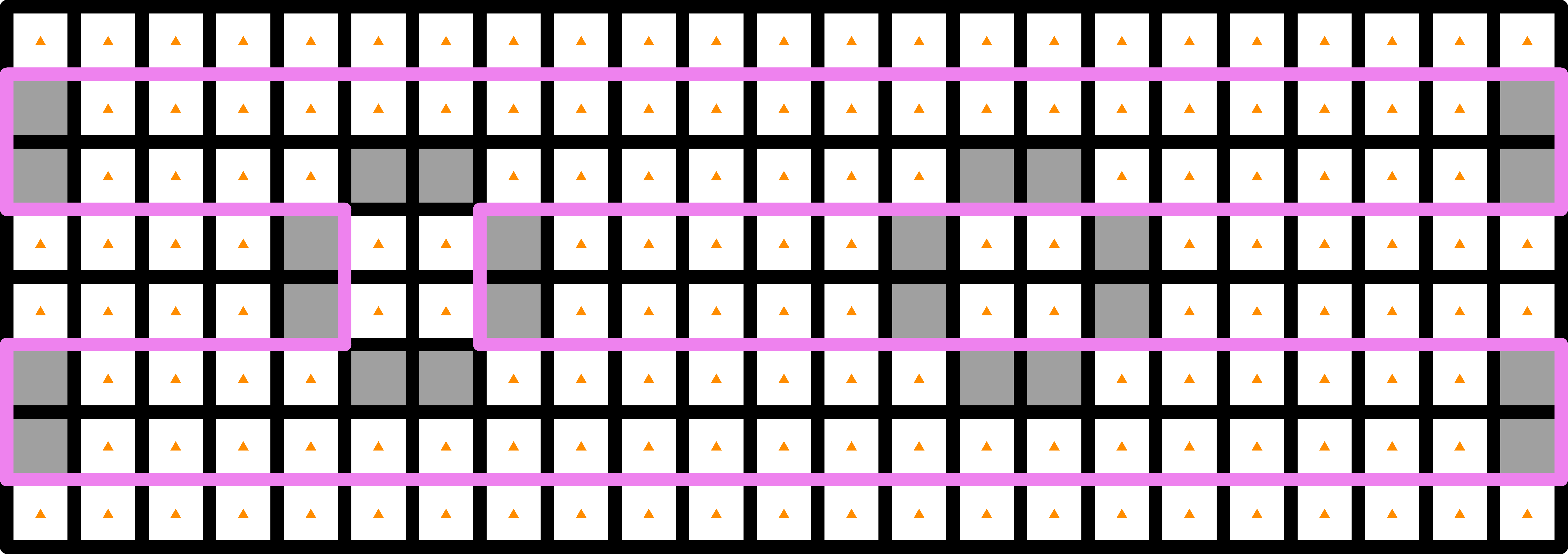}%
  \hspace{.5cm}%
}\\
\subcaptionbox{\label{1-triangles lines twice joined} Lines joined twice (invalid).}{%
  \hspace{.5cm}%
  \includegraphics[width=.75\textwidth]{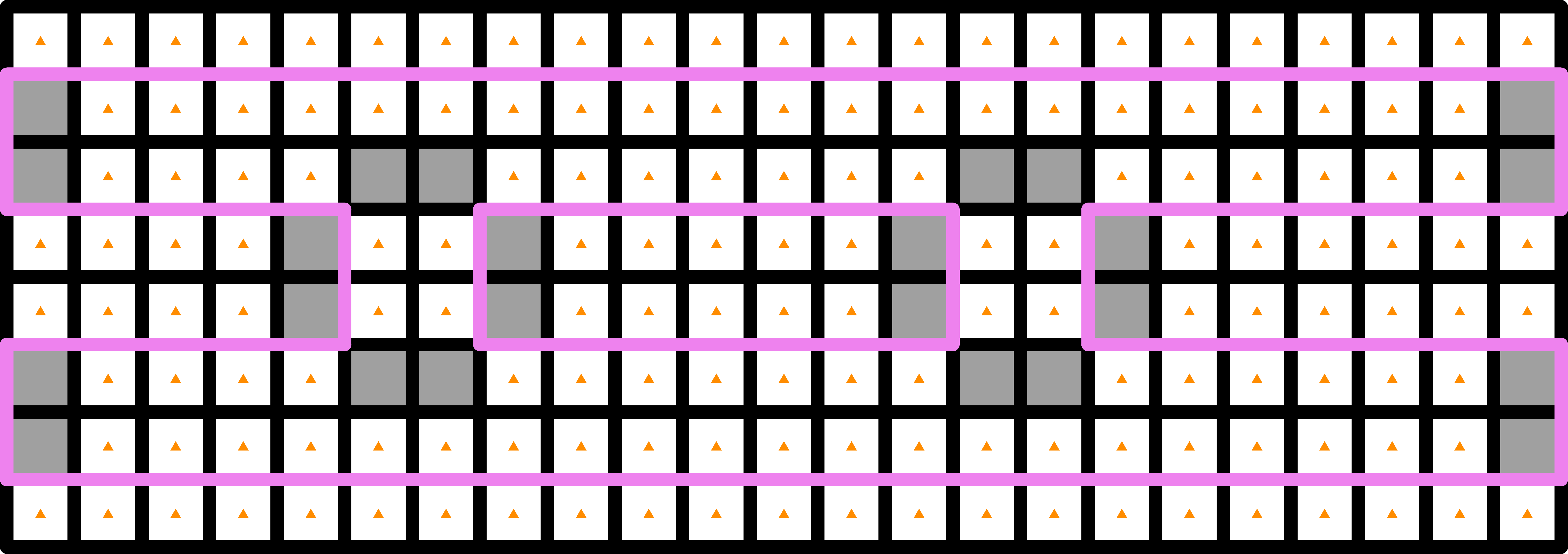}%
  \hspace{.5cm}%
}
\caption{Free-join gadget and how it interacts with the left and right
  boundaries.  In fact, only one of the two $4 \times 4$ circles are needed for
  this gadget, but this form illustrates the negative effect of joining twice.}
\label{1-triangles joining lines}
\end{figure}
  
%Thus, any solution to the leftmost and rightmost three columns looks like Figure~\ref{1-triangles edges solved}: a path in the special top row pair and disconnected cycles in the right-joined row pairs, separated by right-separated row pairs.  If these cycles can be joined into a path, it must be using the VCC and UVC gadgets.

%Most of the puzzle consists of alternating right-separated and right-joined
%row pairs.  A right-separated row pair is simply two consecutive rows of cells
%entirely full of 1-triangle clues.  A right-joined row pair is two consecutive
%rows of cells full of 1-triangle clues except that their leftmost and
%rightmost cells are empty.  \xxx[Jeffrey]{``right-separated'' means ``needs
%connection'' and ``right-joined'' means ``a cycle in any local solution''.
%Please think of better names and find-and-replace}

\paragraph{3-set gadgets.}

We do not place any gadgets in the topmost two rows of cells
(and thus do not touch the topmost horizontal line), effectively shifting
the parity of lines in all gadgets relative to the left and right sides
of the puzzle.
In the remaining rows, the top $28 m$ rows of the puzzle represent the $m$
3-sets, where each group of $28$ consecutive rows represents one $3$-set.
Intuitively, the top $4$ rows of a $3$-set are ``support'' rows
(containing exactly $2$ lines),
while the following $24 = 8 \cdot 3$ rows leave $8$ rows
(containing $4$ lines) for each of the $3$ connections to elements.
The primary \emph{3-set gadget} is shown in Figure~\ref{1-triangles gadget 1};
we place it near the left side of every 3-set.
This gadget will enforce that each set is either entirely unchosen
(Figure~\ref{1-triangles gadget 1 all solution}) or entire chosen
(Figure~\ref{1-triangles gadget 1 nothing solution}).
We also place a free-join gadget to join together the two horizontal lines
in the four support (top) rows of each 3-set.

 \begin{figure}
\centering
% 1-triangles main gadget
\subcaptionbox{\label{1-triangles gadget 1 unsolved} Unsolved.}{%
  \hspace{.1cm}%
  \includegraphics[width=.28\textwidth]{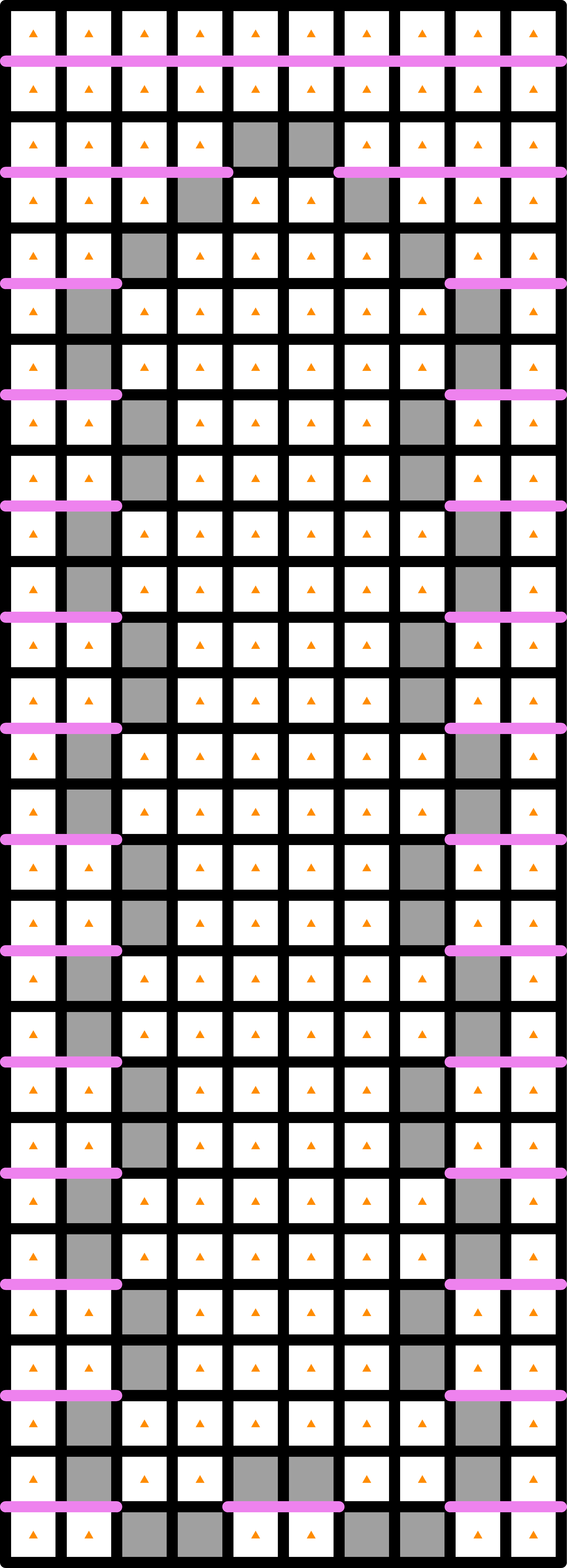}%
  \hspace{.1cm}%
}~~~
\subcaptionbox{\label{1-triangles gadget 1 all solution} Solution when the set is unchosen.}{%
  \hspace{.1cm}%
  \includegraphics[width=.28\textwidth]{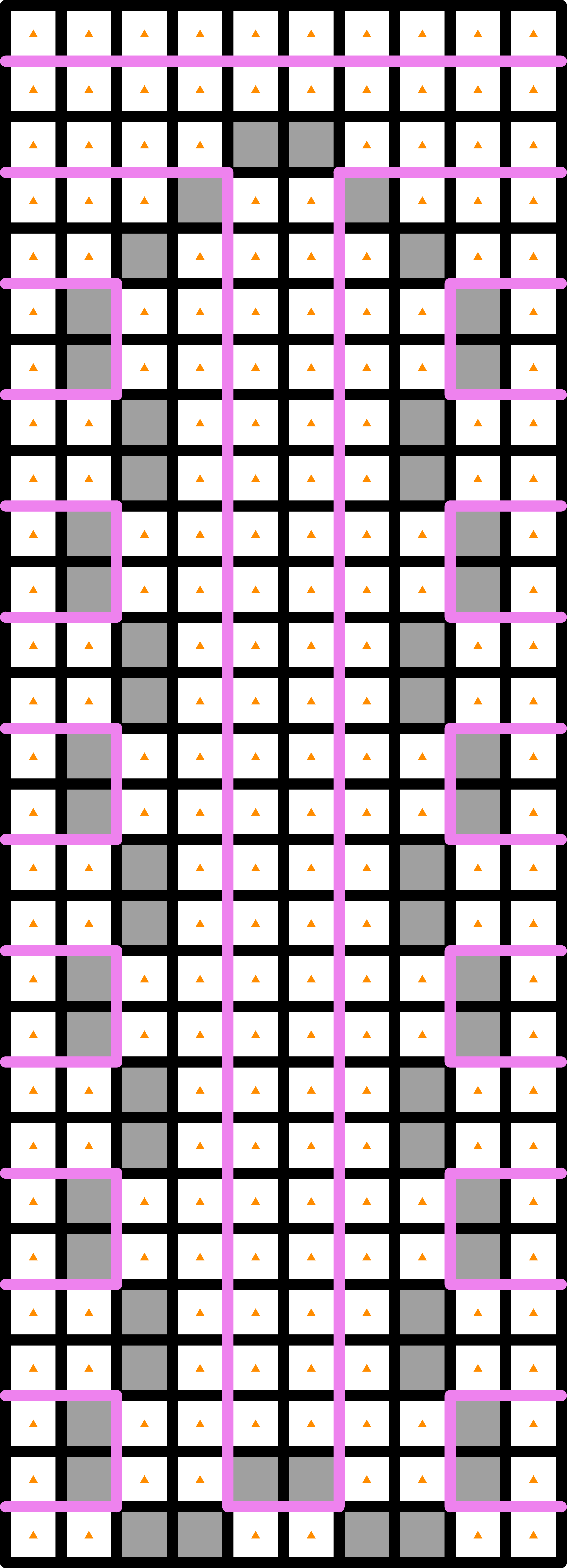}%
  \hspace{.1cm}%
}~~~
\subcaptionbox{\label{1-triangles gadget 1 nothing solution} Solution when the set is chosen.}{%
  \hspace{.1cm}%
  \includegraphics[width=.28\textwidth]{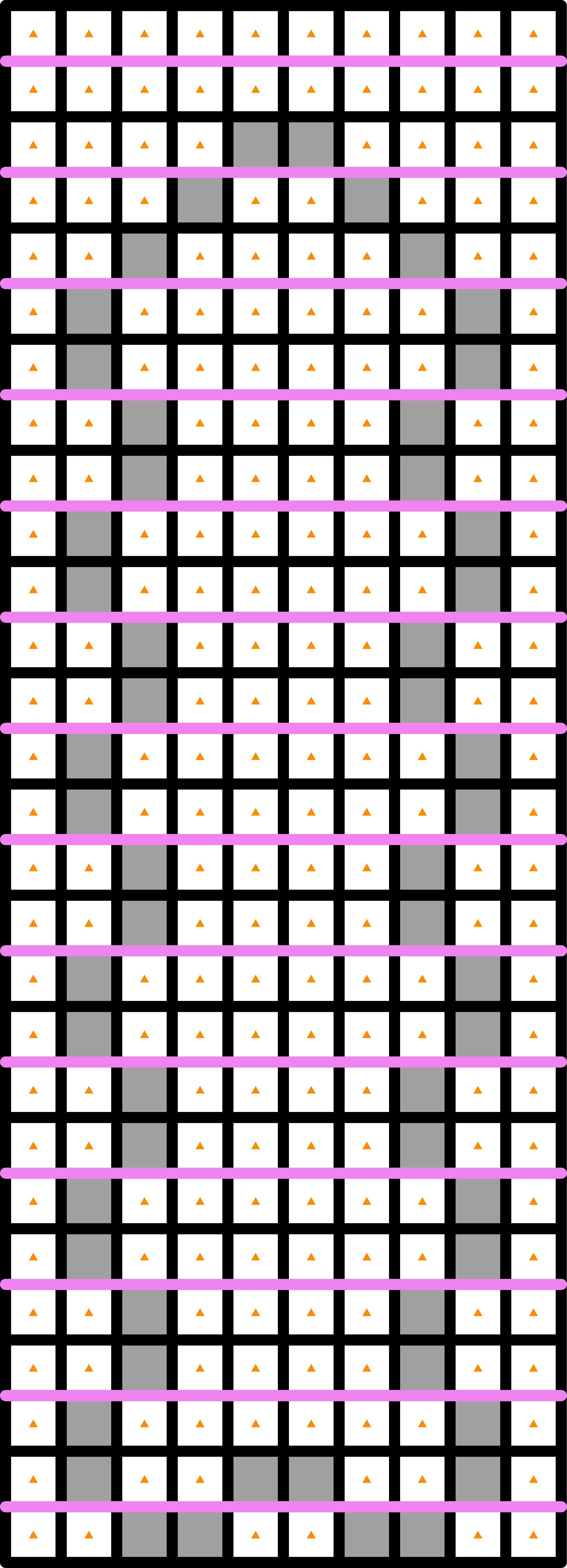}%
  \hspace{.1cm}%
}
%\caption{Unused variable cleanup gadget.}
\caption{The 3-set gadget consists of $28$ rows of cells:
  $4$ ``support'' on top to enable this gadget, followed by $8 \cdot 3$ rows,
  where each group of $8$ rows connects to one element.}
\label{1-triangles gadget 1}
\end{figure}  

\paragraph{Element gadgets.}

The following $8 n$ rows, which reach down to all but the bottommost two
rows of the puzzle, represent the elements.
Each element is represented simply by $8$ consecutive rows
(containing $4$ horizontal lines).

\paragraph{Connecting sets to elements.}

For each occurrence of an element in a set, we connect the $8$ rows of the
element to the corresponding $8$ rows of the set
(among the $24 = 3 \cdot 8$ nonsupport rows)
using the \emph{set--element connection gadget}
in Figure~\ref{1-triangles_gadget_2}.
This gadget has two solutions.
When the set is not chosen, all lines proceed horizontally through the gadget.
When the set is chosen, the solution is as shown in
Figure~\ref{1-triangles_gadget_2_all_solution}, which connects together
adjacent pairs of horizontal lines among the $4$ lines in the set and
the $4$ lines in the element, while leaving all other horizontal lines
unaffected (topologically in terms of which ends they connect,
even though they bend geometrically).
Each element connects to potentially several sets in this way,
but referring to Figure~\ref{1-triangles joining lines},
the element's lines will be properly connected if and only if
exactly one set containing the element is chosen while the rest are unchosen.

 \begin{figure}
\centering
% 1-triangles main gadget
\subcaptionbox{\label{1-triangles_gadget_2_unsolved} Unsolved.}{%
  \hspace{.1cm}%
  \includegraphics[width=.75\textwidth]{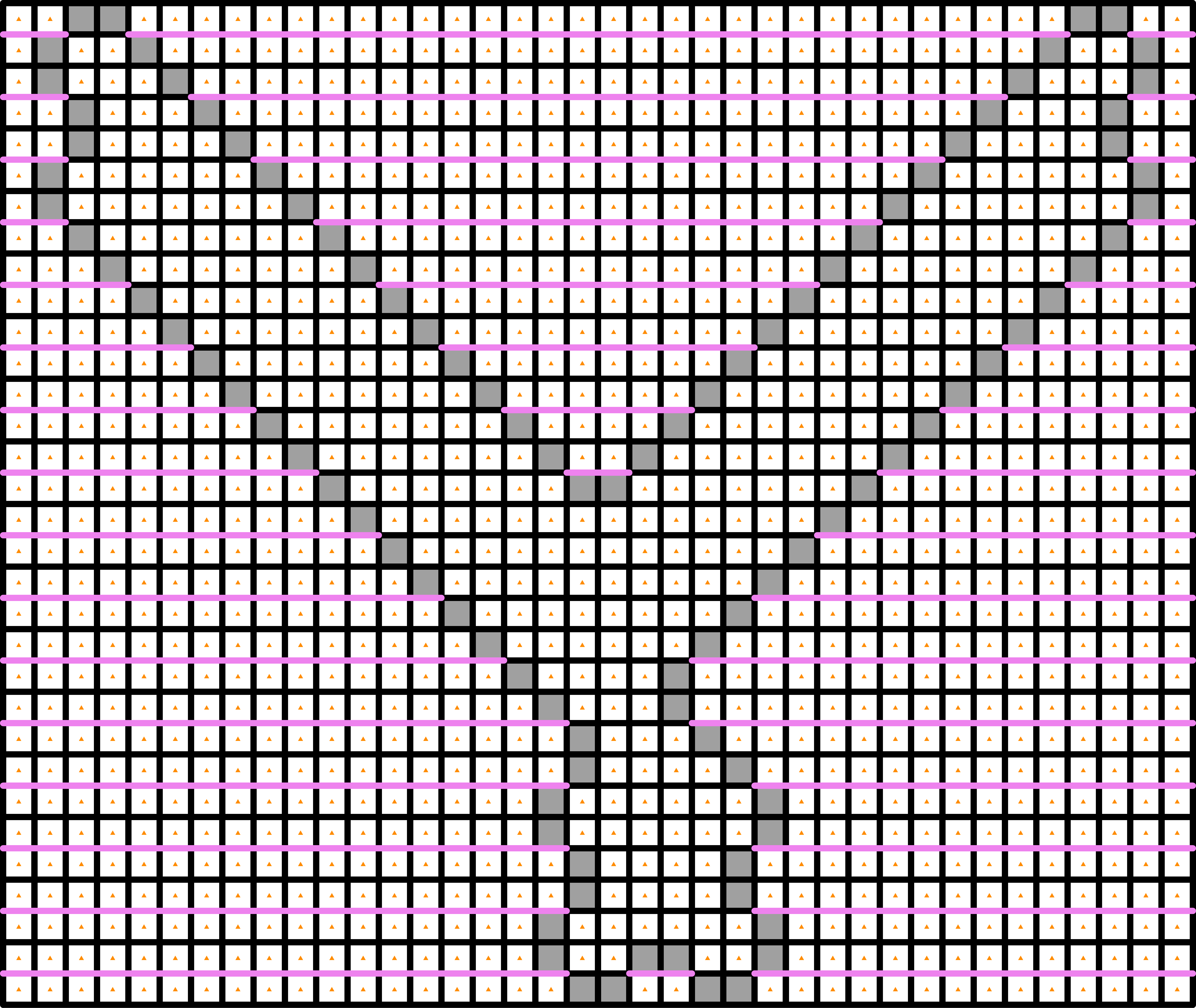}%
  \hspace{.1cm}%
}

\subcaptionbox{\label{1-triangles_gadget_2_all_solution} Solution when the set is chosen.}{
  \hspace{.1cm}%
  \includegraphics[width=.75\textwidth]{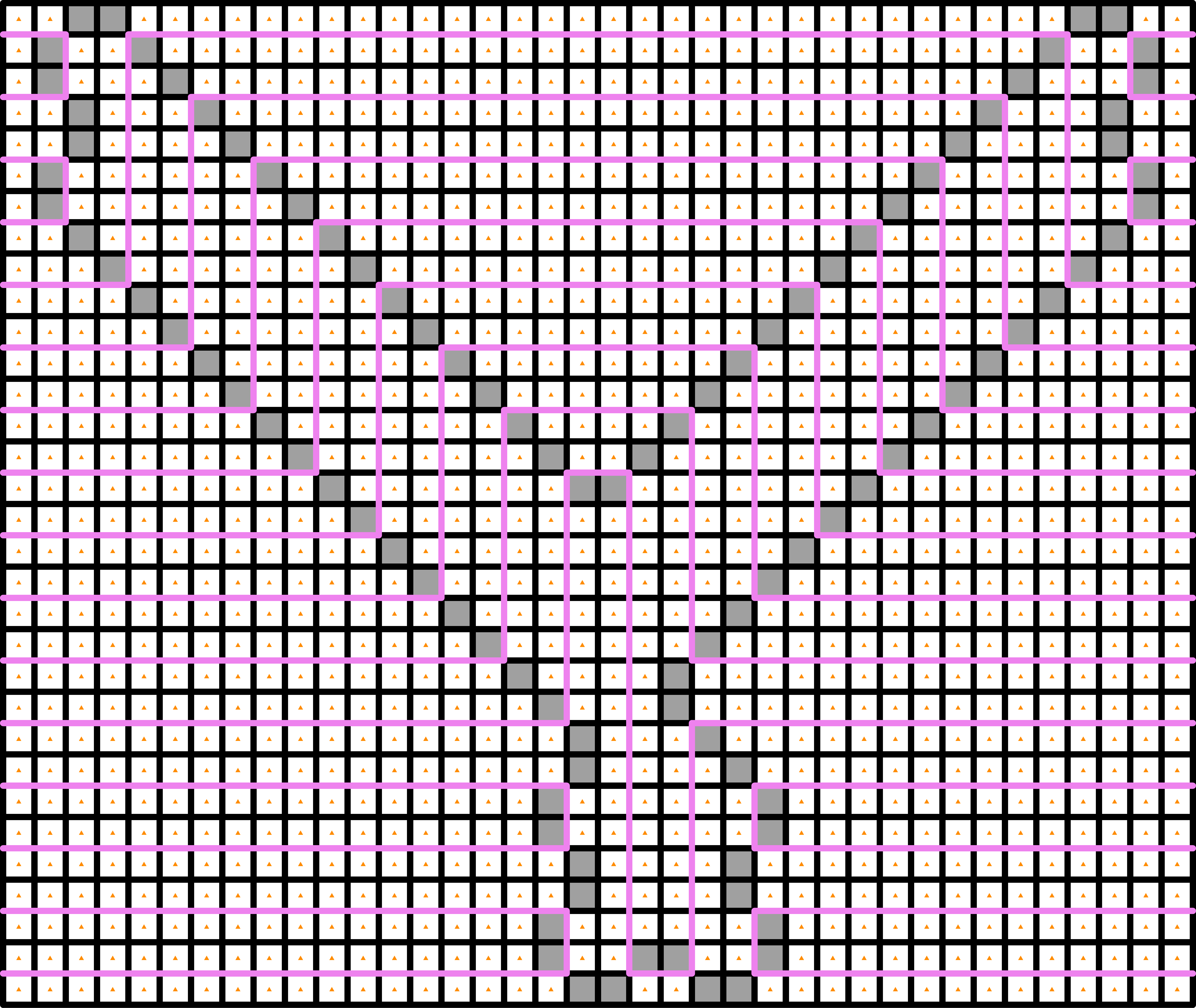}%
  \hspace{.1cm}%
}
%\subcaptionbox{``Nothing'' or ``horizontal'' solution.}{
%\hspace{.1cm}
% \includegraphics[width=.28\textwidth]{figures/1_triangles_Gadget_1_Solution_1.pdf}
%\hspace{.1cm}
%\label{1-triangles_gadget_2_nothing solution}
%}
\caption{Set--element connection gadget.  When this set is unchosen, the solution simply proceeds horizontally across the gadget.}
\label{1-triangles_gadget_2}
\end{figure}

%Note that the set--element connection gadget extends one row above and one row
%below the right-separated row pairs it potentially connects.
The set--element connection gadget can be extended to larger sizes than
drawn in Figure~\ref{1-triangles_gadget_2} by adding an equal number of
rows and columns, and extending the pattern of the four ``diagonals'' of empty
cells in the top half of the gadget.  Similarly, it can be reduced in size.
This resizing enables us to connect together any set's $8$ rows
with any element's $8$ rows.
Because the width of the gadget depends upon its height, the nonoverlapping
placement of these gadgets affects the total number of columns in the puzzle.
Because we use only a polynomial number of polynomial-sized gadgets,
the constructed puzzle remains polynomial size.

\paragraph{Equivalence to X3C.}
As argued above, the partial solution in Figure~\ref{1-triangles edges solved}
joins together into a single solution path if and only if
every even-numbered horizontal line (skipping the first line)
has exactly one join with the line immediately below it
(as in Figure~\ref{1-triangles joining lines}).
The top two lines of each 3-set can always join together
(exactly once) via their free-join gadget (and in no other way).
The remaining ten lines of each 3-set are joined exactly once
if and only if the 3-set gadget and all three incident set--element
connection gadgets either all use the ``unchosen'' solution
(in which case the connections are from the 3-set gadget)
or all use the ``chosen'' solution
(in which case the connections are from the set--element connection gadgets).
The four lines of each element are joined exactly once
if and only if exactly one of the incident set--element connection gadgets
uses the ``chosen'' solution.
Therefore, solving the Witness puzzle is equivalent to
solving the X3C instance.
%
%\xxx{This reduction is parsimonious, too: there's no flexibility in any of the gadget solutions}
\end{proof}

\subsection{2-Triangle Clues}
\label{sec:2-triangles}

As with 1-triangle clues, any rectangle containing only 2-triangle clues can be locally satisfied by a set of disjoint path segments containing either all the vertical edges or all the horizontal edges in the rectangle.  However, compared to 1-triangle-only puzzles, locally-satisfying paths can turn with more flexibility even in areas completely filled with 2-triangle clues.  Accordingly, our proof for puzzles containing only 2-triangle clues is based around connecting concentric cycles instead of straight path segments.

\begin{theorem}
\label{thm:triangles2}
It is NP-complete to solve Witness puzzles containing only 2-triangle clues.
\end{theorem}

\begin{proofsketch}
We reduce from X3C, making use of the fact that the solution path must be a single closed path.  The puzzle is almost completely filled with 2-triangle clues.  Local conditions alone force any solution to form disconnected concentric cycles, but we can build gadgets by deleting 2-triangle clues.  Elements are represented by rows in the bottom quadrant of the puzzle, where sets are represented by gadgets intersecting these rows (see Figure~\ref{2-triangles full puzzle}).  Used set gadgets connect the concentric cycles of their element rows, and cleanup gadgets allow connecting the cycles not corresponding to elements.  The cycles can be connected to form a single solution path exactly when the X3C instance has a solution.
\end{proofsketch}

\begin{proof}
As in the 1-triangle reduction of Theorem~\ref{thm:triangles1},
we reduce from X3C.
Consider an instance $(X,C)$ with $|X| = n$ elements and $|C|=m$
cardinality-3 subsets of $X$, and assume without loss of generality that
$X = \{0,1,\dots,n-1\}$.

Similar to the previous reduction, we start by considering a puzzle completely filled with 2-triangle clues together with a clue-satisfying set of disjoint paths, namely concentric squares in a square grid, and placing gadgets (arrangements of empty cells) such that these paths can be joined into a full solution path if and only if the X3C instance has a solution.

\paragraph{Wave propagation.}

\def\2triscale{0.3}
\begin{figure}
\centering
\subcaptionbox{\label{fig:2-triangles-waveprop} A turn between clues forces further turns in both directions along the ray.}{
  \includegraphics[scale=\2triscale]{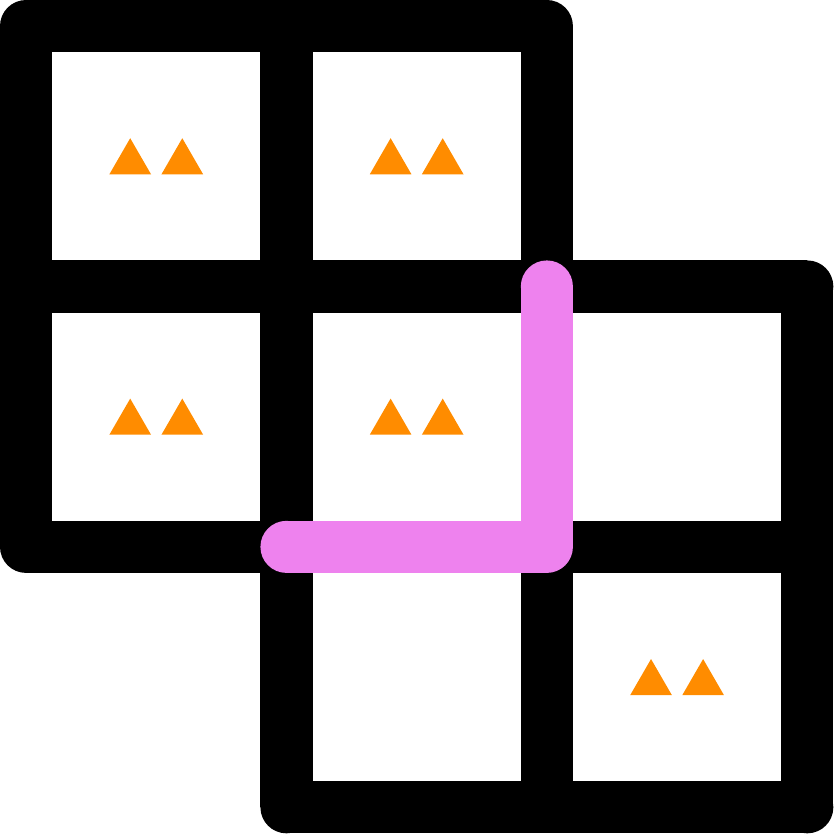}
  \hspace{0.3cm}
  \includegraphics[scale=\2triscale]{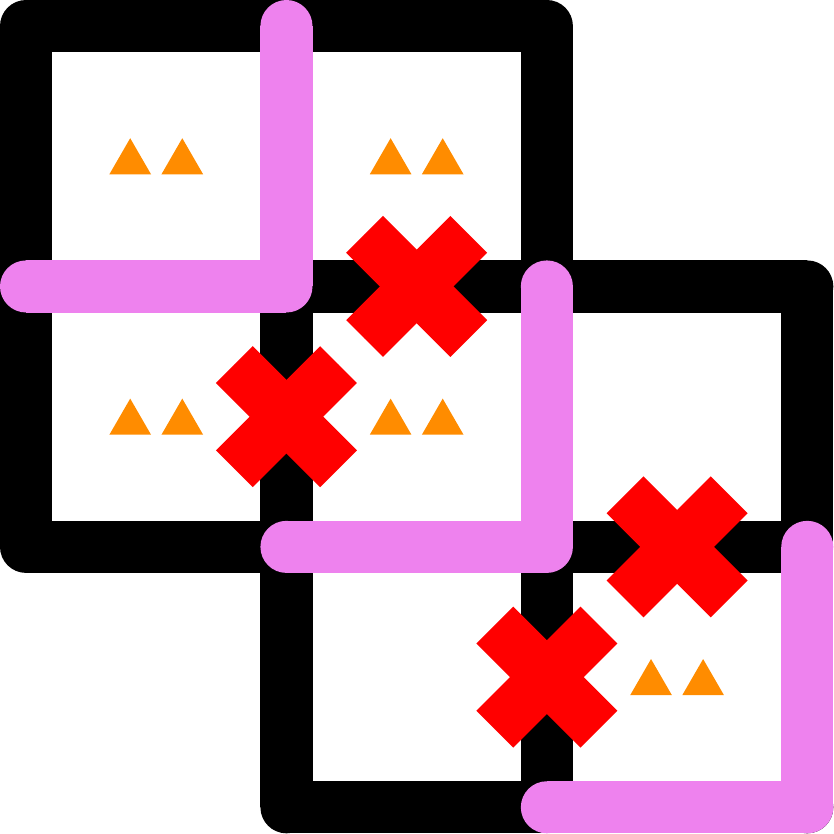}
}~~~~~~~
\subcaptionbox{\label{fig:2-triangles-corner} The corners of the puzzle always emit one or three rays.}{
  \includegraphics[scale=\2triscale]{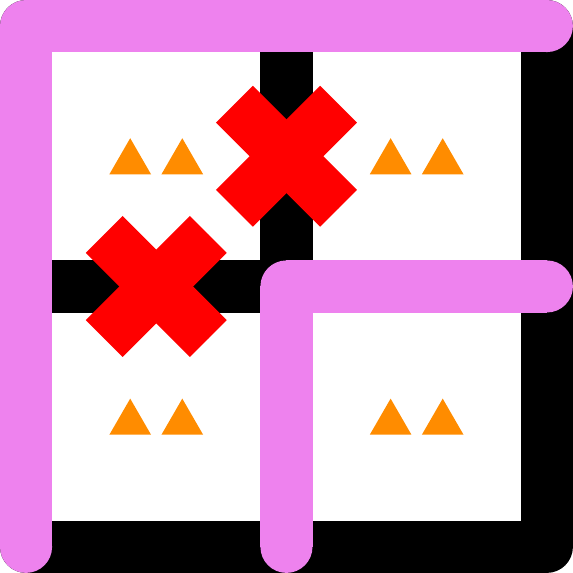}
  \hspace{0.3cm}
  \includegraphics[scale=\2triscale]{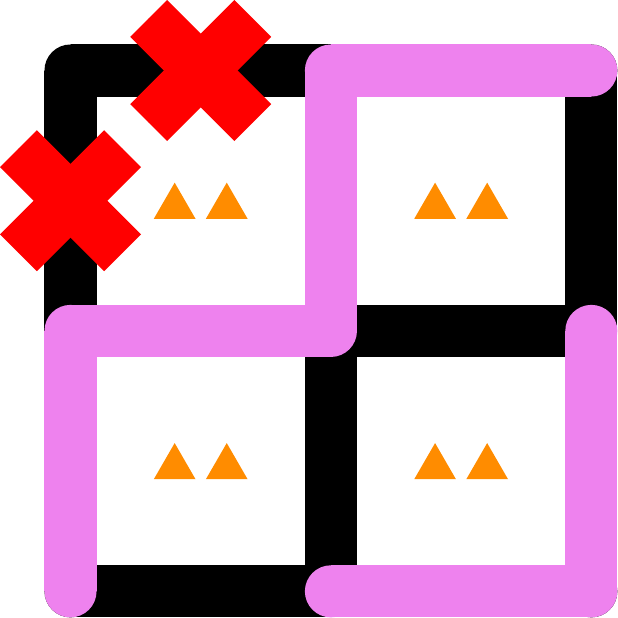}
}

\subcaptionbox{\label{fig:2-triangles-reflection} Rays can reflect off the puzzle boundary in two ways (here, the top boundary).}{
  \includegraphics[scale=\2triscale]{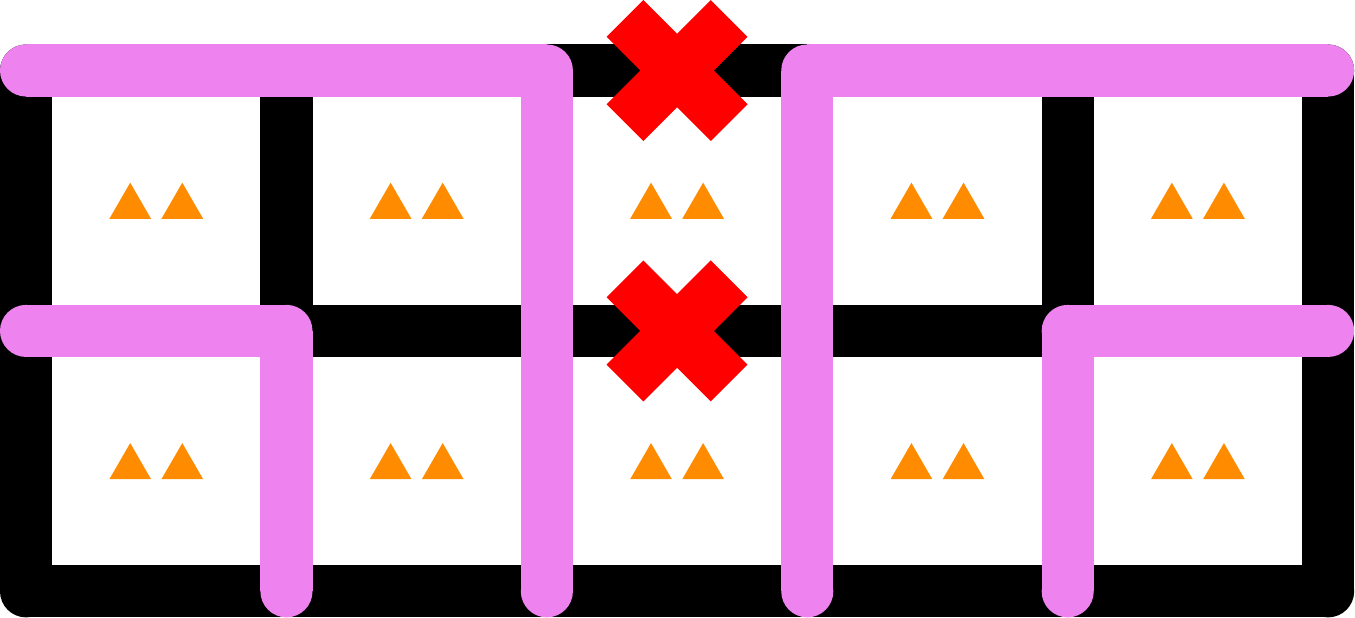}
  \hspace{0.2cm}
  \includegraphics[scale=\2triscale]{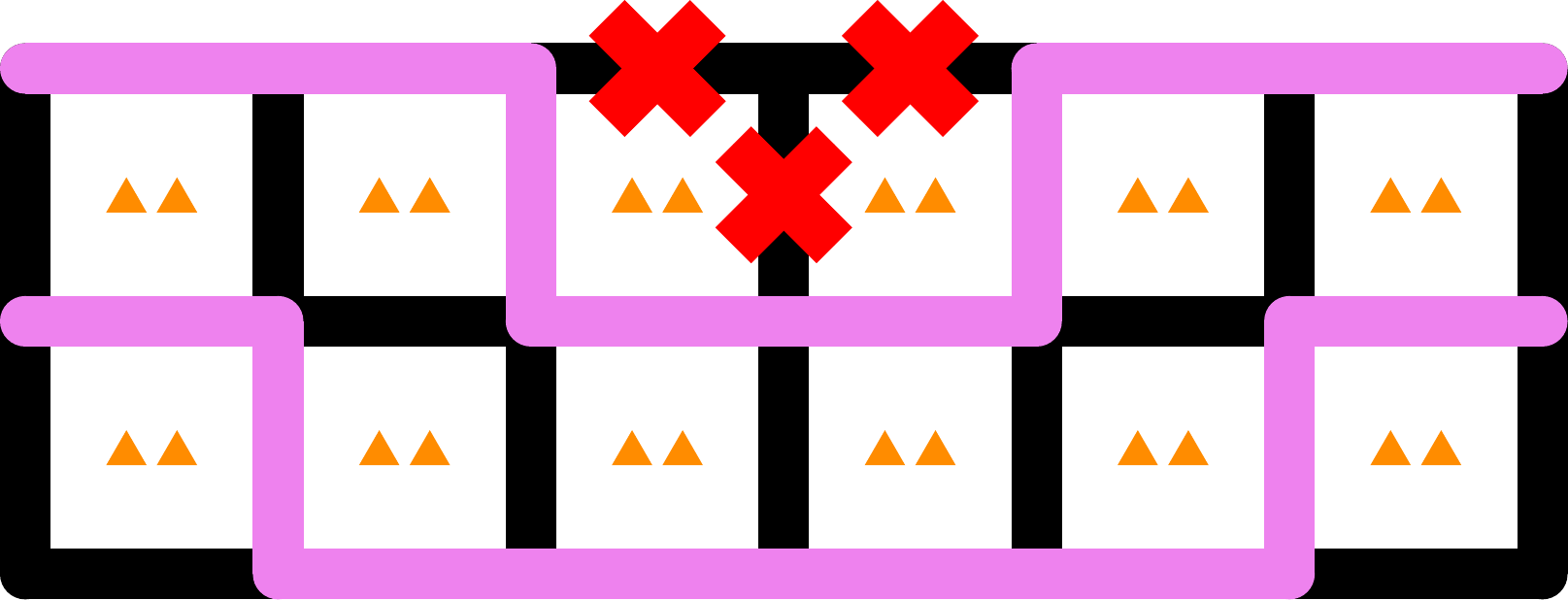}
}

\caption{Local rules for 2-triangle clues regarding wave propagation.}
\label{2-triangles diag prop}
\end{figure}

A path can turn in puzzle areas filled with 2-triangle clues, but when it does so, it becomes a wavefront that continues to propagate, turning along a ray directed based on where the turns ``point'' (see Figure~\ref{fig:2-triangles-waveprop}).  When a path turns on a 2-triangle clue cell, and the ray points to a diagonally adjacent 2-triangle clue cell, then the path must also turn in the same direction on that cell, extending the ray outwards.  Similarly, the wave propagates inward (opposite the ray) when there are only 2-triangle clues nearby, as the path would otherwise produce an isolated $2 \times 2$ square.  The solution path necessarily turns at the corner of the puzzle, so each corner emits one ray (directed into the corner) or three rays (one out and two in) along the main diagonals of the puzzle (see the top-right of Figure~\ref{fig:2-triangles-corner}).

Wave propagation imposes structure on any possible solution path: in an area full of 2-triangle clues, the solution path divides the area into horizontally- and vertically-oriented subareas (containing parallel horizontal or vertical path edges respectively) separated by the rays of the waves.  By the propagation rules, rays can only start and end at either an empty cell or the boundary of the puzzle, and cannot intersect other rays propagating in a different direction except at an empty cell or at the boundary of the puzzle.

\paragraph{Overall layout.}

Figure~\ref{2-triangles full puzzle} shows the overall layout of the produced Witness puzzle. The rays emitted from the corners of the puzzle divides the puzzle into four quadrants.  Because the puzzle is full of 2-triangle clues except in our gadgets, the path cannot turn except along those main diagonal rays, so the initial ``solution'' to the puzzle is a set of concentric square cycles.  In the remainder of the proof, we add gadgets (remove 2-triangle clues) to enable connecting these cycles into a single path exactly when the X3C instance has a solution.

We associate some pairs of concentric squares with elements in the X3C instance and add X3C gadgets that allow connecting the square pairs associated with each of the elements in a 3-set.  The filler diamond gadgets connect the cycles between the element-representing cycles.  The bulk connector gadgets connect the cycles in the rest of the puzzle and resolve the collision of the main diagonal rays at the very center of the puzzle.

\begin{figure}
\centering
\includegraphics[width=.7\textwidth]{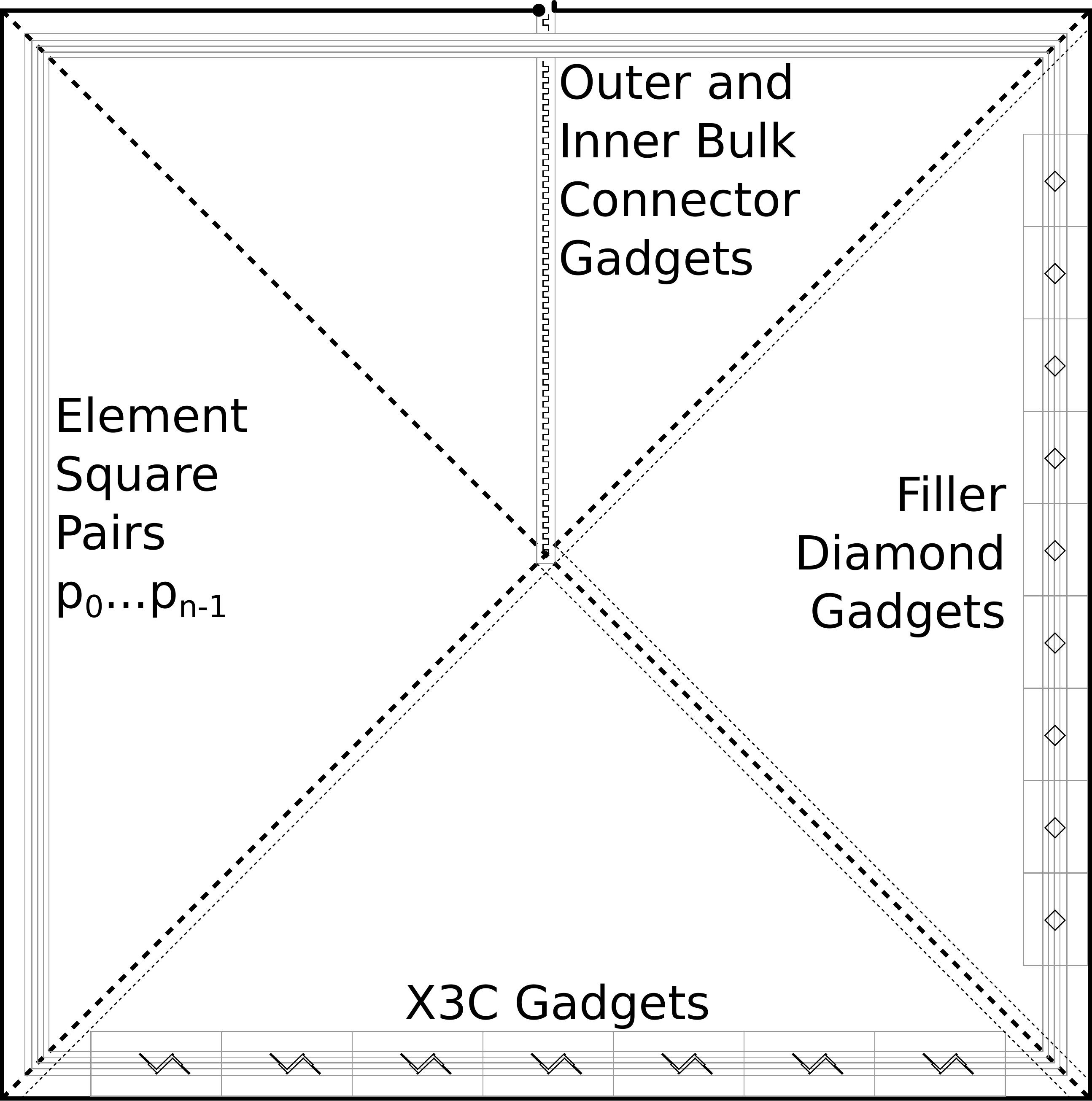}
\caption{The layout of gadgets within the X3C reduction puzzle. Quadrants defined by corner rays separated by thick dashed lines, and ray escape risk areas by thin dashed lines.}
\label{2-triangles full puzzle}
\end{figure}

Each of these gadgets contains some empty cells, so there is a risk that rays could start in one gadget and end in another, or travel between distant empty cells in the same gadget in unintended ways, disrupting the pattern of concentric squares.  To contain the rays, we place gadgets such that there are no empty cells outside the gadget along the diagonals of the gadget's empty cells (see Figure~\ref{2-triangles x3c exterior}).  We also take into account possible reflections off of the puzzle boundary (see Figure~\ref{fig:2-triangles-reflection}) by adding buffer space below the gadget and protecting additional diagonals to its left and right.  By placing gadgets in this way, any ray escaping a gadget will crash into one of the main diagonal rays (possibly after reflecting off of the puzzle boundary); as there are no empty cells along the main diagonals (except the very center cell), there is no way to avoid violating the 2-triangle clue at the ray intersection.

We isolate the different types of gadgets from one another by placing them in different quadrants (separated by a main diagonal).  Besides being convenient for layout, using separate quadrants is necessary because the bulk connector gadget contains empty cells on most diagonals.

\begin{figure}
\centering
\includegraphics[width=.7\textwidth]{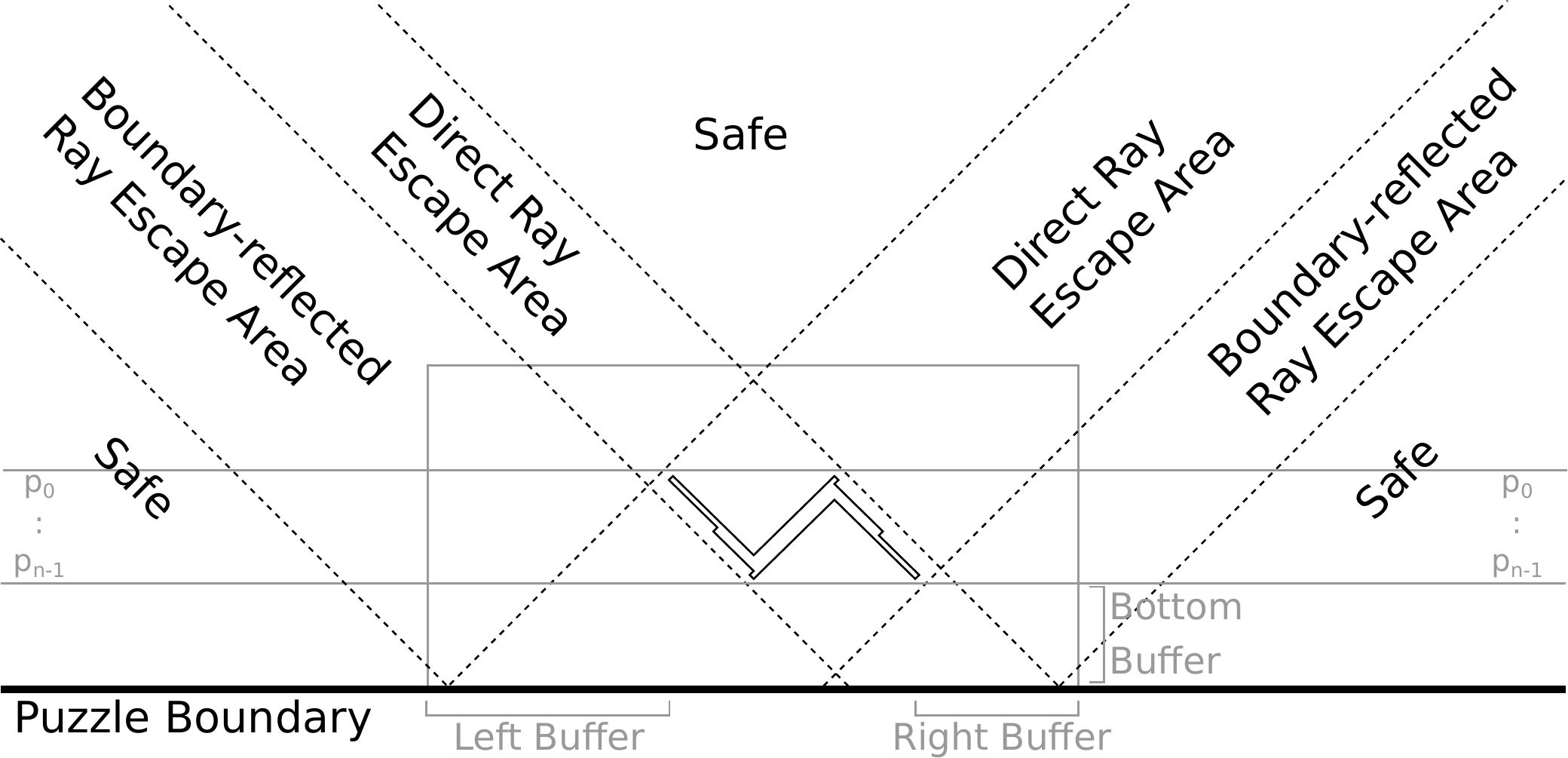}
\caption{How to safely place a gadget (e.g., an X3C gadget) anywhere within $p_0$ to $p_{n-1}$. Possible ray escape areas are delimited by dashed lines, and the reserved space for the gadget by a gray box.}
\label{2-triangles x3c exterior}
\end{figure}

\begin{figure}
\centering
\includegraphics[scale=\2triscale]{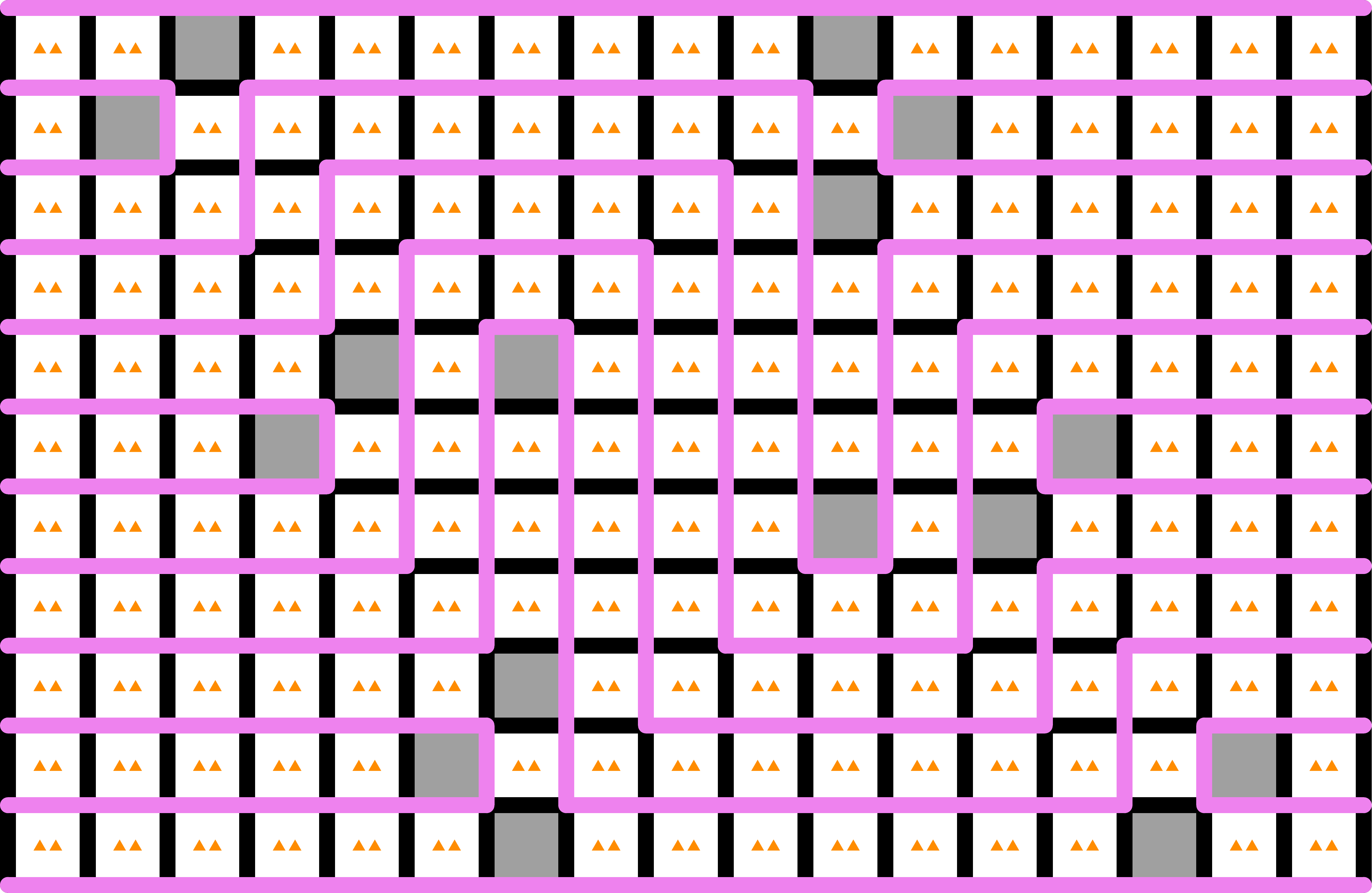}
\caption{The X3C gadget for $\{x,x+1,x+2\}$, solved to join 3 pairs of adjacent paths $p_x$, $p_{x+1}$, and $p_{x+2}$. All other traversing paths' connections are unaffected.  In the other solution, all paths continue horizontally across the gadget.
%The vertically oriented interior of the gadget is shaded for clarity.
}
\label{2-triangles x3c used}
\end{figure}

\begin{figure}
\centering
\includegraphics[width=.7\textwidth]{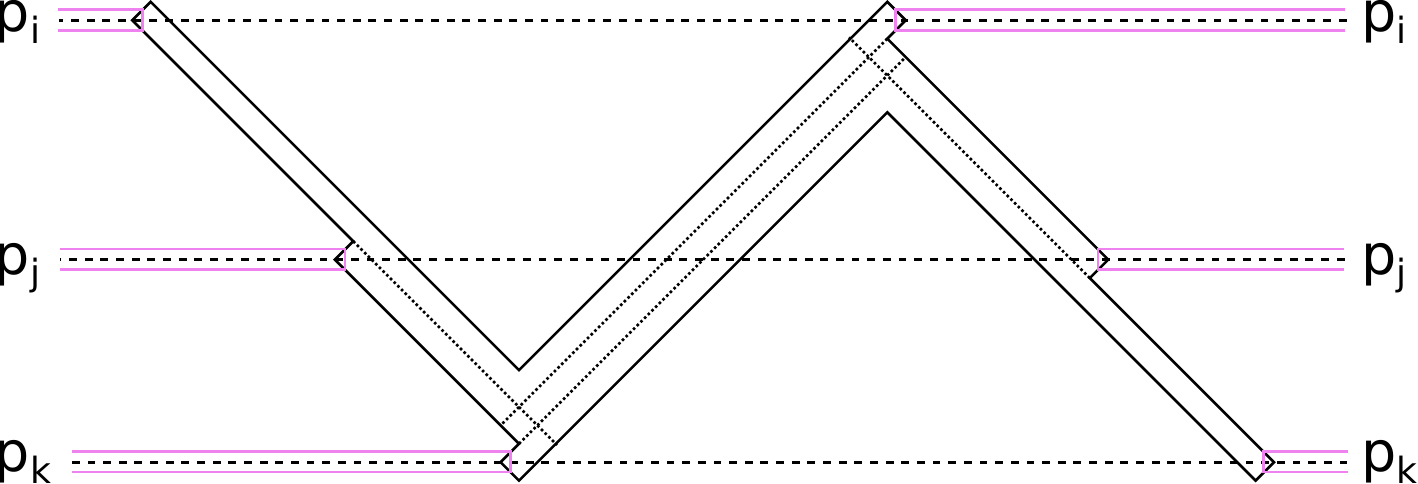}
\caption{A blueprint for constructing an X3C gadget for any $\{i,j,k\}$. The polygon represents the border of the vertically oriented interior of the gadget: its points are the centers of the empty cells, and its edges are the rays, all on unique diagonals of the grid. The dotted lines mark the width of each interior section.}
\label{2-triangles x3c blueprint}
\end{figure}

\paragraph{X3C gadget.} For each $x \in X$, we associate an adjacent pair of concentric squares $p_x$ with radii $r_x$ and $r_x+1$, separated from other pairs such that $r_{x+1} = 4 + r_x$.
%In our construction, these squares may only connect into a single solution path by using X3C gadgets, which represent the subsets in~$C$, whereas all other adjacent squares may freely join into large paths using diamond gadgets and a connector gadget.
For each $\{i,j,k\} \in C$, we place an X3C gadget, which is built to permit one nontrivial solution where each of $p_i$, $p_j$, and $p_k$ merge, representing using the set in the cover.  The trivial solution, with all paths continuing horizontally through the gadget, represents not using the set in the cover.
Figure~\ref{2-triangles x3c used} shows the smallest X3C gadget, and a general X3C gadget can be constructed by stretching it to cover any $p_i$, $p_j$, and $p_k$ pairs, as detailed in Figure~\ref{2-triangles x3c blueprint}.
Each empty cell in an X3C gadget shares a diagonal with exactly two other empty cells, each on a different diagonal.  Because our overall layout leaves no way to terminate a ray that leaves the gadget, when the nontrivial solution is used, the rays are forced to trace out the unique polygon that merges the three pairs of paths and preserves connectivity of other paths.

\paragraph{Diamond gadget.}  The diamond gadget, shown in Figure~\ref{fig:2-triangles-diamond}, is the simplest gadget that can connect adjacent cycles in a controlled way.  Again, due to our overall layout, only the trivial fully-horizontal/vertical local solution and the nontrivial local solution shown in the figure are possible in a global solution.  We place diamond gadgets in the right quadrant to connect the three adjacent pairs between each $p_x$ and $p_{x+1}$.  As these gadgets are the only way to connect those pairs, the nontrivial local solution is always used.

\begin{figure}
\centering
\subcaptionbox{\label{fig:2-triangles-diamond-unused} Unused (not joining any cycles).}{
  \includegraphics[scale=\2triscale]{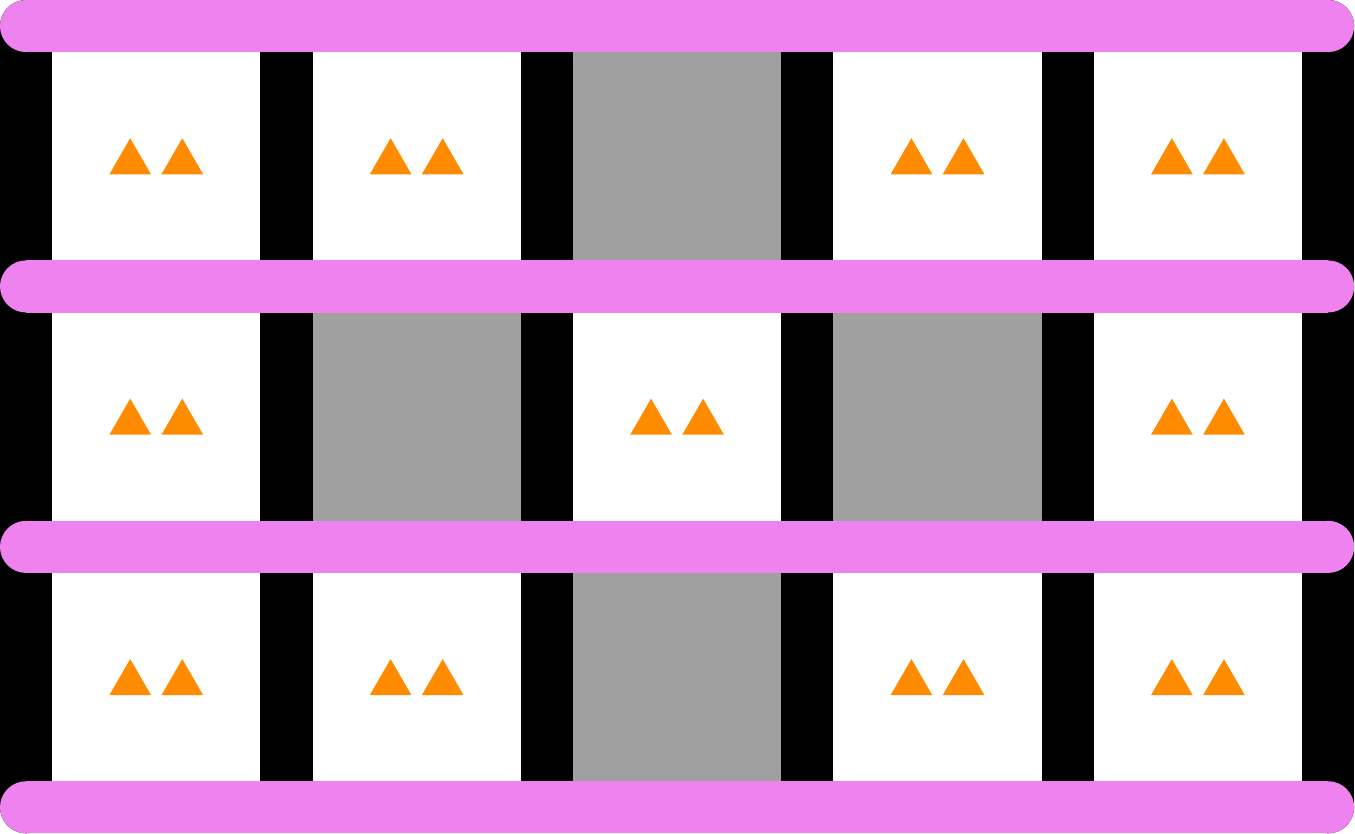}
}~~~~
\subcaptionbox{\label{fig:2-triangles-diamond-used} Used to join the center pair of cycles into a single cycle.}{
  \includegraphics[scale=\2triscale]{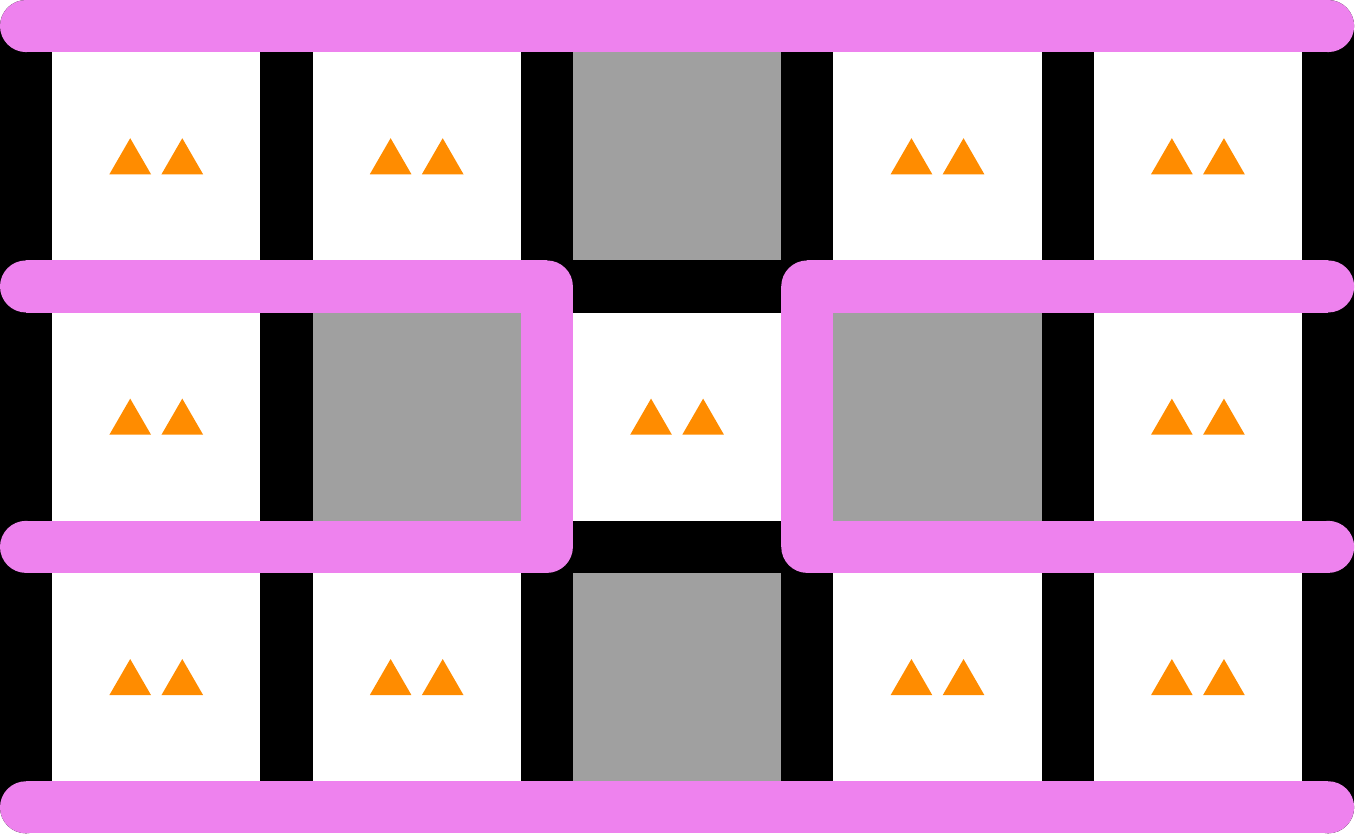}
}
\caption{The two local solutions of a diamond gadget.}
\label{fig:2-triangles-diamond}
\end{figure}

\paragraph{Bulk Connector Gadget.} The bulk connector gadget is an extensible pattern connecting all concentric squares smaller or larger than a desired radius (see Figure~\ref{fig:2-triangles-bulk}).  The X3C gadgets take care of connecting $p_0,\dots,p_{n-1}$, and the diamond gadgets connect the pairs in between; we use bulk connector gadgets in the top quadrant to join all squares with radii from $1$ to $r_0$ and from $r_{n-1}+1$ to the boundary.  This places an empty cell at the center of the puzzle to terminate the main diagonal rays.
%TODO: We could really use more description of the Bulk Connector Gadget, it took jaysonl a while to understand why it worked.

\begin{figure}
\centering
% Writing note: These won't fit side-by-side at \2triscale, and in this orientation we want them to be side by side.  There is a complicated solution at https://tex.stackexchange.com/q/318067 to find a common scale factor fitting \textwidth, but here we just make something up.
\subcaptionbox{\label{fig:2-triangles-bulk-center} The inner bulk connector gadget.  The leftmost empty cell is the puzzle's center cell.}{
  \includegraphics[scale=0.25]{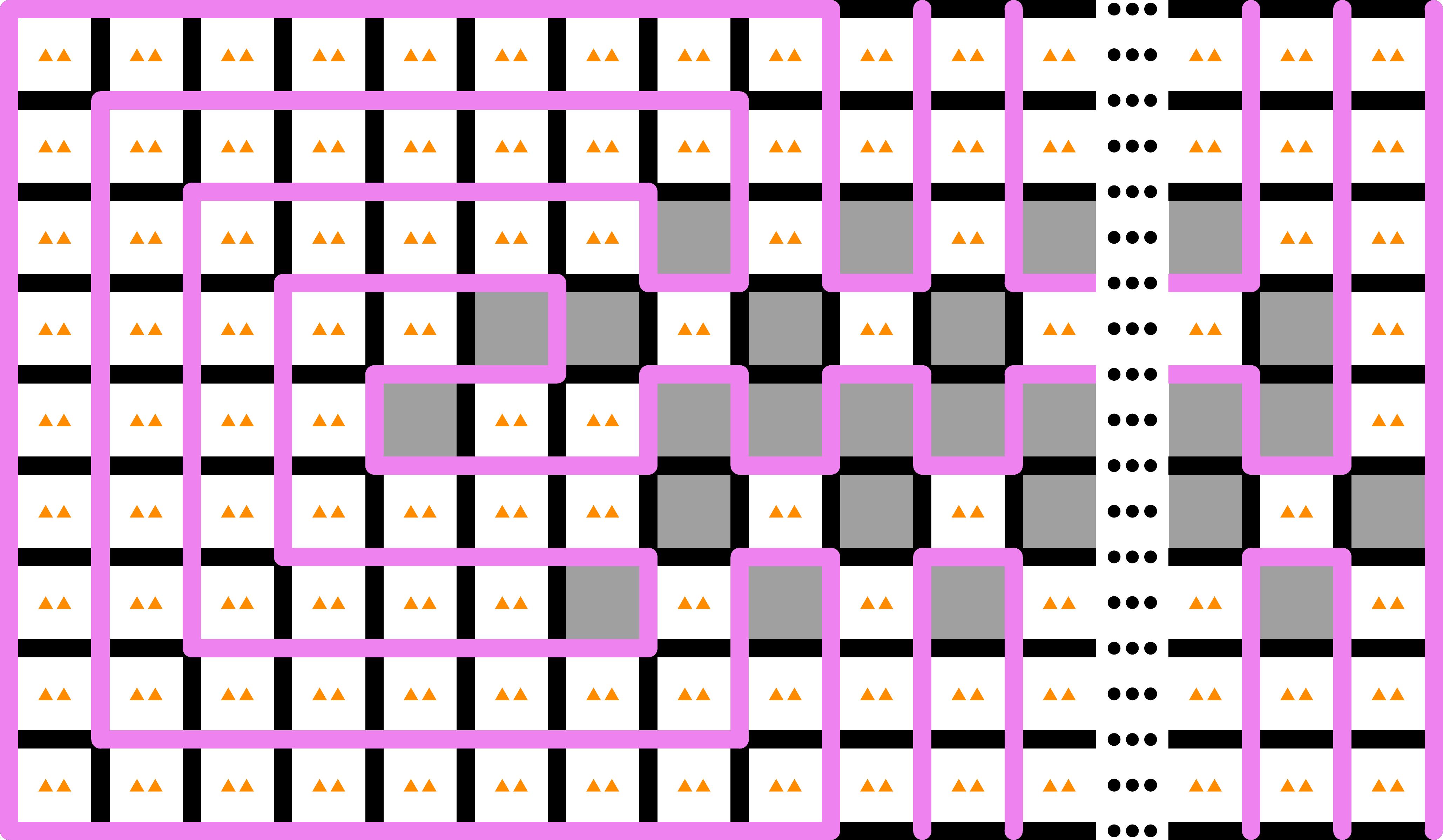}
}~~~~~~~~
\subcaptionbox{\label{fig:2-triangles-bulk-startend} The outer bulk connector gadget, incorporating the start circle and end cap.}{
  \includegraphics[scale=0.25]{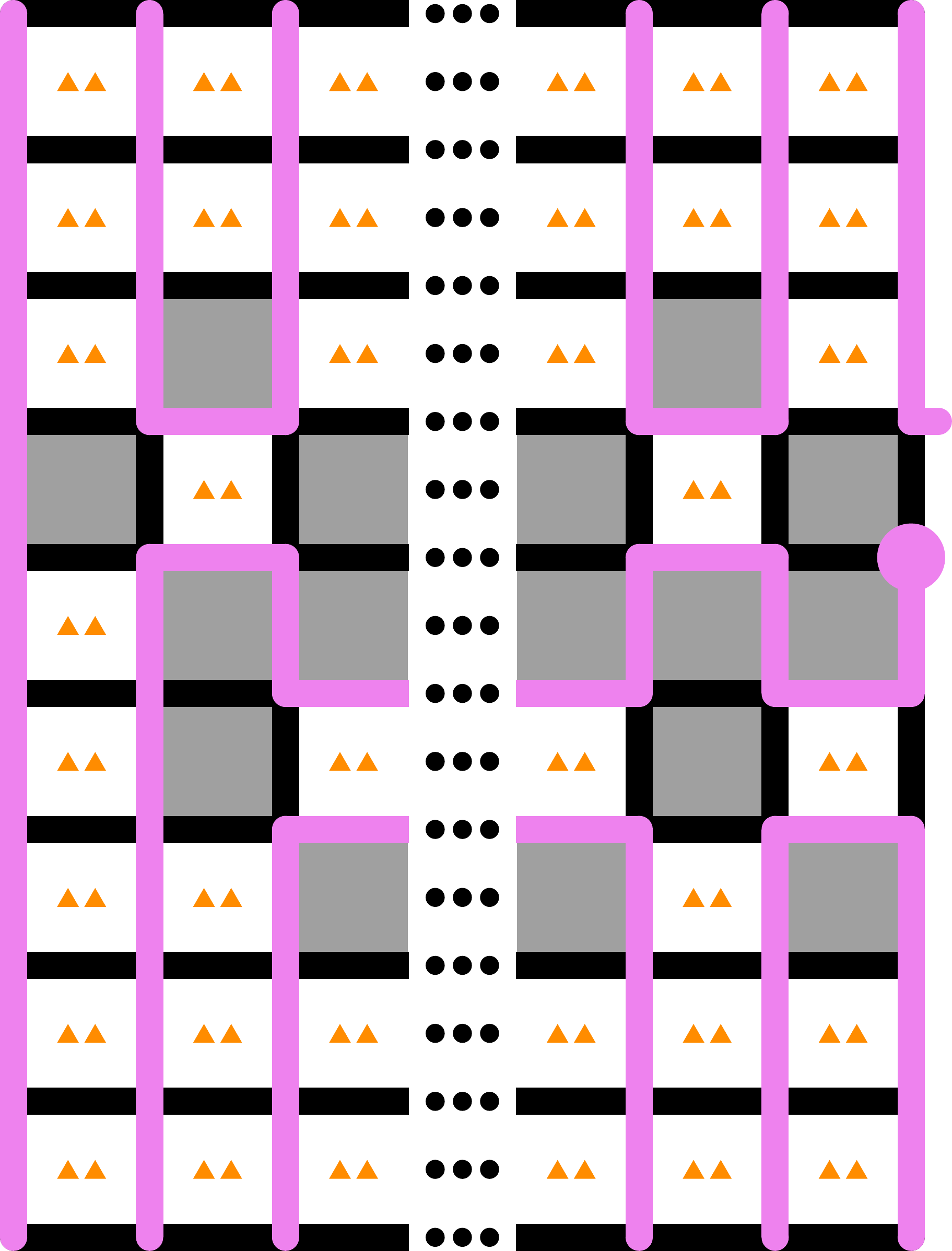}
}
\caption{The bulk connector gadget used to connect all concentric squares smaller or larger than a desired radius.  The bulk connector gadget can be instantiated in any orientation; in the orientation shown here, vertical pairs are connected.  The concentric squares representing elements lie between the inner and outer bulk connector gadgets.}
\label{fig:2-triangles-bulk}
\end{figure}

%%%%%%%%%%%%%%%% proof argument

\paragraph{Analysis.} If $(X,C)$ is a \textsc{yes} instance of X3C, then the puzzle can be solved using the partition $P \subseteq C$ of $X$. Start with the concentric squares set of disjoint paths, which satisfies all two-triangle clues, then use the bulk connector gadget and all filler diamond gadgets to join all adjacent squares except $p_0,\dots,p_{n-1}$. For each $\{i,j,k\} \in P$, use the corresponding X3C gadget to join $p_i$, $p_j$, and $p_k$. Because each $x \in X$ appears in exactly one $S \in P$, each $p_x$ will be joined exactly once. Because using a gadget preserves clue satisfaction, and every concentric square is joined into one path, the path is a solution to the Witness puzzle.

Likewise, if the puzzle has a solution path, then $(X,C)$ must be a \textsc{yes} instance of X3C. As argued, gadgets are sufficiently spaced apart and the main diagonals of the puzzle isolate each quadrant, so any ray leaving a gadget would intersect with a main diagonal ray (possibly after reflecting off of the puzzle boundary), violating the 2-triangle clue at the intersection.
Because X3C and diamond gadgets only have one nontrivial solution not emitting any rays, and using X3C gadgets is the only way to join adjacent concentric squares $p_0,\dots,p_{n-1}$ without violating 2-triangle clues, the solution path defines a $P \subseteq C$ by its use of the X3C gadgets. $P$ must cover every $x \in X$ otherwise the solution path would be disjoint, and it cannot double-cover any $x \in X$ otherwise the solution path would create a disconnected cycle out of the concentric squares in $p_x$, so $P$ is a partition of~$X$.

With respect to the complexity of the reduction, notice that the bounding box of the X3C gadget for $\{i,j,k\}$ is $\Theta(k-i) \times \Theta(k-i)$ cells, for $i < j < k \leq n-1$. The size of the buffer space surrounding each X3C gadget is $\Theta(n) \times \Theta(n)$ cells and the total lineup of X3C gadgets is $\Theta(mn)$ cells wide and $\Theta(n)$ tall. Similarly, each of the $3n$ filler diamond gadgets takes up $\Theta(n) \times \Theta(n)$ cells of buffer space, for a total of $\Theta(n^2)$ cells tall and $\Theta(n)$ cells wide. Thus, the side length of the puzzle is $\Theta(nm+n^2)$ cells to fit the simple-to-construct gadgets along the border without interference, and therefore the puzzle can be constructed in polynomial time.
\end{proof}

\subsection{3-Triangle Clues}

3-triangle clues admit a much simpler construction than those needed for 1-triangles and 2-triangles, as 3-triangles are constraining enough to build gadgets whose properties can be verified purely locally. The proof roughly follows the Hamiltonicity framework of Section~\ref{sec:Hamiltonicity Reduction Framework}, using adjacent 3-triangle clues to force the solution path to form impassable ``walls''.

\both{
\begin{theorem}
\label{thm:triangles3}
It is NP-complete to solve Witness puzzles containing only 3-triangle clues.
\end{theorem}
}

\ifabstract
\begin{proofsketch}
We use the Hamiltonicity framework. Adjacent 3-triangle clues must be traversed consecutively by the solution path, so we can use them to for the solution path to trace the boundary of each chamber. Figure~\ref{3-triangles chamber} shows the construction of a chamber.
\end{proofsketch}
\fi

\begin{figure}
\centering
% 3-chamber with boundaries
\subcaptionbox{Unsolved.}{%
  \hspace{.5cm}%
  \includegraphics[width=.4\textwidth]{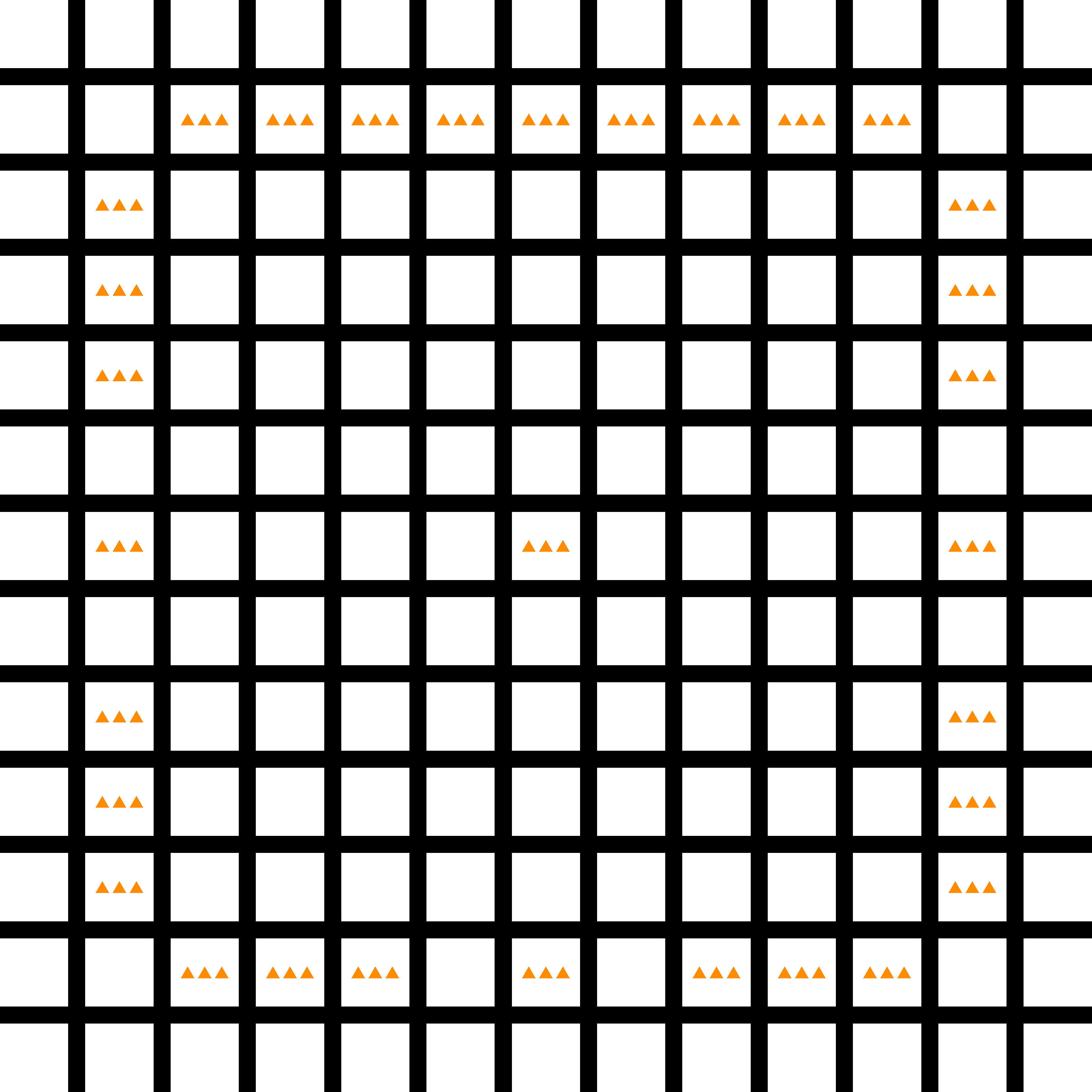}%
  \hspace{.5cm}%
  \label{3-triangles chamber unsolved}%
}~~~
\subcaptionbox{One possible solution, using the left and bottom edges.}{%
  \hspace{.5cm}%
  \includegraphics[width=.4\textwidth]{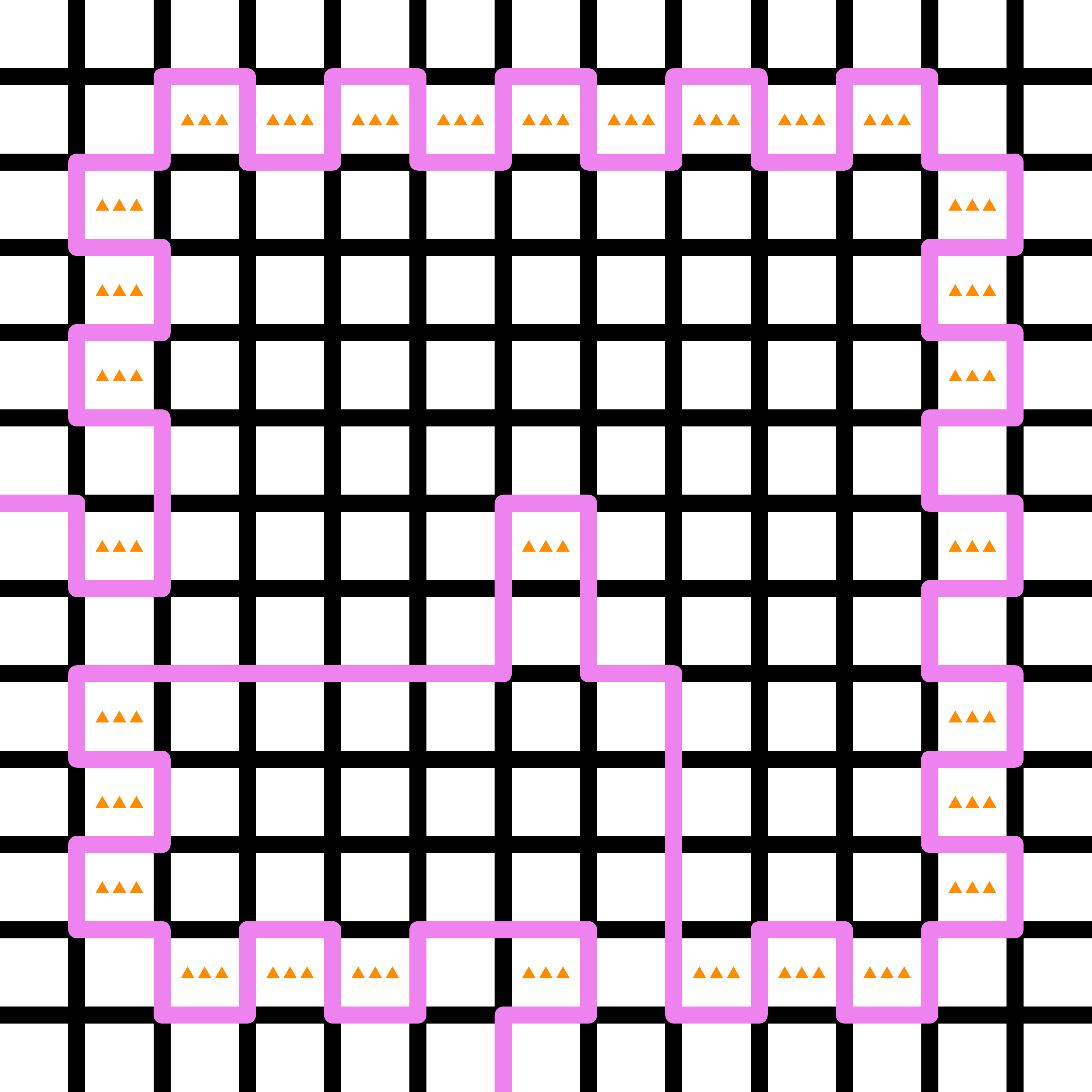}%
  \hspace{.5cm}%
  \label{1-triangles chamber example sol}%
}
\caption{A chamber with edges to the left, right, and below.}
\label{3-triangles chamber}
\end{figure}

\later{
\begin{proof}
We use a variation of the Hamiltonicity framework---reducing from
Hamiltonicity in a maximum-degree-$3$ grid graph, $G$, by scaling by a constant
factor $12$, and representing each vertex by a square-region
\emph{chamber}---but we will represent edges slightly differently than hallways.
Each chamber is a $11 \times 11$ grid of cells.
Initially, we place 3-triangle clues in the center cell
and in every boundary cell except the four corners.
Because the scale factor is $12$, there is exactly one row or column of
empty cells between adjacent chambers.
If there is an edge between vertices in $G$ corresponding to two chambers,
then we remove the 3-triangle clues from the two cells adjacent to the center
cell of the walls of both chambers (i.e., the fourth and sixth cells of each
wall), as depicted in Figure~\ref{3-triangles boundary}.

The rightmost bottommost chamber is the special start/end chamber.
We remove the center 3-triangle clue on its bottom edge, and place the
start circle and end cap just below, as can be seen in the full graph instance
shown in Figure~\ref{3-triangles whole solution}.

Now we prove that every solution to this Witness puzzle corresponds to a
Hamiltonian path.  The key observation is that $k$ consecutive 3-triangle clues
(in a row or column) force the solution path to have one of two local zig-zag
patterns, shown in Figure~\ref{3-triangles local}.  

  \begin{figure}
\centering
% 3-triangles forcing behavior
\subcaptionbox{One possible local solution.}{%
  \hspace{.5cm}%
  \includegraphics[width=.4\textwidth]{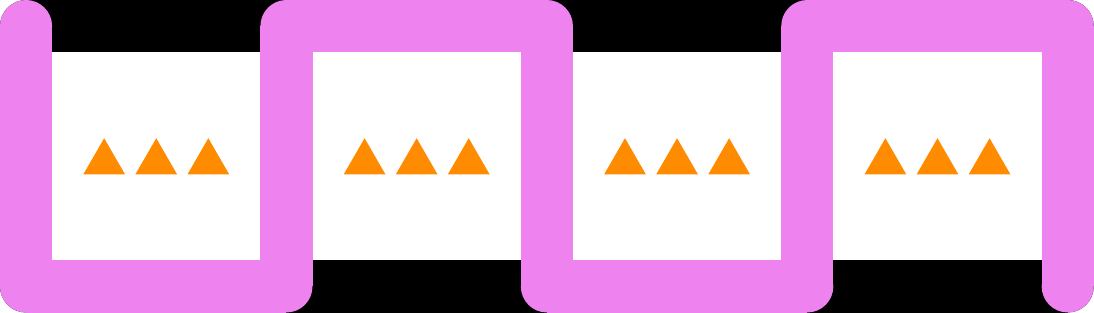}%
  \hspace{.5cm}%
  \label{3-triangles string of threes 1}%
}~~~
\subcaptionbox{The other possible local solution.}{%
  \hspace{.5cm}%
  \includegraphics[width=.4\textwidth]{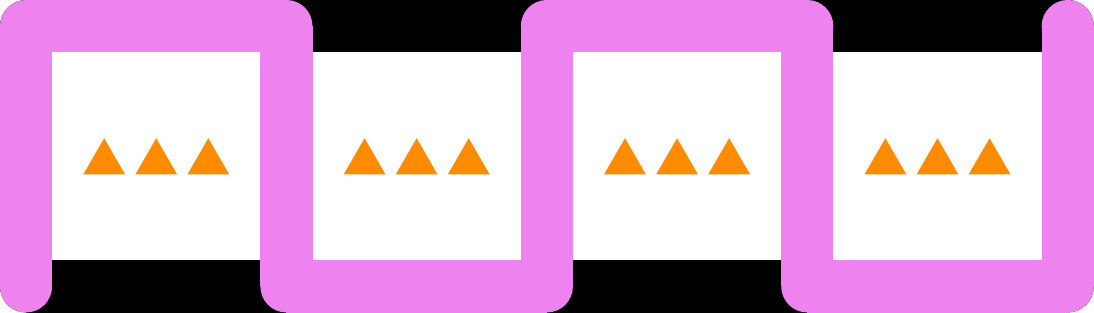}%
  \hspace{.5cm}%
  \label{1-triangles string of threes 2}%
}
\caption{The only local solutions to a string of contiguous 3-triangle clues.}
\label{3-triangles local}
\end{figure}

Each chamber gadget is surrounded by such consecutive sequences of
3-triangle clues.  Pairs of these sequences meet diagonally adjacently
at each corner of the chamber, which in fact forces the solution path to use
a particular zig-zag for these two sequences.  As a result, the solution path's
behavior is forced around all 3-triangle clues except the lone clues in
the center of each side of a chamber corresponding to an edge in $G$. 

Every chamber must be visited in order to satisfy its center 3-labeled cell. When a chamber is visited, all of its 3-labeled wall cells can be satisfied by following the walls around from the entry point to just before the intended exit point, then traversing through the center 3-labeled cell back to the other side of the entry point, and then around to and out the exit point, as shown in each chamber of Figure~\ref{3-triangles whole solution}.

If there is an edge in $G$ between vertices corresponding to two chambers $a$ and $b$, then the solution path can travel from $a$ to $b$ (or vice versa) as shown in Figure~\ref{3-triangles boundary used edge}. 

\begin{figure}
\centering
% 3-triangles chamber boundary
\subcaptionbox{\label{3-triangles boundary unsolved} Boundary between two chambers (unsolved).}{%
  \includegraphics[width=.28\textwidth]{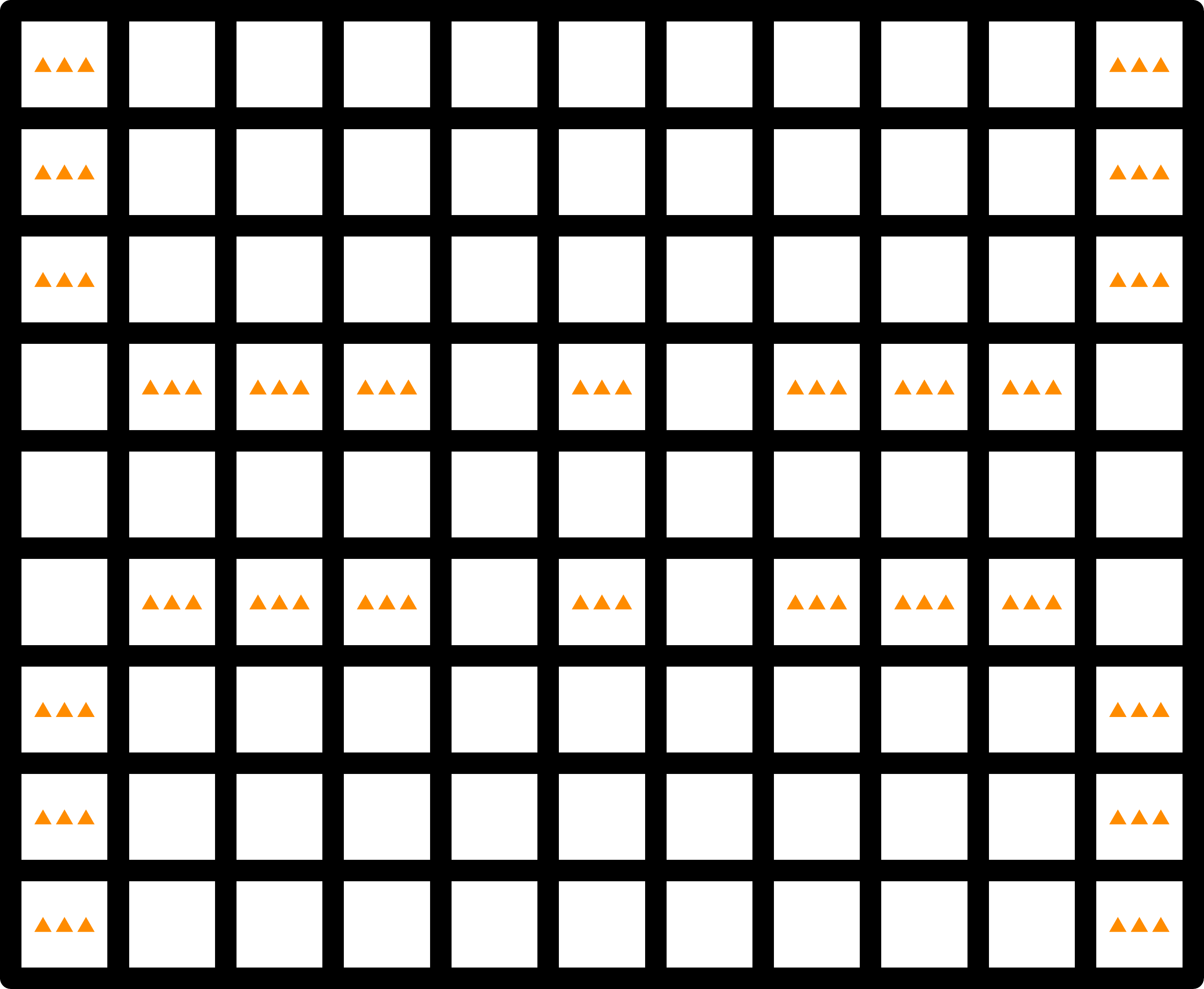}%
}
\hfil\hfil
\subcaptionbox{\label{3-triangles boundary unused edge} Solution when edge is unused.}{%
  \includegraphics[width=.28\textwidth]{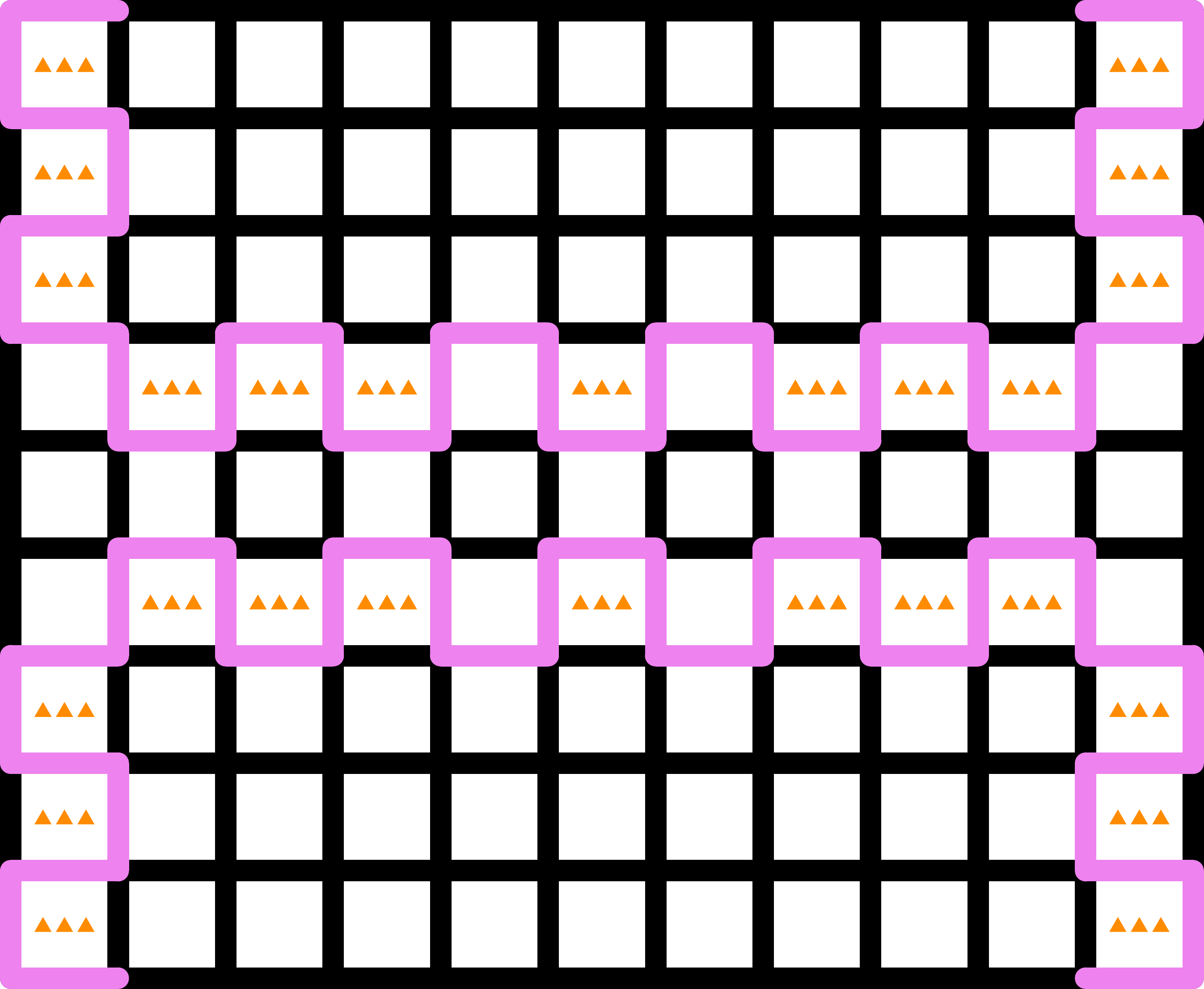}%
}
\hfil\hfil
\subcaptionbox{\label{3-triangles boundary used edge} Solution when edge is used.}{%
  \includegraphics[width=.28\textwidth]{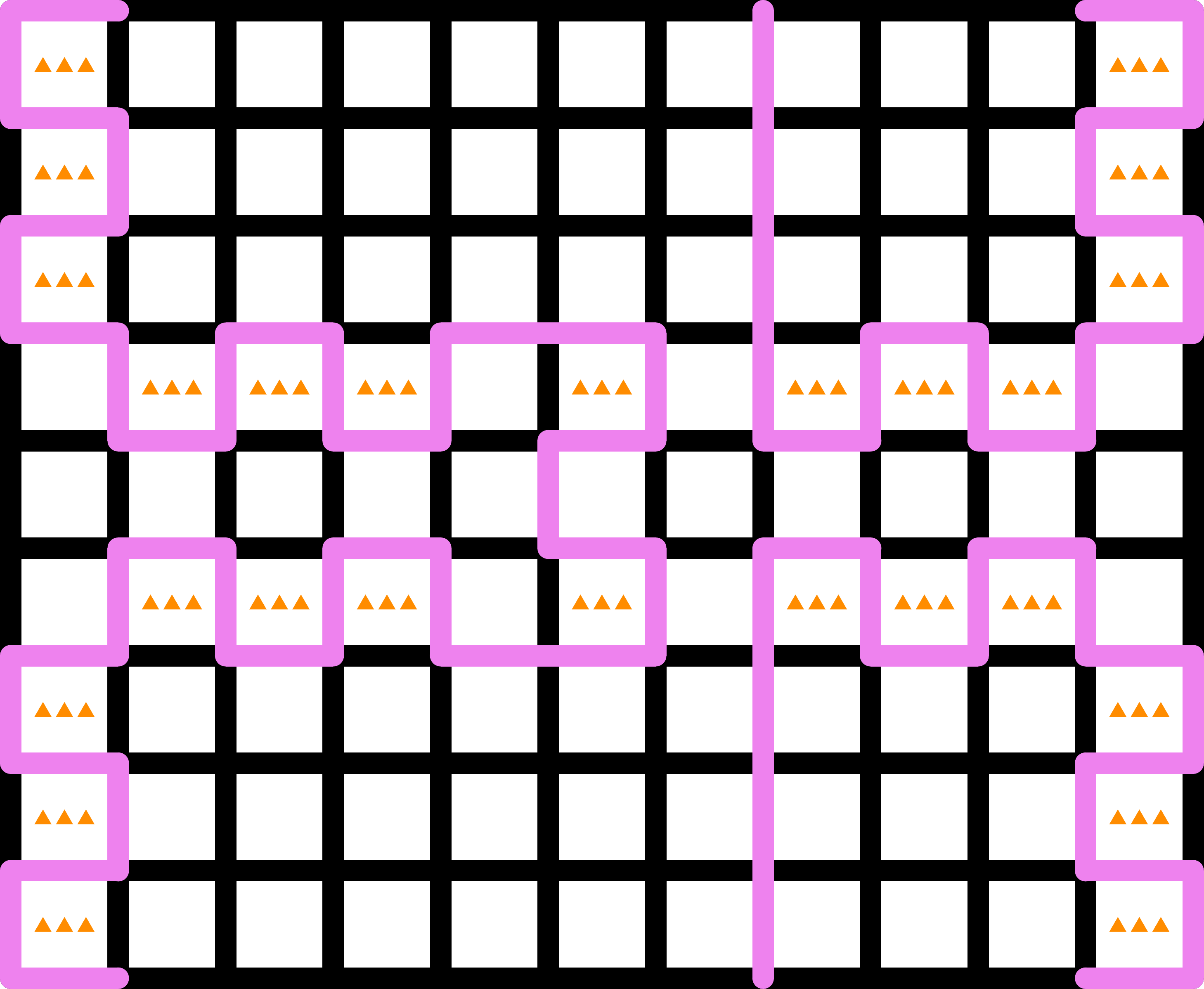}
}
%\xxx{caption text misbehaving?}
\caption{The boundary between two chambers with an edge between them.}
\label{3-triangles boundary}
\end{figure}

Examining Figure~\ref{3-triangles chamber}, a simple case analysis shows that the path satisfying the two 3-triangle cells comprising each corner of a chamber is completely forced, and this in turn forces the path around the rest of the 3-labeled cells comprising the walls.  Similarly, case analysis on the middle of the walls in Figure~\ref{3-triangles boundary} shows that the path cannot escape to another chamber, so the path can only travel between adjacent chambers.  Finally, there isn't enough space to use an edge more than once, and therefore exactly two edges incident to each chamber must be traversed, because $G$ has max degree 3 and if the solution path traverses only one edge incident to a chamber, it could not have left that chamber and therefore could not have made it to the end cap.
\end{proof}

\begin{figure}[t]
\centering
% 3-chamber with boundaries
\includegraphics[width=.9\textwidth]{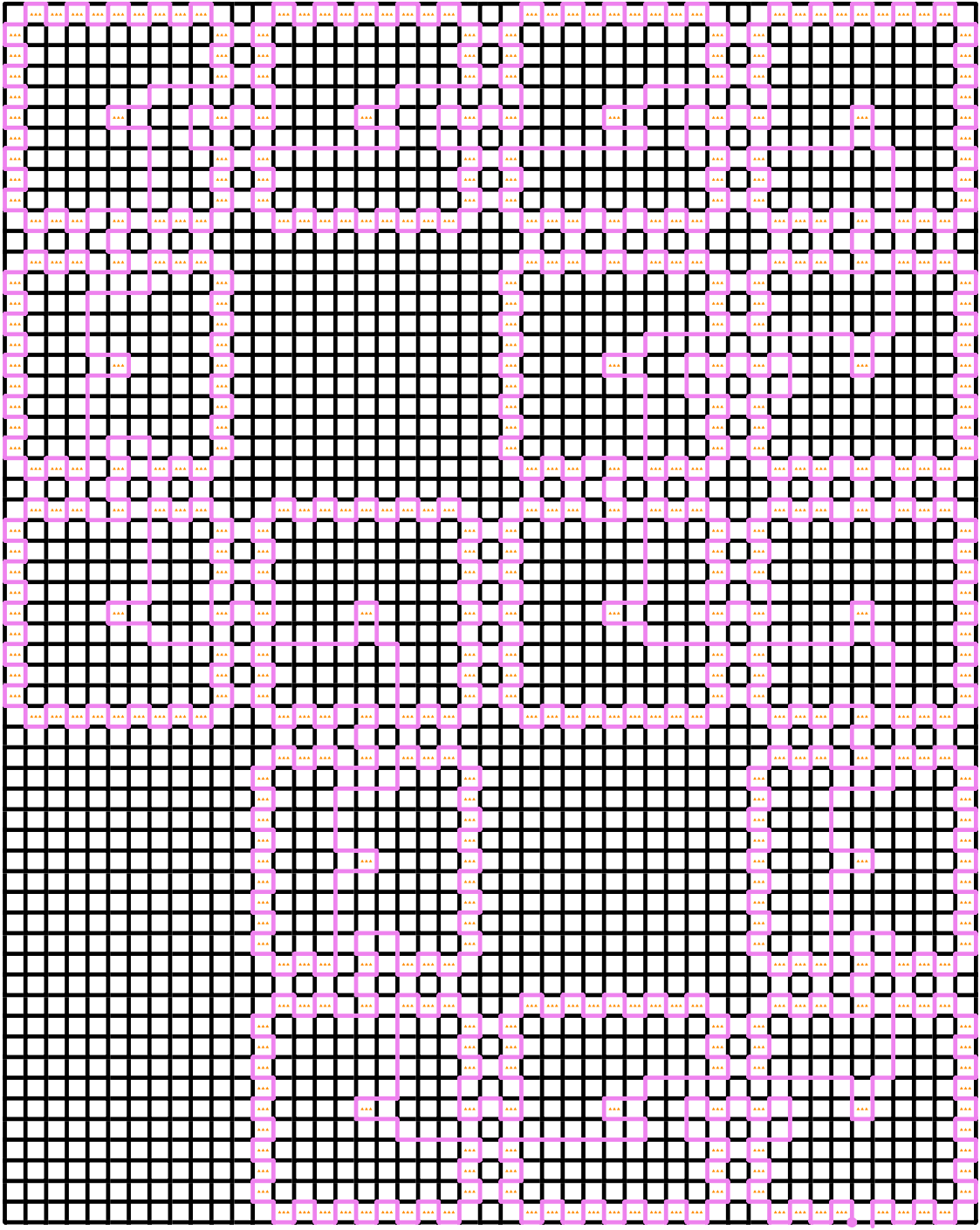}
\caption{A complete 3-triangles Hamiltonicity framework instance with solution.}
\label{3-triangles whole solution}
\end{figure}
}

\section{Polyominoes}

\abstractlater{
  \section{Proofs: Polyominoes}
  \label{appendix:polyominoes}
}

This section covers various types of \emph{polyomino} and \emph{antipolyomino} clues, giving both positive and negative results. Polyomino clues can generally be characterized
by the size and shape of the polyomino and whether or not they can be
rotated (\dominovert{} vs.\ \dominofree{}). For each region, it must be
possible to place all polyominoes and antipolyominoes depicted in that region's clues (not necessarily
within the region) so that for some $i \in \{0,1\}$, each cell inside the region is covered by exactly $i$
more polyomino than antipolyomino and each cell outside the region is
covered by the same number of polyominoes and antipolyominoes.
In Section~\ref{sec:Monominoes and Broken Edges}, we prove that monomino
clues alone can be solved in polynomial time, even in the presence of
broken edges, generalizing a result of \cite{eurocg}.
In Sections~\ref{sec:Rotatable Dominoes}, \ref{sec:Monominoes + Antimonominoes}, and \ref{sec:Nonrotatable Dominoes} we give several negative results showing that some of the simplest (anti)polyomino clues suffice for NP-completeness.

%We give one positive result
%showing that puzzles containing only broken edges and monomino clues (\monomino) are solvable in polynomial time, 
%and several negative results showing that essentially all polyomino clue sets containing more than just monominoes are NP-complete. 
%Proofs omitted from this section are available in Appendix~\ref{appendix:polyominoes}.

\subsection{Monominoes and Broken Edges}
\label{sec:Monominoes and Broken Edges}

In this section, we prove the following theorem:

\begin{theorem}
\label{MonominoesAndBrokenEdgesTheorem}
Witness puzzles containing only monominoes and broken edges are in P.
\end{theorem}

Concurrent work~\cite{eurocg} gives a polynomial-time algorithm for
Witness puzzles where \textit{every} cell has a square clue of one of two colors.
Such puzzles are equivalent to puzzles with only monomino clues, by replacing
one color of square with monominoes and the other color with empty cells
(or vice versa for the reverse reduction).
The only constraint on the two puzzle types is that there can be no region
with a mix of square colors or, equivalently, monomino clues and empty cells.
Accordingly, our proof of Theorem~\ref{MonominoesAndBrokenEdgesTheorem}
is similar.
In our case, we reduce to the result from Section~\ref{sec:boundary hexagons}
that path finding is easy with broken edges and forced edges on the boundary
of the puzzle, which we considered in Section~\ref{sec:boundary hexagons}
in the context of boundary edge hexagons.
%but we need an additional result establishing that a path-finding
%subroutine used in both \cite{eurocg} and our proof remains easy even in the
%presence of broken edges.

%\subsubsection{Dividing into Divisions}
%
%We are now ready to prove Theorem~\ref{MonominoesAndBrokenEdgesTheorem}.

\begin{proof}
When a puzzle contains only monomino clues and broken edges, every region
outlined by the solution path that contains a monomino clue must have a
monomino clue in every cell in the region to satisfy the polyomino area
constraint.  Any solution path must therefore include all of the edges shared
between an empty cell and a cell containing a monomino clue.
If any of these forced edges is broken, or the forced edges form a cycle,
then we can immediately reject the puzzle as a \textsc{no} instance.
In particular, we reject if any connected component of monomino clues
either does not touch the boundary or has a hole of empty cells.
%Because all regions of a Witness solution path touch the boundary, the
%%clue-containing
%regions form a laminar family (their dual is a tree).
%%when the boundary is considered to connect regions.
The solution path is forced to trace the outline of each clue-containing
component except where it meets the boundary.
%starting and ending at the endpoints of its portion of the boundary.
Thus, the problem reduces to connecting the endpoints of these forced
outline paths into a path from the start circle to the end cap.

Some connected components of monomino clues divide the puzzle into disconnected groups of cells.  We call these groups \emph{divisions},
%(in particular avoiding ambiguity with the regions defined by the solution path),
and the connected components of clues that give rise to them \emph{dividers}.
A division may contain more than one region,
but a region is never split across divisions.
Call a division \emph{empty} if it does not contain a monomino clue,
the start circle, or the end cap, and call it \emph{nonempty} otherwise.
Because of our rejection rule above, every division is simply connected.

When a solution path visits the forced outline edges of a divider
(from one of the incident divisions),
the path must visit all of the forced edges
before entering one of the incident divisions,
because it can never exit that division without revisiting a vertex.
In particular, if any divider neighbors three or more nonempty divisions,
then we reject the puzzle as a \textsc{no} instance:
only two divisions can be visited by a solution path
(before and after the forced edges, respectively), and
each nonempty division must be visited to visit either some forced edge,
the start circle, and/or the end cap.

Because each divider neighbors at most two nonempty divisions, we can define
a graph where nodes correspond to the nonempty divisions and
edges correspond to dividers.
Each edge in this graph can only be traversed once by the solution path
because doing so consumes all the forced edges surrounding that divider.
Therefore if this graph is not a path, or if the endpoints of this path
are not the divisions containing the start vertex and end cap, respectively,
then we reject the puzzle as a \textsc{no} instance,
because no path from the start vertex to the end cap can visit all of the
nonempty divisions while using each adjacency graph edge only once.

For each division $D$, let $F_D$ be forced edges between an
empty cell and a cell containing a monomino clue within the division, that is,
the subset of the edge outline separating $D$ from its neighboring dividers
and the nonboundary edge outline of any nondivider components of
monomino clues in the division.
Because there is a forced edge where each divider's outline touches the puzzle
boundary, the path can enter and exit this division at each of two possible
vertices (at opposite ends of the dividers).  The start circle is the sole
entry vertex for the first division and the end cap is the sole exit vertex
for the last division.

% writing note: using 'node' for vertices of $G$ to avoid overloading 'vertex'

We build a connectivity graph $G$ whose nodes represent
the entry and exit vertices of all nonempty divisions
(including the start circle and end cap).
For each division $D$, let $C_D$ be the subset of empty cells in~$D$.
Set $C_D$ is connected because any set of cells (containing monominoes)
separating $C_D$ into two components would be a divider.
By our first rejection rule, $C_D$ is also simply connected.
Every edge in $F_D$ is incident to a single empty cell,
so $F_D$ is a subset of the edge outline of~$C_D$.
Hence, $(C_D,F_D)$ is a forced division
as defined in Section~\ref{forced division}.
For each pair $(s,t)$ of entry and exit vertices of~$D$, respectively,
we apply Theorem~\ref{thm:monominoes-dp} to decide whether there is a path
from $s$ and $t$ entirely within $C$ and traversing all forced edges in~$F_D$.
Whenever there is such a path, we add an edge in $G$ between the nodes
corresponding to $s$ and~$t$.

The forced edges of each divider, plus some puzzle boundary edges interior to
the divider, form paths connecting each possible exit vertex to an entry
vertex in the next division, and we add the corresponding edges to~$G$.  Then
we can find a puzzle solution by searching for a path in $G$ from the start
circle to the end cap, and replacing each edge from that path in $G$ with the
Witness puzzle path it represents.  If there is no such path in $G$, then
the Witness puzzle has no solution.
\end{proof}

\subsection{Rotatable Dominoes}
\label{sec:Rotatable Dominoes}

\both{
\begin{theorem} \label{rotating dominoes}
It is NP-complete to solve Witness puzzles containing only rotatable dominoes.
\end{theorem}
}

\ifabstract
\begin{proofsketch}
We reduce from Rectilinear Steiner Tree:
given $n$ points with integer coordinates $(x'_i, y'_i)$ in the plane,
$i \in \{1, 2, \dots, n\}$, and given an integer~$k$,
decide whether there exists a rectilinear tree
connecting the $n$ points having total length at most~$k$.
As illustrated in Figure~\ref{DominoesOverview},
we embed the tree in the cells of a Witness puzzle, putting a domino
clue at each vertex of the tree, which the solution path must therefore visit.
The total number of dominoes is proportional to $k$, such that with careful counting,
the area enclosed by the solution path must ``look like'' a tree of length exactly $k$
in the original Steiner tree instance.
\end{proofsketch}
\fi

\later{
\begin{proof}
We reduce from the rectilinear Steiner tree problem
\cite{Garey-Johnson-1977-steiner}:
given $n$ points with integer coordinates $(x'_i, y'_i)$ in the plane,
$i \in \{1, 2, \dots, n\}$, and given an integer~$k$,
decide whether there exists a rectilinear tree
connecting the $n$ points having total length at most~$k$. Hanan's Lemma
\cite[Theorem 4]{Hanan-1966} states that this is equivalent to the
existence of a rectilinear tree of length at most $k$ on the integer
grid, or even on the grid formed by a horizontal and a vertical line
through each of the $n$ points.
By globally translating the points by $(-\min_i x'_i, -\min_i y'_i)$,
we can find an equivalent instance of $n$ points $(x_i, y_i)$,
$i \in \{1, 2, \dots, n\}$, with all coordinates nonnegative
%satisfying $x_i, y_i \in \{0, 1, \dots, m-1\}$
and where $\min_i y_i = 0$.
Let $m = \max_i \{x_i,y_i\}$.
Because this problem is strongly NP-hard
\cite{Garey-Johnson-1977-steiner},
\cite[Problem ND13, p.~209]{Garey-Johnson-1979},
we can assume $m = n^{O(1)}$.

We construct a $(2m+4k+3) \times (2m+4k+3)$ Witness puzzle
with a $(2m+1) \times (2m+1)$ center~$C_m$;
refer to Figure~\ref{fig:rotatable-dominoes}.
% padded by $2 k+1$ on all four sides.
%on a $(2 m + 2 (2 k+1)) \times (2 m + 2 (2 k+1))$ board.
Within $C_m$, we include exactly $n$ dominoes:
for every point $(x_i, y_i)$ in the rectilinear Steiner
tree input, we include a domino in cell $(2 x_i, 2 y_i)$,
where these coordinates are relative to the lower-left corner of~$C_m$.
Consider the leftmost point with a $y$ coordinate of zero,
whose corresponding domino clue is on the bottom row of~$C_m$;
call this cell the \emph{root}.
Define the \emph{trunk} to be the $2 k+1$ cells in the same column as
and below the root.  We include exactly $(2 k+1) - n$ dominoes in the trunk,
starting at the bottom of the board.
We place the start circle and end cap on the boundary vertices adjacent
to the trunk (in either order).

%\begin{figure}
%\centering
%\includegraphics{figures/DominoesOverview}
%\caption{Overview of the rotatable dominoes NP-completeness proof.}
%\label{DominoesOverview}
%\end{figure}

\begin{figure}
\centering
\subcaptionbox{Unsolved.}{
  \begin{overpic}[width=0.48\linewidth]{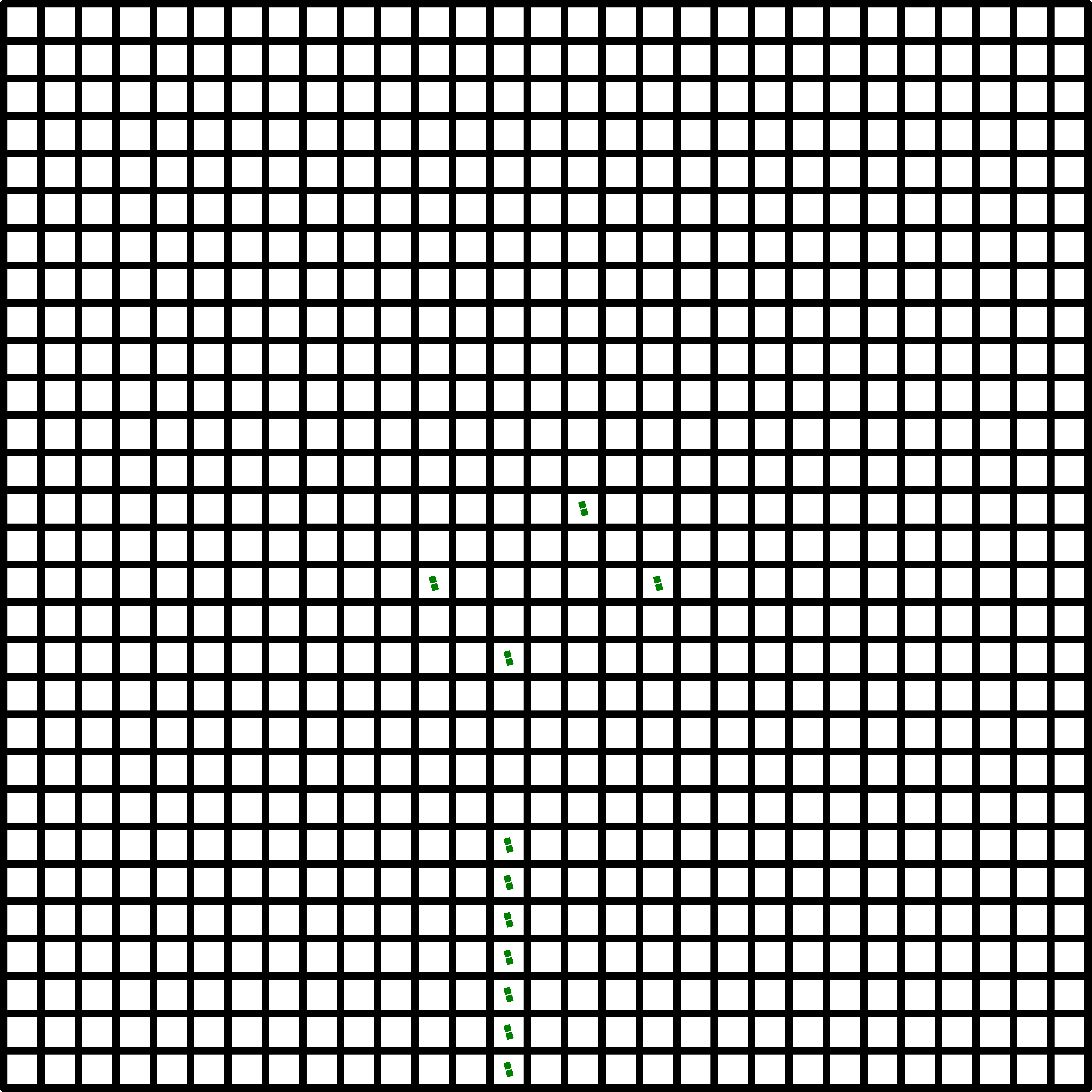}
    \put(0,0){\includegraphics[width=0.48\linewidth]{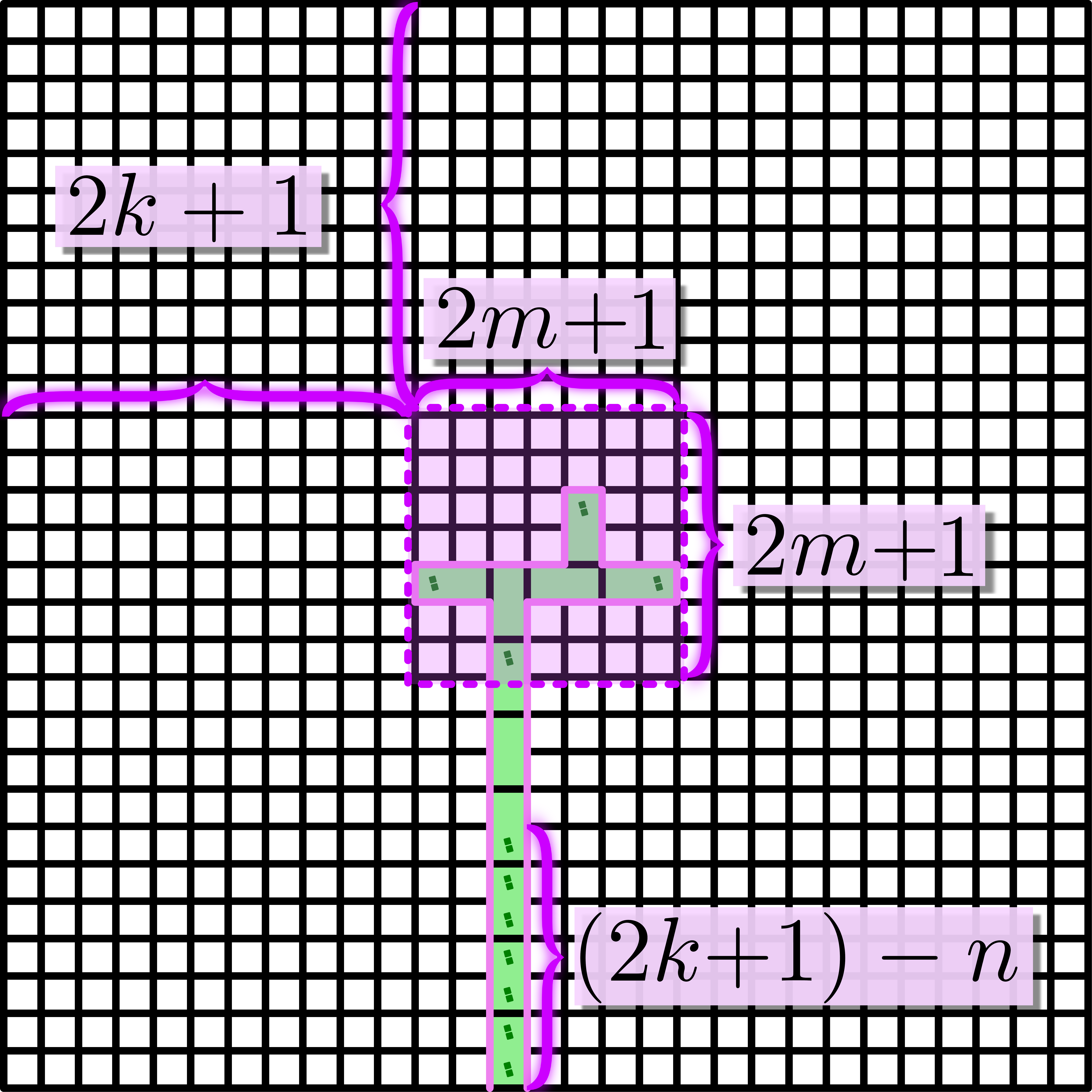}}
  \end{overpic}
}\hfil\hfil
\subcaptionbox{Solved.}{
  \begin{overpic}[width=0.48\linewidth]{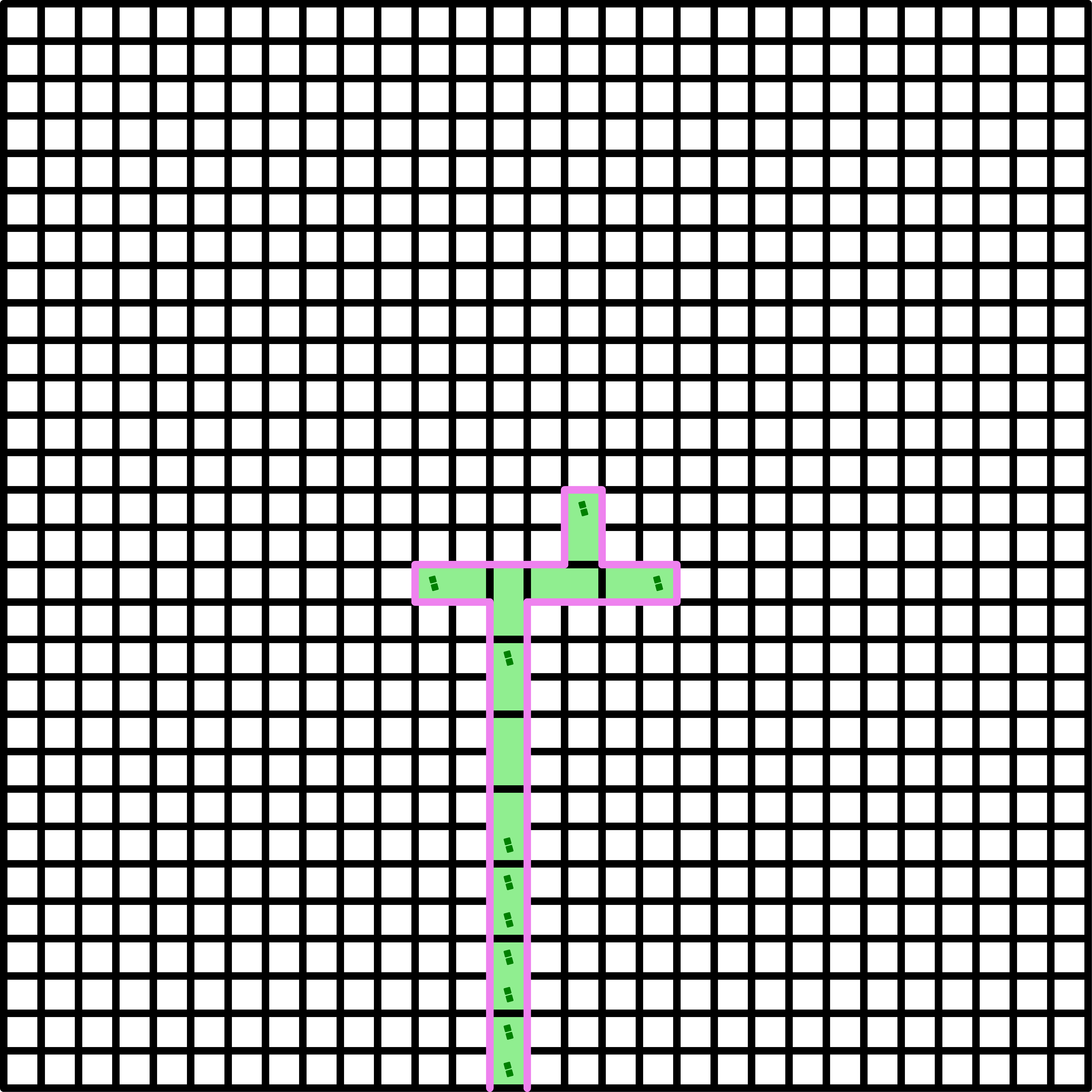}
    %\put(0,0){\includegraphics[width=0.48\linewidth]{figures/steinertree_overlay}}
  \end{overpic}
}
\caption{Example of the rotatable dominoes NP-completeness proof,
  where $k=5$, $n=4$, and $m=3$.}
\label{fig:rotatable-dominoes}
\end{figure}

\paragraph{Steiner tree $\to$ Witness solution.}
First suppose that the rectilinear Steiner tree instance has
a solution $S$ with total length at most~$k$.
By Hanan's Lemma, we can assume that the vertices and horizontal/vertical
segments of $S$ lie on the integer grid, with all $y$ coordinates nonnegative.
Scale the tree $S$ by a factor of~$2$, and consider the set $P$ of
integer grid points that lie on the vertices or segments of this scaled tree.
If two points in $P$ have distance $1$ on the grid, then (by the scaling)
they must have come from a common segment of the tree~$S$.
Thus, $P$ induces a grid graph that is a tree.
This grid tree is the dual of a set $C \subseteq C_m$ of cells.
Because $S$ is a Steiner tree on the $n$ points,
$C$ contains all the cells with domino clues in~$C_m$.
In particular, one vertex of $S$ corresponds to the root clue,
and we define this vertex to be the root of the tree~$S$.

Consider a segment $(u,v)$ of the tree $S$ of length~$\ell$, 
oriented so that $u$ is closer than $v$ to the root.
Vertices $u$ and $v$ map to cells $c_u$ and $c_v$ in $C$
%(with even $x$ and $y$ coordinates)
in the same row or column,
with $2\ell-1$ cells from $C$ between them.
We can tile these $2\ell-1$ cells together with~$c_v$
with exactly $\ell$ dominoes.
By our orientation of segments, dominoes from different segments do not
overlap, so we obtain a tiling of $C$ except the root
using exactly $L$ dominoes.

%Suppose this tree $S$ has $v$ vertices and thus $v-1$ segments
%of total length $L \leq k$.
%Each vertex of $S$ maps to a cell in~$C$, for a total of $v$ cells;
%and each segment of $S$ of length $\ell$ maps to $2 \ell-1$ cells in~$C$,
%excluding the endpoints, for a total of $2 L - (v-1)$ cells.
%Thus, $|C| = 2 L+1 \leq 2 k + 1$.

Next we construct a set $C'$ of cells, by adding to $C$ the $2 k+1$ cells
in the trunk (which includes all remaining domino clues).
The trunk has an odd number of cells in a single column,
so together with the root it can be tiled by exactly $k+1$ dominoes.
In total, we have tiled $C'$ with $L+k+1$ dominoes.
%so $|C'| = 2(L+k+1) \leq 4 k + 2$.

Finally, we construct a set $C''$ by adding to $C'$
the cells of $k-L$ disjoint dominoes immediately right of the trunk
(e.g., starting from the bottom).
Set $C''$ can be tiled by exactly $2 k+1$ dominoes,
is connected and hole-free, and contains all domino clues.
The border of $C''$ is therefore a cycle on the integer grid;
removing the bottommost edge of the trunk leaves a path
which solves the Witness puzzle.

\paragraph{Witness solution $\to$ Steiner tree.}
Now suppose that the Witness puzzle has a solution. 

Consider the cells containing domino clues in $C_m$. Every
such cell is in a region %(a connected component outlined by the Witness
%solution and the boundary of the board)
that borders the boundary of the board and can be tiled by dominoes equal in
number to the number of domino clues inside that region. Each such
region has area strictly greater than $2 k + 1$
because the region must reach from the boundary of the board into~$C_m$.
But the total number of domino clues in the whole board is $2 k + 1$,
so the total area in all such regions is $2 (2 k + 1)$.
Therefore, there is exactly one region $A$ that contains all the domino clues
in~$C_m$.

We now show that there is a connected subset $B_m$ of $C_m$,
of size at most $2 k+1$, containing all the
domino clues of $C_m$. We start with $B_{m+(2k+1)} = A$ and
inductively define sets $B_i$ for all $m \le i \le m+(2k+1)$
such that $B_i$
(1)~is a subset of the $(2i+1) \times (2i+1)$ center $C_i$ of the board,
(2)~is connected,
(3)~contains all the domino clues in $C_m$,
(4)~has size at most $i-m+(2k+1)$, and
(5)~is connected to one of the four boundary sides of~$C_i$.

Given $B_{i+1}$, define $B_i$ to include all the cells of $B_{i+1}$
that are in $C_i$
plus, for each cell of $B_{i+1}$ on the boundary of $C_{i+1}$,
the nearest cell on the boundary of $C_i$.
(1) $B_i$ is, by definition, a subset of $C_i$.
(2) To prove $B_i$ connected, consider the $(2i +3) \times (2i+3)$
grid graph representing the dual of $C_{i+1}$,
with cells of $B_{i+1}$ \emph{marked}. Then transforming from $B_{i+1}$ to
$B_i$ is equivalent to
contracting dual edges with exactly one endpoint on the boundary of $C_{i+1}$
and the eight dual edges with one endpoint at a corner of $C_{i+1}$
and then marking each contracted vertex
if and only if at least one of the endpoints of the contracted edge was marked.
This transformation preserves
connectivity of the marked vertices, so $B_i$ is connected.
(3) $B_{i+1}$ contained all the domino clues in $C_m$,
so because those are all in $C_i$,
$B_i$ still contains them. (4) Because $B_{i+1}$ was
connected to a boundary of the board and $C_m$, it
contains at least one adjacent pair of squares on the boundaries of
$C_{i+1}$ and $C_i$, so at least one dual edge with both endpoints marked is
contracted, making $B_i$ smaller than $B_{i+1}$ by at least
one square. Hence $B_i$ has size at most $|B_{i+1}|-1 \le
(i+1-m+(2k+1))-1 = i-m+(2k+1)$. (5) By the existence of the same edge,
$B_{i+1}$ and hence also $B_i$ has a cell on
the boundary of $C_i$, completing the induction.

We claim that any spanning tree $T$ of the dual graph of $B_m$ is a
factor-2 scaling of a Steiner tree of length at most $k$.
Any such $T$ has $|B_m| \le 2 k+1$ vertices and hence at most
$2 k$ edges. Each edge has length $1$ in the grid, so $T$ is a
tree of total rectilinear length at most $2 k$. In addition, we know that
$T$ touches the points with coordinates $(2 x_i, 2 y_i)$
for $i \in \{1, 2, \dots, n\}$.
Scaling $T$ by a factor of $\frac12$ yields a tree of
total length at most $k$ which touches the points with coordinates
$(x_i, y_i)$ for $i \in \{1, 2, \dots, n\}$,
solving the given rectilinear Steiner tree instance.
\end{proof}
} %\later
\subsection{Monominoes + Antimonominoes}
\label{sec:Monominoes + Antimonominoes}

\both{
\begin{theorem} \label{anti-monominoes}
It is NP-complete to solve Witness puzzles containing only monominoes and
antimonominoes.
\end{theorem}
}

\ifabstract
\begin{proofsketch}
The reduction is very similar to that of Theorem~\ref{rotating dominoes}, except
that the vertices of the Steiner tree contain antimonomino clues, and most of the
other cells contain monomino clues. We force the solution path to partition the puzzle
into two regions, an ``outside'' region which is entirely covered by monominoes,
 and an ``inside'' region which contains exactly as many
antimonominoes as monominoes, thereby satisfying both. We show that doing this
corresponds to a solution to the Steiner tree source instance. 
\end{proofsketch}
\fi

\later{
\begin{proof}
  Like Theorem~\ref{rotating dominoes}, we reduce from the rectilinear
  Steiner tree problem; refer to Figure~\ref{fig:(anti)monominoes}.
  Given an instance $I$ of this problem, let $W$
  be the Witness puzzle that is constructed from $I$ by the reduction in
  Theorem~\ref{rotating dominoes}; $W$ contains only domino
  clues. Modify $W$ as follows: in each cell with a domino clue, instead
  place an antimonomino clue, and in each cell that was previously
  empty, add a monomino clue. This modified puzzle, $W'$, which contains
  only monomino and antimonomino clues, is the output of our reduction.

\begin{figure}
\centering
\subcaptionbox{Unsolved.}{
  \begin{overpic}[width=0.48\linewidth]{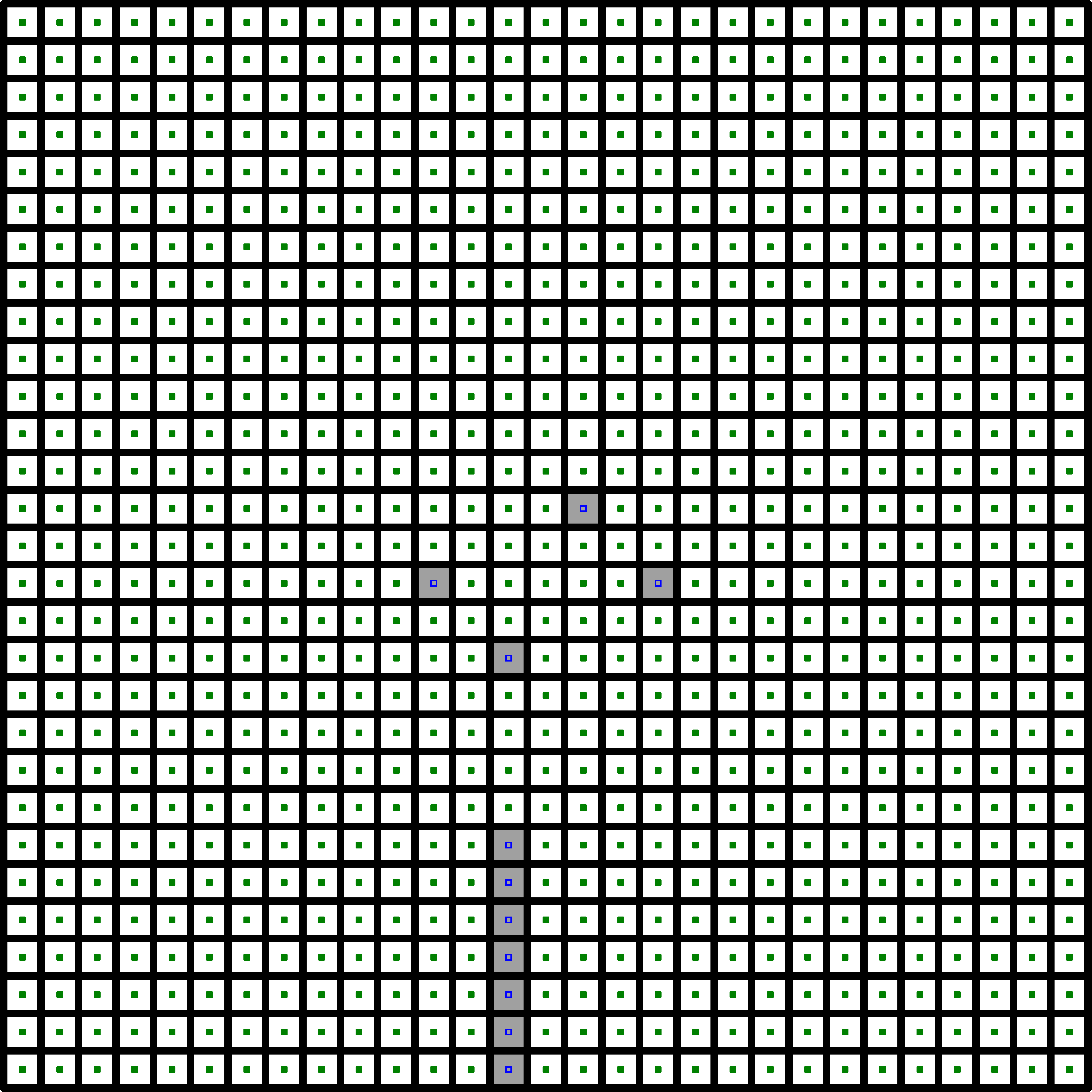}
    \put(0,0){\includegraphics[width=0.48\linewidth]{figures/steinertree_overlay}}
  \end{overpic}
}\hfil\hfil
\subcaptionbox{Solved.}{
  \begin{overpic}[width=0.48\linewidth]{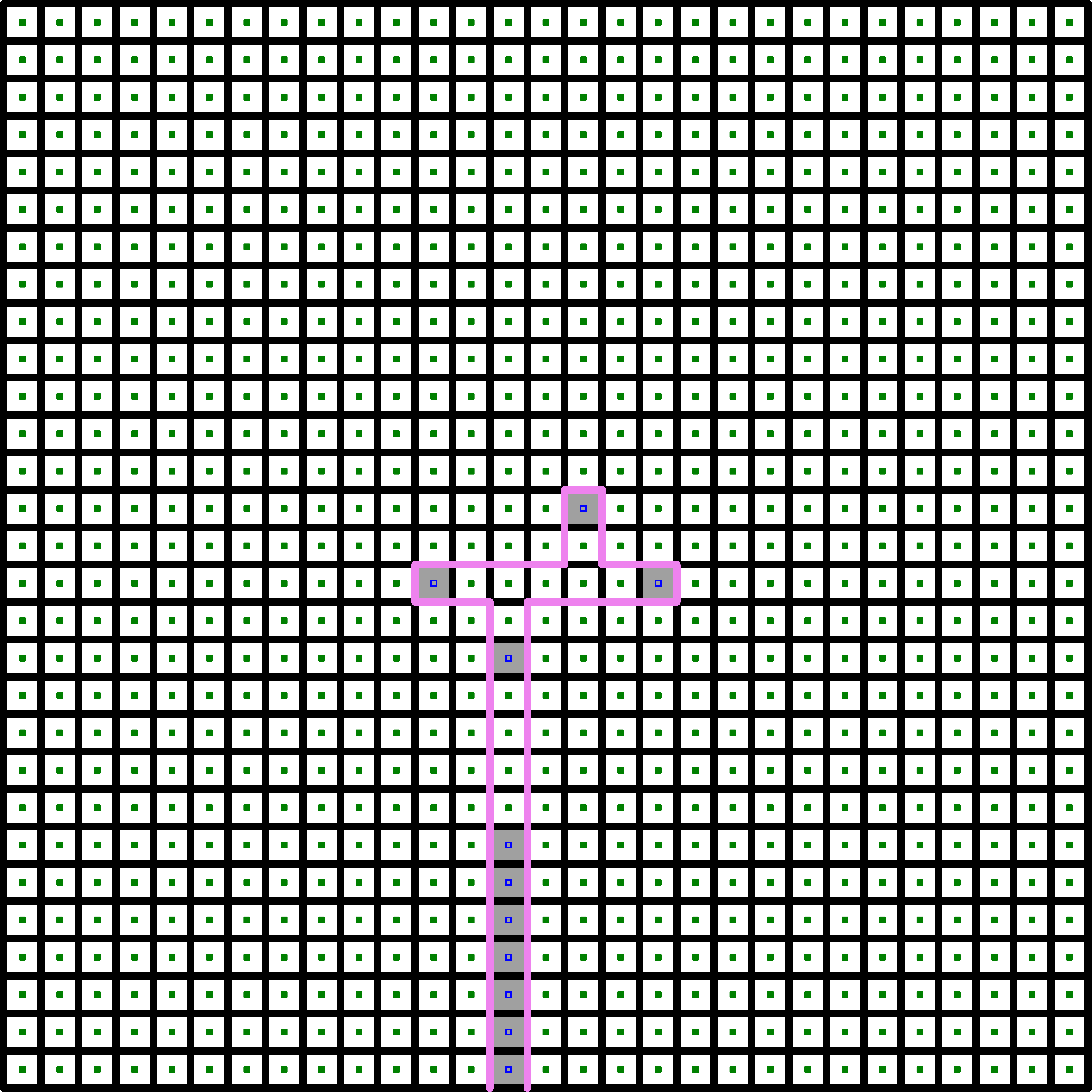}
    %\put(0,0){\includegraphics[width=0.48\linewidth]{figures/steinertree_overlay}}
  \end{overpic}
}
\caption{Example of the monominoes + antimonominoes NP-completeness proof,
  converted from Figure~\ref{fig:rotatable-dominoes}.}
\label{fig:(anti)monominoes}
\end{figure}

  Consider any region in a solution of $W'$ which contains an antimonomino. The total area of the antipolyomino clues in this region is less than the total area of the region. Thus, the total area of the clues must be exactly zero in order to satisfy the polyomino constraint. In other words, the area of the region must be exactly double the number of antimonominoes in the region.

  Therefore every region must either contain only monominoes, or its area must equal twice the number of antimonomino clues (i.e., the region must contain an equal number of monominoes and antimonominoes). In fact, as long as a path satisfies this area condition, it is a solution.

  This condition is strictly weaker than that imposed by $W$, so any solution to $W$ is also a solution to $W'$. So if the Steiner tree instance is solvable, so is~$W'$.

  Conversely, suppose $W'$ is solvable. In the  \textbf{Witness solution} $\to$ \textbf{Steiner tree} direction of the proof of Theorem~\ref{rotating dominoes}, we relied only on area conditions: that the area of every region with a domino equals twice the number of domino clues in the region. This exactly corresponds to the area conditions imposed by $W'$, so an identical proof shows that any solution to $W'$ can be converted into a Steiner tree that solves $I$.
\end{proof}
} %\later

\subsection{Nonrotatable Dominoes}
\label{sec:Nonrotatable Dominoes}

\both{
\begin{theorem} \label{vertical dominoes}
It is NP-complete to solve Witness puzzles containing only nonrotatable
vertical dominoes.
\end{theorem}
}

\ifabstract
\begin{proofsketch}
We reduce from planar rectilinear monotone 3SAT \cite{PlanarRectilinearMonotone}.
We construct variable ``wires'' which are comprised of dominoes arranged on a diagonal
which the solution path must enclose in one of two settings. Each clause needs to ``connect''
to at least one of its literals, but can only get close enough to do so if the corresponding
variable is set appropriately. 
\end{proofsketch}
\fi

\later{
\begin{proof}
We reduce from planar rectilinear monotone 3SAT \cite{PlanarRectilinearMonotone},
which is a special case of
3SAT where the formula can be represented by the following planar structure.
Each variable has an associated \emph{variable segment} on the $x$ axis,
and the variable segments are disjoint.
Each clause has an associated \emph{clause segment} which is horizontal
but not on the $x$ axis.
Each clause segment is connected by vertical segments to the variable segments
corresponding to the literals in the clause.
Clause segments above the $x$ axis represent \emph{positive} clauses whose
literals are all positive (not negated), and clause segments below the
$x$ axis represent \emph{negative} clauses whose literals are all negated.

We will convert such an instance into a Witness puzzle with nonrotatable
vertical domino clues.  Our construction will have the following properties:

\begin{property} \label{no vertical domino clues}
  No two domino clues are in vertically adjacent cells.
\end{property}

\begin{property} \label{boundary domino clues}
  There is exactly one cell $\ell$ in the leftmost column that
  contains a domino clue; and there are no other domino clue cells
  in the rightmost column, top two rows, or bottom two rows.
\end{property}

\begin{property} \label{domino start end}
  The start circle is the bottom-left corner of $\ell$,
  and the end cap the top-left corner of~$\ell$.
\end{property}

\paragraph{Domino tiling problem.}
Before describing the construction, we show how solving a Witness puzzle
satisfying the properties above can be rephrased as solving a special kind
of domino tiling problem: choosing one domino tile to cover each domino-clue
cell (and the cell either above or below it) such that the boundary of the
union of the domino tiles forms a simple cycle.

Any solution path to the constructed Witness puzzle must divide
the board into regions such that each region containing domino clues can be
exactly tiled by a number of domino tiles equal to the number of contained
domino clues.  In such a tiling, every domino-clue cell
must be covered by exactly one domino tile.
By Property~\ref{no vertical domino clues},
no domino tile can cover multiple cells with domino clues.
Therefore, there is a one-to-one correspondence between domino clues and
domino tiles in any solution, where each domino tile covers the cell with
the corresponding domino clue.  Equivalently, a valid such tiling can be
specified by choosing whether the domino tile covering a domino-clue cell
also covers the cell above or below.

By Property~\ref{boundary domino clues}, exactly one
domino tile can be adjacent to the board boundary.  Every tiled region in
the solution must touch the board boundary, and must therefore contain
the unique domino tile adjacent to the board boundary.  Thus, there is
exactly one region containing all the domino clues, which must be tiled by
all corresponding domino tiles.  The solution path must (as a subpath)
outline the part of the boundary of this region that lies interior to the
board.
%enter and exit
%the interior of the board along the top and bottom edges of the domino tile
%covering~$\ell$; before entering and after exiting, the solution path must
%remain on the board boundary.

We claim that solving the Witness puzzle is equivalent to choosing one domino
tile to cover each domino-clue cell (and the cell either above or below it)
such that the boundary of the union of the domino tiles forms a simple cycle.
On the one hand, any simple boundary cycle can be converted into a
solution path by removing the left edge of~$\ell$,
by Property~\ref{domino start end}.
On the other hand, for any solution path, we can take the subpath outlining
the tiled region and add two edges (the left edges of the domino tile
covering~$\ell$) to form a simple cycle that is the boundary of the tiled
region.
% whose interior must be tiled by the correct number of domino tiles.

\paragraph{Variable gadget.}
For a variable appearing in $k$ clauses,
its \emph{variable gadget} consists of $4k$ domino clues in a positive-slope
diagonal line, with each clue immediately above and to the right of the
previous one.
Divide these clues into $k$ \emph{blocks} of four contiguous clues each.
Each block will be used to connect to one of the clause gadgets
corresponding to a clause in which the variable appears.
Suppose the variable appears in $k^+$ positive clauses
$C^+_1, C^+_2, \dots, C^+_{k^+}$
and in $k^-$ negative clauses
$C^-_1, C^-_2, \dots, C^-_{k^-}$
(so $k^+ + k^- = k$).
We assign the leftmost $k^-$ blocks for connections to the negative clauses, followed by the remaining $k^+$ blocks for connections to the positive clauses.

Dominoes in the variable gadget that have the path extending one cell above them (``up'' dominoes) represent an assignment of \textsc{true} to the corresponding variable and dominoes with the path extending one cell below them (``down'' dominoes) represent an assignment of \textsc{false}.  We implement the connections to the clause gadgets such that all other domino clues are at least three cells away from the diagonal.  Under this property, the path may switch mid-variable-gadget from up to down dominoes, but not from down to up (the path would intersect itself).  Because connections to the negative clauses precede connections to the positive clauses, switching from up to down is never beneficial. See Figure~\ref{nonrot-domino-var}.

\begin{figure}
\centering
\subcaptionbox{Unset.}{%
  \hspace{.3cm}%
  \includegraphics[width=.2\textwidth]{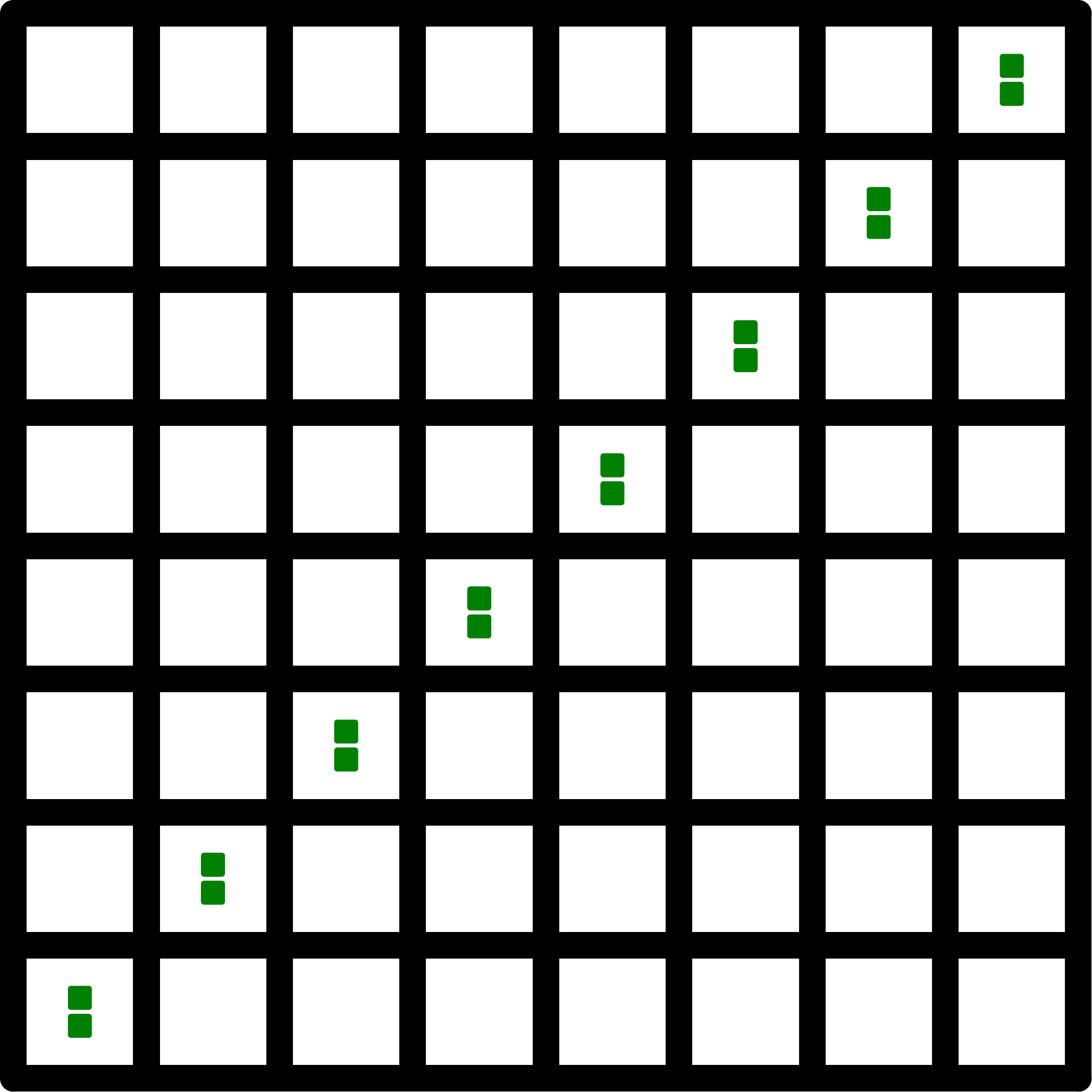}%
  \hspace{.3cm}%
  \label{nonrot-domino-var-unset}%
}~~~
\subcaptionbox{True.}{%
  \hspace{.3cm}%
  \includegraphics[width=.2\textwidth]{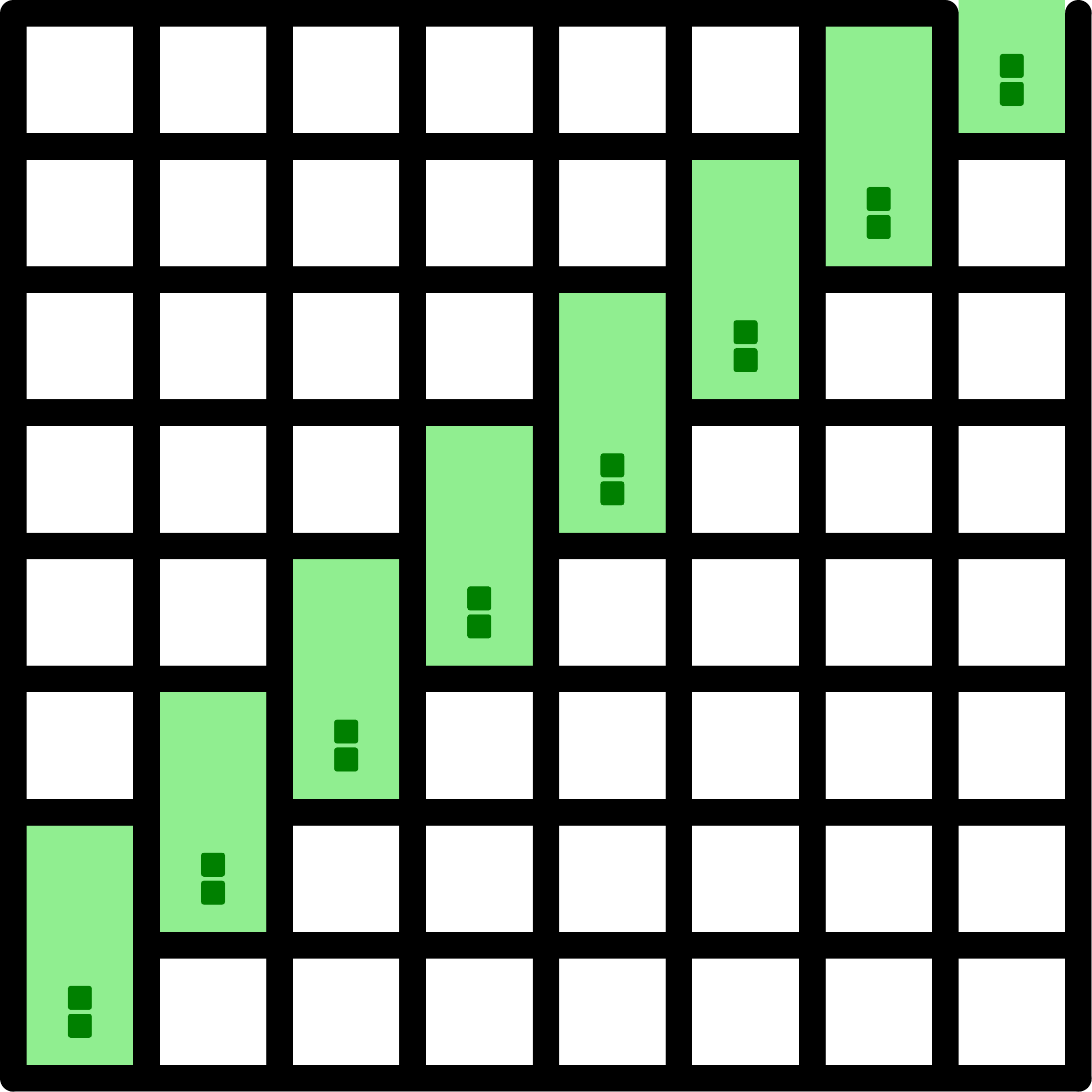}%
  \hspace{.3cm}%
  \label{nonrot-domino-var-true}%
}~~~
\subcaptionbox{False.}{%
  \hspace{.3cm}%
  \includegraphics[width=.2\textwidth]{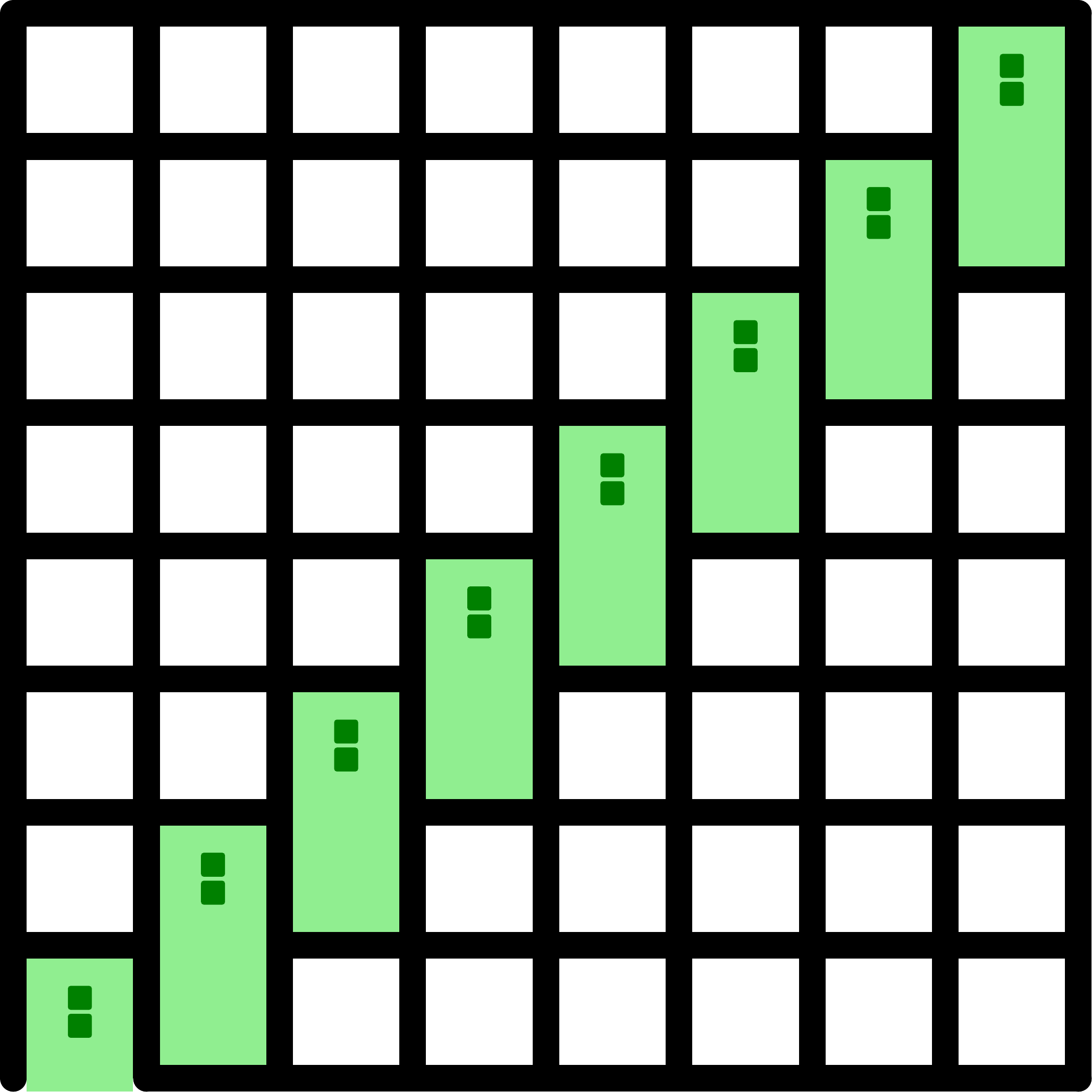}%
  \hspace{.3cm}%
  \label{nonrot-domino-var-false}%
}~~~
\subcaptionbox{False to true (invalid).}{%
  \hspace{.3cm}%
  \includegraphics[width=.2\textwidth]{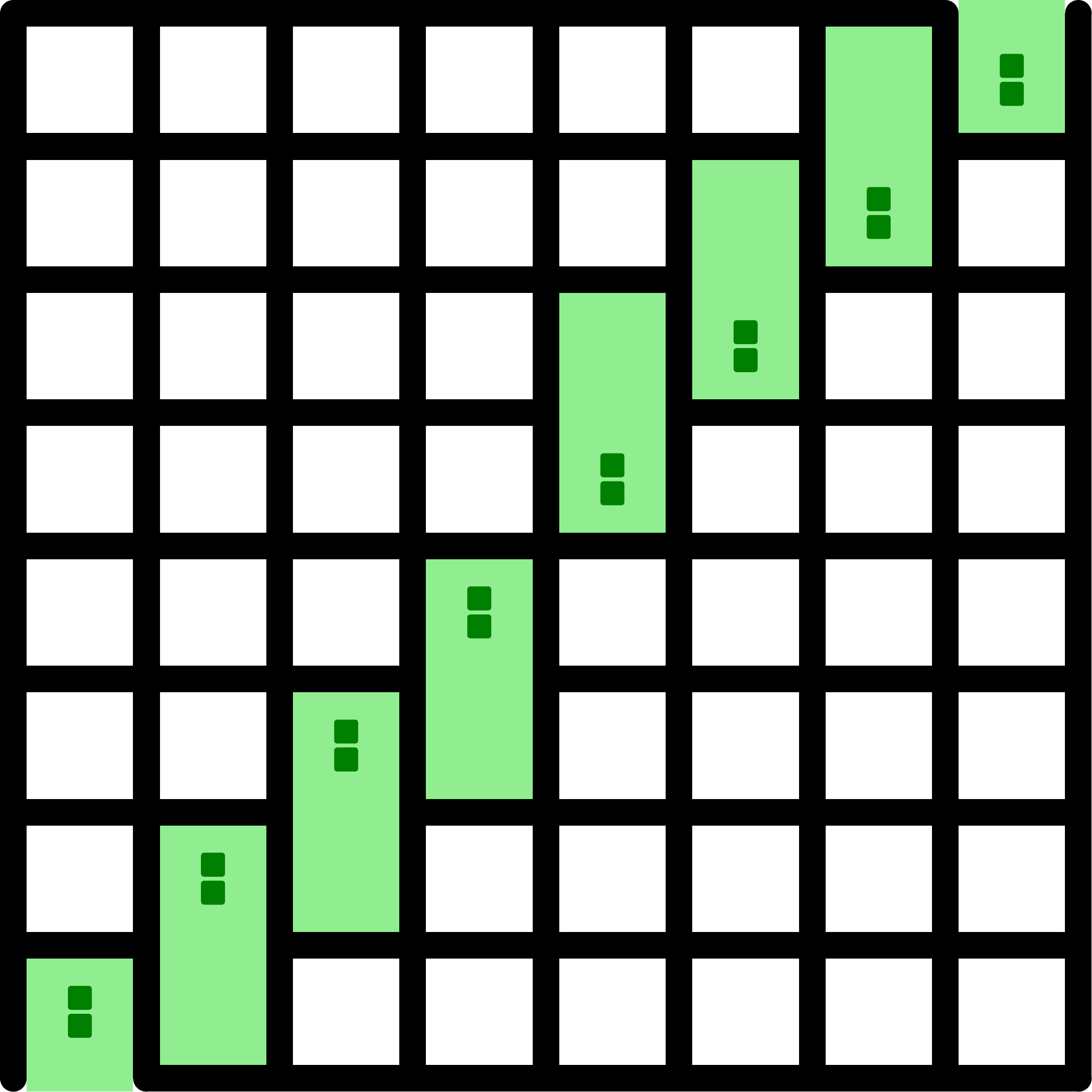}%
  \hspace{.3cm}%
  \label{nonrot-domino-var-inconsistent}%
}~~~
\caption{A 2-block-wide variable gadget and its possible domino tilings.  Switching from false to true mid-gadget is impossible.}
\label{nonrot-domino-var}
\end{figure}

\paragraph{Clause gadget.}
The clause gadget is a horizontal line of domino clues with horizontal extent exactly matching the extent of the blocks allocated from variable gadgets for connections to this clause.

\begin{figure}
\centering
\includegraphics[width=.9\textwidth]{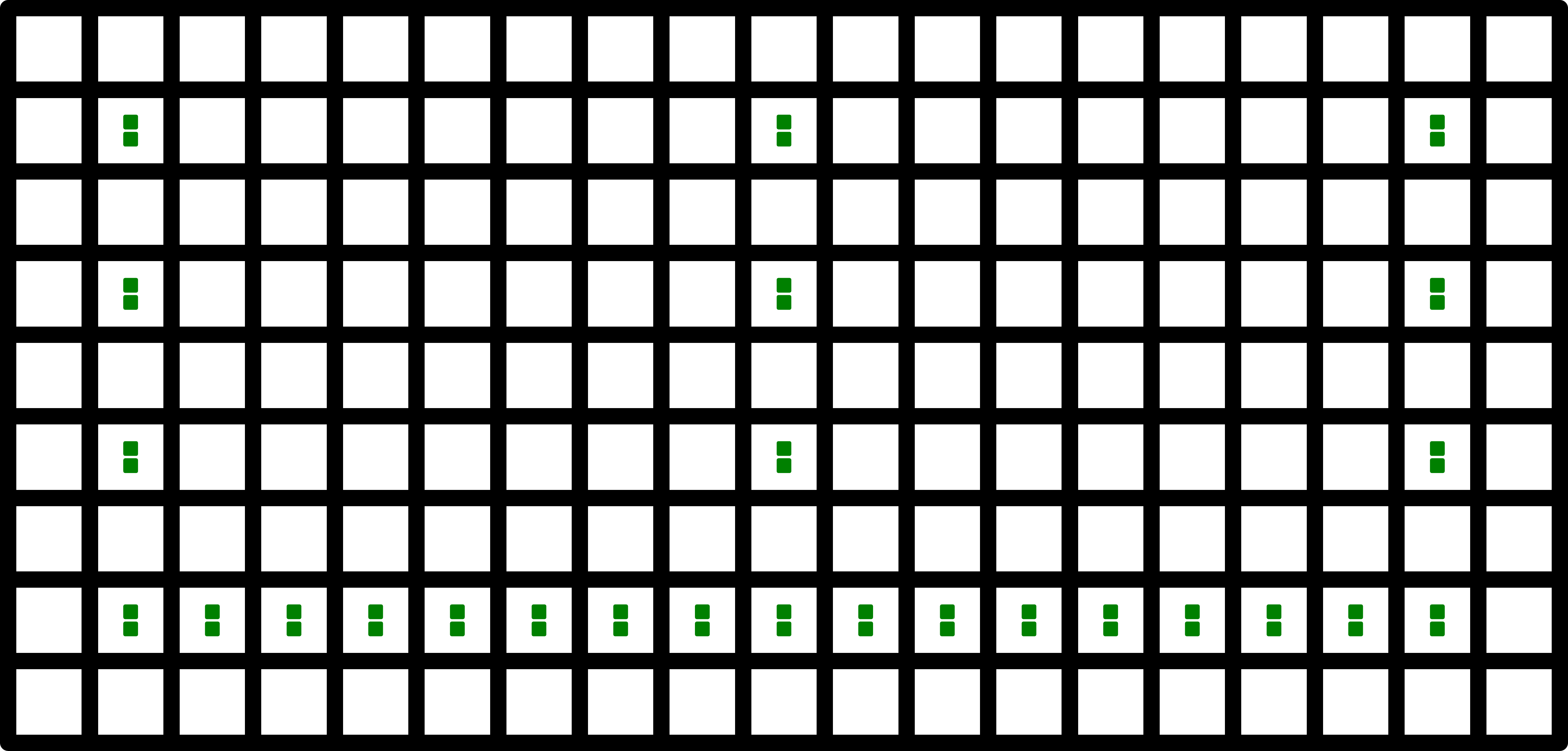}
\caption{A negative clause gadget and its three literal connections.}
\label{nonrot-domino-clause-unset}
\end{figure}

\paragraph{Layout and connections.}
Given an instance $R$ of planar rectilinear monotone 3SAT, we construct a Witness puzzle $W$ whose only clues are $1 \times 2$ nonrotatable dominoes.

We place the variable gadgets in the same order as the corresponding variable intervals appear on the $x$-axis in $I$.  The bottom-left clue of each variable gadget is two cells to the right of the top-right clue of the previous variable gadget, with a domino clue in the cell between allowing the path to switch freely from up to down or down to up at the transition point between gadgets.

We place the clause gadgets with their left and right ends aligned with the left and right ends of the outermost blocks in the variable gadgets they connect to and sorted vertically with respect to the variable gadgets and other clause gadgets in the same way as the corresponding clause line segments in $I$.  We scale the vertical distance between clause line segments such that every pair of clause gadgets is separated by a vertical distance of at least five cells.  (In particular, positive (negative) clause gadgets are above (below) the variable gadgets.)

Each connection between a variable and a clause is allocated to one block of four clues in the corresponding variable gadget.  Of the two center columns of that block, we choose the column in which the variable's clue's vertical parity is opposite to the vertical parity of the line of clues that constitute the clause gadget.  We place domino clues in every square of the clause line's vertical parity in that column between the clause and variable gadgets, except the one such square adjacent to a clue in the variable.

\begin{figure}
\centering
\subcaptionbox{Satisfied.}{%
  \hspace{.3cm}%
  \includegraphics[width=.2\textwidth]{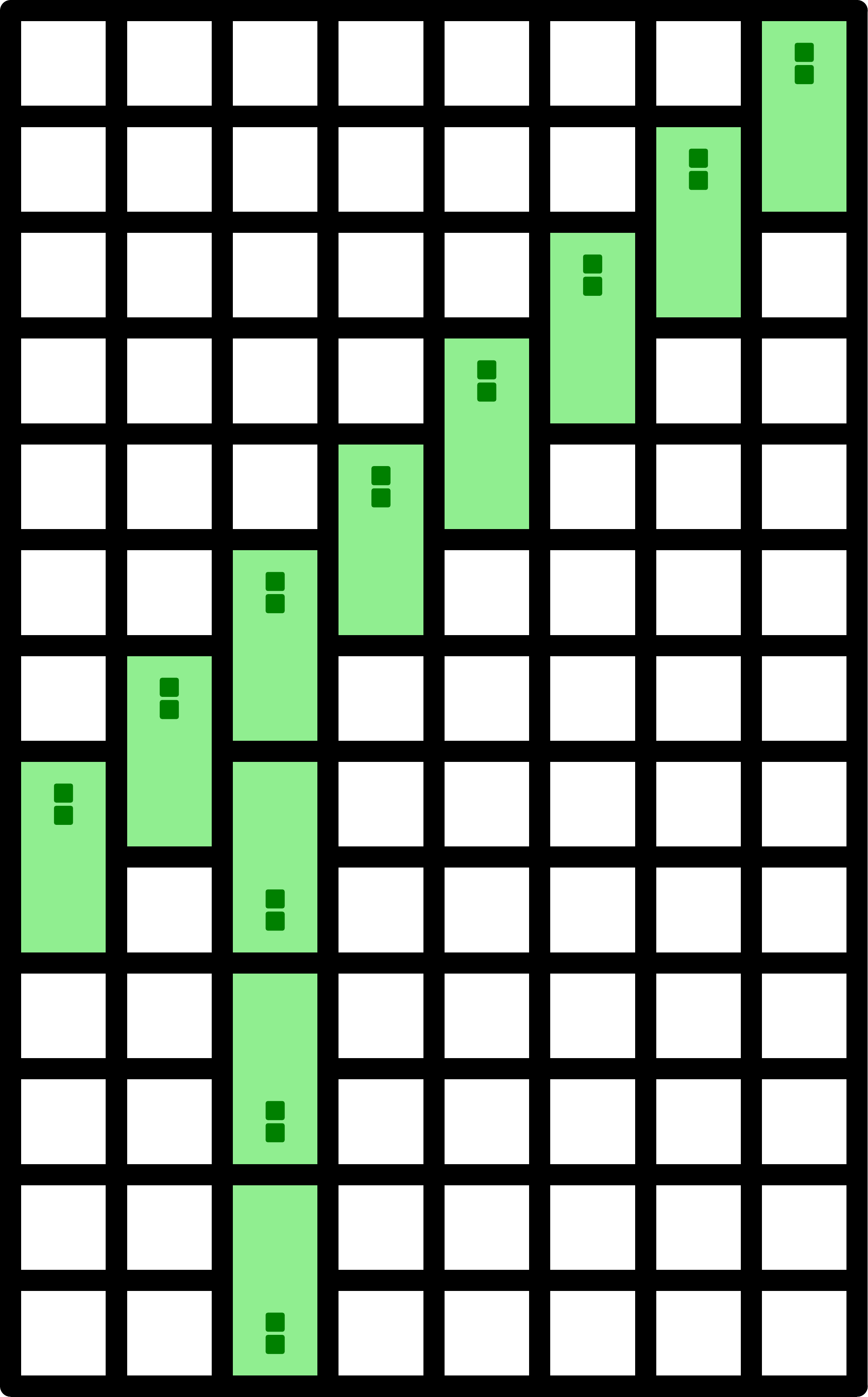}%
  \hspace{.3cm}%
  \label{nonrot-domino-neg-lit-sat}%
}~~~
\subcaptionbox{Unsatisfied.}{%
  \hspace{.3cm}%
  \includegraphics[width=.2\textwidth]{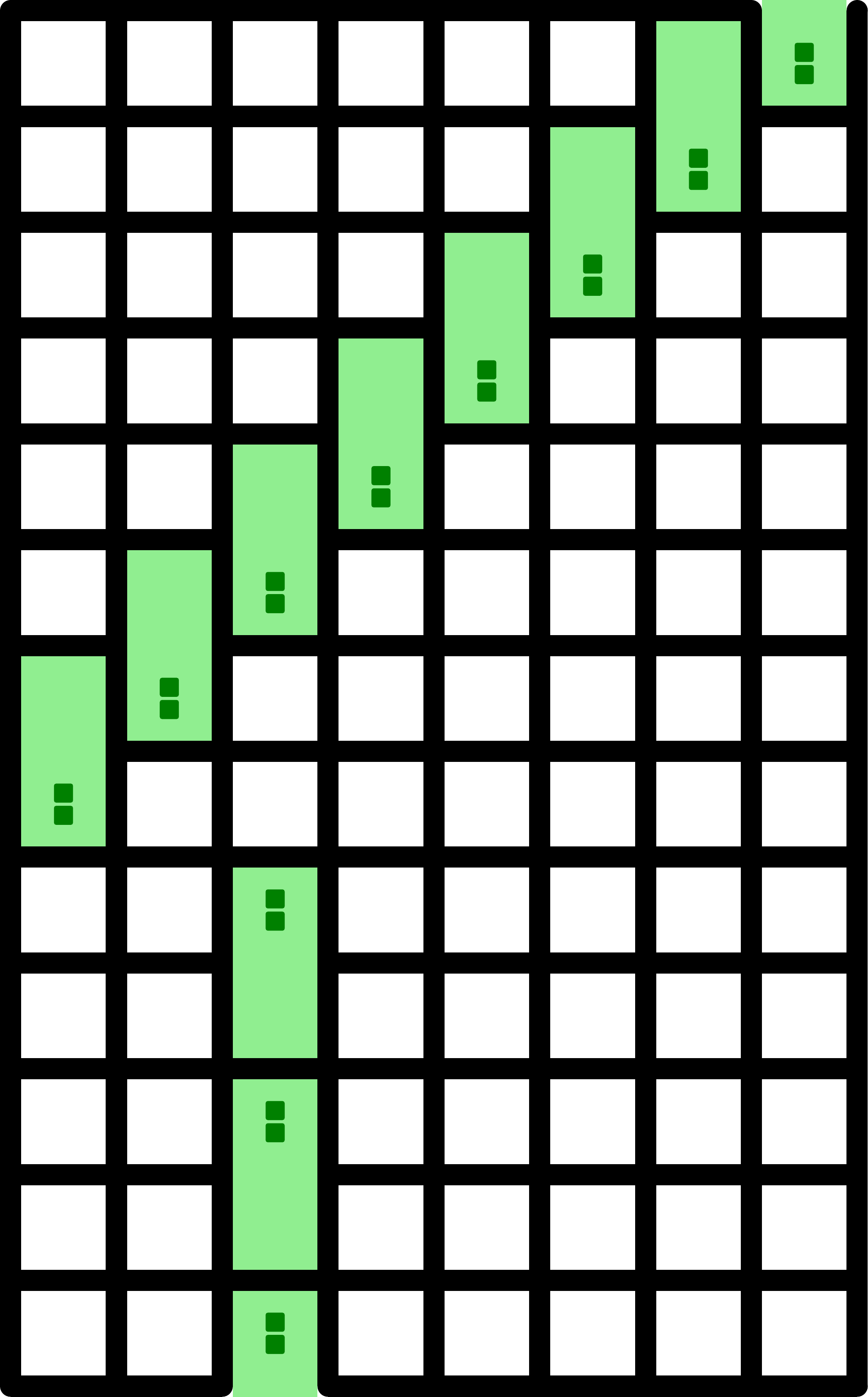}%
  \hspace{.3cm}%
  \label{nonrot-domino-neg-lit-unsat}%
}~~~
\subcaptionbox{Invalid.}{%
  \hspace{.3cm}%
  \includegraphics[width=.2\textwidth]{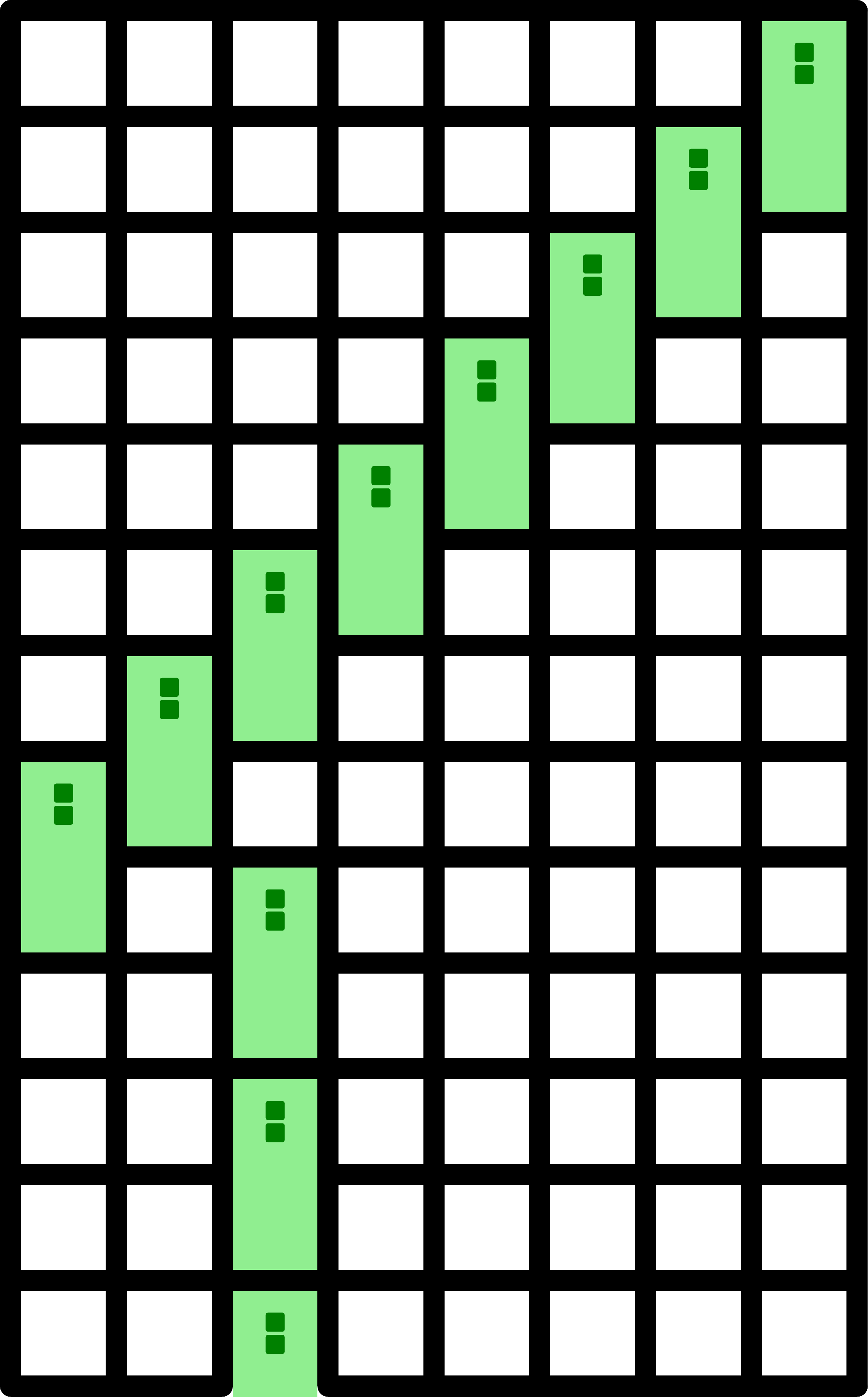}%
  \hspace{.3cm}%
  \label{nonrot-domino-neg-lit-inconsistentf}%
}~~~
\subcaptionbox{Invalid.}{%
  \hspace{.3cm}%
  \includegraphics[width=.2\textwidth]{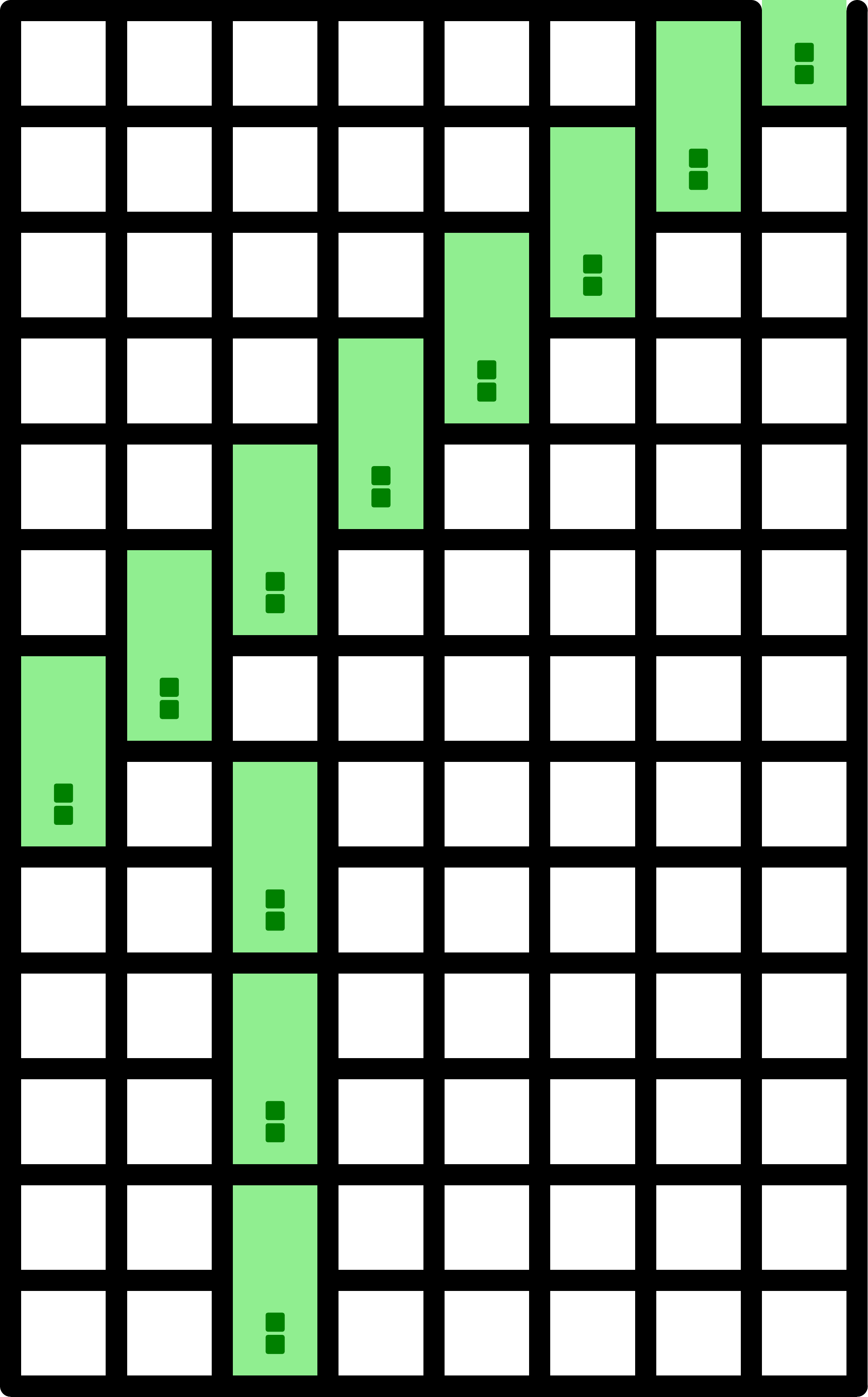}%
  \hspace{.3cm}%
  \label{nonrot-domino-neg-lit-inconsistentt}%
}~~~
\caption{A negative literal connection to a variable gadget.  The other two possibilities are invalid because the solution path would touch itself.}
\label{nonrot-domino-neg-lit}
\end{figure}

We place one domino clue in the square $S$ immediately to the left of the bottom-left clue in the leftmost variable gadget.  The left boundary of the board is immediately left of $S$, while the other three boundary are three empty cells away from the most extreme clue in that direction.  We put the start circle and end cap at the bottom-left and top-left vertices of $S$.

\begin{figure}
\centering
\subcaptionbox{The puzzle.}{%
  \includegraphics[width=.49\textwidth]{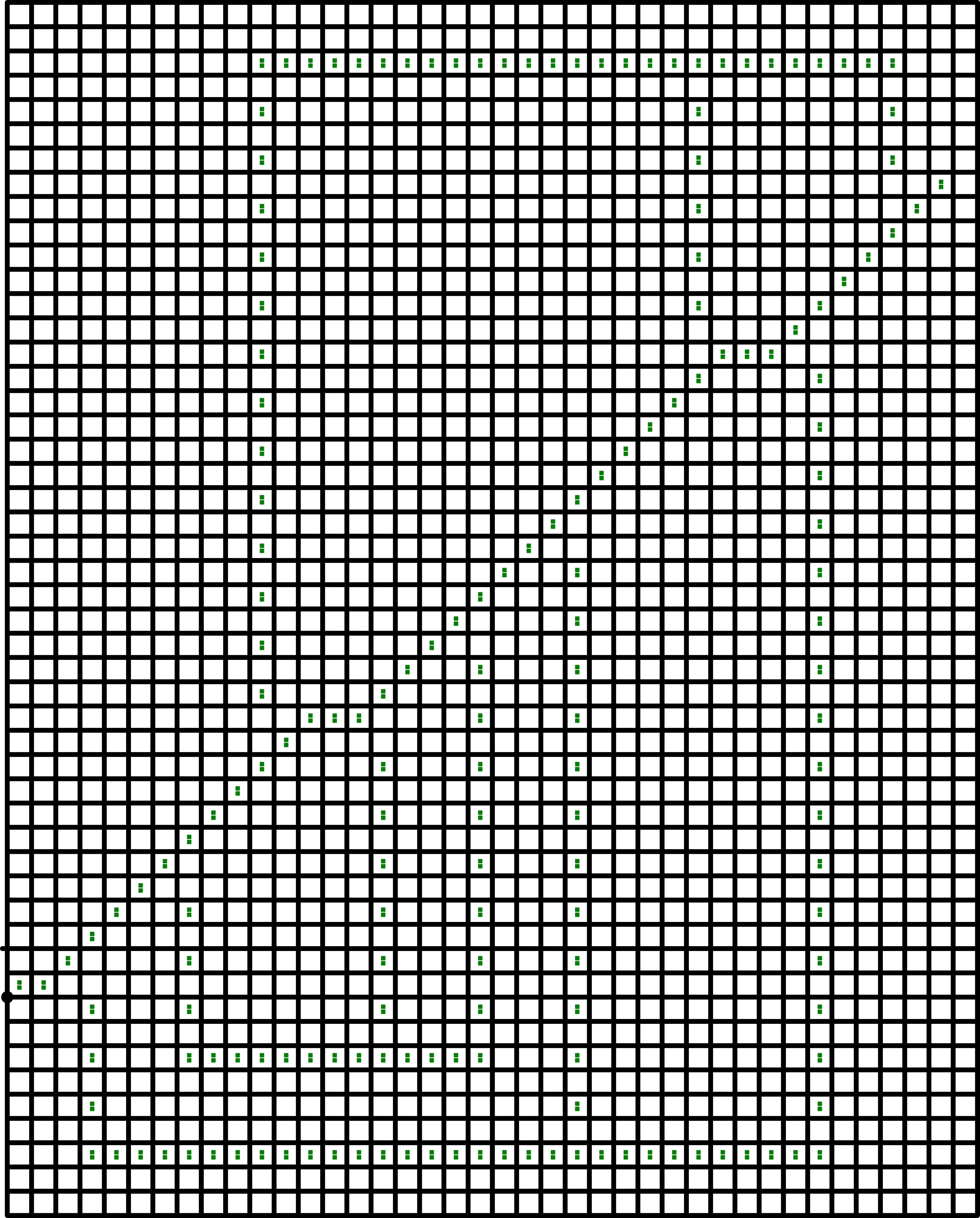}%
  \label{nonrot-domino-worked-unsolved}%
}
\subcaptionbox{The solution.}{%
  \includegraphics[width=.49\textwidth]{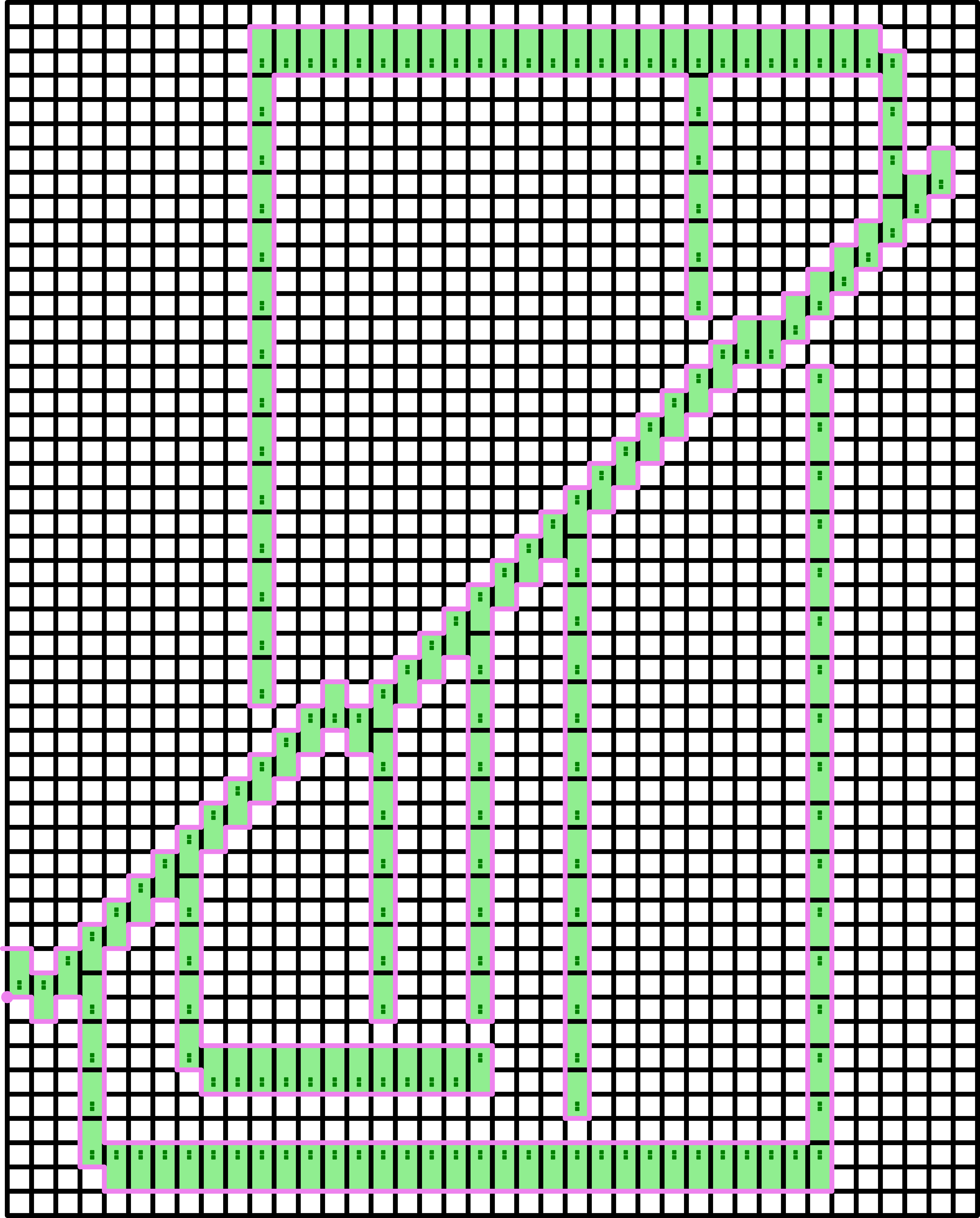}%
  \label{nonrot-domino-worked-solved}%
}
\caption{A Witness puzzle produced from $(x \lor y \lor z) \land (\lnot x \lor \lnot y \lor \lnot z) \land (\lnot x \lor \lnot y \lor \lnot y)$ and its solution ($x$ and $y$ are \textsc{false}, $z$ is \textsc{true}).  Shaded cells show the domino tiling on the interior of the path.}
\label{nonrot-domino-worked}
\end{figure}

\paragraph{Witness solution $\to$ Satisfying assignment.}
Suppose the Witness puzzle produced by our reduction has a solution: a path from the start circle to the end cap such that each induced region that contains domino clues is exactly tiled by that many (vertical) dominoes.  As a result, every square containing a domino clue is covered by a domino.  Because no two domino clues are vertically adjacent, we can associate each clue with the domino that covers it.  By this correspondence, those are the only dominoes in the tiling.

By construction, the domino covering $S$ is the only domino touching the boundary of the board (all other domino clues are too far from the boundary), and so only one region containing domino clues touches the boundary.  Because solution paths are simple paths, every region in a Witness puzzle touches the boundary, so all domino clues are in a single region.

We can read off a truth assignment from the position of the dominoes covering the clues in the variable gadget: \textsc{true} if the domino is up and \textsc{false} if the domino is down.  As explained above, the path may switch from up to down in the middle of a variable gadget, but such switches do not help in satisfying clauses, so we take the first domino in each variable gadget to determine the assignment of that variable.

It remains to show that this assignment satisfies the formula.  Because all domino clues are in a single region, the clause gadgets are all in the region.  For the path to reach and encircle a positive (negative) clause gadget, it must follow a connection up (down) from a variable gadget domino that is up (down) and then return along the other side of that connection; any other path to or from a clause gadget results in a region too large to cover with the domino clues.   That variable gadget, connection and clause gadget correspond to the variable that satisfies the literal in that clause.  Thus, the recovered truth assignment satisfies the formula.

\paragraph{Satisfying assignment $\to$ Witness solution.}
Given a satisfying assignment, we can construct a domino tiling of a tree of cells rooted at $S$, the boundary of which is a solution path.  We begin by covering each domino clue in the variable gadgets with a domino positioned according to the value of that variable (up for \textsc{true}, down for \textsc{false}).  We cover the connections corresponding to satisfied positive literals with up dominoes and satisfied negative literals with down dominoes, such that the connection connects the variable and clause gadgets.  Then for all but one of the satisfying literals, we move the two dominoes furthest from the clause gadget to the other orientation (further away from the variable gadget), disconnecting the connection to avoid forming a closed loop that would make drawing the solution path impossible.  We cover the unsatisfied positive literals with down dominoes and satisfied negative literals with up dominoes, leaving them connected to the clause gadget but not the variable gadget.  Then we cover the remaining clause gadget domino clues with up dominoes for positive clauses and down dominoes for negative clauses.

\begin{figure}
\centering
\includegraphics[width=.9\textwidth]{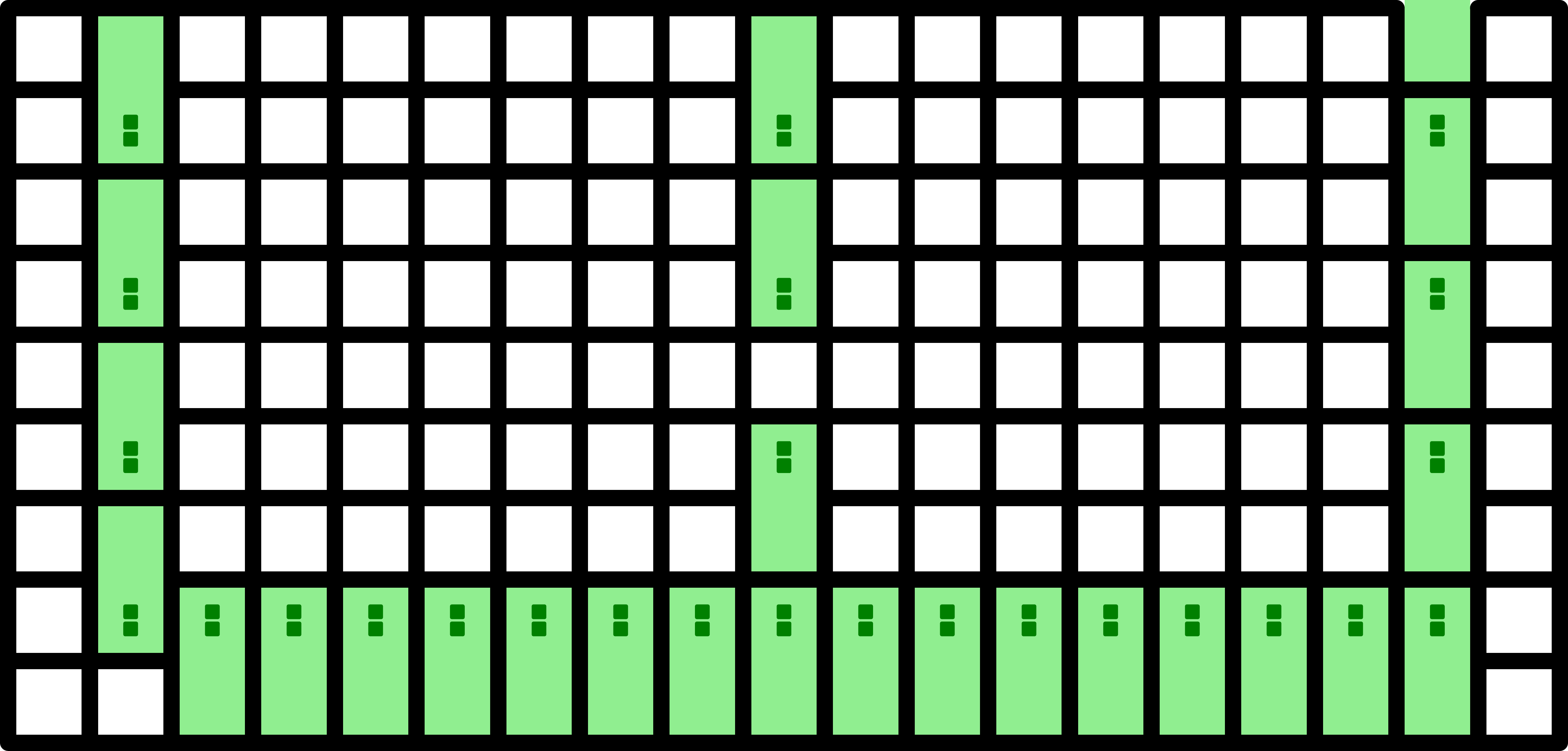}
\caption{A negative clause gadget with the left two literals satisfied and the right literal unsatisfied.  The left and middle literal connections connect to their variable gadgets, while the right literal does not.  The middle literal could be used to satisfy this clause (instead of the left literal) by moving the bottom two dominoes in that column up and the bottom two dominoes of the left literal's connection column down.}
\label{nonrot-domino-clause-sat}
\end{figure}

This domino tiling is a tree: the variable gadgets are connected in a line because they have adjacent domino clues, exactly one connection connects each clause gadget to exactly one variable gadget, and the remaining connections connect to either their variable gadget or their clause gadget (not both).  $S$ is covered by the tiling because it contains a domino clue.  Then the boundary of this domino tiling (deleting $S$'s left edge) is a solution path in the Witness puzzle.
\end{proof}
}%\later

\section{Antibodies}
\label{Antibodies}

\abstractlater{
  \section{Proofs: Antibodies}
  \label{appendix:antibodies}
}

% https://www.reddit.com/r/TheWitness/comments/4sg8zz/any_clue_about_the_meaning_of_the_annulation/

An antibody ($\antibody$) eliminates itself and one other clue in its region.  For
the antibody to be satisfied, this region must \textit{not} be satisfied
without eliminating a clue; that is, the antibody must be necessary.  An antibody may be colored, but its color does not restrict which clues it can eliminate.\footnote{Antibody color matters when checking if the antibody is necessary; a region containing only a star and an antibody of the same color is unsatisfied because the antibody is not necessary.}
Very few Witness puzzles contain multiple antibodies, making the
formal rules for the interactions between antibodies not fully determined by the
in-game puzzles. We believe the following interpretation is a natural
one: each antibody increments a count of clues that must
necessarily be unsatisfied for their containing region to be
satisfied. If there are $k$ antibodies in a region, then there must be
$k$ clues which can be eliminated such that those $k$ clues were
unsatisfied and all other clues were satisfied; furthermore, there must
not have been a set of fewer than $k$ unsatisfied clues such that all
other clues are satisfied\footnote{Whether or not a clue is satisfied is
  usually determined only by the solution path; however, in the case of
  polyominoes and antipolyominoes, there might be several choices of
  packings which satisfy different sets of clues.}. Antibodies cannot
  eliminate other antibodies. The choice of clue to
  eliminate need not be unique; for instance, a region with three white
  stars and one antibody is satisfied, even though the stars are not 
  distinguished. Formally:

\begin{definition}[Simultaneous Antibodies]
A region with $k$ antibody clues is
satisfied if and only if there exists a set $S$ of $k$ non-antibody clues
such that eliminating all clues in $S$ and all $k$ antibodies leaves the region satisfied,
and there does \textit{not} exist a set $S'$ of non-antibody clues with
$|S'| < k$ such that eliminating all clues in $S'$ and only $|S'|$ of the antibodies leaves the region satisfied.
\end{definition}

\ifabstract
For proofs omitted from this section, see Appendix~\ref{appendix:antibodies}.
\fi

%Under the sequential interpretation, an antibody is satisfied if it eliminates another clue and the region would not be satisfied without this removal. % Antibodies are themselves clues, so they can be eliminated by other antibodies.  \xxx{An erased antibody implies no constraint on the puzzle. -- what?}  
%This interpretation is most easily defined in terms of a verification procedure.

%\newcommand\Verify{\operatorname{Verify}}

%\begin{definition}[Sequential Antibodies]
%  A region $R$ containing a set of clues $C$ is satisfied if the following recursive procedure $\Verify(R, C)$ returns true:
%  \begin{enumerate}
%  \item If $C$ contains no antibody clues, verify whether all clues in $C$ are satisfied in $R$, and return this answer.
%  \item If $C$ contains an antibody clue $e$:
%    \begin{enumerate}
%    \item If there are no other clues, return false.
%    \item Evaluate $\Verify(R,C\setminus\{e,c\})$ for all clues $c\in C$ other than $e$. If none of these returns true, return false.
%    \item Evaluate $\Verify(R,C\setminus \{e\})$, and return the opposite.
%  \end{enumerate}
%\end{enumerate}
%\end{definition}

%\xxx{would be great to have a simple example of a case where they differ}

%\xxx{Must an antibody eliminate a unique clue?  We decided against it, but we might want to mention it.  \url{https://coauthor.csail.mit.edu/6.890/m/RWRsWuueQQP9z5S5a}}

%\xxx{discussion of which interpretation applies to what}

\subsection{Positive Results}
This section gives algorithms and arguments showing various types of Witness puzzles can be solved in NP or $\Sigma_2$. We begin with a simple case: puzzles with no antibodies are in NP. Then we show that puzzles with antibodies but without polyominoes or antipolyominoes is also in NP. Next we restrict to a single antibody and add polyomino clues back in, which remains in NP. Finally, we show that puzzles with all clue types are in $\Sigma_2$. These last three results provide tight matches for our lower bounds in Section~\ref{sec:antibody hardness}.

\begin{observation}\label{all-but-antibody-np}
  Witness puzzles containing all clue types except antibodies are in NP.
\end{observation}

\begin{proof}
  For all clues except polyominoes and antipolyominoes,
  a witness to such puzzles is the solution path.
  (Because the solution path must be simple, it has linear size.)
  For each region defined by the solution path that has a polyomino or
  antipolyomino clue, we also include the (anti)polyomino packing as part
  of the witness.  This packing can be encoded as $(x,y)$ coordinates for
  the topmost leftmost pixel of each polyomino and each antipolyomino.
  These coordinates can be assumed to be polynomial in the input size:
  given any solution packing, delete any rows or columns not intersecting
  the board nor the bounding box of any polyomino.  Because some polyomino
  in the packing covers a cell of the region in the board, the other
  polyominoes can be no farther away than the number of polyominoes times the
  maximum size of the board or a polyomino's bounding box.
\end{proof}

\both{
\begin{theorem}\label{all-but-poly-antipoly-np}
  Witness puzzles containing all clue types except polyominoes and antipolyominoes are in NP.
\end{theorem}
}

\ifabstract
\begin{proofsketch}
Other than antibodies, polyominoes, and antipolyominoes, whether or not a clue is satisfied can be easily determined from the solution path. Thus, checking
whether an antibody which eliminates such a clue is necessary is easy. 
\end{proofsketch}
\fi

\later{
\begin{proof}
%\xxx{What if we're using the sequential definition?}
We give a polynomial-time algorithm to verify a claimed solution path.  Each region induced by the path can be verified separately because no clues can interact with clues in other regions.  For regions not containing antibodies, we simply check that all clues in the region are satisfied, which can be done in polynomial time.

An \emph{elimination set} is a set of eliminated non-antibody clues.  An elimination set is \emph{proper} if it has size equal to the number of antibodies in the region and \emph{improper} if it is smaller.  A region is satisfied under an elimination set if the non-eliminated non-antibody clues are satisfied; we say the region is \emph{properly satisfied} if the elimination set is proper and \emph{improperly satisfied} otherwise.  Then the region is satisfied if it is properly satisfied under at least one elimination set (all of the antibodies can be used) and not improperly satisfied under any elimination set (all of the antibodies are necessary).

We proceed by \emph{marking} non-antibody clues that must be in any elimination set that satisfies the region (properly or improperly).  When we have a choice of clues to mark, we mark as few clues as possible, or if it cannot be known how many clues to mark, we evaluate the subproblem resulting from each choice.  After processing all clue types, we check whether the elimination set of the marked clues properly or improperly satisfies the region (or neither).

To begin, we mark all edge and vertex hexagons in the interior of the region (i.e., not visited by the solution path) and all triangle clues not adjacent to the corresponding number of edges in the solution path.  Hexagons and triangles do not interact with each other or other clue types, so we have no choices to make here.

The remaining allowed clue types are squares, stars and antibodies, which interact by their color.  Within a region, clues of these types with the same color are interchangeable because their location is irrelevant to their satisfaction, so we record only how many clues of each type and color have been marked.  We assume for the moment that all antibodies are used and thus eliminate themselves.  Of stars and squares, a star can be satisfied by pairing with a square of the same color, but the presence or absence of stars has no effect on squares, so we consider squares first.  Squares are only satisfied if the region is monochromatic in squares, so we must mark all clues of all but one color of square.  We evaluate the subproblems resulting from each choice of surviving square color.

Only clues of the same color are relevant to satisfying stars, so we process each star color independently.  For each color of star, we have a choice between marking all stars of that color (zero stars survive), all but one star of that color and all but one square of that color (only possible if the star matches the surviving square color), or all but two stars of that color and all squares of that color (if any).  We choose the number of surviving stars that results in marking the fewest clues; if there is a tie, we choose the smaller number of surviving stars.\footnote{We must choose zero surviving stars if possible; between one and two surviving stars there is no difference.}

At this point, all unmarked clues are satisfied.  We now check whether the region is satisfied under the elimination set consisting of the marked clues.

\begin{itemize}
\item If the region contains more marked clues than antibodies, then the region is not satisfied under this elimination set.  We continue with the next choice of square color.

\item If this elimination set is improper, then some of the antibodies are not used, contrary to our assumption.  Unused antibodies do not eliminate themselves, so they may cause a star to be unsatisfied.  Let $u$ be the number of unused antibodies and $v$ be the number of antibody clues not sharing a color with a surviving star clue.  If $u > v$, at least one of the unused antibodies causes a star to be unsatisfied, so the region is not satisfied under this elimination set, and we continue with the next choice of square color.  Otherwise ($u \le v$), the region is improperly satisfied under the marked clues, so we reject the solution path.  Note that our tie-breaking preference for zero surviving stars maximizes $v$ without decreasing $u$, ensuring that we detect improper satisfaction.

\item If this elimination set is proper, then we need to check one special case to discharge our assumption that all antibodies are used: if the region contains an antibody of the same color as a marked star and all other non-antibody clues of that color (if any) are marked, then the region is also satisfied if that star survives and that antibody is unused.  In that case, the region is improperly satisfied after that star is removed from the elimination set, so we reject the solution path.  Our tie-breaking preference for zero surviving stars ensures that we detect this.  That is the only way an unused antibody can satisfy another clue, so otherwise our assumption that all antibodies are used holds, and we record that the region can be properly satisfied.  We still have to evaluate any remaining choices of square color to verify they do not allow the region to be improperly satisfied.
\end{itemize}

If, after checking all choices of surviving square color, we did not reject the solution path, and the region was properly satisfied under at least one elimination set, then this region is satisfied.  The path is a solution exactly when all regions are satisfied.

It remains to show this algorithm runs in polynomial time.  There are only polynomially many regions to check.  For each region, there are only polynomially many squares in the region, and thus at most polynomially many colors of squares.  For each color of square, we do work proportional to the number of colors of stars, but similarly there are only polynomially many stars in the region and so only polynomially many colors of stars.  The work done at the end of each iteration requires only counting the number of clues of a particular type or color that are marked or unmarked.  Then an NP algorithm to solve Witness puzzles containing all clue types except polyominoes and antipolyominoes follows directly from this polynomial-time verification algorithm.
\end{proof}
}

\both{
\begin{theorem} \label{one antibody in NP}
\label{thm:antibody-np}
  Witness puzzles containing all clue types except antipolyominoes and for which at least one solution eliminates at most one polyomino in each region are in NP. 
\end{theorem}
}

% No puzzle in The Witness contains more than one antibody, so all puzzles in The Witness satisfy the promise. % This isn't true for 1 antibody + antipolyominoes
\ifabstract
\begin{proofsketch}
If at least one polyomino is eliminated in a region containing at least two polyominoes and the region is satisfied as a result, 
then the region can't be satisfied without deleting at least one polyomino because the total area of the polyominoes is greater than that of the region, and therefore there is no packing.
\end{proofsketch}
\fi

\later{
\begin{proof}
We give a polynomial-time algorithm to verify a certificate.
Similar to Observation~\ref{all-but-antibody-np}, the certificate includes
a solution path and a polyomino packing for each region defined by the solution
path that has an uneliminated polyomino.  In addition, the certificate
includes which clue gets eliminated by each antibody, satisfying the
promise that at most one polyomino from each region gets eliminated.
%Unlike in Theorem~\ref{all-but-poly-antipoly-np}, the certificate is more than the solution path.  Because polyomino packing is NP-complete\xxx{cite}, we need the certificate to provide a packing witness for each region containing polyomino clues. %encoded as a map from the coordinates of each cell containing a polyomino clue participating in the packing (not eliminated) to the coordinates of the cell occupied in the packing by the topmost leftmost pixel of the polyomino.
(The NP-hardness of polyomino packing \cite{Jigsaw_GC} necessitates this promise; otherwise, verifying a packing resulting from using two antibodies to eliminate two polyominoes would require checking that there is no packing resulting from using just one of the antibodies to eliminate one (larger) polyomino.)
%We do not need the certificate to indicate which clues are eliminated (but because the packing witness only contains the positions of polyominoes participating in the packing, we can deduce which polyomino clue was eliminated in that region, if any).

We first verify the packing witnesses.  The certificate is invalid if it contains a packing witness for a region with no polyomino clues or if the specified arrangement of the polyominoes does not actually pack the region.  Otherwise, for each region containing polyomino clues:
\begin{itemize}
\item If the certificate does not provide a packing witness for this region, this region must contain exactly one polyomino clue and one antibody, and the polyomino clue must not have the shape of the entire region (because then the polyomino is satisfied and can't be eliminated); otherwise, the certificate is invalid.  We record the antibody as eliminating the polyomino clue.
\item If the certificate contains a packing witness omitting the position of more than one polyomino clue (requiring two polyomino clues to be eliminated), the certificate is invalid.  (The solution path in the certificate may be a valid solution to the Witness puzzle, but this algorithm cannot verify it.  By the promise that there exists a solution eliminating at most one polyomino per region, there is some other certificate attesting that this Witness puzzle is a \textsc{yes} instance that this algorithm can verify.)
\item If the certificate contains a packing witness omitting the position of exactly one polyomino clue, this region must contain an antibody; otherwise, this certificate is invalid.  We record the antibody as eliminating the omitted polyomino clue.  (We know the antibody necessarily eliminates the omitted polyomino clue because otherwise the total area of the polyominoes would be greater than the size of the region.)
\item If the certificate contains a packing witness specifying the position of every polyomino clue, the polyominoes enforce no further constraint.
\end{itemize}

Now consider the Witness puzzle obtained by taking the initial instance and removing all polyomino clues and the antibodies we recorded as eliminating polyomino clues.  The resulting puzzle contains neither polyominoes nor antipolyominoes, so we can apply the algorithm given in the proof of Theorem~\ref{all-but-poly-antipoly-np} to verify that the remaining clues in each region are satisfied under the solution path.  The certificate is valid exactly when the solution path is valid for the resulting puzzle.  Then an NP algorithm to solve Witness puzzles containing all clue types except antipolyominoes and for which at least one solution eliminates at most one polyomino in each region follows immediately from this polynomial-time verification algorithm.
\end{proof}
}

\begin{theorem} \label{in Sigma_2}
Witness puzzles containing any set of clue types (including polyominoes, antipolyominoes, and antibodies) are in $\Sigma_2$.
\end{theorem}

\begin{proof}
Solving this Witness puzzle requires picking clues for antibodies to eliminate and finding a path which respects the remaining clues, such that the regions cannot be satisfied if only a subset of antibodies are used to eliminate clues.
Membership in $\Sigma_2$ requires an algorithm which accepts only when there exists a certificate of validity for which there is no certificate of invalidity (i.e., one alternation of $\exists x \forall y$). A certificate of invalidity allows a polynomial-time algorithm to check whether an instance of a given problem is false.
Our certificate of validity is a solution path, a mapping from antibodies to eliminated clues, and a packing witness for any region with at least one uneliminated polyomino.
Our certificate of \emph{invalidity} is the solution path (from the certificate of validity),  a mapping of a \emph{subset} of the antibodies to eliminated clues, and a packing witness for any region with at least one uneliminated polyomino.

Our verification algorithm begins checking the certificate of validity by verifying the packing witnesses and checking that the antibody mapping specifies distinct eliminated clues in the same region as each antibody.  Then we remove all antibody clues, polyomino and antipolyomino clues, and eliminated clues from the Witness puzzle and
(like the proof of Theorem~\ref{all-but-poly-antipoly-np})
verify in polynomial time that the remaining clues in each region are satisfied under the solution path.

%The certificate of invalidity is chosen via universal nondeterminism, so if one exists, it is valid.  
To verify the certificate of invalidity, we again check its packing witnesses and its (partial) antibody mapping.  Then we remove the used antibody clues, polyomino and antipolyomino clues, and eliminated clues from the Witness puzzle.  Because unused antibodies interact with stars (and only stars), we replace any unused antibodies with stars of their color if they are in the same region as an (uneliminated) star of that color, then remove any remaining antibodies.  We verify in polynomial time that the resulting Witness puzzle is satisfied by the solution path.  Our algorithm accepts if and only if the certificate of validity is valid and all certificates of invalidity are invalid.
\end{proof}

\subsection{Negative Results}
\label{sec:antibody hardness}

In this section, we prove that Witness puzzles in general are $\Sigma_2$-complete. We will proceed in two steps, first considering puzzles which have two
(or more) antibodies which might be eliminating polyominoes in the same region, then considering puzzles which have only one antibody but both polyominoes and antipolyominoes. In both cases,
we reduce from \emph{Adversarial-Boundary Edge-Matching}, a one-round two-player game defined as follows:

\begin{problem}[Adversarial-Boundary Edge-Matching]
%\xxx{italic definition formatting looks really bad for something this
%  extended.} %I think it's fine. -Adam

  A \emph{signed color} is a sign ($+$ or $-$) together with an element of
  a set $C$ of colors.  Two signed colors \emph{match} if they have the same
  element of $C$ and the opposite sign.  A \emph{tile} is a unit square with
  a signed color on each of its edges.
 
  An $n\times (2m)$ \emph{boundary-colored board} is an $n\times (2m)$
  rectangle together with a signed color on each of the unit edges
  along its boundary.  Given such a board and a multiset $T$ of $2nm$ tiles,
  a \emph{tiling} is a placement of the tiles at integer locations
  within the rectangle such that two adjacent tiles have matching colors along
  their shared edge, and a tile adjacent to the boundary has a matching
  color along the shared edge.  There are two types of tiling according to
  whether tiles can only be translated or can also be rotated.

  The \emph{adversarial-boundary edge-matching game} is a one-round two-player
  game played on a
  $2n \times m$ boundary-colored board $B$ and a multiset $T$ of $2nm$ tiles.
  Name the unit edges along $B$'s top boundary $e_0, e_1, \ldots, e_{2n}$
  from left to right.  During the first player's turn, for each even $i = 0, 2, 4,
  \dots, 2n-2$, the first player chooses to leave alone or swap the signed
  colors on $e_i$ and $e_{i+1}$.
  During the second player's turn, the second player attempts to tile the resulting
  boundary-colored board $B'$ such that signed colors on coincident edges
  (whether on tiles or on the boundary of $B'$) match.
  If the second player succeeds in tiling, the second player wins;
  otherwise, the first player wins.

  The \emph{adversarial-boundary edge-matching problem} is to decide whether the first player has a winning strategy for a given adversarial-boundary edge-matching game; that is, whether there exists a choice of top-boundary swaps such that there does \textit{not} exist an edge-matching tiling of the resulting boundary-colored board.
\end{problem}

\begin{lemma} \label{Adversarial-Boundary Edge-Matching Problem Sigma_2}
  Adversarial-boundary edge-matching is $\Sigma_2$-hard,
  with or without tile rotation, even when the first player has a losing strategy.
\end{lemma}

See Appendix~\ref{appendix Adversarial-Boundary Edge-Matching Problem Sigma_2}
for the proof.

\both{
\begin{theorem}
  \label{thm:antibody-poly}
  It is $\Sigma_2$-complete to solve Witness puzzles containing
  two antibodies %(under either interpretation)
  and polyominoes.
\end{theorem}
}

%in case we want small-caps problem naming style later
%\newcommand{\abem}{adversarial-boundary edge-matching}
\global\def\abem{adversarial-boundary edge-matching}

\begin{proof}
We reduce from \abem{} with the guarantee that the first player has a losing strategy.  We create a Witness puzzle containing two antibodies.  We will force the solution path to split the puzzle into two regions, with both antibodies in the same region and with part of the solution path encoding top-boundary swaps.  In the construction, it will be easy to find a solution path satisfying all non-antibody clues when both antibodies are used to eliminate clues, but the antibodies themselves are only satisfied if they are necessary.  When only one antibody is used, the remaining polyominoes in one of the regions, together with the solution path, simulate the \abem{} instance.  The remaining polyominoes cannot pack the region (necessitating the second antibody and making the Witness solution valid) exactly when the \abem{} instance is a \textsc{yes} instance.  (In the context of The Witness, the human player is the first player in an \abem{} game, and The Witness is the second player.)

\newcommand{\numcolors}{\lceil \log_2(c+1)\rceil}

\paragraph{Encoding signed colors.}  We encode signed colors on the edges of polyominoes in binary as unit-square tabs (for positive colors) or pockets (for negative colors) \cite[Figure 7]{Jigsaw_GC}.  If the input \abem{} instance has $c$ colors, we need $\numcolors$ bits to encode the color\footnote{We cannot use $0$ as a color because we need at least one tab or pocket to determine the sign.}.  To prevent pockets at the corners of a tile from overlapping, we do not use the $2 \times 2$ squares at each corner to encode colors, so tiles are built out of squares with side length $w = \numcolors + 4$\footnote{At the cost of introducing disconnected polyomino clues, we could leave only one pixel at each corner out of the color encoding; that pixel is disconnected when the colors on its edges both have pockets next to it.}.

\paragraph{Clue sets.} We consider the clues in the Witness puzzle to be grouped into two clue sets, $A$ and $B$, which we place far apart on the board.  We will argue that any solution path must partition the puzzle into two regions, such that each set is fully contained in one of the regions.  The clue sets are shown in Figures~\ref{fig:antibody1-clueseta} and \ref{fig:antibody1-cluesetb}.

\begin{figure}
\centering
\subcaptionbox{\label{fig:antibody1-clueseta-monominoes} Antibodies.}[1in]{%
  \includegraphics{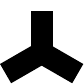}\includegraphics{figures/antibody.pdf}%
}~~~~
\subcaptionbox{Monominoes.}[1in]{%
  \shortstack{\includegraphics[trim=20 20 20 20]{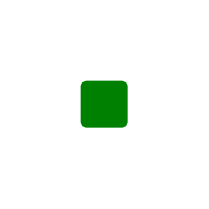}\\\large{$\times~2nw-q$}}%
}~~~~
\subcaptionbox{\label{fig:antibody1-clueseta-tiles} Tile polyominoes.}[1.5in]{%
  \shortstack{\includegraphics[scale=1.5]{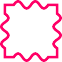}\\\large{$\times~2nm$}}%
}

\newdimen\figheight
\settoheight\figheight{\includegraphics[scale=1.2]{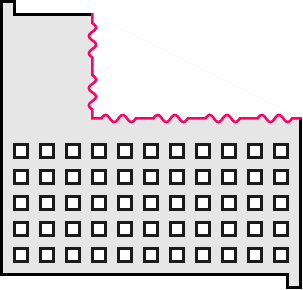}}%
\subcaptionbox{\label{fig:antibody1-clueseta-medium} The medium polyomino.}{%
  \includegraphics[scale=1.2]{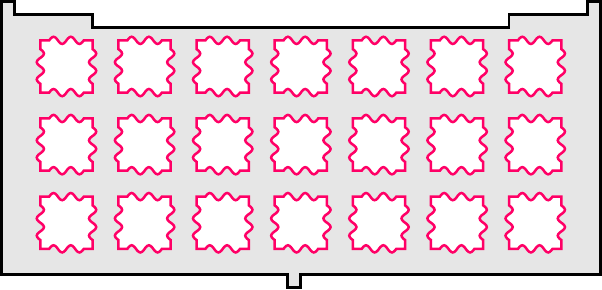}%
}~~~~
\subcaptionbox{\label{fig:antibody1-clueseta-frame-left} The left board-frame polyomino.}{%
  \includegraphics[scale=1.2]{figures/antibody-boardframe-left.pdf}%
}~
\subcaptionbox{\label{fig:antibody1-clueseta-frame-right} The right board-frame polyomino.}{%
  \vtop to \figheight{%
    \hbox{\includegraphics[scale=1.2]{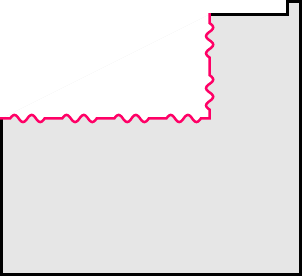}}%
    \vfil
  }%
}
\caption{The contents of clue set A in the proof of Theorem~\ref{thm:antibody-poly} (not to scale).  Pink wavy edges bear tabs and pockets encoding signed colors.}
\label{fig:antibody1-clueseta}
\end{figure}

\begin{figure}
\centering
\subcaptionbox{$2n$ stamps.}{%
  \includegraphics[scale=1.2,valign=c]{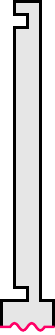}~
  \includegraphics[scale=1.2,valign=c]{figures/antibody-stamp1.pdf}~~
  \includegraphics[scale=1.2,valign=c]{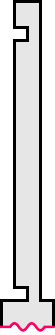}~
  \includegraphics[scale=1.2,valign=c]{figures/antibody-stamp2.pdf}~~
  \includegraphics[scale=1.2,valign=c]{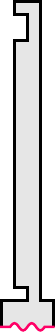}~
  \includegraphics[scale=1.2,valign=c]{figures/antibody-stamp3.pdf}~~
  \includegraphics[scale=1.2,valign=c]{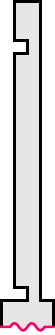}~
  \includegraphics[scale=1.2,valign=c]{figures/antibody-stamp4.pdf}%
  \label{fig:antibody1-cluesetb-stamps}%
}~~~~\subcaptionbox{The large polyomino.}{%
  \includegraphics[scale=1.2,valign=c]{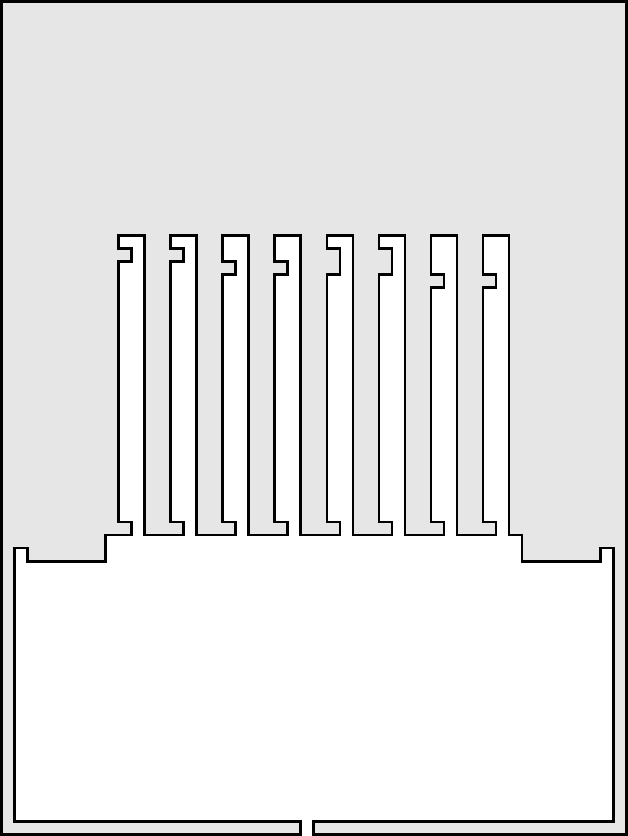}%
  \label{fig:antibody1-cluesetb-large}%
}
\caption{The contents of clue set B in the proof of Theorem~\ref{thm:antibody-poly} (not to scale).  Pink wavy edges bear tabs and pockets encoding signed colors.}
\label{fig:antibody1-cluesetb}
\end{figure}

\begin{figure}
\centering
\captionsetup[subfigure]{justification=raggedright}
\subcaptionbox{\label{fig:antibody1-trivial-packing} The trivial packing of the puzzle after eliminating both board-frame polyominoes.  The medium polyomino slots inside the large polyomino and the tiles fill the medium polyomino's holes.  The stamps fill in their matching handle slots in the large polyomino and the monominoes fill in the pockets and any extra space in the stamp accommodation zone.  This packing is always possible.}{%
  \includegraphics[scale=1.2]{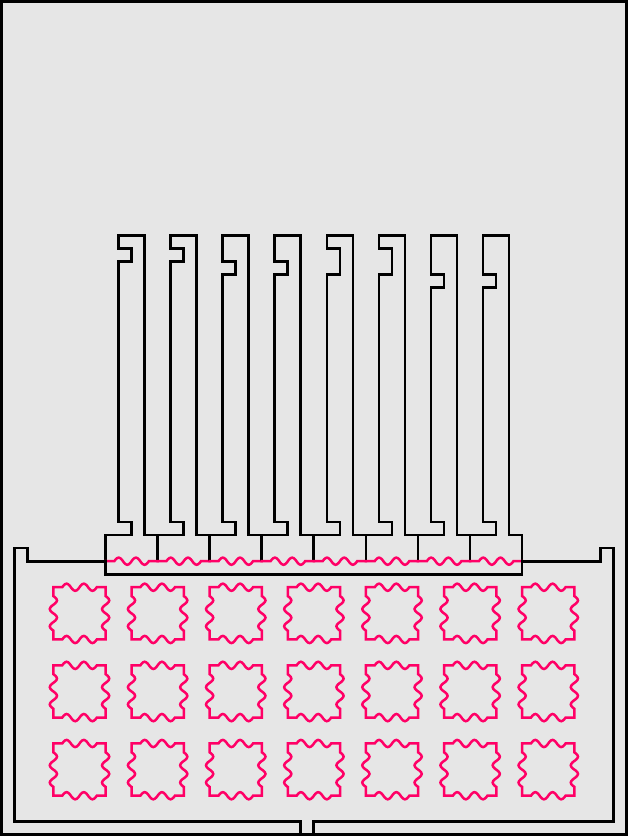}
}~~
\subcaptionbox{\label{fig:antibody1-intended-packing} The intended packing of the puzzle after eliminating the medium polyomino (not to scale).  The left and right board-frame polyominoes slot inside the large polyomino, and the monominoes fill the holes in the left board-frame polyomino.  The stamps fill in their matching handle slots in the large polyomino, leaving only the boundary-colored board for the simulated \abem{} instance.}{
  \includegraphics[scale=1.2]{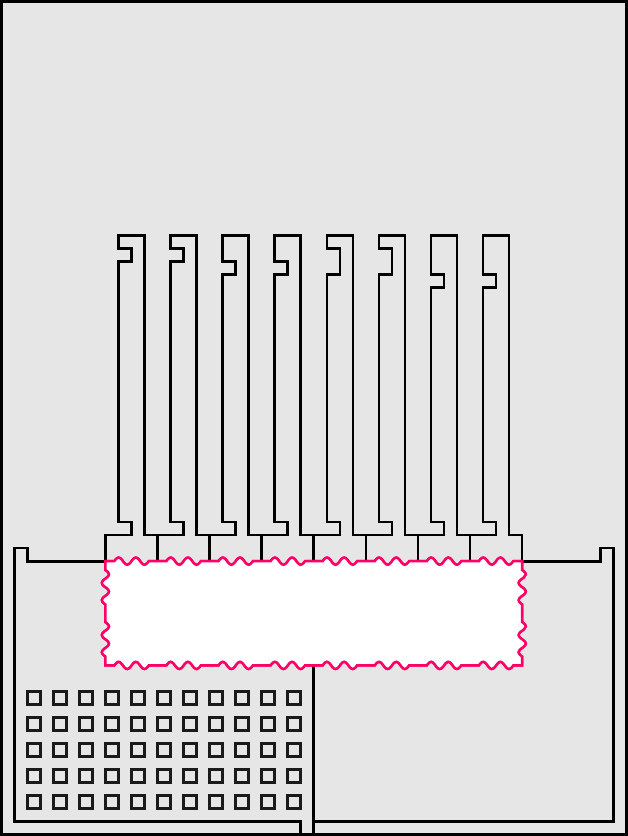}
}
\caption{The two possible polyomino packings.}
\label{fig:antibody1-packings}
\end{figure}

Clue set $A$ contains:
\begin{itemize}
\item Two antibodies.

\item $2nw - q$ monominoes, where $q$ is the total number of pockets minus the total number of tabs across the ``dies'' of the ``stamps'' in clue set $B$ (see below).  There are $2n$ stamps each having up to $\numcolors$ tabs or pockets, so the total number of monominoes is between $2nw - 2n\numcolors = 8n$ and $2nw + 2n\numcolors = 4nw - 8n$ inclusive.

\item A $w \times w$ square polyomino for each of the $2nm$ tiles in the \abem{} instance.  The edges of each polyomino are modified with tabs and pockets encoding the signed colors on the corresponding edges of the corresponding tile.  Call the upper-left corner of the $w \times w$ square the \emph{key pixel} of that polyomino (even if tabs caused other pixels to be further up or to the left).

\item A ``medium'' sized polyomino formed from a $2n(w+3) - 1 \times m(w+3) + 3$ rectangle polyomino.  Cut a hole out of this rectangle in the image of each tile polyomino, aligning the key pixel of each tile polyomino to a $2n \times m$ grid with upper-left point at the fourth row, second column of the rectangle and $w+3$ intervals between rows and columns.  Regardless of the pattern of tabs and pockets on each tile, this spacing ensures at least two rows of pixels above the top row of tile-shaped holes, at least one row on each other side, and at least one row between adjacent holes.  Then add pixels\xxx{how many?} above the upper-leftmost and upper-rightmost pixel of the rectangle (the \emph{horns}) and below the middle-bottommost pixel of the rectangle (the \emph{tail}).  Finally, cut $2nw$ pixels out of the top row of the rectangle starting from the third pixel; this cutout is the \emph{stamp accommodation zone}.

\item Two \emph{board-frame} polyominoes.  Again, starting from a $2n(w+3) - 1 \times m(w+3) + 3$ rectangle polyomino, add horns and tail pixels in the same locations.  Then cut out a $2nw \times mw$ rectangle whose upper-left pixel is the third pixel in the top row of the rectangle.  The left, right and bottom edges of this cutout are modified with tabs and pockets encoding the signed colors on the corresponding sides of the boundary-colored board in the \abem{} instance.  Split the polyomino vertically along the column of edges immediately to the right of the tail pixel.

Finally, for each monomino in this clue set, cut a pixel out of the left
board-frame polyomino, starting from the second-bottommost pixel in the
second column, continuing across every other column, then continuing
with the fourth-bottommost pixel in the second column, and so on.  The
left board-frame polyomino has width $nw+3n$, we cut pixels out of every
other column, and we do not cut holes in its left or right columns, so
we cut pixels out of $\frac{nw + 3n - 2}{2}$ columns.  Below the
$mw$-tall cutout and allowing two rows to ensure cut pixels do not join
with pockets encoding signed colors along the edges of the cutout, we
can cut pixels out of $\frac{3w - 1}{2}$ rows (or $\frac{3w}{2}$,
depending on parity).  This allows up to $(\frac{nw + 3n -2}{2})(\frac{3w-1}{2}) = \frac{n(w-4)^2 + 2w(nw-3) + 13n + 2}{4} + 4nw-8n$ pixels to be cut out, but there are at most $4nw - 8n$ monominoes, so we can always cut enough pixels without interfering with any other cuts. %\xxx{or at least WolframAlpha says it's true for $n \ge 1, w \ge 5$}
\end{itemize}

Clue set $B$ contains:
\begin{itemize}
\item A \emph{stamp} polyomino for each of the $2n$ edge segments of the top edge of the boundary-colored board.  Each stamp is composed of a $w \times 2$ rectangle modified to encode the signed color on the corresponding edge segment (called the \emph{die}), a pixel centered above that rectangle, and a $2 \times h$ rectangular \emph{handle} whose bottom-right pixel is immediately above that pixel, where $h = \max(m(w+3) + 7, n)$.  Stamps corresponding to 1-indexed edge segments $2i$ and $2i+1$ have pockets encoding $i$ in binary cut into the left edge of their handle, starting from the second-to-top row of the handle. %always possible because of the max(n) -- don't need it otherwise

\item A ``large'' sized polyomino built from a $2n(w+3)+1 \times t$ rectangular polyomino, where $t$ is the total area of all other polyominoes so far defined.  Modify this polyomino by cutting out the middle pixel of the bottom row, the $2n(w+3) - 1 \times m(w+3) + 3$ horizontally-centered rectangle immediately above that removed pixel, and the pixels above the upper-left and upper-right removed pixels.  (That is, cut out space for the medium polyomino, including the horns and tail but not including the stamp accommodation zone.)  Then cut out the image of each stamp in the order of their corresponding edge segments in the \abem{} instance, aligning the leftmost-bottom pixel of the first stamp's die two pixels to the right of the upper-left removed pixel and aligning successive dies immediately adjacent to one another.
\end{itemize}

\paragraph{Puzzle.} The Witness puzzle is a $2n(w+3)+1 \times t$ rectangle.  The start circle and end cap are at the middle two vertices of the bottom row of vertices.

\paragraph{Placement of $A$ clues.} We place a monomino from clue set $A$ in the cell having the start circle and end cap as vertices, then place an antibody above that monomino, surrounded by a monomino in each of its other three neighbors.  We then place the other antibody, surrounded by monominoes in its neighboring cells, three cells above the first antibody.  (See Figure~\ref{antibody-monomino-surround}.)  It is always possible to surround the antibodies in this way because there are at least $8n$ monominoes.  We place the remaining clues from clue set $A$ inside the $2n(w+3) - 1 \times m(w+3) + 3$ rectangle one row above the bottom of the puzzle; this is always possible because $|A| \le 4nw-8n+2nm+5$.

\begin{figure}
\centering
\includegraphics[scale=0.5]{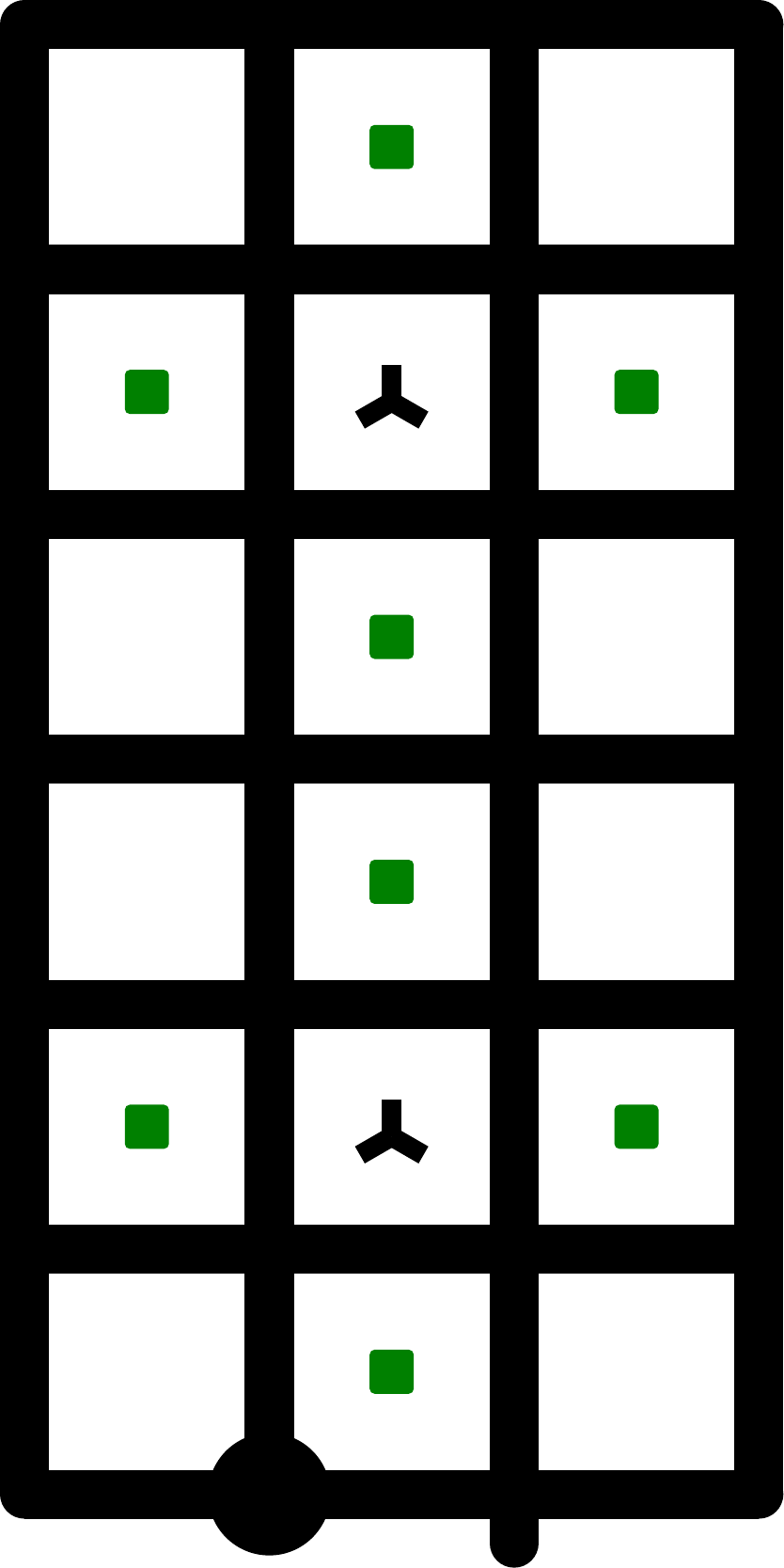}
\caption{Because both antibodies are surrounded by monominoes, any region containing an antibody also contains at least one monomino.}
\label{antibody-monomino-surround}
\end{figure}

\paragraph{Placement of $B$ clues.} We place the large polyomino clue in the upper-left cell of the board and the stamp clues in the $2n$ cells to its right.

\paragraph{Argument.} In any solution to the resulting puzzle, the large polyomino is not eliminated.  If it were, it must be in the same region as an antibody.  Because each antibody is surrounded by monomino clues, the number of polyomino clues in this region is strictly greater than the number of antibodies, so the region must be packed by the non-eliminated polyomino clues.  The nearest (upper) antibody is $t-4$ columns and $nw+3n$ rows away from the large polyomino clue, so this region has area at least $t$.  Recall that $t$ is the total area of all polyomino clues except the large polyomino.  If the large polyomino is eliminated, there is no way to pack this region, even if all other polyomino clues are used.

The large polyomino is as wide and as tall as the entire puzzle, so it has a unique placement.  The large polyomino intersects its bounding box everywhere except one unit-length edge aligned with the start vertex and end cap, so any solution path can only touch the boundary at the start and end.  Thus, the solution path divides the puzzle into at most two regions (an inside and an outside).

Suppose the solution path places the entire puzzle into a single region; that is, suppose the solution path proceeds (in either direction) from the start vertex to the end cap without leaving the boundary.  Then by the assumption that the first player has a losing strategy in the input \abem{} instance, we can pack the region while eliminating only one clue.  The large polyomino's placement is fixed.  We eliminate the medium polyomino, place the two board-frame polyominoes inside the large polyomino, and place the monominoes in the pixels cut out of the left board-frame polyomino.  It remains to place the stamps and tiles.  By the assumption, there is a losing set of top-boundary swaps; we swap the corresponding pairs of stamps when placing them into the cutouts in the large polyomino, and then place the tiles in the remaining uncovered area bordered by the board-frame polyominoes and stamp dies.  Because we satisfied all non-antibody constraints after eliminating only one clue, the unused antibody is unsatisfied, so any path resulting in a single region is not a solution to the puzzle.  Thus, there are exactly two regions.

The cells containing the stamp clues are covered by the large polyomino, so any solution places the stamps in the same region as the large polyomino.  The handles of the stamps are taller than the cutout in the bottom-middle of the large polyomino, so they must instead be placed in the stamp-shaped cutouts in the large polyomino.  The pockets cut into the left edges of the handles ensure that stamps can only swap places corresponding to top-boundary swaps in the \abem{} instance.

All clues in set $A$ are in the other region.  The monomino clue in the
cell having both the start circle and end cap as vertices cannot be in
the same region as the large polyomino (else the path could not divide
the puzzle into two regions).  Because each antibody is surrounded by
monomino clues, the number of polyomino clues in this region is strictly
greater than the number of antibodies, so the region must be packed by
the non-eliminated polyomino clues.  When both antibodies are used to
eliminate clues, they must eliminate both board-frame polyominoes, and
when only one is used, it must eliminate the medium polyomino; any other
elimination 
%(including eliminating one antibody with the other under the sequential
%interpretation) 
leaves polyomino clues with too much or too little area to pack the area of the puzzle not yet covered by the large polyomino or the stamps.  Thus, either the medium polyomino or both board-frame polyominoes will not be eliminated.  The medium polyomino and board-frame polyominoes have unique placements within the large polyomino determined by the horns and tail.  The intersection of the outlines of these placements covers all the $A$ clues, so they are all in the same other region.

By this division of the clues into regions, any solution path traces the inner boundary of the large polyomino and the dies of the stamps (possibly after swapping some pairs).  It remains to show that the path is a solution exactly when the implied set of top-boundary swaps is a winning strategy in the \abem{} instance.

When using both antibodies to eliminate the board-frame polyominoes, the remaining polyominoes always pack their region (see Figure~\ref{fig:antibody1-trivial-packing}).  The medium polyomino's placement is fixed by the horns and tail; the stamp accommodation zone ensures this placement is legal regardless of the pattern of tabs on the dies of the stamps.  The tile polyominoes fit directly into the cutouts in the medium polyomino and there are exactly enough monominoes to fill in the uncovered area in the stamp accommodation zone and the pockets of the dies.

The path is a solution only if both antibodies are necessary.  When using one antibody to eliminate the medium polyomino, the board-frame polyominoes' position is forced by the horns and tail.  The monominoes are the only way to fill the single-pixel holes in the left board-frame polyomino and there are exactly enough monominoes to do so.  Then the dies of the stamps and the edges of the rectangular cutout in the board-frame polyominoes models the boundary-colored board of the input \abem{} instance (see Figure~\ref{fig:antibody1-intended-packing}).  The tile polyominoes cannot pack this area, necessitating the second antibody and making the path a solution, exactly when the set of top-boundary swaps is a winning strategy in the \abem{} instance.
\end{proof}

\both{
\begin{theorem}
  \label{thm:antibody-antipoly}
  It is $\Sigma_2$-complete to solve Witness puzzles containing
  one antibody, polyominoes and antipolyominoes.
\end{theorem}
}

\later{
\begin{proof}
As in the previous proof, we reduce from \abem{} (though we do not need the first player to have a losing strategy), and the reduction is similar.  The primary difference is that the medium polyomino is also the (singular) board-frame polyomino.

\paragraph{Clue sets.} As before, we have two clue sets.  Clue set $A$ is shown in Figure~\ref{fig:antibody2-clueseta}.  Clue set $B$ is nearly the same as it is in the previous proof (see Figure~\ref{fig:antibody1-cluesetb}), only with the large polyomino being slightly wider (not visibly different at the scale of the figure).

Clue set $A$ contains:

\begin{figure}
\centering
\subcaptionbox{\label{fig:antibody2-clueseta-monominoes} One antibody.}[1.2in]{%
  \includegraphics{figures/antibody.pdf}%
}~
\subcaptionbox{\label{fig:antibody2-clueseta-tiles} Tile polyominoes.}[1.5in]{%
  \shortstack{\includegraphics[scale=1.5]{figures/antibody-tile.pdf}\\\large{$\times~2nm$}}%
}~
\newdimen\figheight
\settoheight\figheight{\includegraphics[scale=1.2]{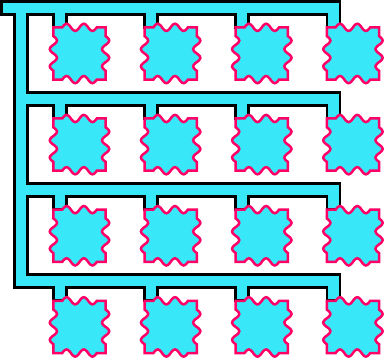}}%
\subcaptionbox{\label{fig:antibody2-antikit} The antikit antipolyomino.}{%
  \includegraphics[scale=1.2]{figures/antibody-antikit.pdf}%
}~~~~
\subcaptionbox{\label{fig:antibody2-sprue} The sprue polyomino.}{%
  \vtop to \figheight{%
    \hbox{\includegraphics[scale=1.2]{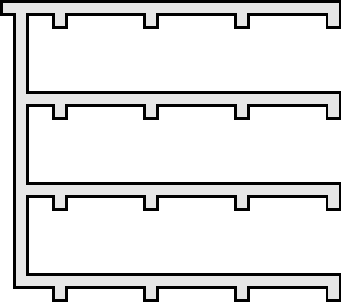}}%
    \vfil
  }%
}

\subcaptionbox{\label{fig:antibody2-clueseta-medium} The medium polyomino.  The attached kit polyomino and sprue cutout are to scale with each other, but not with the body of the medium polyomino.}{%
  \includegraphics[scale=1.2]{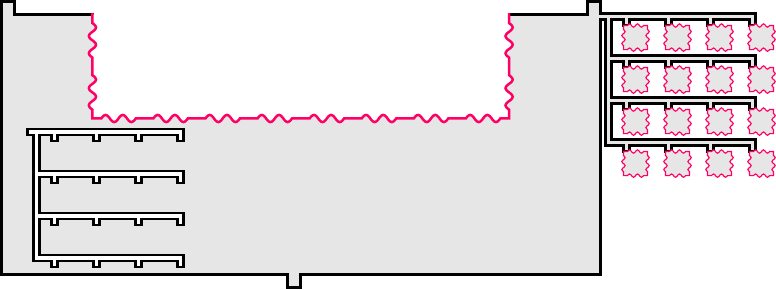}%
}
\caption{The contents of clue set A in the proof of Theorem~\ref{thm:antibody-antipoly} (not to scale).  Pink wavy edges bear tabs and pockets encoding signed colors.}
\label{fig:antibody2-clueseta}
\end{figure}

\begin{figure}
\centering
\captionsetup[subfigure]{justification=raggedright}
\subcaptionbox{\label{fig:antibody2-unconstrained-packing} The trivial packing of the puzzle after eliminating the medium polyomino.  The sprue and tiles annihilate with the antikit antipolyomino.  The bottom region is unconstrained because it does not contain any surviving polyomino clues.  After the stamps are placed in their handle slots, the path simply traces the exterior of the surviving polyominoes.  This packing is always possible.}{%
  \includegraphics[scale=1.2]{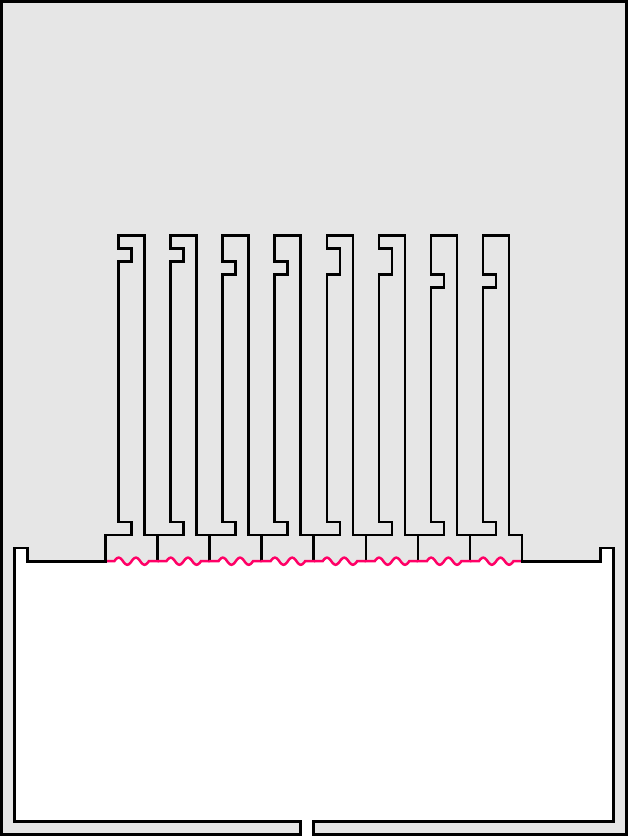}
}~~
\subcaptionbox{\label{fig:antibody2-noelimination-packing} When the antibody is not used, the antikit antipolyomino annihilates part of the medium polyomino, allowing it to fit inside the large polyomino.  The sprue polyomino fills the cutout in the medium polyomino.  The stamps fill in their matching handle slots in the large polyomino, leaving only the boundary-colored board for the simulated \abem{} instance.}{
  \includegraphics[scale=1.2]{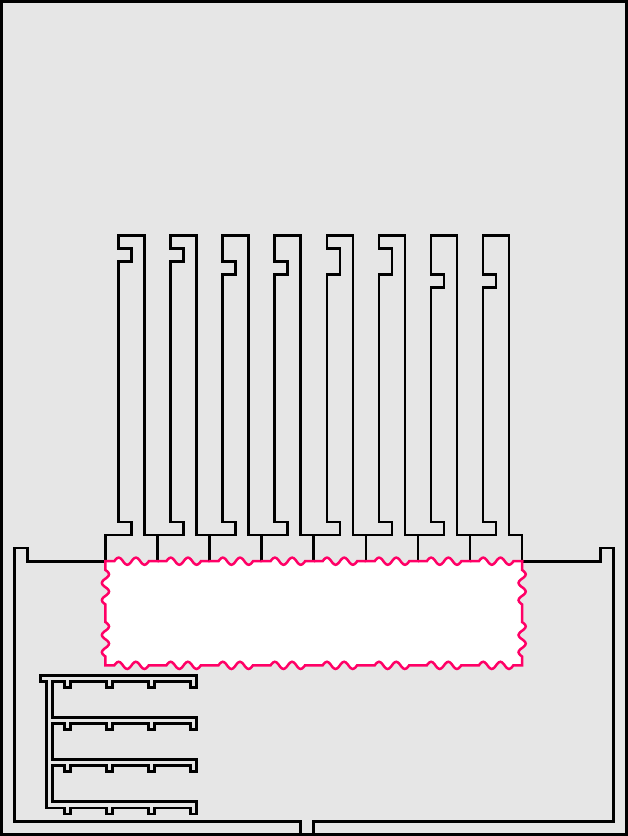}
}
\caption{The two possible polyomino packings.}
\label{fig:antibody2-packings}
\end{figure}

\begin{itemize}
\item One antibody.

\item A $w \times w$ square polyomino for each of the $2nm$ tiles in the \abem{} instance.  The edges of each polyomino are modified with tabs and pockets encoding the signed colors on the corresponding edges of the corresponding tile.  Call the upper-left corner of the $w \times w$ square the \emph{key pixel} of that polyomino (even if tabs caused other pixels to be further up or to the left).

\item An \emph{antikit} antipolyomino, which we define in terms of a kit polyomino (which is not itself a clue).  Start with $m$ copies of a $2(n-1)(w+5)+3 \times 1$ rectangle, vertically aligned and spaced at $w+5$ intervals.  Connect these rectangles with $1 \times w+4$ rectangles in the left column.  Align the key pixels of the tile polyominoes to a $2n \times m$ grid spaced at $w+5$ intervals, then place this grid such that the key pixel of the upper-left polyomino is three pixels below and to the right of the upper-left pixel of the polyomino being built.  (A tile polyomino having tabs on each edge is $w+2$ wide and tall, but there is $w+4$ width and height in the grid, so this placement is always possible.)  Connect each tile polyomino to the rectangle above it by adding two pixels immediately above the key pixel.  Finally, add one pixel immediately to the left of the upper-left pixel.  The antikit antipolyomino is just the antipolyomino with an antipixel for each pixel in the kit polyomino.  (See Figure~\ref{fig:antibody2-antikit}.)

\item A \emph{sprue} polyomino shaped like the sprue\footnote{In molding, the sprue is the waste material that cools in the channels through which the material is poured or injected into the mold.  Specifically, in plastic model kits, the sprue is the frame from which the model pieces are snapped out.}, the non-tile-polyomino area of the kit polyomino.

\item A medium polyomino built from a $2n(w+5) + 5 \times m(2w+5) + 2$ rectangle polyomino.  Add a pixel directly below the center pixel of the rectangle (the tail).  Cut out a $2nw \times mw$ rectangle whose upper-left pixel is the third pixel in the top row of the rectangle.  The left, right and bottom edges of this cutout are modified with tabs and pockets encoding the signed colors on the corresponding sides of the boundary-colored board in the \abem{} instance.  Then add the kit polyomino, placing the leftmost pixel of the kit to the right of the upper-rightmost pixel of the rectangle.  Then cut out the shape of the sprue polyomino, aligning the leftmost pixel in the second column and the bottommost pixels in the second-to-bottom row of the rectangle. \xxx{If not for the tail, we could put the sprue cutout in the bottom row, making the medium polyomino hole-free.}  This is always possible because the sprue is bounded by a $2(n-1)(w+5)+4 \times (m-1)(w+5) + 2$ rectangle, so the sprue cutout fits within the medium polyomino even after the board cutout.
\end{itemize}

Clue set $B$ contains:

\begin{itemize}
\item A \emph{stamp} polyomino for each of the $2n$ edge segments of the top edge of the boundary-colored board.  The stamps are exactly as specified in the previous proof.

\item A large polyomino built from a $2n(w+5)+7 \times t$ rectangular polyomino, where $t$ is the total area of all other polyominoes so far defined. %\xxx{we don't use this in this proof, so we could give a tighter size}
Modify this polyomino by cutting out the middle pixel of the bottom row and the $2n(w+5) + 5 \times m(2w+5) + 2$ horizontally-centered rectangle immediately above the removed pixel.  Then cut out the image of each stamp in the order of their corresponding edge segments in the \abem{} instance, aligning the leftmost-bottom pixel of the first stamp's die one pixel above and two pixels to the right of the upper-left pixel of the rectangular cutout and aligning successive dies immediately adjacent to one another.  (This is the same procedure as in the previous proof, but with slightly different dimensions.)
\end{itemize}

\paragraph{Puzzle.} The Witness puzzle is a $2n(w+5)+7 \times t$ rectangle.  The start circle and end cap are at the middle two vertices of the bottom row of vertices.

\paragraph{Placement of $A$ clues.} We place the antibody in the cell having the start circle and end cap as vertices, then place the medium polyomino above the antibody.  We place the remaining $2nm+2$ clues from clue set $A$ contiguously, directly to the left and right of the medium polyomino, half on each side.

\paragraph{Placement of $B$ clues.} We place the large polyomino clue in the upper-left cell of the board and the stamp clues in the $2n$ cells to its right.

\paragraph{Argument.} The medium polyomino clue is wider than the puzzle, so any solution must eliminate it.  There is only one antibody, so the large polyomino is not eliminated in any solution.

The large polyomino is as wide and as tall as the entire puzzle, so it has a unique placement.  The large polyomino intersects its bounding box everywhere except one unit-length edge aligned with the start vertex and end cap, so any solution path can only touch the boundary at the start and end.  Thus, the solution path divides the puzzle into at most two regions (an inside and an outside).

Suppose the solution path places the entire puzzle into a single region; that is, suppose the solution path proceeds (in either direction) from the start vertex to the end cap without leaving the boundary.  After eliminating the medium polyomino, the remaining polyominoes and antipolyomino have less net area than the puzzle, so any path resulting in a single region is not a solution to the puzzle.  Thus, there are exactly two regions.

The cells containing the stamp clues are covered by the large polyomino, so any solution path places the stamps in the same region as the large polyomino.  The handles of the stamps are taller than the cutout in the bottom-middle of the large polyomino, so they must instead be placed in the stamp-shaped cutouts in the large polyomino.  The pockets cut into the left edges of the handles ensure that stamps can only swap places corresponding to top-boundary swaps in the \abem{} instance.

All clues in set $A$ are in the other region.  The antibody clue is in the cell having both the start circle and end cap is vertices, so it cannot be in the same region as the large polyomino (else the path could not divide the puzzle into two regions).  The medium polyomino is eliminated, so it must be in the same region as the antibody.  The remaining $A$ clues (the sprue polyomino, tile polyominoes and antikit antipolyomino) have net area 0.  If a polyomino from this group is in the first region, the other region's clues require negative area to be satisfied; if the antikit is in the first region, the other region's clues do not have enough area to pack the area of the puzzle not already covered by the large polyomino or the stamps.  Thus, the remaining clues are in the other region (and, because there are no other non-eliminated polyomino clues in the region, are satisfied).

By this division of the clues into regions, any solution path traces the inner boundary of the large polyomino and the dies of the stamps (possibly after swapping some pairs).  It remains to show that the path is a solution exactly when the implied set of top-boundary swaps is a winning strategy in the \abem{} instance.

By the above arguments, any solution path of this form satisfies all non-antibody clues when the antibody is used to eliminate the medium polyomino (see Figure~\ref{fig:antibody2-unconstrained-packing}), but the antibody is only satisfied if it is necessary.  If the antibody is not used to eliminate the medium polyomino, the antikit antipolyomino must annihilate the portion of the medium polyomino shaped like a kit polyomino.  The rest of the medium polyomino is exactly the right size to fit inside the large polyomino and the sprue polyomino fills the sprue cutout in the medium polyomino.  Then the dies of the stamps and the edges of the rectangular cutout in the medium polyomino models the boundary-colored board of the input \abem{} instance (see Figure~\ref{fig:antibody2-noelimination-packing}).  The tile polyominoes cannot pack this area, necessitating the antibody and making the path a solution, exactly when the set of top-boundary swaps is a winning strategy in the \abem{} instance.
\end{proof}
}

\ifabstract
\begin{proofsketch}
Both Theorems~\ref{thm:antibody-poly} and~\ref{thm:antibody-antipoly} follow the same sketch. We reduce from a new $\Sigma_2$-complete problem, \abem{}, in which the first player may swap some edge segments of a rectangular signed edge-matching puzzle, then the second player tries to solve the resulting puzzle.  In the Witness puzzles produced by our reduction, it's trivial to satisfy the non-antibody clues when all antibodies are used, but the antibodies themselves are only satisfied if they are necessary.  The polyomino pieces are designed to simulate the \abem{} instance, forming a playfield whose inner surface encodes the signed colors of the edge-matching puzzle.  The top of the playfield is defined by the solution path (constrained to only swap segments as allowed in the \abem{} instance), so the player solving the Witness puzzle takes the place of the first player in the \abem{} instance.  The Witness puzzle verifier then attempts to solve the edge-matching puzzle by packing polyomino clues into the playfield region.  This is impossible, necessitating all the antibodies and making the path a solution, exactly when the edge segment swaps are a winning strategy in the \abem{} instance.
\end{proofsketch}
\fi

By Theorem~\ref{one antibody in NP}, Theorem~\ref{thm:antibody-poly} and Theorem~\ref{thm:antibody-antipoly} are tight.

%\section{Parameterized Puzzles}
%
%\subsection{Thin Puzzles}
%
%\begin{theorem}
%  Witness puzzles with broken edges, hexagons, triangles, squares, and stars
%  can be solved in polynomial time for a constant number of rows in the grid
%  (membership in XP).  If the number of squares and stars is $O(1)$,
%  then we get an FPT algorithm.
%\end{theorem}
%\begin{proof}
%\xxx{old text}
%For $k \times n$, we can achieve $2^{O(k)}n$ running time. Probably also works with squares (but if $c$ colors, get a factor of $c^{O(k)}$).
%
%Subproblems defined by suffix of columns to right of a vertical cut ($n$), plus where path crosses the cut ($2$k), plus partial pairing of these crossings on the left side (laminar matching, so $2^O(k)$), plus how many stars (0, 1, or 2) are in each face defined by this pairing and the connection of the up to two ends to the boundary ($3^k$). So total number of subproblems is $2^{O(k)}n$.
%
%The algorithm proceeds left to right, considering all possible things that can happen in the column (in each pixel, one of four turns or one of two straights or blank), then checking compatibility, no violation of maze constraints, updating faces and star counts, etc.
%
%If there are multiple start/end candidate locations, we will need an ``end counter'' in the subproblem spec that guarantees exactly two ends in the path.
%\end{proof}
%
%\xxx{Open: polyomino clues}
%
%\subsection{Short Solution Paths}
%
%\xxx{but this is meta}

\section{Nonclue Constraints} \label{sec:nonclue-constraints}

In this section, we analyze Witness puzzle constraints that are not represented by clues on vertices, edges, or cells. This is not an exhaustive list of elements included in puzzles in The Witness. In most cases, these elements do not appear to change the computational complexity of the problem. However, containment in L for broken edges is no longer obvious in many cases. In addition, it is unclear how to modify some of our reductions to obey the symmetry constraint, though we believe the problem still remains hard for all prior cases.

\subsection{Visual Obstruction}

Some panels are placed in the 3D world in such a way that obstacles obstruct the player's vision of some of the panel.  Because the player can only draw paths along edges they can see, the obstructions constrain the panel's solution set, similar to broken edges. Indeed, we show how to reduce visual obstructions to broken-edge puzzles:

\begin{theorem} \label{thm:visual obstruction}
Witness puzzles constrained only by visual obstruction are in P.
\end{theorem}
\begin{proof}
We can determine the combinatorially different viewing positions (where different sets of vertices and/or edges are obstructed) in $O(n^3)$ time by building the arrangement of the planes extending the obstacle faces~\cite{DBLP:journals/siamcomp/EdelsbrunnerOS86}.  For each viewing position reachable in the environment, we can reduce the visual obstruction puzzle to an equivalent puzzle using only broken edges: starting from a puzzle with the same size, start circles, and end caps, break each obstructed edge and break all edges incident to each obstructed vertex.\footnote{Some puzzles in The Witness have end caps halfway along edges instead of at vertices.  If only one half-edge or adjacent vertex is obstructed for such an end cap, replace the edge with an end cap at the other vertex instead of breaking edges.}  This set of broken edges exactly captures the obstruction's restrictions on path drawing.  Puzzles containing only broken edges are in L (Observation~\ref{thm:broken-edges}), and there are only polynomially many viewing positions, so we can solve from all positions in polynomial time.  The visual obstruction puzzle is a \textsc{yes} instance exactly when at least one of the broken-edge puzzles is a \textsc{yes} instance.
\end{proof}

Some puzzles constrained only by visual obstruction from The Witness can obviously only be viewed from a single position in the game world, so building the plane arrangement is unnecessary.  Such puzzles are in L if there exists an L algorithm to decide whether the projection onto the puzzle panel of some face of the obstacle overlaps a particular edge or vertex.

More generally, for any puzzle type allowing broken edges from P or larger
complexity classes, adding visual obstructions results in puzzles in the same
complexity class, as we can simply try all of the polynomially many
viewpoints by the proof of Theorem~\ref{thm:visual obstruction}.

\subsection{Symmetry}

In symmetry puzzles, there are an even number of start circles (usually exactly two), and drawing a path from one of them also causes one reflectionally symmetric path (reflected across one axis, or reflected across both axes which is equivalent to $180^\circ$ rotation) to be drawn from another start circle.  The paths must not intersect and both paths must reach an end cap.  The paths induce regions that must satisfy all clues just as in no-symmetry puzzles.  In some symmetry puzzles, the paths are colored yellow and blue, as are some vertex or edge hexagons; each path must visit the hexagons of its color (and any hexagons of other colors must be visited by either path).  Path color is not relevant for other clue types.

We now discuss the implications of adding symmetry to each of the clue sets we have considered so far in this paper.

\label{Containment}
\paragraph{Containment.}  Our proofs for containment in NP or in $\Sigma_2$ are based on algorithms for verifying certificates.  Our verifiers can be straightforwardly modified to check that the paths in the certificates are indeed symmetric, do not intersect, and visit all hexagons of their color (if any).  From then on it is irrelevant that there are two paths instead of one, so the containments still hold when the symmetry constraint is added.

Observation~\ref{thm:broken-edges} establishes that puzzles with only broken edges, even with multiple start circles and end caps, are in L.  For each start circle, we can determine the symmetric start circle in L, and the connectivity algorithm can follow the progress of both paths (if necessary, by running another instance of the algorithm in lockstep) to ensure paths are extended only if both edges are not broken.  However, it is not obvious how to ensure that the paths do not intersect each other (as there is not enough space to store them).  Thus, we have only that Witness puzzles with only broken edges and symmetry are in P.

Theorem~\ref{thm:monominoes-dp} and
Theorem~\ref{MonominoesAndBrokenEdgesTheorem}
establish polynomial-time algorithms for edge hexagons on the boundary of
a puzzle and monominoes, respectively.  While we suspect that these algorithms
can be extended to solve symmetric puzzles of these types, we leave these
as open problems.

\paragraph{Hardness.}  To prove that adding symmetry does not make puzzles easier, we need to modify the constructions from our proofs to tolerate the addition of a symmetric copy.  For most of our constructions, this is straightforward:

\begin{itemize}
\item \textbf{Vertex hexagons and broken edges} (Observation~\ref{thm:broken edges+vertex hexagons}): Under the symmetry constraint, Hamiltonian path is no longer a strict special case.  We can add a symmetric copy of the puzzle if we separate it from the original by breaking all the edges in a row or column between the copies.  This wall of broken edges ensures the two paths cannot interact, making Hamiltonian path a special case again.
\item \textbf{Edge hexagons} (Theorem~\ref{thm:edge hexagons}): We can add a symmetric copy in any position with a row and/or column of empty cells between them to ensure that the paths do not touch.  The rightmost bottommost chamber, which contains the start circle and end cap, is on the bottom boundary of the board, so the symmetric copy's start circle and end cap are also on the boundary of the board (possibly the top boundary).  Then each path is confined to its copy by the existing argument.
\item \textbf{Squares} (Theorem~\ref{thm:squares 2 color}): If we consistently use the same colors inside and outside the path, we can add a symmetric copy of the puzzle in any position.  The boundary of red squares ensures that each path is confined to its copy of the puzzle.
\item \textbf{Rotatable dominoes} (Theorem~\ref{rotating dominoes}): We can add a symmetric copy of the puzzle in any position.  The distance between the copies is twice the existing distance from the Steiner tree area to the boundary, so the copies cannot interact.  Then hardness preservation for monominoes and antimonominoes (Theorem~\ref{anti-monominoes}) follows in the same way as the no-symmetry construction.
\item \textbf{Nonrotatable dominoes} (Theorem~\ref{vertical dominoes}): We can add a symmetric copy of the puzzle without affecting the proof if we add padding rows and/or columns to ensure that there are at least five empty cells between the closest domino clues.  This preserves Properties~\ref{no vertical domino clues}, \ref{boundary domino clues}, and~\ref{domino start end} (rotated and/or reflected in the symmetric copy), so solving the puzzle still requires solving the domino tiling problem.  Then because there are at least five empty cells between clues, the symmetric paths cannot cross from their tree of domino clues to the other.
\item \textbf{1-triangle clues} (Theorem~\ref{thm:triangles1}): We can add a symmetric copy above the original construction.  The original copy's path is forced to proceed from the start circle entirely across the puzzle, forming a barrier preventing the two paths from interfering.
\item \textbf{3-triangle clues} (Theorem~\ref{thm:triangles3}): We can add a symmetric copy of the puzzle in any position because each path is confined to the chambers of its copy.
\end{itemize}

But for others of our constructions, adapting to a symmetry puzzle is not so simple, and we leave solutions as open problems:

\begin{itemize}
\item \textbf{Stars} (Theorem~\ref{thm:stars}): While Figure~\ref{star scheme} is not to scale, it does accurately depict that there is plentiful blank space outside the simulated squares instance; note that the number of cells occupied by the red and blue stars is equal inside and outside the simulated instance, but the outside stars are arranged in a line instead of a rectangle.  Thus, there is space for the symmetric path to disrupt the argument about what colors of stars must be in which region.  It may be possible to fix this by packing the outside stars more densely, so that the simulated squares instance occupies the full width or height of the puzzle, blocking the symmetric path from interfering.
\item \textbf{2-triangle clues} (Theorem~\ref{thm:triangles2}): Our construction relies on forcing paths throughout the entire puzzle, particularly wave propagation from the four corners of the puzzle.  To add a symmetric copy, we need some way to force simulated corners exactly halfway across the puzzle.
\item \textbf{Antibodies} (Theorems~\ref{thm:antibody-poly} and \ref{thm:antibody-antipoly}):  Obviously, the symmetric copy includes a copy of the antibody clue(s), so adding symmetry weakens the theorems to require four and two antibodies for hardness (respectively).  Substantively, both proofs rely on the large polyomino not being eliminated and having a unique placement.  With a symmetric copy, this is no longer true: one large polyomino could be eliminated and the other placed in the middle of the board instead of up against the boundary, or the two large polyominoes could exchange the intended placements.  Then arguments that polyominoes must be in the same region because the placement of the large polyomino overlaps the cells containing their clues no longer hold.
\end{itemize}

\subsection{Intersection}

In intersection puzzles, multiple puzzles must be solved simultaneously by the same solution path.  That is, the set of solutions to an intersection puzzle is the intersection of the sets of solutions to its component puzzles.  The Witness contains one intersection puzzle in the mountain consisting of six panels.  Drawing a path on any of these panels causes the path to be drawn on all of them, and the game only accepts a path if it satisfies all the panels.  (Initially only the first panel is energized, with each additional panel activating when all the previous panels are solved, but this is not relevant for the complexity analysis.)

Intersection puzzles are typically in the same complexity class as the hardest of their component puzzles.
Solving an intersection puzzle containing only broken edges is the same as solving a single broken-edges puzzle (which is in L, Observation~\ref{thm:broken-edges}) where an edge is broken if it is broken in any component puzzle; the path existence algorithm simply checks that the edge is present in every component puzzle.
If the hardest component puzzle is in NP or $\Sigma_2$, membership for the intersection puzzle follows by combining the certificate verification algorithms (as in Section~\ref{Containment}).
%immediately from these classes' closure under intersection.
An interesting challenge is monomino puzzles
(Theorem~\ref{MonominoesAndBrokenEdgesTheorem});
it is unclear whether intersection puzzles of monominoes can still
be solved in polynomial time.

\subsection{Recursion}

In recursive puzzles, one or more panels are embedded in the cells of a larger panel.  The inner panels contain non-antibody clues and at least one antibody.  The inner panels are solved independently and any surviving clue(s) from each inner panel ``bubble up'' to the cell in the outer puzzle containing that panel.  If a surviving clue is a rotatable polyomino clue, it bubbles up as a nonrotatable clue oriented as it was used in the inner panel.\footnote{See \url{https://gaming.stackexchange.com/q/253987}.  As this puzzle (on the bottom floor of the mountain) is the sole recursive puzzle in The Witness, another consistent interpretation is that the solution paths of the inner puzzles are what matters, not the surviving clues, but that definition does not generalize to nonpolyomino clue types.}  Once all of the inner panels have been solved, the outer panel is solved just like a normal puzzle with the surviving clues in those cells (along with any other clues in the outer puzzle).
As antibodies are only satisfied if they are necessary, an inner panel always provides the same number of clues to its outer panel regardless of how it is solved.

%\xxx{This paragraph might be too much detail. -jaysonl agrees, this can probably be cut} The Witness contains one recursive puzzle, the first puzzle on the bottom floor of the mountain.  There are four inner panels, each containing two polyomino clues and one antibody.  \xxx{The outer panel is also a symmetry puzzle, but it doesn't matter.}  As this is the only recursive puzzle, there is no precedent for inner panels that provide more than one surviving clue to the outer panel.  Some clue pairs can never be satisfied (e.g., squares of multiple colors, or more than two stars of the same color), but others are fine (e.g., polyominoes).  Nor is there precedent for inner panels not containing an antibody.  Antibodyless inner panels with only one clue are superfluous (no different from just having that clue in the outer panel without the inner panel), and with multiple clues they serve only to bypass the one-clue-per-cell limit.  As antibodies are only satisfied if they are necessary, an inner panel always provides the same number of clues to its outer panel regardless of how it is solved.  Finally, one could extend to multiple levels of recursion.

Because clues are only promoted from inner to outer levels, recursive puzzles are contained in the larger of the largest class among the complexity classes of their inner puzzles and the largest class possible for the outer puzzle (considering all possible sets of surviving clues).
%This bound is slightly loose because the possible outer puzzles might not all fall in the same class.
For example, consider an outer puzzle containing two antibodies and inner puzzles each containing one antibody, one polyomino clue, and one nonpolyomino clue.  Then the inner puzzles are in NP, and when the inner puzzles all provide the nonpolyomino clue, the resulting outer puzzle contains only antibodies and nonpolyominoes, and so is also in NP. But because some possible resulting outer puzzles may contain polyominoes, we obtain the weaker bound that the recursive puzzle is in $\Sigma_2$.
%Obtaining an accurate bound requires solving the inner puzzles. \xxx{seems to be confusing instances with problems}

\section{Metapuzzles}
\label{Metapuzzles}

\abstractlater{
  \section{Proofs: Metapuzzles}
  \label{appendix:meta}
}

\begin{wrapfigure}{r}{3.5in}
  \centering
  \vspace*{-8ex}
  \includegraphics[width=\linewidth]{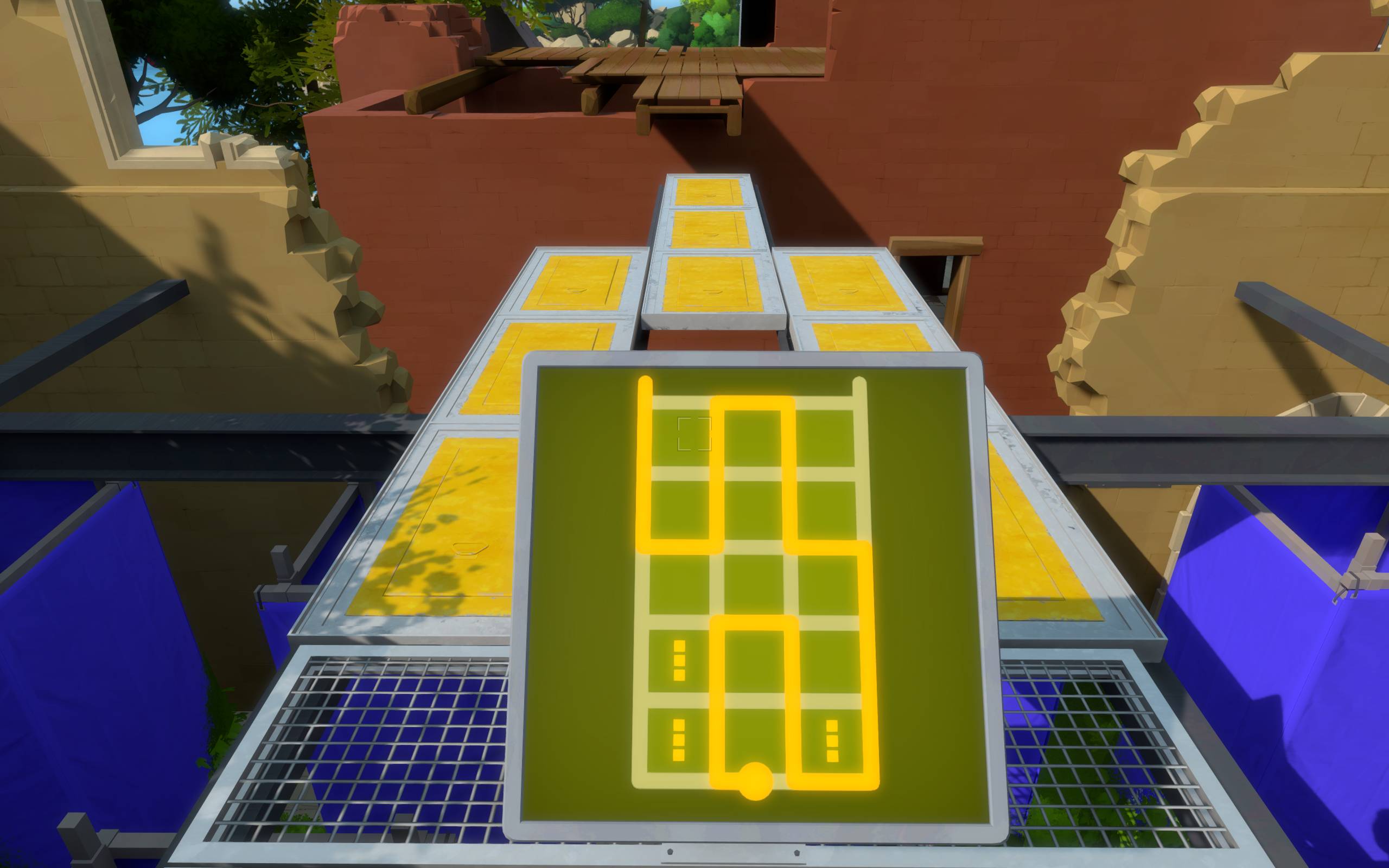}
  \caption{A screenshot from The Witness, featuring a panel whose solution
    controls a sliding bridge.}
  \vspace*{-3ex}
  \label{screenshot metapuzzle}
\end{wrapfigure}

In this section, we analyze several of the \emph{metapuzzles} that appear in The Witness.  Metapuzzles are puzzles which have one or more puzzle panels as a subcomponent of the puzzle, and in which solving the puzzle panel affects the surrounding world in a way that depends on the choice of solution that was used to solve the panel.
Figure~\ref{screenshot metapuzzle} shows one example from The Witness.
\ifabstract
Proofs omitted from this section are available in Appendix~\ref{appendix:meta}.
\fi

%\xxx{A collection of puzzles that don't directly interact/are separately controlled, but change the environment so as to require solving them in sequence}

Unlike the single-panel path-finding puzzles, metapuzzles are naturally a kind of reconfiguration puzzle with no obvious bound on the number of moves (solving and re-solving panels). Indeed, we show PSPACE-completeness for three different metapuzzles: sliding bridges, elevators and ramps, and power cables and doors. Refer to Table~\ref{3D results}.

Our PSPACE-hardness results use the \emph{door framework} of \cite{Nintendo_TCS}. Specifically, \cite[Section~2.2]{Nintendo_TCS} shows that the following gadgets suffice to prove a one-robot motion planning puzzle PSPACE-hard:
\begin{enumerate}
\item \textbf{One-way}: allows the robot/player to traverse from a point A to a point B but not from B to A.
\item \textbf{Door}: has three paths, \emph{open}, \emph{traverse},
  \emph{close} and two states \emph{open} and \emph{close}. The
  \emph{open} and \emph{close} paths transition the gadget to the
  \emph{open} and \emph{close} states respectively. The \emph{traverse}
  path can be traversed if and only if the door is in the \emph{open}
  state and is blocked otherwise.
\item \textbf{Crossover}: has two non-interacting paths which can be
  independently traversed and which geometrically cross (in projection).
  This gadget is trivial to build in a 3D game like The Witness with
  ramps and tunnels. (The Witness is inherently nonplanar.)
\end{enumerate}

\subsection{Sliding Bridges}
The first such metapuzzle we will discuss are the \emph{sliding bridges} found in the marsh area. In this metapuzzle, each bridge has a corresponding puzzle panel, and solving the puzzle causes the bridge to move into the position depicted by the outline of the solution path. The following theorem demonstrates that, regardless of the difficulty of the puzzle panels (i.e., even if it is easy to find all solutions of each individual panel), it is PSPACE-complete to solve sliding bridge metapuzzles.

\both{
\begin{theorem}
It is PSPACE-complete to solve Witness metapuzzles containing sliding bridges.
\label{sliding-bridges}
\end{theorem}
}

\ifabstract
\begin{proofsketch}
We straightforwardly construct the \emph{one-way} and \emph{door}
gadgets of~\cite{Nintendo_TCS}, which are known to be sufficient for
PSPACE-completeness. \xxx{The full proof can be found in Appendix \ZZZ}
\end{proofsketch}
\fi

\later{
\begin{proof}
We apply the door framework described above using the gadgets in
Figure~\ref{sliding-bridge-gadgets}.
%Figures~\ref{sliding-bridge-diode} and~\ref{sliding-bridge-lock-door}
%show the one-way and door gadgets, respectively.

\begin{figure}
\centering
\subcaptionbox{\label{sliding-bridge-diode} One-way gadget}
  {\includegraphics[scale=0.45]{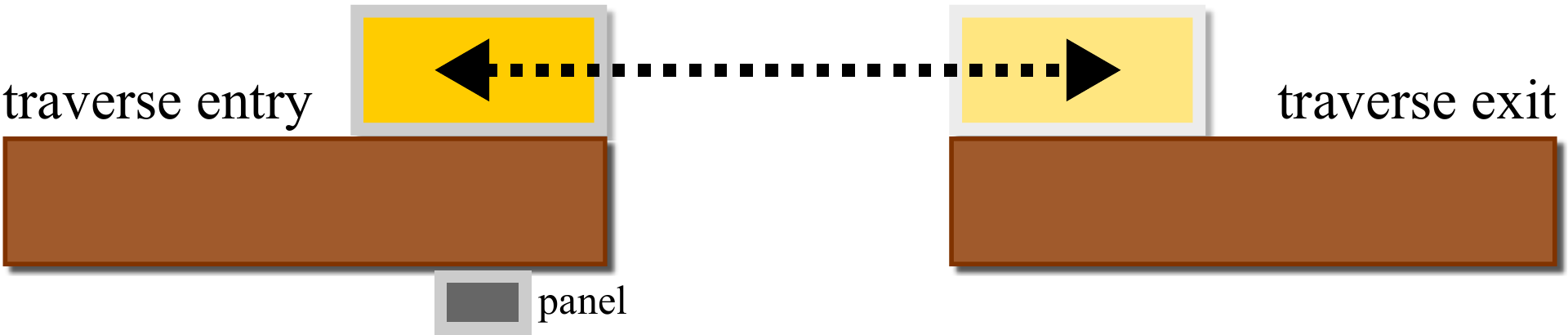}}\hfill
\subcaptionbox{\label{sliding-bridge-lock-door} Door gadget}
  {\includegraphics[scale=0.45]{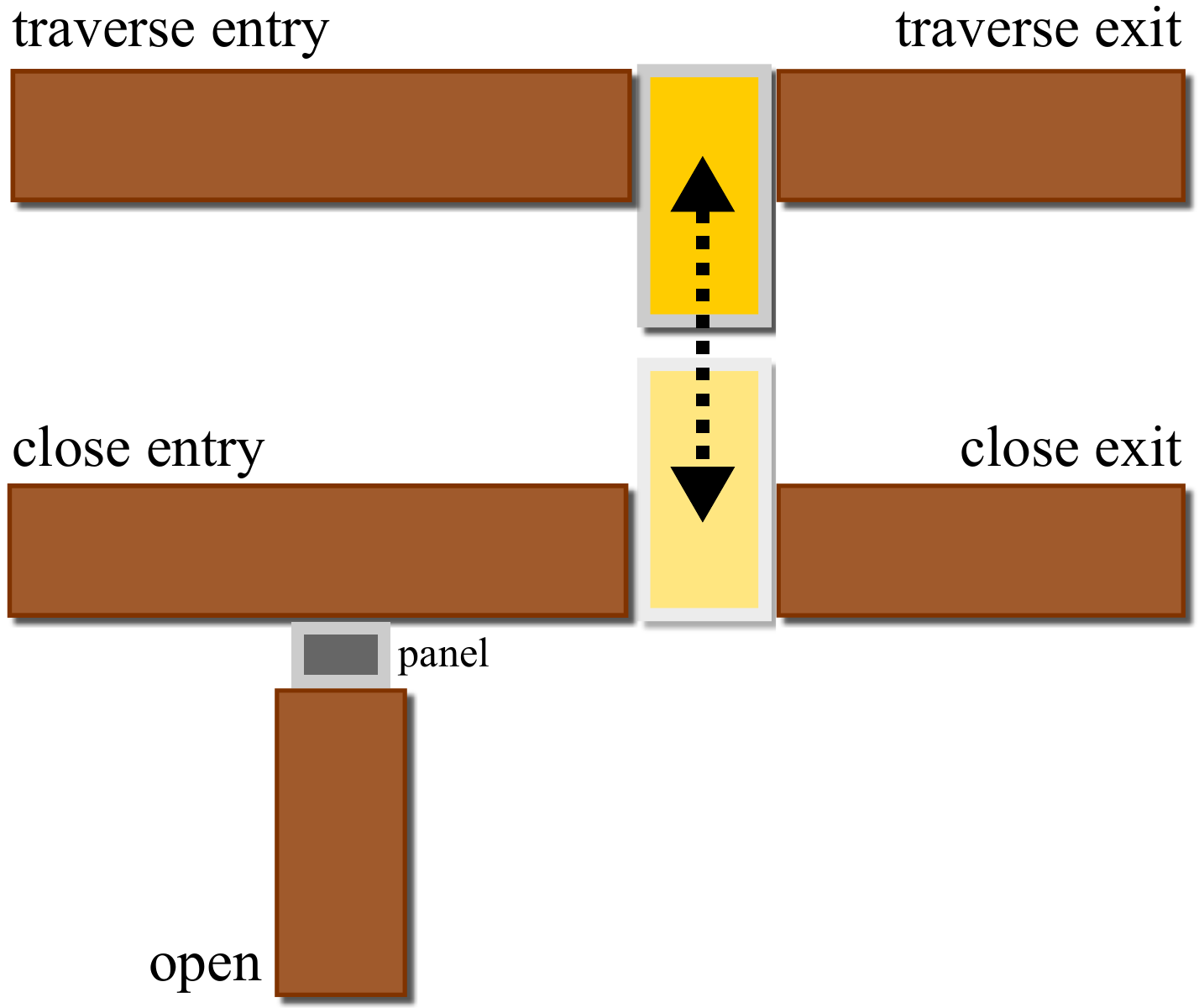}}
\caption{Sliding-bridge gadgets for Theorem~\ref{sliding-bridges},
         drawn from a birds-eye view.}
\label{sliding-bridge-gadgets}
\end{figure}

The one-way gadget consists of a sliding bridge that can be adjacent to one of two docks; see Figure~\ref{sliding-bridge-diode}.  The entry dock has a puzzle panel which controls the bridge.  When the player is at the entry dock, they can change the bridge's position by solving the panel.  To cross to the exit dock, the player stands on the bridge and solves the panel to set the bridge's position to the exit dock.  The panel is not viewable from the exit dock, so once the bridge reaches the exit dock, the player cannot return to the entry dock.

The door gadget consists of five docks, named \emph{traverse-entry}, \emph{traverse-exit}, \emph{close-entry}, \emph{close-exit} and \emph{open}, and a single bridge with two states, spanning the gap between either the two traverse docks or the two close docks; see Figure~\ref{sliding-bridge-lock-door}. There is also a one-way gadget on either side of the close-entry and close-exit so the player can only enter from the close-entry side and exit from the close-exit side. In either bridge state, the player can cross between the docks adjacent to the bridge, but not the other pair of docks.  We place the controlling puzzle panel visible from only the \emph{open} and \emph{close-entry} docks, allowing the player to set the state of the bridge from those two docks only.  The panel is far enough away from the end of the \emph{close-entry} dock to prevent the player from running onto the bridge after sending the bridge to the gap between traverse docks.  The player opens the door by using the panel from the \emph{open} dock to move the bridge to the traverse docks, and can then traverse that path; to cross between the close docks, the player is forced to use the panel from the \emph{close-entry} dock to move the bridge to the close docks, closing the door.  

Using these gadgets, we can construct any instance of TQBF using the framework from~\cite{Nintendo_TCS} and so The Witness metapuzzles with just sliding bridges are PSPACE-hard. Sliding bridge puzzles have only a polynomial amount of state (given by the position of each bridge and the position of the player) so these puzzles are also in PSPACE. Thus, sliding bridge puzzles are PSPACE-complete. 
\end{proof}
}

\subsection{Elevators and Ramps}

Another metapuzzle which appears in The Witness consists of groups of platforms that move vertically at one or both ends to form an elevator or ramp, controlled by the path drawn on puzzle panels.  Because the player cannot jump or fall in The Witness, the player can walk onto an elevator platform only if it is at the same height as the player.  The player can adjust the height of the platforms from anywhere with line-of-sight to the controlling panel, including while on the platforms themselves.  Groups of these elevator puzzles are PSPACE-complete by a similar argument to sliding bridges puzzles, and indeed the sliding ramp in the sawmill can also be used in our construction for sliding bridges.  Besides the sawmill, the other building in the quarry contains a ramp and an elevator.  The marsh contains a single puzzle with a $3 \times 3$ grid of elevators controlled by two identical panels; as a metapuzzle, our puzzle could be built out of multiple marsh puzzles with two platforms and one panel each.

\both{
\begin{theorem}
\label{thm:elevator}
It is PSPACE-complete to solve Witness metapuzzles containing elevator reconfiguration, even when each panel controls at most one elevator.
\end{theorem}
}

\ifabstract
\begin{proofsketch}
We construct \emph{one-way} and \emph{door} gadgets in a manner similar
to the reduction in Theorem~\ref{sliding-bridges}. \xxx{The full proof
  can be found in Appendix \ZZZ}
\end{proofsketch}
\fi

\later{
\begin{proof}
Again we apply the door framework, building our one-way and door gadgets by
modifying the gadgets from the sliding-bridges proof.

For the one-way gadget, we replace the bridge with an elevator
controlled by a panel on the lower level. When the elevator is on the
lower level, the player can access the puzzle panel to raise the
elevator and reach the upper level. From the upper level, the player
cannot see the panel at all and cannot move the elevator to reach the
lower level. This gadget's requirement that the player travel from the
lower level to the upper level is not a constraint, as we can freely
add elevators controlled by a panel visible from both levels.

\begin{figure}
\centering
\includegraphics[width=0.4\textwidth]{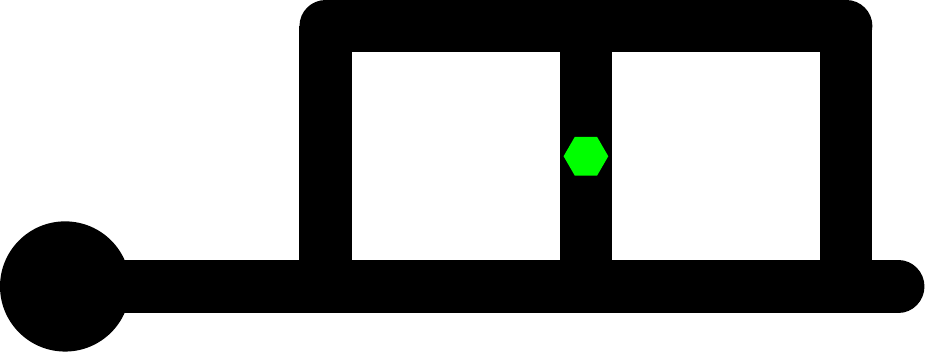}
\caption{Elevator control puzzle panel for the door gadget for Theorem~\ref{thm:elevator}.}
\label{fig:elevator}
\end{figure}

For the door gadget, we use the puzzle panel in Figure~\ref{fig:elevator} to
allow the player to flip between the two states. (We show a panel using an
edge hexagon, but equivalent (possibly larger) panels can be constructed using
squares, stars, triangles, or monominoes instead.) Now the moving bridge in
the previous reduction (Figure~\ref{sliding-bridge-lock-door}) is replaced
with two up/down platforms which are constrained by the puzzle panel such that
exactly one of them can be up at any given time, and the player can only
traverse between an entry and exit dock if the corresponding block is down. 
\end{proof}
}

\subsection{Power Cables and Doors}

In the introductory area of The Witness, there are panels with two solutions, each of which activates a power cable.  Activated cables can power one other panel (allowing it to the solved) or one door (opening it).  If a cable connected to a door is depowered, the door closes.  Cables cannot be split and panels can power at most one cable at a time.

\both{
\begin{theorem} \label{thm:power cables}
It is PSPACE-complete to solve Witness metapuzzles containing power cables and doors.
\end{theorem}
}

\ifabstract
\begin{proofsketch}
Once again, we construct \emph{one-way} and \emph{door} gadgets, with the slight complication that all powered doors in The Witness are initially closed, so we need to give the player a way to open exactly the set of doors
which are initially open in the source instance. \xxx{The full proof can
  be found in Appendix \ZZZ}
\end{proofsketch}
\fi

\later{
\begin{proof}
Again we wish to apply the door framework.
One difficulty is that power cables in The Witness are initially off until their powering puzzle has been solved, so initially all doors are closed.  We use an additional construction to allow the player to open the doors that are initially open in the one-way-and-doors instance we are reducing from.

%The one-way gadget is shown in Figure~\ref{powered-door-diode}.
This gadget is akin to an airlock: there are two doors in sequence (an entry and an exit) which are both initially closed.  There is a puzzle panel just outside the entry door with two solutions connected by power cables to the entry and exit doors.  When the player approaches the one-way gadget, the player can open the entry door (closing the exit door if it is open) using the panel and walk through the door.  To leave via the exit door, the player must solve the panel with the other solution while standing inside the airlock looking through the open entry door, thereby opening the exit door but closing the entry door.  Once the entry door is closed, the panel is not visible from the airlock, preventing the player from traversing from the exit to the entry door.

%\begin{figure}
%\centering
%\includegraphics[width=\textwidth]{figures/powered-door-diode.pdf}
%\caption{Powered door metapuzzle one-way
%  gadget.}
%\label{powered-door-diode}
%\end{figure}

The door gadget is almost identical to that of Theorem~\ref{sliding-bridges} (shown in Figure~\ref{sliding-bridge-lock-door}), except that instead of a moving bridge, the puzzle panel is connected by cables to two doors located in the positions where the bridge could have been. As before, in order to traverse the \emph{close} path, the player must solve the panel puzzle such that the bottom door is open, thereby closing the top door. The \emph{open} path can be used to open the top door, and finally, the \emph{traverse} path is traversable if and only if the top door is open.

The one-way-and-doors instance we are reducing from may have some of the door gadgets initially open.  Because The Witness is three-dimensional, we can build a skybridge above the rest of the puzzle allowing the player to open any door gadgets that should be initially open.  For each such door gadget, a path branches off the skybridge to an overlook point with line of sight to the puzzle panel controlling that gadget.  The player can use this overlook to open the door from the skybridge.  Using walls along the sides of the skybridge (including its branches), we block line of sight to the panels controlling door gadgets that should be initially closed.  The skybridge ends in a one-way gadget leading to the starting point of the one-way-and-doors instance, ensuring the player cannot return to the skybridge.  Note that setting a door gadget to the \emph{open} state allows the player to traverse strictly more paths than the \emph{closed} state, so the player cannot benefit from leaving any skybridge-visible doors closed.
\end{proof}
}

\subsection{Light Bridges}

Light-bridge metapuzzles consist of a large void between two platforms, with a controller panel at one or both ends.  Drawing a path on the controller panel causes a ``hard light'' bridge to follow a corresponding path between the two platforms.  As with all panels, solution paths remain displayed on the panel after the player finishes drawing them, so the player can traverse the light bridge to the other platform.  If there is a controller panel on that side of the void, the player can draw another path to create another light bridge.  Each controller panel can only control the bridge originating on its side, but both paths appear on the panel, so the first bridge's path constrains the second, and any clues in the panels must be satisfied in each configuration.  The objective of the metapuzzle is to position the bridge(s) so that the player can view (and solve) panels located around the void, or to use the bridge(s) to reach an exit platform on a third side of the void.

The Witness has two light-bridge metapuzzles inside the mountain.  The first metapuzzle has a single bridge.  There is a pillar in the center of the void that blocks the player from traversing straight through it, but permits the player to enter from the bridge and exit at a 90-degree angle, slightly divorcing the player's path from the path drawn on the panel.  Besides panels on the platform across the void, there are additional panels mounted on the walls of the void area, so a secondary goal is to have the bridge pass in front of these panels (as if the panel has ``virtual'' edge hexagons).  Neither of these features is interesting for hardness.

The second light bridge metapuzzle has two bridges.  The objective of this metapuzzle is to exit through a door on the side of the void.  This door is only open when both bridges are in place, and even if it were always open, it is not possible to draw a solution path on the first controller panel directly to the door (because the panel's clues cannot be satisfied).  Instead the player must cross to the other platform, draw the second bridge, return to the first platform and adjust the first bridge, and so on.  The optimal solution requires crossing and returning twice before the first bridge can be drawn to the exit.

The hardness of light-bridge reconfiguration metapuzzles remains open.  We speculate that they are PSPACE-complete, as many reconfiguration problems are.

\begin{restatable}{open}{lightbridge} \label{open:light bridge}
Is it PSPACE-complete to solve Witness light-bridge reconfiguration metapuzzles?
\end{restatable}

\section{Puzzle Design Problem}
\label{sec:puzzle design}

%\xxx{express, specify, force all used interchangeably to describe what a set of clues does.  should be consistent, maybe make it a defined term}

The previous sections analyzed the hardness of solving Witness puzzles.  In this section, we consider the \emph{Witness puzzle design problem}: creating a Witness puzzle having a specified set of solution paths.  While most of the panels in The Witness only need to be interesting to solve, the controller panels for metapuzzles depend critically on having a specific set of solutions corresponding to the positions of the environmental objects they control.  That is, the panel controlling a sliding bridge metapuzzle that brings the player from one platform to another must have solutions representing the bridge being at each end.  The controller also must \textit{not} have a solution representing the bridge in the middle, lest the player get stuck on one platform by moving the bridge to the middle while not on it.

%--- Erik removed the following. Designing 3D shapes to block specific places
%    from given perspectives is easy, not hard.
%While the nonclue constraints discussed in Section~\ref{sec:nonclue-constraints} do not make solving puzzles harder, they become more interesting in the puzzle design problem.  For example, adding the visual obstruction constraint to a puzzle now requires designing a 3D obstacle, which may be difficult\xxx{can we say something about computational geometry here?}, particularly for panels designed to be solved multiple times from different positions (such as the three-position obstruction puzzle at the mountain summit).

\begin{restatable}{open}{puzzledesign} \label{open:puzzle-design}
  What is the complexity of the Witness puzzle design problem?
\end{restatable}

\subsection{Path Universality}

More specifically, it is natural to ask whether there can exist an algorithm using only a subset of the available clue types and nonclue constraints.  Ideally, we would hope for a small clue set having \emph{path universality}, the ability to express all possible sets of solution paths.  Unfortunately, an information-theoretic argument rules out path universality for reasonable clue sets.

\begin{observation} \label{obs:lotsa paths}
  A rectangular board with $n$ cells admits $2^{\Theta(n)}$ simple paths and
  $2^{2^{\Theta(n)}}$ sets of simple paths, even if we fix the endpoints to
  opposite corners of the puzzle.
\end{observation}
\begin{proof}
  Fixing the endpoints reduces the number of paths by $\Theta(n^2)$,
  so we focus on the fixed-endpoint case.
  %--- Argument for arbitrary endpoints:
  For the upper bound, a simple path makes $O(n)$ steps with at most three
  choices (turn left, turn right, or go straight) at each step, beyond
  the first step with four choices (north, west, south, east), so there
  are $O(3^n)$ paths.
  %--- The following is true only for boundary-to-boundary paths:
  %For the upper bound, a path is determined by which cells are on the left
  %side of the path, so there are at most $2^n$ paths.
  For the lower bound, we consider the case that the fixed endpoints are
  diagonally opposite, namely top-left and bottom-right,
  and describe a family of $2^{\Omega(n)}$ such paths.
  Suppose the puzzle is $x \times y$, and assume by symmetry that $x \geq y$.
  We start with $2^{x-1}$ possible paths within the top row of cells,
  from the top-left corner to the bottom-right corner of the row,
  by choosing whether each cell except the rightmost is left or right of the
  path (while the rightmost is forced to be on the left, to ensure the correct
  endpoint).
  Then we continue the path down one unit, and apply the same construction to
  the third row of cells from the top-right corner to the bottom-left corner.
  We repeat this process, zig-zagging back and forth, until we must stop
  to leave room to get to the bottom-right corner of the puzzle.
  In total, we obtain $$(2^{x-1})^{2 \lfloor (y-1)/4 \rfloor + 1}
  \geq 2^{(x-1) ((y-1)/2+1)} = 2^{(x-1) (y+1/2)}
  %= 2^{xy+x/2-y+1/2}
  = 2^{\Omega(n)}$$ paths.
  % y=1 -> 1 horizontal traversal
  % y=2 -> 1
  % y=3 -> 1
  % y=4 -> 2 (treating as 1 in formula above)
  % y=5 -> 3
  % y=6 -> 3
  % y=7 -> 3
  % y=8 -> 4 (treating as 3 in formula above)
  %
  %Consider a $1 \times c$ (cells) puzzle and assume the path
  %starts in the upper left and moves right.  Each time the path moves right, it
  %has the option to switch rows, except possibly in the rightmost column where
  %it switches only if it entered on the bottom row.  Thus, there are at least
  %$2^c$ paths in this puzzle.
  Combining the upper and lower bounds gives $2^{\Theta(n)}$ paths,
  and so the powerset of the set of paths has size
  $2^{2^{\Theta(n)}}$.
\end{proof}

For clues of constant size, there are $c = O(1)$ choices of clue for each
vertex, edge, and cell (including not placing a clue at that position), and
$O(n)$ clue positions, so there are only $c^{O(n)}$ puzzles.  Even if we allow
polyominoes and antipolyominoes of size $k(n)$, there are at most
$c^{O(n\cdot k(n))}$ puzzles.
By Observation~\ref{obs:lotsa paths}, most sets of solution paths
do not have a corresponding puzzle.

\subsection{$k$-Path Universality}

Instead, we must settle for \emph{$k$-path universality}, the ability to
express all sets of up to $k$ solution paths, for some small~$k$.
Here we focus on $k=1$.

\begin{observation}
Broken edges are $1$-path-universal.
\end{observation}
\begin{proof}
Break all edges not in the desired solution path.
\end{proof}

\begin{observation}
Edge hexagons are $1$-path-universal.
\end{observation}
\begin{proof}
Place edge hexagons on exactly those edges in the desired solution path.  When the path visits an edge, it visits one of the vertices incident to the next edge, so if the path ever deviates from the edges with hexagons, at least one of the hexagons cannot be visited because one of its incident vertices has already been visited.
\end{proof}

Many clue sets are not even $1$-path-universal, often because they cannot distinguish between different paths along the boundary.  Consider the empty puzzle in Figure~\ref{fig:short-vs-long-empty}.  The following observations are based on trying to add clues that distinguish the \emph{short path} in Figure~\ref{fig:short-vs-long-short} from the \emph{long path} in Figure~\ref{fig:short-vs-long-long}.

\begin{figure}
\centering
\subcaptionbox{\label{fig:short-vs-long-empty} Empty puzzle}{%
  \includegraphics[width=.3\textwidth]{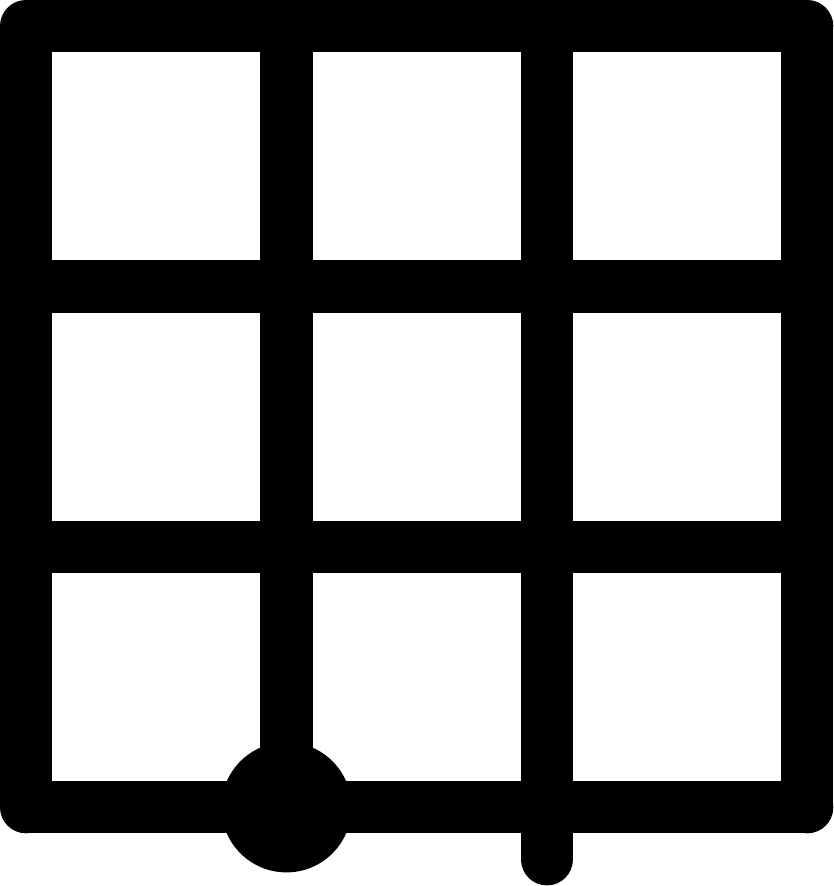}%
}~~~~
\subcaptionbox{\label{fig:short-vs-long-short} Short path}{%
  \includegraphics[width=.3\textwidth]{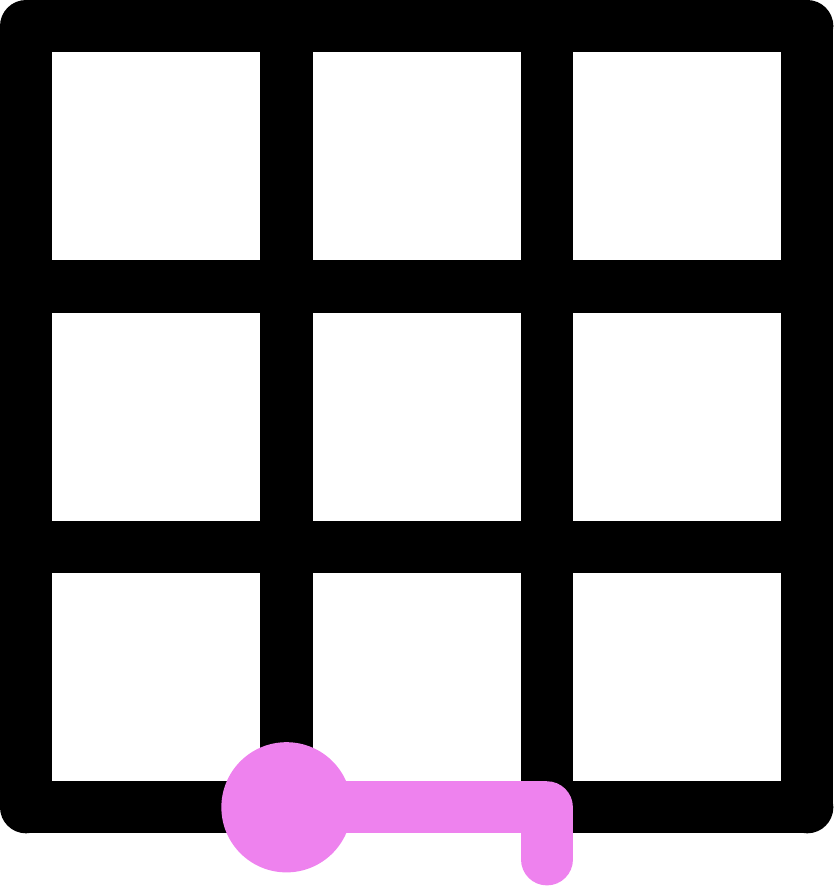}%
}~~~~
\subcaptionbox{\label{fig:short-vs-long-long} Long path}{%
  \includegraphics[width=.3\textwidth]{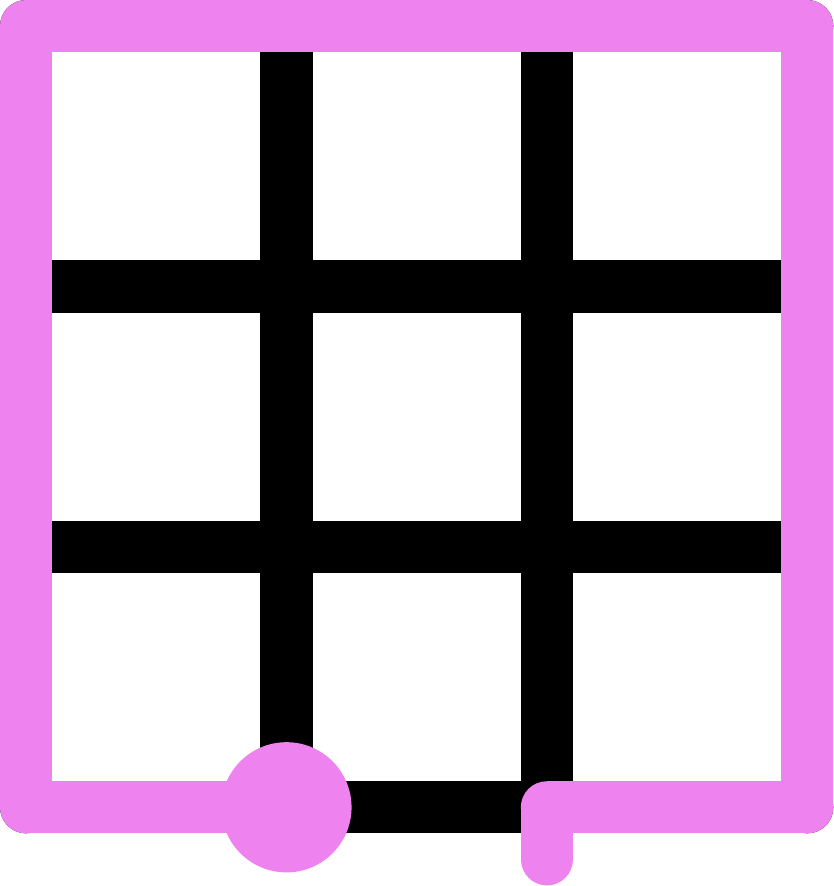}%
}
\caption{Path universality requires distinguishing the short and long boundary paths.}
\label{fig:short-vs-long}
\end{figure}

\begin{observation}
Squares, stars, polyominoes, and antipolyominoes, in any combination, are not $1$-path-universal.
\end{observation}
\begin{proof}
These clue types cannot distinguish the short and long paths because their satisfaction depends only on the other clues contained in or the shape of their region, not the path that induced the region.
%Squares of only one color impose no constraint (Observation~\ref{obs:1 square}), but squares of two colors define at least two regions.
%If there are exactly two clues of each star color, all stars can be in the same region regardless of the path's direction; otherwise, there must be at least two regions.
%Both the short and long paths result in the same region (the entire puzzle), so polyominoes and antipolyominoes either pack the region or not, regardless of the path.
\end{proof}

\begin{observation}
Vertex hexagons are not $1$-path-universal.
\end{observation}
\begin{proof}
Placing vertex hexagons at each boundary vertex (Figure~\ref{fig:vertex-hex-nonuniv-boundary}) rules out the short path, but does not force the long path because the path can enter the interior and return to the boundary at the next hexagon (Figure~\ref{fig:vertex-hex-nonuniv-interior}).  Vertex hexagons are entirely redundant for enforcing the short path because the vertices bearing the start circle and end cap are already required to appear in the path.
\end{proof}

\begin{figure}
\centering
\subcaptionbox{\label{fig:vertex-hex-nonuniv-boundary} Puzzle}{%
  \includegraphics[width=.3\textwidth]{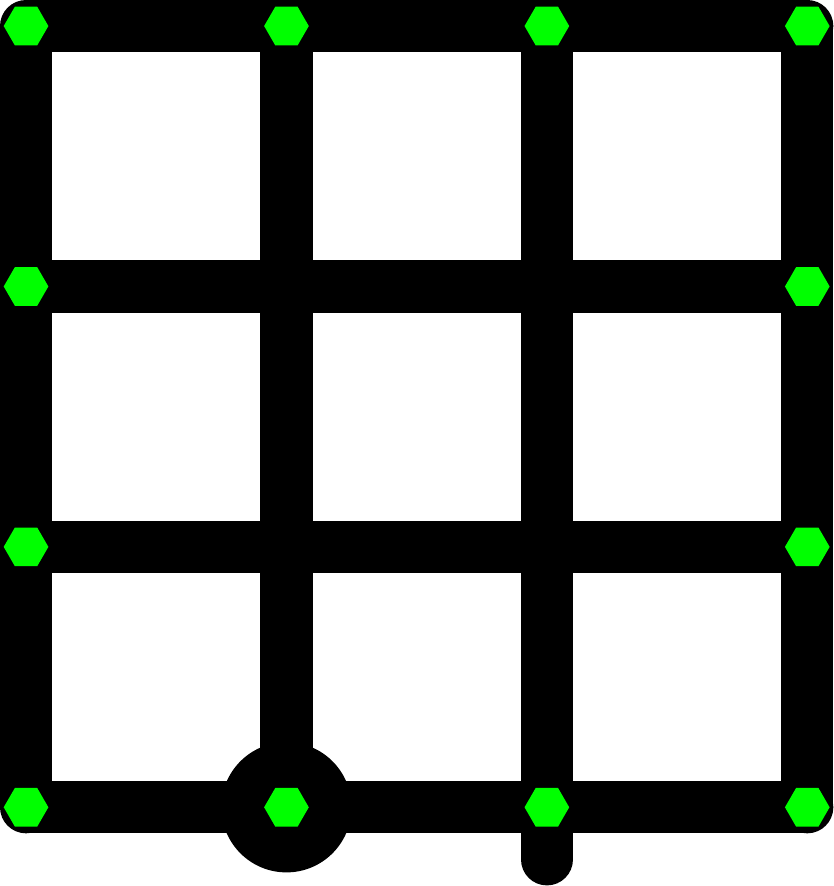}%
}~~~~
\subcaptionbox{\label{fig:vertex-hex-nonuniv-interior} Unintended solution}{%
  \includegraphics[width=.3\textwidth]{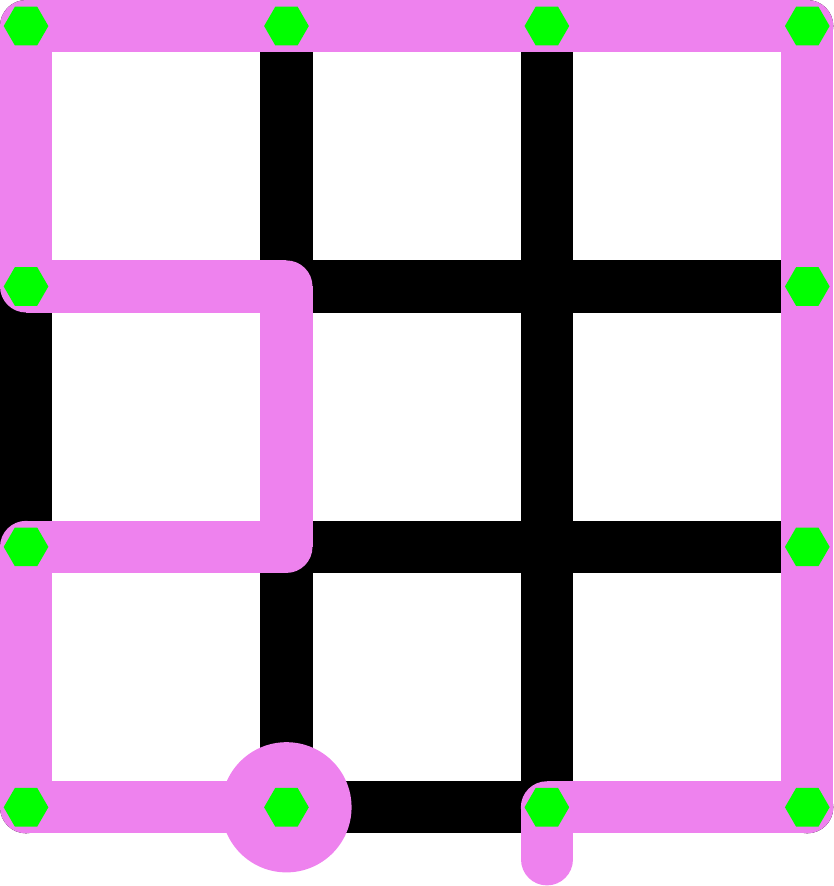}%
}
\caption{Vertex hexagons can rule out the short path, but not force the long path.}
\label{fig:vertex-hex-nonuniv}
\end{figure}

Vertex hexagons combined with antibodies can force the short and long paths as shown in Figure~\ref{fig:vertex-hex-antibody}.  This example shows that adding antibodies can increase the set of expressible paths.  However this combination is still not $1$-path-universal as vertex hexagons cannot distinguish the paths in Figure~\ref{fig:vertex-hex-zigzag} and there are no unsatisfied hexagons to be eliminated by antibodies.

\begin{figure}
\centering
\subcaptionbox{\label{fig:vertex-hex-antibody-short} Short path}{%
  \includegraphics[width=.3\textwidth]{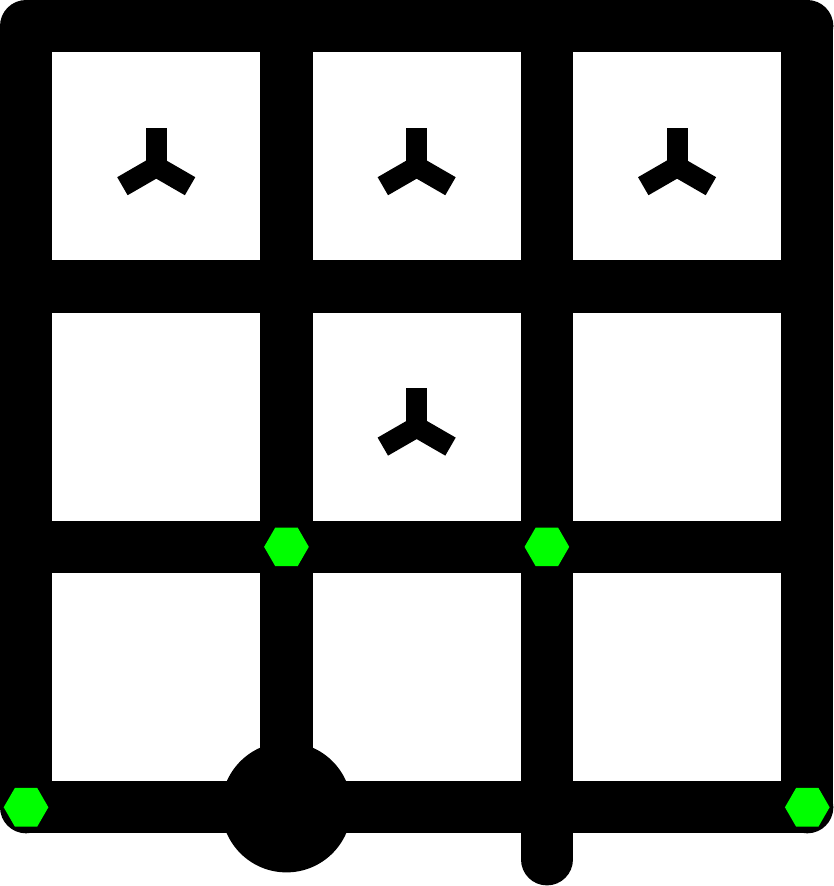}%
}~~~~
\subcaptionbox{\label{fig:vertex-hex-antibody-long} Long path}{%
  \includegraphics[width=.3\textwidth]{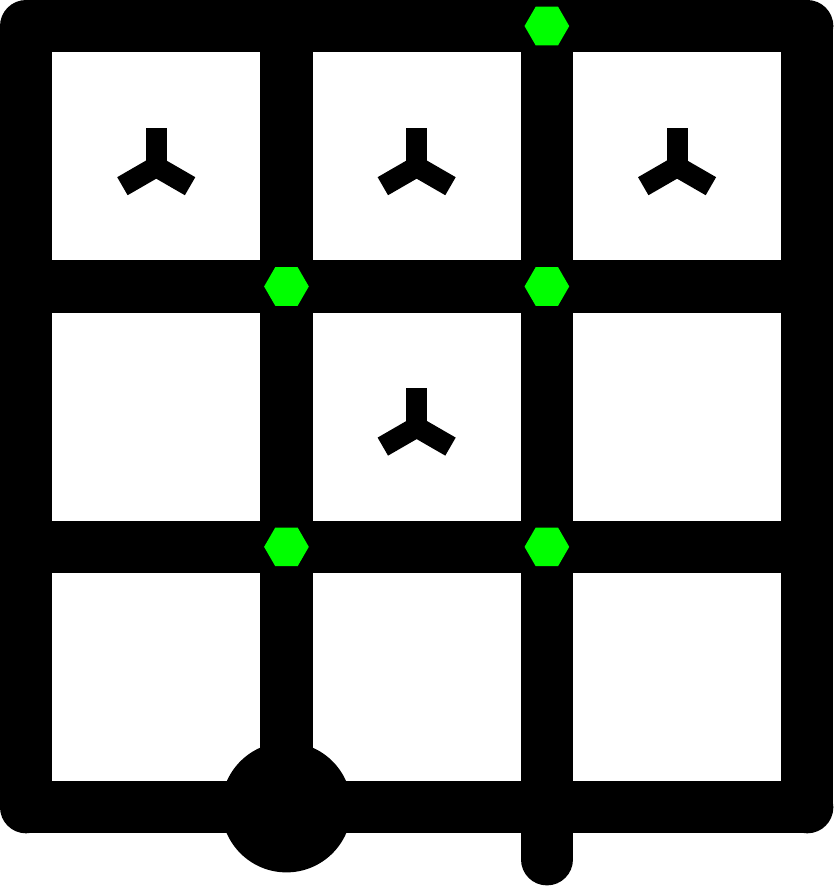}%
}
\caption{Vertex hexagons with antibodies can force the short and long paths.}
\label{fig:vertex-hex-antibody}
\end{figure}

\begin{figure}
\centering
\subcaptionbox{\label{fig:vertex-hex-zigzag-horiz} Horizontal zigzag}{%
  \includegraphics[width=.3\textwidth]{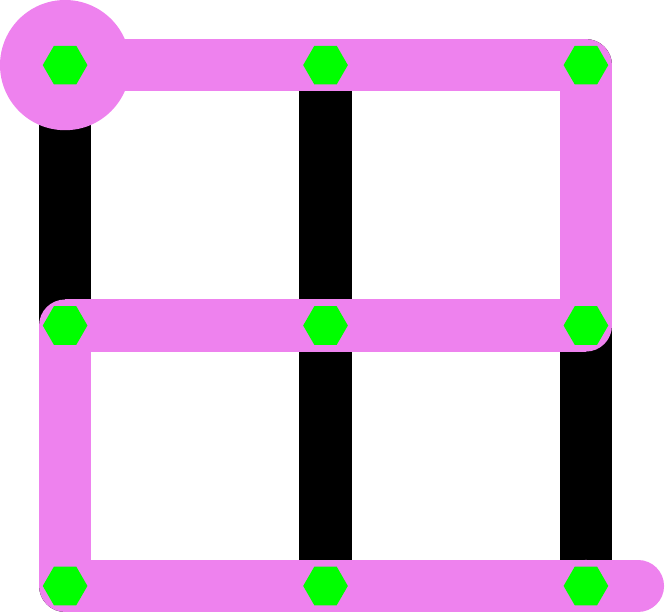}%
}~~~~
\subcaptionbox{\label{fig:vertex-hex-zigzag-vert} Vertical zigzag}{%
  \includegraphics[width=.3\textwidth]{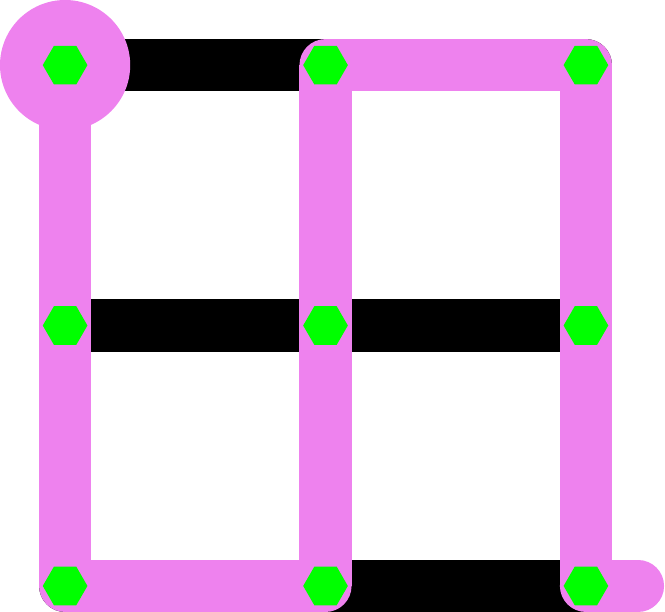}%
}
\caption{Vertex hexagons cannot distinguish these solutions, and antibodies cannot help.}
\label{fig:vertex-hex-zigzag}
\end{figure}

\begin{observation}
$1$-triangle clues are not $1$-path-universal.
\end{observation}
\begin{proof}
Adding a 1-triangle clue in the cell adjacent to both the start circle and end cap rules out the long path, but doesn't force the short path, as shown by the unintended solution in Figure~\ref{fig:1tri-nonuniv}.
\end{proof}

\begin{figure}
\centering
\includegraphics[width=.3\textwidth]{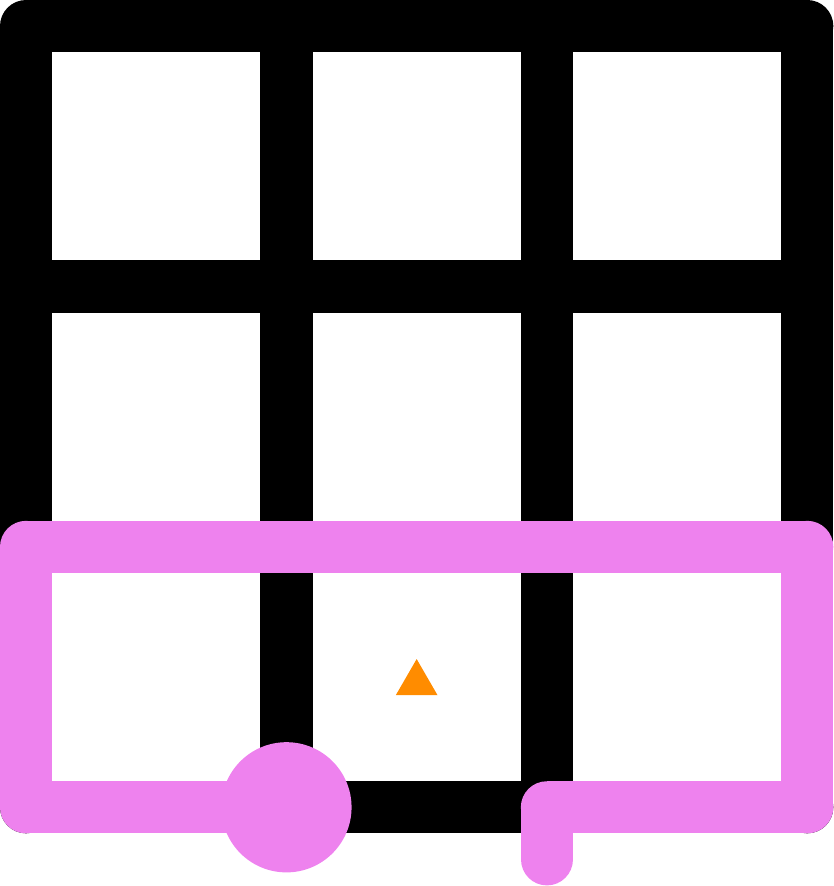}
\caption{1-triangle clues can rule out the long path, but not force the short path.}
\label{fig:1tri-nonuniv}
\end{figure}

\subsection{Region Universality}

As paths along the boundary are the source of problems, we can settle for
specifying the same interior edges as a given path, but possibly with
different boundary edges.  This is equivalent to the ability to specify all
\emph{region decompositions} (sets of regions inducible by some path and whose
union is the entire puzzle), so we call it \emph{region universality}.
There are $O(n^2)$ paths consistent with a given region decomposition
(see Lemma~\ref{thm:decomp-to-path} below),
so Observation~\ref{obs:lotsa paths} also roughly counts the
number of (sets of) region decompositions, so the same information-theoretic
argument makes general region universality impossible.
Therefore, we consider \emph{$k$-region universality} where up to $k$ region
decompositions must be realized, and focus on $k=1$.

%Given a set of regions we can easily check whether there is a path that
%induces it.

\begin{lemma} \label{thm:decomp-to-path}
In polynomial time, we can decide whether a given set
$D = \{D_1, D_2, \dots, D_m\}$ of disjoint sets of cells is a
region decomposition, and if so,
produce a path that induces that region decomposition.
\end{lemma}
\begin{proof}
If the union of the sets of cells is not the set containing all cells in the puzzle, $D$ is not a region decomposition.  If any $D_i$ does not contain a cell adjacent to the boundary, $D$ is not a region decomposition, because the path along $D_i$'s outline is a cycle.  If any vertex in the puzzle is incident to cells in more than two regions, $D$ is not a region decomposition, because more than two edges incident to that vertex are forced to be in any solution.

Otherwise, we will search for a path $P$ inducing $D$.  Initially, $P$ consists of all edges between adjacent pairs of cells not in the same region (all paths inducing $D$ must contain these edges).  Because every $D_i$ contains at least one cell on the boundary, each path segment in $P$ begins and ends at a boundary vertex.  To avoid inducing additional regions, the edges we add to complete $P$ must lie entirely on the boundary.

For each pair $(s,t)$ of endpoints of (possibly different) segments, consider the other endpoint $s'$ of the segment starting at $s$.  The path can go clockwise or counterclockwise along the boundary, connecting $s'$ to the next endpoint in that direction.  After this choice, the boundary segments must alternate between not being in the path and being in the path, so that every endpoint except $s$ and $t$ has degree 2.  If there is such a choice of $s$, $t$ and boundary segment parity, then the resulting path induces $D$; otherwise, $D$ is not a region decomposition.  There are only $2i$ choices for $s$ and $t$ and two choices for parity, so this algorithm runs in polynomial time.
\end{proof}

Any $1$-path-universal set of clues is also $1$-region-universal.
We revisit the examples above that are not $1$-path-universal:

\begin{observation}
Vertex hexagons are not $1$-region-universal.
\end{observation}
\begin{proof}
The puzzle in Figure~\ref{fig:vertex-hex-zigzag} contains the maximal set of vertex hexagons and admits at least two region decompositions (induced by the solutions shown in the figure).  Removing vertex hexagons from a puzzle can only add additional solutions and region decompositions, not remove them, so vertex hexagons cannot express only one of the region decompositions shown in the figure.
\end{proof}

\begin{observation}
Squares are not $1$-region-universal.
\end{observation}
\begin{proof}
Squares cannot force the entire puzzle to be a single region.  If the puzzle contains only squares of a single color, the squares do not constrain the path or the region decomposition.  If the puzzle contains multiple colors of squares, all solutions induce at least two regions.
\end{proof}

\begin{observation}
Stars are not $1$-region-universal.
\end{observation}
\begin{proof}
Stars cannot express some region decompositions in which a region consists of a single cell.  Consider a desired decomposition in which the bottom-left cell forms a singleton region.  Any cell containing a star cannot be in a singleton region in any solution, so the bottom-left cell must be empty.  But stars are not affected by empty cells in their region, so if there is a solution path placing the bottom-left cell alone in a region, the path remains a solution after being locally modified to place the bottom-left cell in the adjacent region instead.
\end{proof}

\begin{restatable}{open}{polyominoregionuniversality} \label{open:poly-region-univ}
Are polyominoes $1$-region-universal?
\end{restatable}

It is tempting to think that monominoes are $1$-region-universal by
2-coloring the regions and placing monominoes in every cell in one color class.
While this forces any solution path to visit all the specified region
boundaries, it leaves the path free to further subdivide those regions.

\paragraph{Nonclue constraints.}

Previously uninteresting nonclue constraints become more interesting for universality.  For example, intersection puzzles inherently take the intersection of the solution sets of their component puzzles, so adding intersection to a set of clue types may allow specifying more solution sets than before, possibly making them universal.

\begin{theorem}
\label{thm:intersection-disjointness-cohesion-universal}
Intersection puzzles having
\begin{itemize}
\item one component puzzle containing only squares of $2$ colors, or
\item one component puzzle containing only monominoes;
\end{itemize}
and
\begin{itemize}
\item two component puzzles containing only stars of $q$ colors, or
\item $q-1$ component puzzles containing only stars of $1$ color, or
\item one component puzzle containing nonmonomino polyominoes, or
\item $q-1$ component puzzles containing only monominoes and antimonominoes, or
\item $q-1$ component puzzles containing only antibodies and antimonominoes,
\end{itemize}
are $1$-region universal, where $r$ is the number of regions in the input region decomposition and $q$ is half of the maximum number of exterior cells in any region in the input region decomposition, rounded down.
\end{theorem}
\begin{proofsketch}
The clue types in the first set of component puzzles constrain groups of cells to be in different regions (\emph{disjointness}).  Conversely, the clue types in the second set of component puzzles constrain groups of cells to be in the same region (\emph{cohesion}).  Intersection puzzles allow us to overcome the one-clue-per-cell limit to fully specify the region decomposition using these groupwise constraints.  For colored clue types, there is a trade-off between the number of components and the number of colors.
\end{proofsketch}
\begin{proof}
% Writing note: intent is for 'R_i' to mean the input regions and 'region' to mean what a possible solution could result in.  That way statements such as "Thus, all cells in each $R_i$ must be in the same region" are meaningful.
We give an algorithm taking as input a region decomposition $D = \{R_1, R_2, \dots, R_r\}$ and producing an intersection puzzle for which all solutions induce $D$ (and there is at least one such solution).  Each component puzzle is shaped like the union of the $R_i$.  By Lemma~\ref{thm:decomp-to-path}, we can produce a path $P$ inducing $D$.  In every component puzzle, we place the start circle and end cap at the endpoints of $P$.

\paragraph{Disjointness.}  The component puzzle enforcing disjointness constrain the solution path to induce $D$ or a refinement of $D$ (splitting some $R_i$ into multiple regions).

\begin{itemize}
\item \textbf{Squares of $2$ colors:} We 2-color $D$ and create one component puzzle in which each cell contains a square of that cell's color.  Regions must be monochromatic in squares, so this puzzle enforces disjointness.

\item \textbf{Monominoes:} We 2-color $D$ and create one component puzzle having monomino clues in all the cells in one color class, leaving the cells in the other class empty.  By the polyomino area constraint, no solution can place a cell containing a monomino clue in the same region as an empty cell, so this puzzle enforces disjointness.
\end{itemize}

\paragraph{Cohesion.}  The component puzzle(s) enforcing cohesion constrain the solution to place cells in the same $R_i$ in the same region.  Note that our constructions for cohesion can assume disjointness, because it is enforced by the first component puzzle.

The constructions for some sets depend on a constant $q$, defined as follows.  First, define the \emph{interior cells} of $R_i$ to be the cells in $R_i$ having all four neighbors also in~$R_i$, and call the remaining cells of $R_i$ \emph{non-interior}.  Then $q$ is half of the maximum number of non-interior cells in any $R_i$, rounded down.

\begin{figure}
\centering
\subcaptionbox{\label{fig:cohesion-starsmulti-numbering} Numbering the non-interior cells.  The interior cells are shaded.}{%
  \includegraphics[width=.3\textwidth]{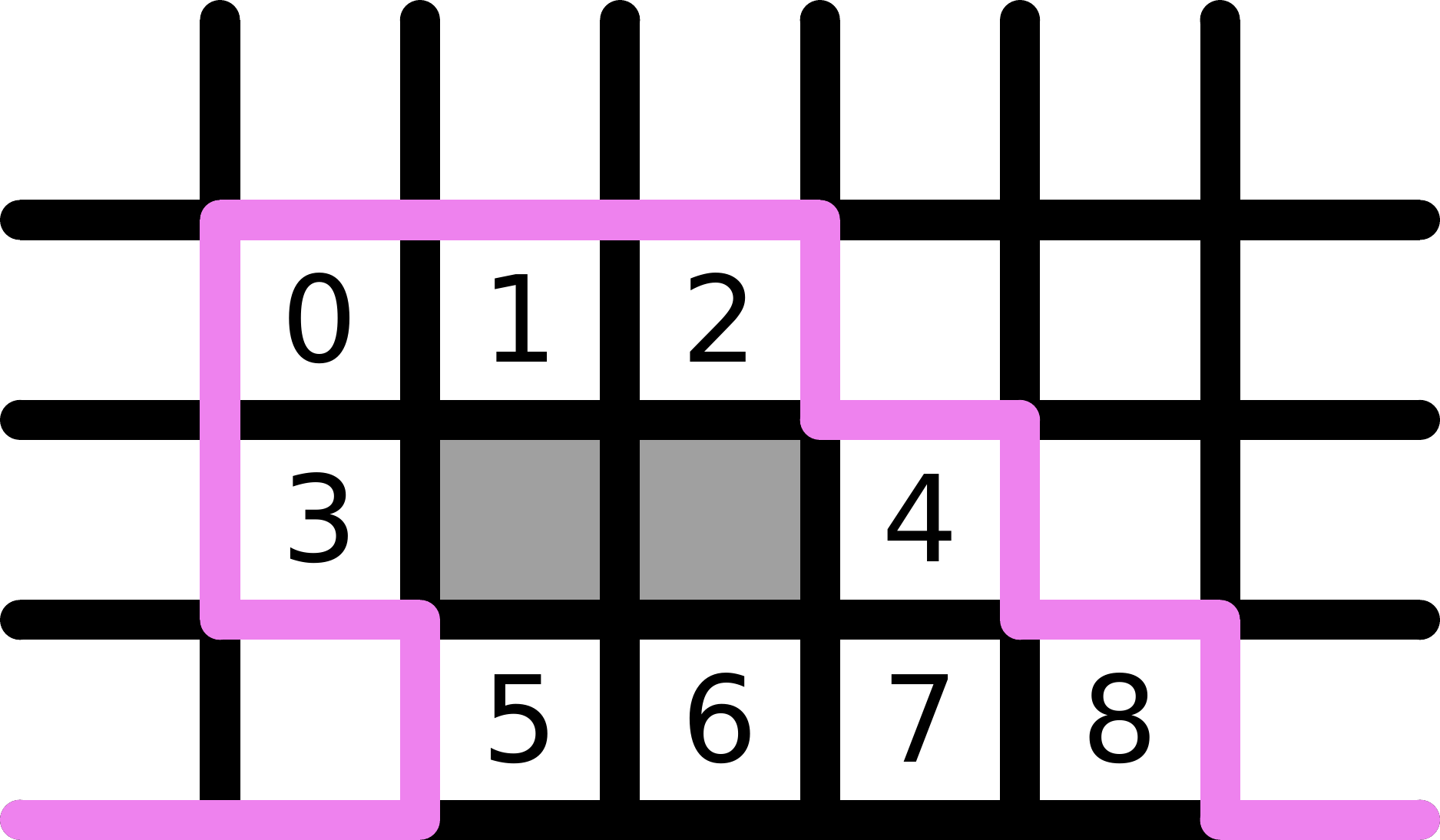}%
}~~~~
\subcaptionbox{\label{fig:cohesion-starsmulti-boards} The two component puzzles implementing cohesion.}{%
  \includegraphics[width=.3\textwidth]{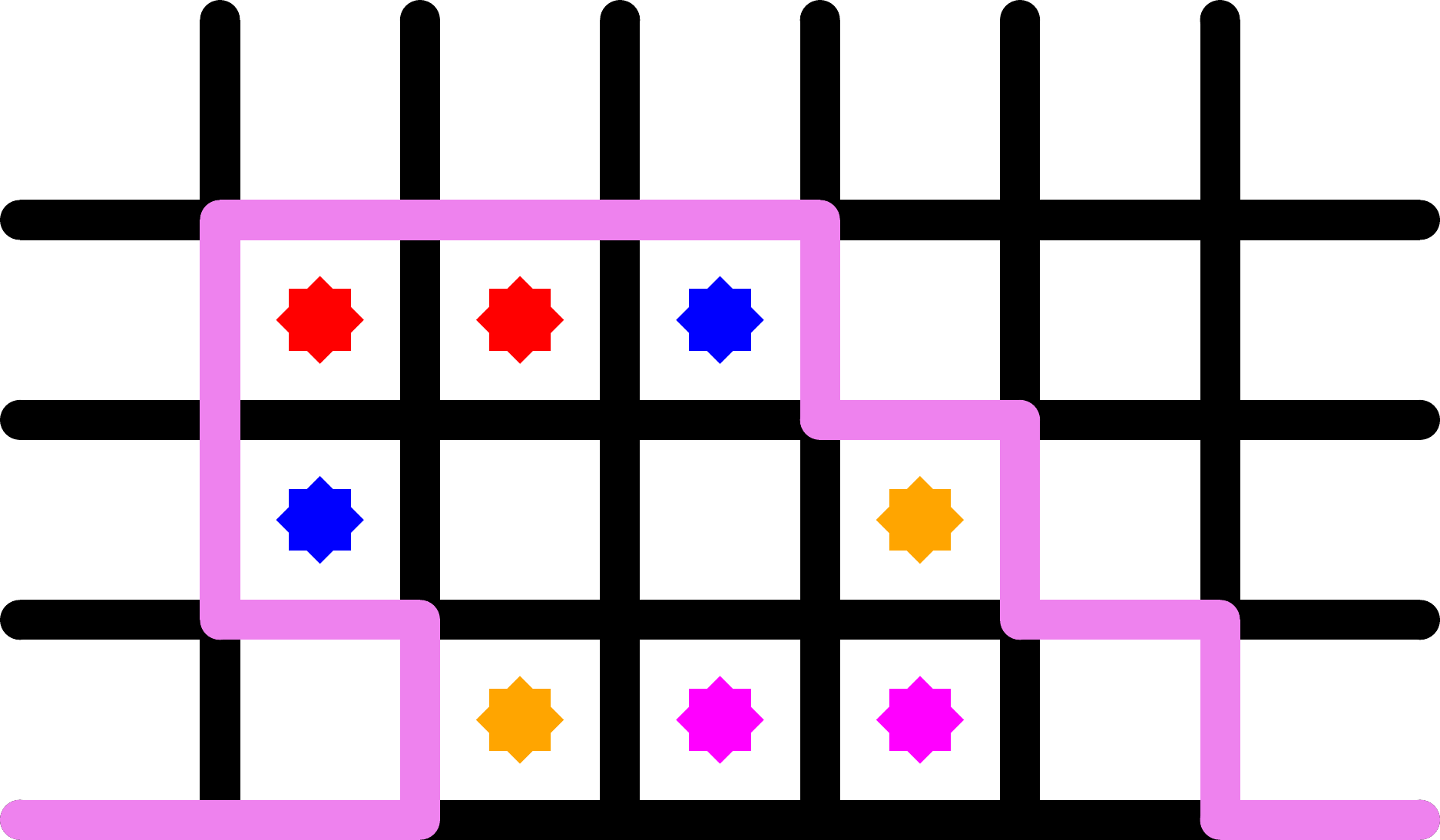}~~%
  \includegraphics[width=.3\textwidth]{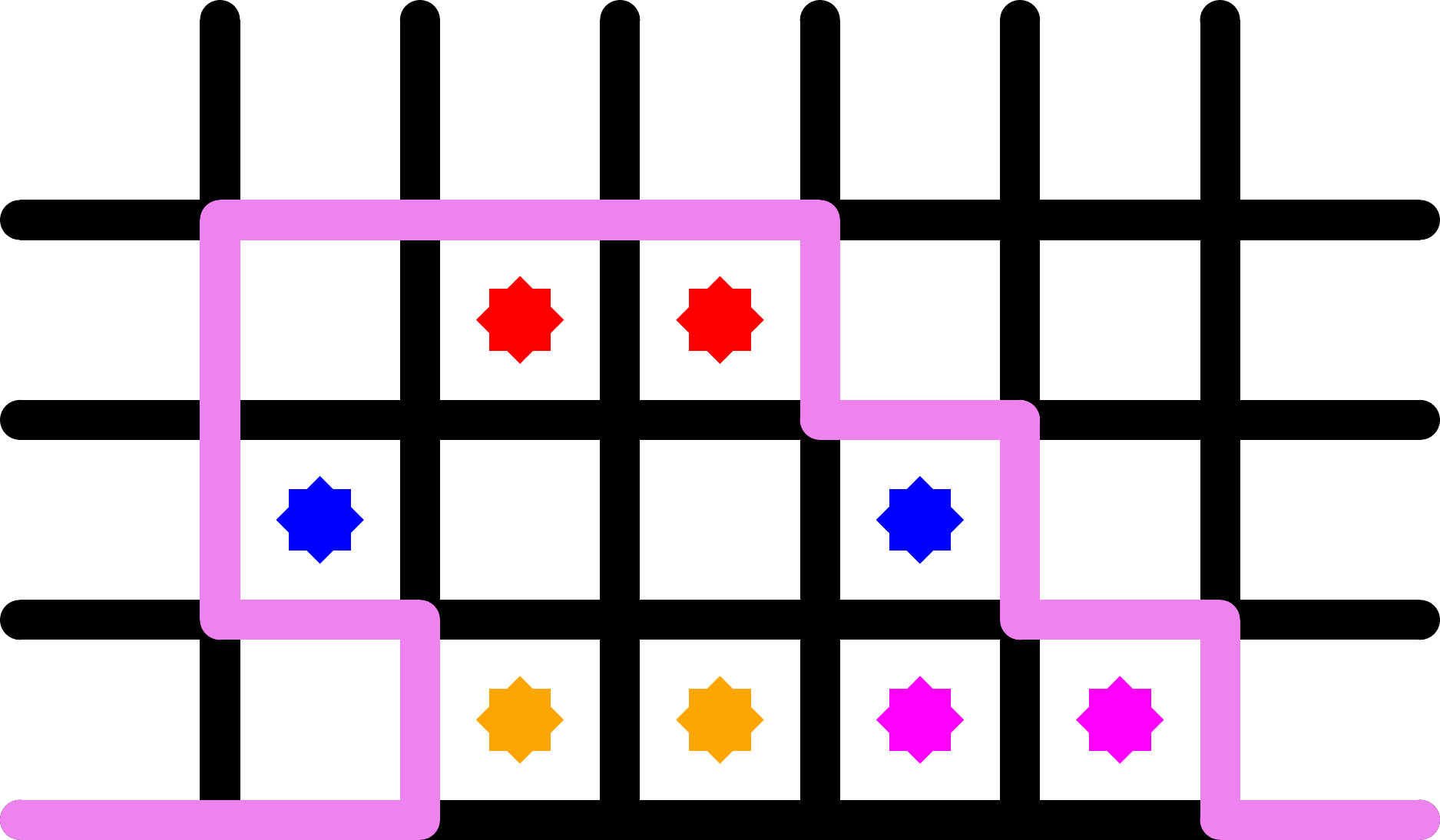}%
}
\caption{An example of implementing cohesion using stars of $q$ colors.}
\label{fig:cohesion-starsmulti}
\end{figure}

\begin{itemize}
\item \textbf{Stars of $q$ colors:} Create two component puzzles.  Arbitrarily number the non-interior cells of each $R_i$ starting at $0$.  In each $R_i$, place stars of color $c$ in the even-parity cell pairs $(2c, 2c+1)$ in the first component puzzle and in the odd-parity cell pairs $(2c+1, 2c+2)$ in the second component puzzle.  See Figure~\ref{fig:cohesion-starsmulti} for an example.

There may be up to $2r$ stars of each color in each component puzzle, but any path satisfying disjointness cannot pair any star in $R_i$ with a star in $R_j$ for $i \neq j$, so each star color is effectively unique within $R_i$.  Then cells numbered $2c$ and $2c+1$ must be in the same region to satisfy the stars in the first component puzzle, and cells numbered $2c+1$ and $2c+2$ must be in the same region to satisfy the stars in the second component puzzle.  Then by transitivity, all non-interior cells of each $R_i$ must be in the same region in any solution.  Any path under which the non-interior cells are in the same region also places the interior cells in that region because the edges separating the interior cells from the non-interior cells constitute a cycle.  Thus, all cells in each $R_i$ must be in the same region.

\item \textbf{Stars of $1$ color:} Create $q-1$ component puzzles.  As in the previous construction, number the non-interior cells of each $R_i$ starting at $0$.  In puzzle $p$, in each $R_i$ containing a $(p, p+1)$ cell pair, place a star in those cells.

In each puzzle, each $R_i$ contains only two stars, so those cells must be in the same region in any solution, and so by transitivity all non-interior cells in each $R_i$ are in the same region.  Then the interior cells are also in the region by the previous argument, so all cells in each $R_i$ must be in the same region.

\item \textbf{Nonmonomino polyominoes:} Create one component puzzle.  For each $R_i$ containing more than one cell, place a polyomino clue shaped exactly like $R_i$ in the leftmost bottommost cell in $R_i$.  A path satisfying disjointness can only satisfy these polyomino clues by refraining from further dividing the $R_i$.  (We do not place monomino clues in single-cell $R_i$ because they cannot be divided further in any case.)
% Writing note: there's nothing special about the leftmost bottommost; I just wanted to give a specific place.

\item \textbf{Monominoes and antimonominoes:} Use the construction for stars of $1$ color, except when placing pairs of stars, place a monomino and an antimonomino instead.  In each of the $q-1$ component puzzles, each $R_i$ contains only a single monomino and antimonomino, so they must be in the same region in any solution.  Then follow the same argument.

\item \textbf{Antibodies and antimonominoes:} Use the construction for stars of $1$ color, except when placing pairs of stars, place an antibody and an antimonomino instead.  In each of the $q-1$ component puzzles, each $R_i$ contains only a single antibody and antimonomino, so they must be in the same region in any solution.  Then follow the same argument.
\end{itemize}

The first set of component puzzles constrains any solution to place cells in different $R_i$ in different regions and the second set constrains any solution to place cells in the same $R_i$ in the same region.  Thus, any solution path must induce $D$, and $P$ is such a path.
\end{proof}

\subsection{Further Variants}

All of the analysis above looked for exact matches to a path or region
decomposition.  We can also consider applying a scale factor, where each
edge in the specified path is represented by $s$ consecutive edges,
or each cell in the specified region decomposition is represented by an
$s \times s$ square, in the produced puzzle.
%As many clue types are
%nonuniversal due to difficulties specifying paths on the boundary, we could
%also add an inset from the boundary.

Returning to the motivation of specifying $k > 1$ paths or region
decompositions, metapuzzles differ in the size of the sets they require.
Light-bridge metapuzzles, for example, require only a small constant number of
paths, because the player will be spending time traversing the bridge and
interacting with any panels or other objects in the void.  But the moving ramp
and log claw in the sawmill can be moved to a linear number of positions,
so we might hope for $O(n)$-path or -region universality.
%\xxx{Of course any concrete instance is constant...}

\begin{restatable}{open}{universalityclassification} \label{open:universality-classification}
%Which sets of clue types and nonclue constraints are universal?
Classify sets of clue types and/or nonclue constraints
by their degree of path and region universality, or more generally,
characterize the sets of paths and regions they can express.
%by each set of clue types and nonclue constraints.
\end{restatable}

\section{Open Problems}
\label{sec:conclusion}

We restate here the open problems defined in the previous sections, then briefly suggest additional possible lines of inquiry.

\vertexhexagons*
\star*
\lightbridge*
\puzzledesign*
\polyominoregionuniversality*
\universalityclassification*

\paragraph{Parsimony.}

Which hardness reductions can be made (or are already) parsimonious, and thus
prove \#P-hardness and ASP-hardness?
For example, Theorem~\ref{thm:triangles1} is a parsimonious reduction
from X3C, and X3C is \#P- and ASP-complete
\cite{Hunt-Marathe-Radhakrishnan-Stearns-1998}.
On the other hand, Theorem~\ref{thm:squares 2 color} is a parsimonious reduction
from \trvb, but is \trvb\ \#P/ASP-hard?

\paragraph{Fixed-parameter tractability with respect to board height.}

We conjecture that many Witness puzzles can be solved in polynomial time (XP)
when there are a constant number of rows in the puzzle, using dynamic
programming over the possible vertical cross-sections and how they are
topologically connected on either side.
In particular, we believe this to be possible for broken edges, hexagons,
triangles, squares, stars, and constant-sized polyominoes and antipolyominoes.
This leaves three clue types:
antibodies and unbounded polyominoes and antipolyominoes.
Antibodies seem difficult to handle because of the need to verify that they
are necessary.
Unbounded polyomino clues are NP-hard in $2 \times n$
puzzles by a simple reduction from 3-Partition (similar to
\cite[Theorem~3]{Ebbesen-Fischer-Witt-2011}), so are not in XP unless P = NP.

Beyond just membership in XP as above, are any of these problems
fixed-parameter tractable with respect to the number of rows?
For example, we conjecture that this is possible for broken edges, hexagons,
triangles, and squares and stars of $O(1)$ colors.

%  Witness puzzles with broken edges, hexagons, triangles, squares, and stars
%  can be solved in polynomial time for a constant number of rows in the grid
%  (membership in XP).  If the number of squares and stars is $O(1)$,
%  then we get an FPT algorithm.

\section*{Acknowledgments}

We thank Jason Ku and Quanquan Liu for helpful discussions related to this paper.
Most figures were produced using SVG Tiler
(\url{https://github.com/edemaine/svgtiler}),
based on SVG source code from The Windmill
(\url{https://github.com/thefifthmatt/windmill-client/}),
a clone of Witness-style puzzles distributed under Apache License 2.0.

%\xxx[Quanquan]{Please add me back to the author's list only if vertex hexagons is solved AND it uses the corner to corner result substantially.}

% bibliography
\bibliographystyle{alpha}
\bibliography{Witness}

%\newpage

\appendix\magicappendix

\section{Adversarial-Boundary Edge-Matching Problem is $\Sigma_2$-complete}
\label{appendix Adversarial-Boundary Edge-Matching Problem Sigma_2}

\begin{proof}[Proof of Lemma~\ref{Adversarial-Boundary Edge-Matching Problem Sigma_2}]
  We reduce from from QSAT$_2$,
  which is the $\Sigma_2$-complete problem of deciding
  a Boolean statement of the form
  $\exists x_1 : \exists x_2 : \cdots : \exists x_n :
   \forall y_1 : \forall y_2 : \cdots : \forall y_n :
   f(x_1, x_2, \dots, x_n; y_1, y_2, \dots, y_n)$
  where $f$ is a Boolean formula using \textsc{and} ($\wedge$),
  \textsc{or} ($\vee$), and/or \textsc{not} ($\neg$).  We convert this formula into a circuit, lay out the circuit on a square grid, and implement each circuit element as a set of tiles.  The first player's boundary-edge swaps encode a setting of true or false for the first player's variables.  Then, as part of solving the edge-matching problem, the second player must exhibit a setting of their variables that makes the formula false; otherwise the first player wins.

\paragraph{Conversion to circuit.} We convert the formula into a Boolean circuit with \textsc{and},
  \textsc{or}, and \textsc{not} gates for the Boolean operations, plus
  $2n$ \textsc{input} gates to represent the $x_i$s and $y_i$s,
  binary \textsc{fan-out} gadgets to make copies of these values,
  and one \textsc{output} gate to represent the formula output
  (which must be \textsc{true}).

\paragraph{Circuit layout.} We construct a planar layout of this circuit on an integer grid.  Each cell in the grid views itself as being ``connected'' to zero, one, two, or three of its edge neighbors:
\begin{itemize}
\item Each \textsc{and} and \textsc{or} gate maps to a three-neighbor cell (two neighbors for the inputs, one neighbor for the output).
\item Each \textsc{fan-out} gate maps to a three-neighbor cell (one neighbor for the input, two neighbors for the outputs).
\item Each \textsc{not} gate maps to a two-neighbor cell (one neighbor for the input, one neighbor for the negated output).
\item Each \textsc{input} gate maps to a one-neighbor cell (one neighbor for the output).
\item Each \textsc{output} gate maps to a one-neighbor cell (one neighbor for the input).
\item Each wire between gates maps to a (possibly empty) path of two-neighbor \textsc{wire} cells (one neighbor for the input, one neighbor for the copied output).
\item All other cells are zero-neighbor \textsc{null} cells.
\end{itemize}

  We construct the layout as follows.  Let $g$ be the number of
  (\textsc{and}, \textsc{or}, \textsc{fan-out}, \textsc{not}, \textsc{input},
  and \textsc{output}) gates.
  Because the graph has maximum degree $3$, we can find an
  orthogonal layout in a $({g \over 2}+1) \times {g \over 2}$ grid of cells
  \cite{Papakostas-Tollis-2005}.

  We want each $x_i$ input to appear in the top row of the grid, in an even column, with its output on the south side, with all other cells in the top row (including all odd columns) being \textsc{null} cells.  To guarantee this,
  we add a top row, move each $x_i$ \textsc{input} gadget
  to an even column in this top row, and connect the old location of that gadget
  to the new location via an orthogonally routed edge
  by adding $O(1)$ new empty rows and columns.

  Then we replace each orthogonally routed edge by a path of \textsc{wire} cells.
  Additionally, we scale the coordinates by a constant factor,
  and replace each crossover between pairs of edges
  with a crossover gadget built out of \textsc{and} and \textsc{not} gates,
  implementing a \textsc{nand} crossover which can be built from the standard
  3-\textsc{xor} crossover \cite{Goldschlager-1977}.

\paragraph{Constructing the board.} The \abem{} board is divided into two regions, an upper \emph{circuit} region and a lower \emph{garbage} region.  We emit a row of tiles called the \emph{wasteline} that is forced to separate the regions (explained later).  The circuit region is tileable only if the circuit is satisfied.  The garbage region is always tileable; it provides space for tiles that weren't used in the circuit region to be placed.

\paragraph{Circuit region.} In the circuit region, each cell in the circuit layout is mapped to a set of tiles, one of which will be used in that position in the solution.  The other tile choices for that cell will be garbage collected in the garbage region.  The first player's boundary-edge swaps set the $x_i$ inputs to the circuit.  The second player will then try to solve the resulting problem by choosing a value for the $y_i$ inputs, then evaluating the circuit; all other tile choices are forced by the inputs.

In the circuit region of the board, all tile edges encode a circuit signal from the set $\{\text{true}, \text{false}, \text{null}\}$.  To preserve the circuit layout in all legal placements of the tiles, the top and bottom edges of each tile in the circuit region also encode their position.  The tile at position $(x, y)$ (where the origin is the upper-left tile) encodes $(x, y)$ on its top edge and $(x, y+1)$ on its bottom edge.  In this way, we can associate each circuit cell with the set of tiles that may appear in the corresponding slot in the tiling problem.

The set of tiles produced for each cell depends on its type.  Unless explicitly mentioned otherwise, the position encoding on the top and bottom edges is as above.

\begin{itemize}
\item \textsc{and}, \textsc{or}: four tiles encoding the truth-table operation on the edges connecting to other cells.  For example, an \textsc{and} cell with inputs on its top and left edges and output on its right edge would produce four tiles with $(F, F, F)$, $(F, T, F)$, $(T, F, F)$ and $(T, T, T)$ as the circuit signals on its top, left and right edges respectively, with null on its bottom edge.
\item \textsc{not}, \textsc{wire}, \textsc{fanout}: two tiles encoding their truth table on the edges connecting to other cells and null on the other edges.
\item \textsc{output}: one tile with the false circuit signal on its connected edge and null on its other edges.  Because there is no tile for this cell with true on its connected edge, the tiling can only be completed when the circuit evaluates to false.
\item \textsc{input} representing $x_i$: two tiles that propagate the circuit signal selected by the first player's top-boundary swap.  By the layout process above, these cells are always in the top row and an even column.  These tiles encode their column pair number instead of their position on their top edge, along with a circuit signal (so either $(\frac{x}{2}, T)$ or $(\frac{x}{2}, F)$).  The bottom edge uses the usual position encoding and propagates the circuit signal (so $((x, 1), T)$ or $((x, 1), F)$, respectively).  Their left and right edges always encode the null circuit signal.
\item \textsc{input} representing $y_i$: two tiles, one with true and the other with false on the connecting edge, and null on the other edges.
\item \textsc{null}: one tile with the null circuit signal on all edges.  For cells in the top row, the top color encodes $\lfloor \frac{x}{2} \rfloor$ instead of the usual $(x, 0)$ position, so that the first player's swapping does not affect the tile's placeability.  Cells in the top row immediately to the right of an \textsc{input} cell representing $x_i$ instead have two tiles in their set, encoding true and false circuit signals on their top edge (and null on their other edges), to terminate whichever signal the first player did not assign to $x_i$.
\end{itemize}

\def\tilescale{0.68}
\begin{figure}
\centering
\subcaptionbox{\textsc{and} with inputs from up and left and output down.}{%
  \includegraphics[scale=\tilescale]{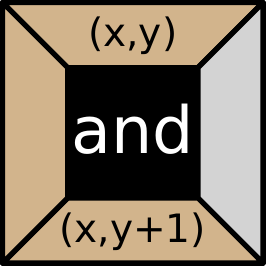}
  \includegraphics[scale=\tilescale]{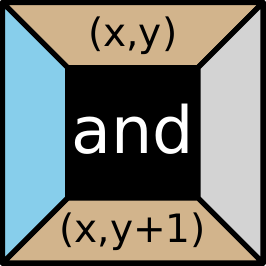}
  \includegraphics[scale=\tilescale]{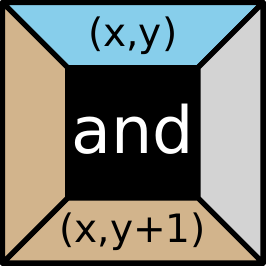}
  \includegraphics[scale=\tilescale]{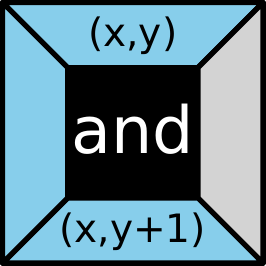}
}\hfil
\subcaptionbox{\label{or-tiles} \textsc{or} with inputs from up and left and output down.}{%
  \includegraphics[scale=\tilescale]{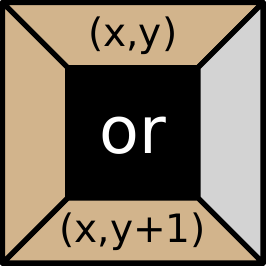}
  \includegraphics[scale=\tilescale]{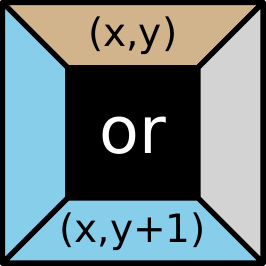}
  \includegraphics[scale=\tilescale]{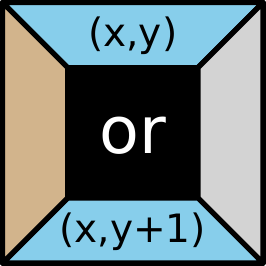}
  \includegraphics[scale=\tilescale]{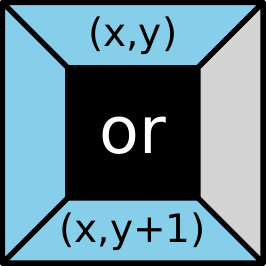}
}

\subcaptionbox{Horizontal \textsc{not}.}{
  \includegraphics[scale=\tilescale]{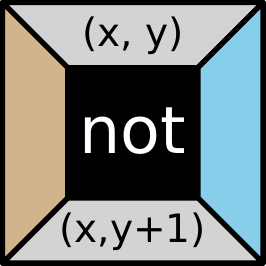}
  \includegraphics[scale=\tilescale]{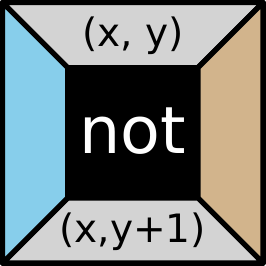}
}\hfil
\subcaptionbox{Horizontal \textsc{wire}.}{
  \includegraphics[scale=\tilescale]{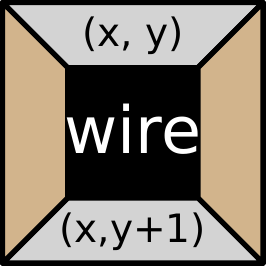}
  \includegraphics[scale=\tilescale]{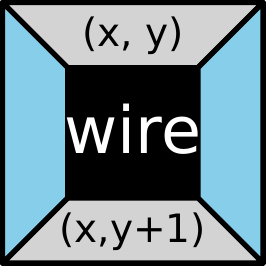}
}\hfil
\subcaptionbox{\textsc{fanout}.}{
  \includegraphics[scale=\tilescale]{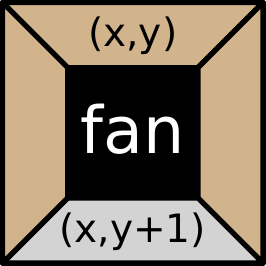}
  \includegraphics[scale=\tilescale]{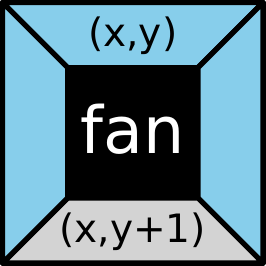}
}\hfil
\subcaptionbox{\textsc{null}.}{
  \includegraphics[scale=\tilescale]{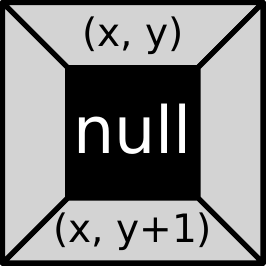}
}

\subcaptionbox{\textsc{input} representing $x_i$.}{
  \includegraphics[scale=\tilescale]{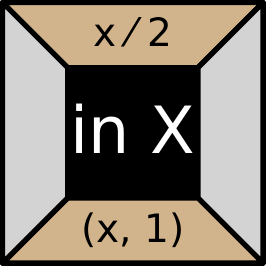}
  \includegraphics[scale=\tilescale]{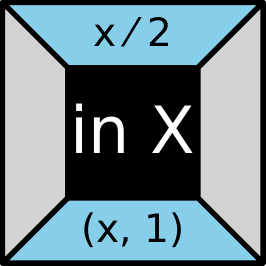}
}\hfil
\subcaptionbox{Special \textsc{null}s to the right of \textsc{input}s representing $x_i$.}{
  \includegraphics[scale=\tilescale]{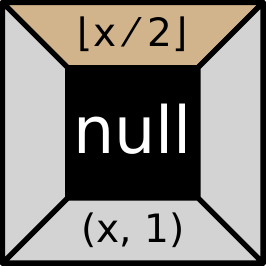}
  \includegraphics[scale=\tilescale]{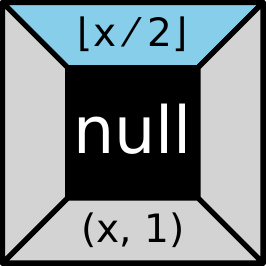}
}\hfil
\subcaptionbox{\textsc{input} representing $y_i$, output on the right.}{
  \includegraphics[scale=\tilescale]{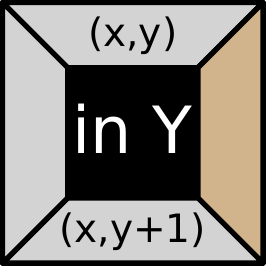}
  \includegraphics[scale=\tilescale]{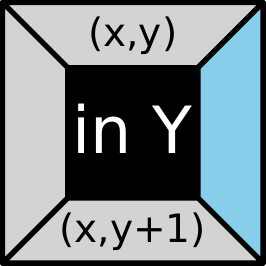}
}\hfil
\subcaptionbox{\textsc{output}, input from up.}{
  \includegraphics[scale=\tilescale]{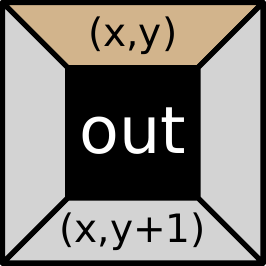}
}
\caption{Examples of tile sets produced for each type of cell in the circuit region.  $x$ and $y$ represent the position of the cell and tan, blue and gray represent the false, true and null circuit signals respectively.  Except for the \textsc{input}s representing $x_i$ and their special \textsc{null} neighbors, each tile type can be instantiated with connecting edges in any direction.}
\end{figure}

\paragraph{Wasteline.} The wasteline is a row of tiles that transitions from the circuit region to the garbage region.  The top of each tile encodes the null circuit signal and the position of the cell in the same way as in the circuit region.  The left and right edges in each tile are unique colors encoding the column of the tile, forcing the wasteline to occur in order as a contiguous row in the tiling.  The bottom of each tile is the garbage color~$g$.

\paragraph{Garbage region.} The garbage region provides space for the unused tiles from each set to be placed.  For each tile $t$ in the circuit region (even if it's the only tile in its set), we emit four garbage tiles, each having one of $t$'s edge colors on one side and $g$ on the other three sides, and four tiles having $g$ on all four edges, allowing that tile to be encapsulated into a $3 \times 3$ block having $g$ on its exterior, as in Figure~\ref{gc-3x3}.

\begin{figure}
\centering
\subcaptionbox{\label{gc-3x3} One of the three $3 \times 3$ garbage blocks encapsulating the unused tiles.}{
  $\vcenter{\hbox{\includegraphics[scale=\tilescale]{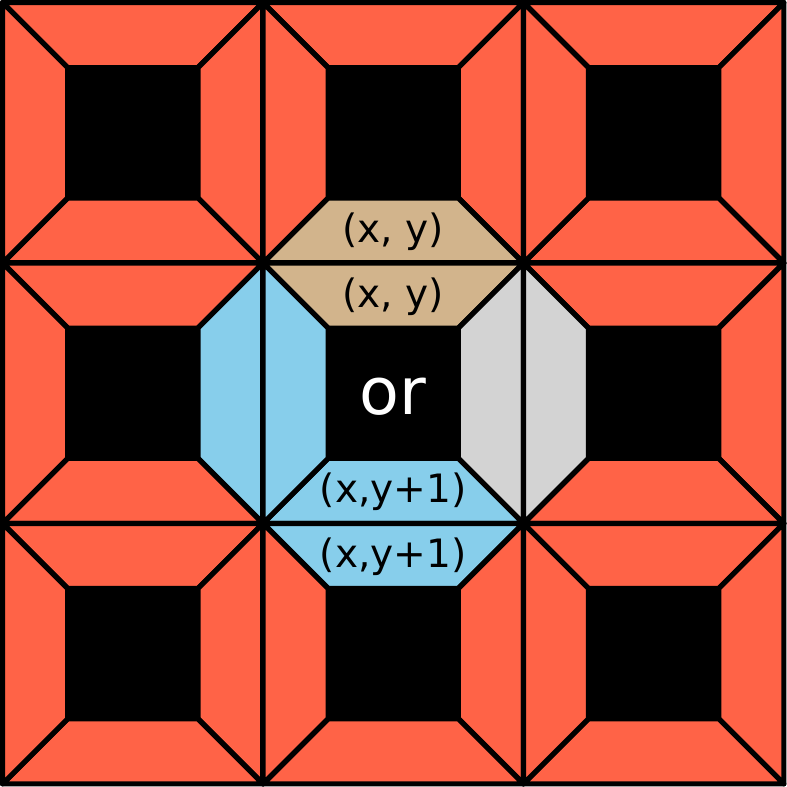}}}$
}\hfil
\subcaptionbox{\label{gc-2x2} The garbage dominoes collecting the garbage collection tiles that would have been used for the both-inputs-true \textsc{or} tile had it not been placed in the circuit region.}{
  $
  \vcenter{\hbox{\includegraphics[scale=\tilescale]{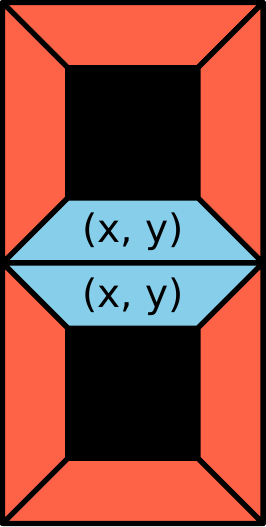}}}
  ~
  \vcenter{\hbox{\includegraphics[scale=\tilescale]{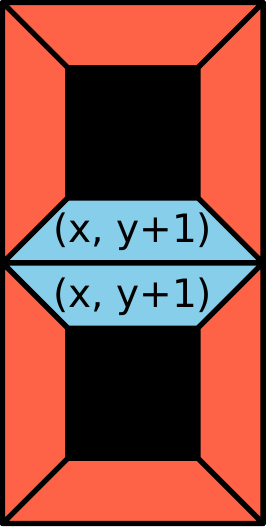}}}
  ~
  \vcenter{\hbox{\includegraphics[scale=\tilescale]{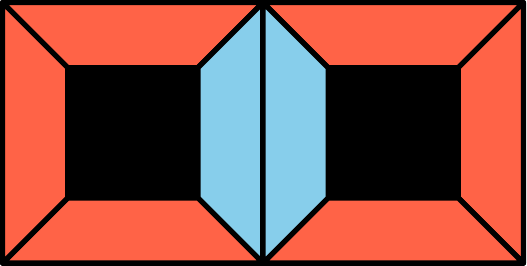}}%
    \vskip 0.5ex
    \hbox{\includegraphics[scale=\tilescale]{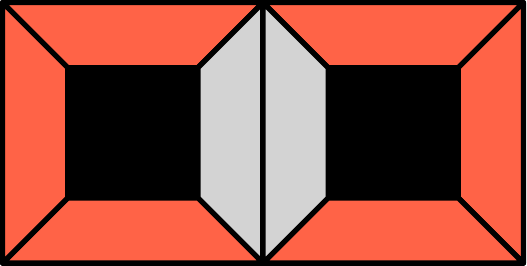}}}
  $
}
\caption{Collecting garbage for the \textsc{or} cell with tile set shown in Figure~\ref{or-tiles}, assuming the rightmost (both inputs true) tile was placed in the circuit region.}
\end{figure}

Consider the boundary edges of the circuit region (that is, the entire top board boundary, parts of the left and right board boundary, and the top edges of the wasteline).  In a valid tiling, these edges will always match a circuit tile edge, so for each of those edges we remove a corresponding garbage collection tile (with that boundary color and three $g$ edges) that we just produced.

We then need to provide a way to garbage collect the garbage collection tiles for the used tiles.  Each nonboundary edge of each used tile corresponds to an edge of a neighboring used tile, so we can pair the corresponding garbage collection tiles to create a vertical or horizontal domino having $g$ on its exterior, as in Figure~\ref{gc-2x2}.

Finally, we need to provide exactly enough tiles having $g$ on all four sides to fill the garbage region.  If the circuit region is $w$ tiles wide and $h_c$ tiles tall (assuming $w \geq 3$) and we produced $t$ total circuit tiles, we have $t - wh_c$ $3 \times 3$ blocks, $2wh_c - 2w$ vertical dominoes and $2wh_c - 2h_c$ horizontal dominoes to pack.  Using a simple strategy of packing the $3 \times 3$ blocks horizontally in as many rows as necessary, followed by the vertical dominoes and horizontal dominoes, each type starting on a new row, the garbage region has height $h_g = 3\lceil \frac{3(t-wh_c)}{w} \rceil + 4(h_c - 1) + \lceil \frac{4h_c(w-1))}{w}\rceil$.  We provide $wh_g - 9(t - wh_c) - 2(4wh_c - 2w - 2h_c)$ tiles with $g$ on all four sides to pack the remaining area in the garbage region.

\paragraph{Boundary-colored board.} The board's width is determined by the width of the circuit region, which is in turn determined by the circuit layout.  The height is the circuit region height, plus 1 for the wasteline, plus the garbage region height (i.e., $h_c + 1 + h_g$).  The top edge of the board is colored $(\lfloor \frac{x}{2} \rfloor, T)$ and $(\lfloor \frac{x}{2} \rfloor, F)$ above the \textsc{input} cells corresponding to $x_i$ variables and the \textsc{null} cells immediately to their right, and $(\lfloor \frac{x}{2} \rfloor, N)$ elsewhere.  The left and right edges of the board above the wasteline have the null circuit signal color.  The left and right edges at the wasteline have the unique colors on the left (right) edge of the leftmost (rightmost) wasteline tile.  Below the wasteline, the left, right and bottom edges are all colored $g$.

\paragraph{Argument.} The first player swaps pairs of top-boundary edges, then the second player attempts to tile the resulting board.

Garbage tiles cannot be placed in the circuit region.  Because each garbage tile has $g$ on three or four of its edges, there is no placement of garbage tiles with non-$g$ colors on its entire exterior.  Thus, if any garbage tiles are placed in the circuit region, the entire region must be filled with garbage tiles, but this is impossible because the corners of the boundary of the circuit region have two non-$g$ edges.

The wasteline's placement is fixed by the unique colors on the left and right boundary.  The wasteline tiles encode their position on their top edge, so this also fixes all the tiles above in place; each tile in the circuit region must be from the set created for the corresponding circuit cell.  Then the tiles corresponding to the $x_i$ \textsc{input} cells must correctly conduct the circuit signal (true or false) from the boundary.  (Swaps of the boundary above the \textsc{null} cells in the top row change nothing because those edges have the same color.)

If the second player knows an assignment to the $y_i$ variables that makes the circuit evaluate to false, they can select the tiles conducting that circuit signal from the sets corresponding to their \textsc{input} cells.  After that, the tiling in the circuit region is forced; from the sets of tiles corresponding to gate and wire cells, only one tile can be placed (based on the inputs from its connecting cells), and its output connection forces further cells.  If the circuit indeed evaluates to false, then the tiling can be completed at the \textsc{output} cell, whose only tile requires its input to be false.  In that case, the remaining tiles can be garbage collected (by the construction of the garbage region given above), and the second player wins the game.

Otherwise, if the circuit evaluates to true under all $y_i$ variable settings, there is no way to place a tile that connects to the \textsc{output} cell's tile, and the first player wins the game.

\paragraph{Losing strategy guarantee.} We can easily adapt this proof to guarantee that the first player always has a \emph{losing} strategy in the produced \abem{} instance (regardless of whether they have a winning strategy).  Simply amend the QSAT$_2$ instance by adding an additional existential variable $x_{n+1}$ which is conjoined with the formula ($f(\dots) \land x_{n+1}$).  Performing the boundary-edge swap corresponding to setting $x_{n+1}$ to \textsc{false} is a losing strategy.

\paragraph{Tile rotation.} The reduction above remains valid when tile rotation is allowed.  The wasteline's position and orientation are fixed by the unique boundary colors, then the position encoding on the top boundary and top edge of the wasteline prevents the circuit tiles from being rotated.  Allowing the garbage collection tiles to rotate is harmless because they still cannot be placed in the circuit region.

\paragraph{Signed colors.} The reduction above produces an unsigned edge-matching instance, but every unsigned edge-matching puzzle can be converted into a signed edge-matching puzzle with a constant-factor increase in size and the number of colors \cite[Theorem 4]{Jigsaw_GC}.
\end{proof}

\end{document}